\RequirePackage{silence}
\documentclass[phd,lfcs,twoside,logo,centrechapter,abbrevs,11pt,listsintoc,parskip,notimes,normalheadings]{infthesis}

\usepackage{titlesec}
\titleformat{\chapter}[display]{\fontsize{18pt}{21pt}\bfseries \sffamily \centering}{\chaptertitlename\ \thechapter}{0pt}{\ifnum\thechapter>0 {\includegraphics[height=0.05\textwidth]{figures/knot\thechapter.png}\\} \fi}{}

\usepackage[T1]{fontenc}    
\usepackage{hyperref}       
\usepackage{url}            
\usepackage{booktabs}       
\usepackage{amsfonts}       
\usepackage{nicefrac}       
\usepackage{microtype}      
\usepackage{xcolor}         
\usepackage{amsmath}
\usepackage{amssymb}
\usepackage{amsthm}
\usepackage[framemethod=TikZ]{mdframed}
\usepackage{thmtools}
\usepackage{bm}
\usepackage[ruled,vlined,linesnumbered,algochapter]{algorithm2e}
\usepackage{caption}
\usepackage{subcaption}
\usepackage{algorithmic}
\usepackage{wrapfig}
\usepackage{multicol,multirow}
\usepackage{afterpage}
\usepackage{graphicx}
\usepackage{wrapfig}
\usepackage{mathtools}
\usepackage{csquotes}
\usepackage{epsfig}
\usepackage{placeins}
\usepackage[export]{adjustbox}
\usepackage{tikzit}

\tikzstyle{basic}=[fill=white, draw=black, shape=circle]
\tikzstyle{square}=[fill=white, draw=black, shape=rectangle]
\tikzstyle{big dashed}=[fill=white, draw=black, shape=circle, minimum width=1cm, dashed]
\tikzstyle{vertical ellipse dashed}=[fill=none, draw=blue, minimum width=0.75cm, minimum height=3cm, ellipse, dashed, tikzit shape=rectangle, tikzit draw=blue, tikzit fill=white, line width=1pt]
\tikzstyle{small vertical ellipse dashed}=[fill=none, draw=blue, shape=circle, tikzit fill=white, tikzit draw=blue, dashed, minimum width=0.75cm, minimum height=1.5cm, tikzit shape=rectangle, ellipse, line width=1pt]
\tikzstyle{tiny vertical ellipse dashed}=[fill=none, draw=blue, shape=circle, tikzit fill=white, ellipse, dashed, minimum width=0.75cm, minimum height=1cm, tikzit shape=rectangle, line width=1pt]
\tikzstyle{red}=[fill=red, draw=black, shape=circle]
\tikzstyle{green}=[fill={rgb,255: red,0; green,128; blue,128}, draw=black, shape=circle]
\tikzstyle{blue}=[fill=blue, draw=black, shape=circle]
\tikzstyle{huge dashed}=[fill=white, draw=black, shape=circle, dashed, minimum width=2cm]
\tikzstyle{medium}=[fill=white, draw=black, shape=circle, minimum width=1cm]
\tikzstyle{pale green}=[fill={rgb,255: red,173; green,231; blue,0}, draw=black, shape=circle, minimum width=1cm]
\tikzstyle{horizontal ellipse dashed}=[fill=white, draw=black, tikzit draw=magenta, tikzit shape=rectangle, minimum width=3cm, minimum height=0.75cm, ellipse, dashed]
\tikzstyle{minsize}=[fill=white, draw=black, shape=circle, minimum width=0.75cm]
\tikzstyle{horizontal ellipse green}=[fill={rgb,255: red,191; green,255; blue,0}, draw=black, tikzit draw={rgb,255: red,191; green,255; blue,0}, tikzit shape=rectangle, minimum width=3cm, minimum height=0.75cm, ellipse, dashed]
\tikzstyle{horizontal ellipse blue}=[fill={rgb,255: red,107; green,203; blue,255}, draw=black, tikzit draw=blue, tikzit shape=rectangle, minimum width=3cm, minimum height=0.75cm, ellipse, dashed]
\tikzstyle{smallblack}=[fill=black, draw=black, shape=circle, inner sep=0 pt, minimum size=3 pt]
\tikzstyle{smallSquare}=[fill=white, draw=black, shape=rectangle, inner sep=0 pt, minimum size=6 pt]
\tikzstyle{smallCircle}=[fill=white, draw=black, shape=circle, inner sep=0 pt, minimum size=20 pt]
\tikzstyle{big vertical ellipse dashed}=[fill=none, draw=blue, shape=circle, tikzit shape=rectangle, ellipse, dashed, minimum width=0.95cm, minimum height=3.7cm, line width=1pt]
\tikzstyle{smallred}=[fill=red, draw=red, shape=circle, inner sep=0 pt, minimum size=3 pt]
\tikzstyle{smallblue}=[fill=blue, draw=blue, shape=circle, inner sep=0pt, minimum size=3pt]
\tikzstyle{small green}=[fill={rgb,255: red,0; green,107; blue,61}, draw={rgb,255: red,0; green,107; blue,61}, shape=circle, inner sep=0pt, minimum size=3pt]
\tikzstyle{med red}=[fill=red, draw=red, shape=circle, inner sep=0pt, minimum size=5pt]
\tikzstyle{med blue}=[fill=blue, draw=blue, shape=circle, inner sep=0pt, minimum size=5pt]
\tikzstyle{med green}=[fill={rgb,255: red,0; green,107; blue,61}, draw={rgb,255: red,0; green,107; blue,61}, shape=circle, inner sep=0pt, minimum size=5pt]
\tikzstyle{med black}=[fill=black, draw=black, shape=circle, inner sep=0pt, minimum size=5pt]

\tikzstyle{directed}=[->, line width=1pt]
\tikzstyle{undirected}=[-, line width=1pt]
\tikzstyle{directed red}=[draw=red, ->, line width=1pt]
\tikzstyle{directed green}=[draw={rgb,255: red,0; green,128; blue,128}, ->, line width=1pt]
\tikzstyle{directed blue}=[draw=blue, ->, line width=1pt]
\tikzstyle{directed purple}=[draw={rgb,255: red,128; green,0; blue,128}, ->, line width=1pt]
\tikzstyle{undirected red}=[-, draw=red, line width=1pt]
\tikzstyle{undirected green}=[-, draw={rgb,255: red,0; green,107; blue,61}, line width=1pt]
\tikzstyle{undirected blue}=[-, draw=blue, line width=1pt]
\tikzstyle{undirected purple}=[-, draw={rgb,255: red,128; green,0; blue,128}, line width=1pt]
\tikzstyle{undirected dashed}=[-, line width=1pt, dashed]
\tikzstyle{orange dashed}=[-, draw={rgb,255: red,255; green,128; blue,0}, dashed, line width=1.5pt]
\tikzstyle{directed dash}=[->, dashed]
\tikzstyle{blue dashed}=[-, draw=blue, dashed, line width=1pt]
\tikzstyle{green dashed}=[-, draw={rgb,255: red,0; green,162; blue,0}, dashed, line width=1pt]
\tikzstyle{blue filled}=[-, fill={blue!20}, draw=blue, line width=1pt, opacity=0.5, tikzit fill=white]
\tikzstyle{red filled}=[-, fill={red!20}, line width=1pt, draw=red, opacity=0.5, tikzit fill=white]
\tikzstyle{green filled}=[-, line width=1pt, draw={rgb,255: red,0; green,107; blue,61}, opacity=0.5, tikzit fill=white, fill={rgb,255: red,149; green,255; blue,179}]
\tikzstyle{orange filled}=[-, fill={orange!20}, draw=orange, line width=1pt, opacity=0.5, tikzit fill=white]
\tikzstyle{undirected dashed}=[-, draw=black, dashed, line width=1pt]
\tikzstyle{thick}=[-, line width=3pt]
\tikzstyle{red dashed}=[-, line width=1pt, dashed, draw=red]
\tikzstyle{bluefill}=[-, fill={rgb,255: red,179; green,179; blue,255}, line width=1pt]
\tikzstyle{redfill}=[-, fill={rgb,255: red,255; green,156; blue,158}, line width=1pt]
\tikzstyle{greenfill}=[-, fill={rgb,255: red,174; green,255; blue,177}, line width=1pt]
\tikzstyle{whitefill}=[-, fill=white, draw=black, line width=1pt]
\tikzstyle{white fill}=[-, draw=none, fill=white, tikzit draw={rgb,255: red,24; green,255; blue,12}]
\tikzstyle{red dashed}=[-, draw={rgb,255: red,255; green,0; blue,4}, line width=1 pt, dashed]

\newtheorem{lemma}{Lemma} [chapter]

\newtheorem{definition}{Definition} [chapter]
\newtheorem{corollary}{Corollary} [chapter]
\newtheorem{proposition}{Proposition} [chapter]
\newtheorem{example}{Example} [chapter]
\newtheorem{openquestion}{Open Question}
\newtheorem{mainresult}{Main Result}

\newcommand{\transpose}{\intercal}                    
\newcommand{\conjtranspose}{\dagger}                  
\newcommand{\inner}[2]{\left\langle #1 , #2 \right\rangle} 
\newcommand{\norm}[1]{\left\| #1\right\|}                  
\newcommand{\abs}[1]{\left\lvert#1\right\rvert}
\DeclareMathOperator*{\argmin}{arg\,min}        
\DeclareMathOperator*{\argmax}{arg\,max}        
\newcommand{\tr}{\mathrm{tr}}                   

\newcommand{\supp}{\mathrm{supp}}
\newcommand{\sign}{\mathrm{sign}}
\newcommand{\poly}{\mathrm{poly}}
\renewcommand{\P}{\mathsf{P}}
\newcommand{\NP}{\mathsf{NP}}

\newcommand{\vol}{\mathrm{vol}}

\newcommand{\bigo}[1]{O\!\left(#1\right)}
\newcommand{\bigomega}[1]{\Omega\!\left(#1\right)}
\newcommand{\bigtheta}[1]{\Theta\!\left(#1\right)}
\newcommand{\bigotilde}[1]{\widetilde{O}\!\left(#1\right)}
\newcommand{\polylog}{\mathrm{polylog}}

\newcommand{\E}{\mathbb{E}}
\newcommand{\p}{\mathbb{P}}

\newcommand{\R}{\mathbb{R}}
\newcommand{\C}{\mathbb{C}}
\newcommand{\Z}{\mathbb{Z}}
\newcommand{\N}{\mathbb{N}}
\newcommand{\union}{\cup}
\newcommand{\intersect}{\cap}
\newcommand{\cardinality}[1]{\abs{#1}}
\newcommand{\setcomplement}[1]{\overline{#1}}

\newcommand{\symmetricdiff}{\triangle}

\newcommand{\lur}{\setl \union \setr}   
\newcommand{\sweepset}[2]{\sets_{#2}^{#1}} 
\newcommand{\pjsweep}{\sweepset{\vecp}{j}} 
\newcommand{\lscurve}[2][\vecp]{#1\!\left[#2\right]} 
\newcommand{\geqve}{\graphg = (\vertexset, \edgeset)} 
\newcommand{\geqvew}{\graphg = (\vertexset, \edgeset, \weight)} 
\newcommand{\geqvewg}{\graphg = (\vertexsetg, \edgesetg, \weight_\graphg)} 
\newcommand{\In}{\mathrm{in}} 
\newcommand{\Out}{\mathrm{out}}
\newcommand{\LS}{Lov\'asz-Simonovits} 
\newcommand{\rank}{\mathrm{rank}}                   
\newcommand{\esc}{\mathrm{esc}}
\newcommand{\dfdt}{\frac{\mathrm{d} \vecf_t}{\mathrm{d} t}}
\newcommand{\peye}{\vecp^{(i)}}
\newcommand{\pj}{\vecp^{(j)}}

\newcommand{\APT}{\textsf{APT}}

\newcommand{\lambdaeigs}{\lambda_1 \leq \lambda_2 \leq \ldots \leq \lambda_n}

\newcommand{\barx}{\Bar{\vecx}}
\newcommand{\barxi}{\barx^{(i)}}
\newcommand{\barxj}{\barx^{(j)}}
\newcommand{\barg}{\Bar{\vecg}}
\newcommand{\hatg}{\widehat{\vecg}}
\newcommand{\hatf}{\widehat{\vecf}}
 
\newcommand{\algsweepset}{\textsc{SweepSet}}
\newcommand{\algspectralcluster}{\textsc{SpectralCluster}}
\newcommand{\alggeneratesample}{\textsc{GenerateSample}}
\newcommand{\alglocbipartdc}{\textsc{LocBipartDC}}
\newcommand{\algaprdc}{\textsc{ApproximatePagerankDC}}
\newcommand{\algaprdcshort}{\textsc{APRDC}}
\newcommand{\algdcpush}{\textsc{DCPush}}
\newcommand{\algevocutdirected}{\textsc{EvoCutDirected}}

\newcommand{\algcomputechangerate}{\textsc{ComputeChangeRate}}
\newcommand{\alglp}{\textsc{LP}}
\newcommand{\alglbdc}{\textsc{LBDC}}
\newcommand{\algecd}{\textsc{ECD}}
\newcommand{\algclsz}{\textsc{CLSZ}}
\newcommand{\algcliquecut}{\textsc{CliqueCut}}
\newcommand{\algcc}{\textsc{CC}}

\newcommand{\diffalgname}{\textsc{FindBipartiteComponents}}
\newcommand{\diffalgshortname}{\textsc{FBC}}
\newcommand{\diffalgapproxname}{\textsc{FBCApprox}}
\newcommand{\diffalgapproxshortname}{\textsc{FBCA}}

\definecolor{indiagreen}{rgb}{0.07, 0.53, 0.03}

\newcommand{\twopartdef}[4]
{
	\left\{
		\begin{array}{ll}
			#1 & \mbox{if } #2 \\
			#3 & \mbox{if } #4
		\end{array}
	\right.
}

\newcommand{\twopartdefow}[3]
{
	\left\{
		\begin{array}{ll}
			#1 & \mbox{if } #2 \\
			#3 & \mbox{otherwise}
		\end{array}
	\right.
}

\newcommand{\twopartdefowfs}[3]
{
	\left\{
		\begin{array}{ll}
			#1 & \mbox{if } #2, \\
			#3 & \mbox{otherwise.}
		\end{array}
	\right.
}

\newcommand{\threepartdefow}[5]
{
	\left\{
		\begin{array}{lll}
			#1 & \mbox{if } #2 \\
			#3 & \mbox{if } #4 \\
			#5 & \mbox{otherwise}
		\end{array}
	\right.
}
\newcommand{\allnotation}[1]{#1}

\renewcommand{\vec}[1]{{\allnotation{\bm{#1}}}}
\newcommand{\vecf}{\vec{f}}
\newcommand{\vecg}{\vec{g}}
\newcommand{\vecx}{\vec{x}}
\newcommand{\vecy}{\vec{y}}

\newcommand{\vecu}{\vec{u}}
\newcommand{\vecv}{\vec{v}}
\newcommand{\vecr}{\vec{r}}
\newcommand{\vecp}{\vec{p}}
\newcommand{\vecs}{\vec{s}}
\newcommand{\veca}{\vec{a}}
\newcommand{\vecb}{\vec{b}}
\newcommand{\vecw}{\vec{w}}
\newcommand{\vecc}{\vec{c}}
\newcommand{\constvec}{\vec{1}}
\newcommand{\zerovec}{\vec{0}}
\newcommand{\indicatorvec}{\vec{\chi}}
\newcommand{\statdist}{\vec{\pi}}

\newcommand{\graph}[1]{{\allnotation{#1}}}
\newcommand{\graphg}{\graph{G}}
\newcommand{\graphh}{\graph{H}}
\newcommand{\graphm}{\graph{M}}

\newcommand{\set}[1]{{\allnotation{#1}}}
\newcommand{\setv}{\set{V}}
\newcommand{\sete}{\set{E}}
\newcommand{\sets}{\set{S}}

\newcommand{\setl}{\set{L}}
\newcommand{\setr}{\set{R}}
\newcommand{\seta}{\set{A}}
\newcommand{\setb}{\set{B}}
\newcommand{\setc}{\set{C}}
\newcommand{\setu}{\set{U}}
\newcommand{\setp}{\set{P}}
\newcommand{\setm}{\set{M}}
\newcommand{\vertexset}{\setv}
\newcommand{\vertexsetg}{\vertexset_\graphg}

\newcommand{\edgeset}{\sete}
\newcommand{\edgesetg}{\edgeset_\graphg}

\newcommand{\mat}[1]{{\allnotation{\mathbf{#1}}}}

\newcommand{\matm}{\mat{M}}

\newcommand{\lap}{\mat{L}}
\newcommand{\lapn}{\mat{N}}
\newcommand{\signlap}{\mat{J}}
\newcommand{\signlapn}{\mat{Z}}
\newcommand{\degm}{\mat{D}}
\newcommand{\degmhalf}{\degm^{\allnotation{\frac{1}{2}}}}
\newcommand{\degmhalfneg}{\degm^{\allnotation{-\frac{1}{2}}}}
\newcommand{\adj}{\mat{A}}
\newcommand{\adjn}{\mat{B}}
\newcommand{\adjh}{\allnotation{\mat{A}_h}}
\newcommand{\lazywalkm}{\mat{W}}
\newcommand{\walkm}{\mat{X}}

\newcommand{\identity}{\mat{I}}

\newcommand{\imag}{\allnotation{\mathit{i}}}
\newcommand{\conjugate}[1]{\allnotation{\overline{#1}}}

\newcommand{\pagerank}{Pagerank}
\newcommand{\lovasz}{Lov\'asz}
\renewcommand{\etal}{et al.}

\renewcommand{\deg}{{\allnotation{d}}}
\newcommand{\weight}{{\allnotation{w}}}
\newcommand{\bipart}{{\allnotation{\beta}}}
\newcommand{\cond}{{\allnotation{\Phi}}}
\newcommand{\kcond}{{\allnotation{\rho}}}
\newcommand{\kcondg}{{\allnotation{\kcond_\graphg}}}
\newcommand{\neighbors}{{\allnotation{N}}}
\newcommand{\connected}{{\allnotation{\sim}}}

\newcommand{\ppr}{\mathrm{ppr}}
\newcommand{\apr}{\mathrm{apr}}

\newcommand{\discrep}{{\allnotation{\Delta}}}

\newcommand{\firstdef}[1]{\textbf{#1}}

\newcommand{\iverson}[1]{\pmb{\left[\vphantom{#1}\right.} #1 \pmb{\left.\vphantom{#1}\right]}}

\definecolor{thmblue}{RGB}{224, 231, 255}
\definecolor{thmoutline}{RGB}{199, 212, 255}
\declaretheoremstyle[
    headfont=\bfseries, 
    notebraces={[}{]},
    bodyfont=\normalfont\itshape,
    spacebelow=\parsep,
    spaceabove=\parsep,
    mdframed={
        backgroundcolor=thmblue, 
            linecolor=thmoutline, 
            innertopmargin=6pt,
            roundcorner=5pt, 
            innerbottommargin=6pt, 
            skipabove=\parsep, 
            skipbelow=\parsep } 
]{keythmstyle}
\declaretheorem[
    name=Theorem,
    numberwithin=chapter
]{theorem}

\title{On Learning the Structure of Clusters in Graphs}
\author{Peter Macgregor}

\abstract{%
    Graph clustering is a fundamental problem in unsupervised learning, with numerous applications in computer science and in analysing real-world data. In many real-world applications, we find that the clusters have a significant high-level structure. This is often overlooked in the design and analysis of graph clustering algorithms which make strong simplifying assumptions about the structure of the graph. This thesis addresses the natural question of whether the structure of clusters can be learned efficiently and describes four new algorithmic results for learning such structure in graphs and hypergraphs.
    
    The first part of the thesis studies the classical spectral clustering algorithm, and presents a tighter analysis on its performance. This result explains why it works under a much weaker and more natural condition than the ones studied in the literature, and helps to close the gap between the theoretical guarantees of the spectral clustering algorithm and its excellent empirical performance.
   
    The second part of the thesis builds on the theoretical guarantees of the previous part and shows that, when the clusters of the underlying graph have certain structures, spectral clustering with fewer than $k$ eigenvectors is able to produce better output than classical spectral clustering in which $k$ eigenvectors are employed, where $k$ is the number of clusters. This presents the first work that discusses and analyses the performance of spectral clustering with fewer than $k$ eigenvectors, and shows that general structures of clusters can be learned with spectral methods.
    
    The third part of the thesis considers efficient learning of the  structure of clusters with local algorithms, whose runtime depends only on the size of the target clusters and is independent of the underlying input graph. While the objective of classical local clustering algorithms is to find a cluster which is sparsely connected to the rest of the graph, this part of the thesis presents a local algorithm that finds a pair of clusters which are \emph{densely} connected to each other. This result demonstrates that certain structures of clusters can be learned efficiently in the local setting, even in the massive graphs which are ubiquitous in real-world applications.
    
    The final part of the thesis studies the problem of learning densely connected clusters in hypergraphs. The developed algorithm is based on a new heat diffusion process, whose analysis extends a sequence of recent work on the spectral theory of hypergraphs. It allows the structure of clusters to be learned in datasets modelling higher-order relations of objects and can be applied to efficiently analyse many complex datasets occurring in practice.
    
    All of the presented theoretical  results are further  extensively evaluated on both synthetic and real-word datasets of different domains, including image classification and segmentation, migration networks, co-authorship networks, and natural language processing. These experimental results demonstrate that the newly developed algorithms are practical, effective, and immediately applicable for learning the structure of clusters in real-world data.
}

\begin{document}

\begin{preliminary}

\maketitle

\begin{acknowledgements}
I must begin by thanking my supervisor, He Sun, for his invaluable support throughout this process.
I have immensely benefited from his encouragement, experience, knowledge, and friendship.
Sincere thanks are also due to the many other friends I made along the way.
Thanks are particularly in order for the walks and online hang-outs which helped me through the COVID-19 lockdowns.
Their ongoing friendship makes life infinitely more enjoyable.

I'd like to thank my family for being as wonderful as anyone could reasonably hope for.
Thanks in particular to my parents.
I would not be writing this thesis were it not for their encouragement of my interest in mathematics from a young age and for their continuing and unconditional support.
Thanks to Dad for the endless conversations about mathematical puzzles and curiosities~\cite{macgregorStaticsDynamicsDeeply1990}.
Thanks to Mum for putting up with us, and for her uncanny ability to spot typos hidden deep in the technical sections of my papers.

Thanks to Granny for permission to use her photo in the introduction.
I extend to her all my love and congratulations on the occasion of her 90th birthday on 17th August 2022.
\afterpage{\null\newpage}
\end{acknowledgements}

\standarddeclaration


\setcounter{tocdepth}{2}
\tableofcontents

\listoffigures
\listoftables
\listofalgorithms
\addcontentsline{toc}{chapter}{List of Algorithms}

\end{preliminary}


\chapter{Introduction} \label{chap:intro}
Graphs are a natural tool to represent many datasets and have proved invaluable in many diverse areas of data science such as
recommendation systems~\cite{andersenTrustbasedRecommendationSystems2008, abbassiRecommenderSystemBased2007}, image segmentation~\cite{shiNormalizedCutsImage2000, sunDistributedGraphClustering2019, zhangImageSegmentationBased2021}, and bioinformatics~\cite{highamEpidemicsHypergraphsSpectral, goglevaKnowledgeGraphbasedRecommendation2022}.
One of the most fundamental operations in these applications is \emph{graph clustering}, in which the goal is to partition the vertices of a graph into disjoint clusters.
A cluster is usually defined to be a set of vertices in the graph such that there are many edges connecting vertices inside the cluster, while there are few edges connecting the cluster to the rest of the graph.
However, this definition fails to tell the whole story.
In many real-world applications, we find that the clusters also have a significant high-level structure which is often overlooked in the design and analysis of graph clustering algorithms.
Throughout this thesis, we will discuss several methods for learning this high-level structure of clusters, and for exploiting this structure to develop improved algorithms for graph clustering.

In order to better understand the structure of clusters in graphs, let us consider some examples constructed from real-world data.
We begin by considering the well-known MNIST dataset which consists of 70,000 images of handwritten digits from 0 to 9.
Each image has $28 \times 28$ greyscale pixels, and each pixel takes a value from 0 to 255.
This dataset is not given to us as a graph, but we can employ a common tool to convert almost any generic dataset into a graph: the \emph{$k$ nearest neighbour graph}.
To do this, we consider each image in the dataset to be a vertex in the graph, and   connect each vertex $v$ to the vertices representing the images which are the `most similar' to the image represented by $v$.
This construction is described more formally in Chapter~\ref{chap:meta}.
Figure~\ref{fig:intro:mnist_knn} illustrates the $k$ nearest neighbour graph for a randomly chosen subset of the data.
\begin{figure}[t]
    \centering
    \hspace{0.05\textwidth}
    \begin{subfigure}{0.45\textwidth}  
    \centering
    \includegraphics[height=0.22\textheight]{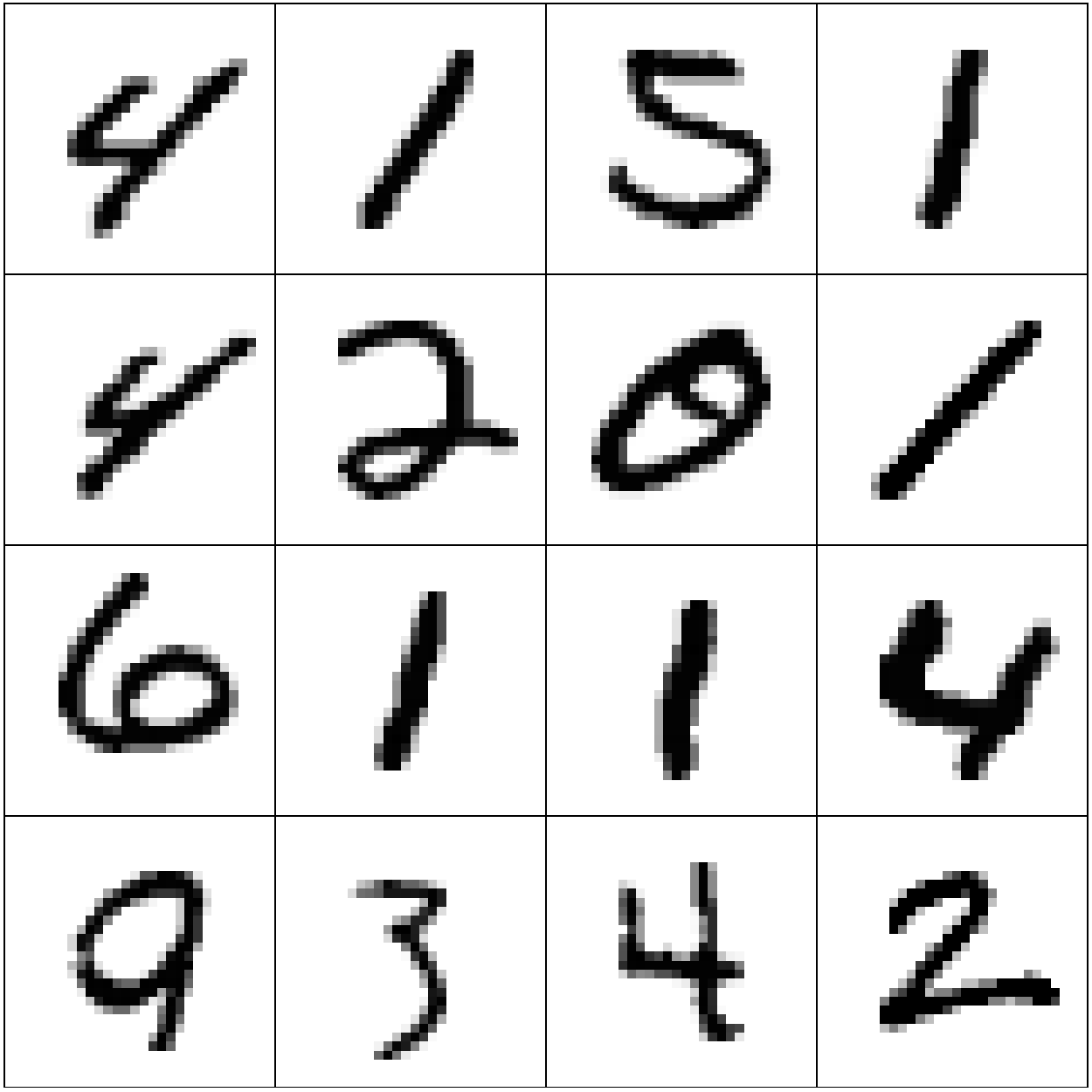}
    \caption{}
    \end{subfigure}
    \hspace{0.02\textwidth}
    \begin{subfigure}{0.45\textwidth}  
    \centering
    \includegraphics[height=0.22\textheight]{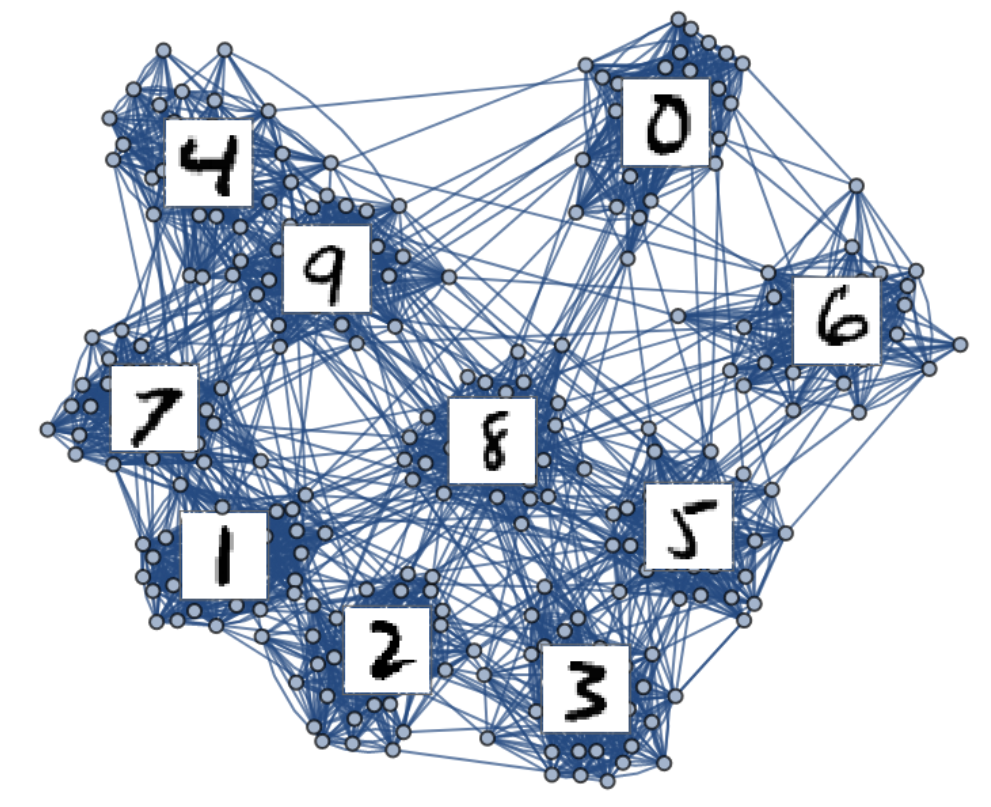}
    \caption{}
    \end{subfigure}
    \caption[The MNIST dataset]{\textbf{(a)} Example images from the MNIST dataset. \textbf{(b)} The $k$ nearest neighbour graph of a random subset of the MNIST dataset, for $k = 3$. The graph has a cluster corresponding to each digit.
    }
    \label{fig:intro:mnist_knn}
\end{figure}
Notice that in the constructed graph there are 10 clusters, which correspond to the ground truth clusters of images representing each digit.
A vertex in a given cluster tends to be more strongly connected with other vertices inside the same cluster than with vertices in other clusters.
There is also additional high-level structure due to the relative similarity of different digits.
For example, the cluster representing the digit 4 has more connections with the cluster representing 9 than the one representing 6.
This is the \emph{high-level structure} which has been overlooked in previous clustering algorithms, and we study such structures in more detail in Chapter~\ref{chap:meta}.

We will also look at another learning task for image data: image segmentation.
In this case, we are given an arbitrary image, and the goal is to partition the pixels into some number of clusters such that each cluster corresponds to a different region of the image.
This time, we will construct a graph from a single image in the following way:
we first represent each pixel in the image by a point $(r, g, b, x, y)^\transpose \in \R^5$ where $r, g, b \in [1, 255]$ are the red, green, and blue values of the pixel and $x$ and $y$ are the coordinates of the pixel in the image.
Then, we construct the similarity graph by taking each pixel to be a vertex in the graph, and for every pair of pixels $\vecu, \vecv \in \R^5$, we add an edge with weight $\exp(- \norm{\vecu - \vecv}^2 / 2 \sigma^2)$ for some parameter $\sigma$.
Notice that this graph combines the spatial and colour information of each pixel, such that we might expect each region of the image to be a separate cluster in the graph.
We can then use a graph clustering algorithm to find these regions and segment the image.
Once again, we can reasonably expect that the clusters in the graph constructed from the image will have some high-level structure.
For example, a cluster representing a given region of the image will be more closely related to clusters of pixels with similar colours, or clusters representing neighbouring regions.
We will see in Chapter~\ref{chap:meta} that this structure of clusters allows us to develop an improved graph-based image segmentation algorithm than the ones previously studied in the literature, and the results of this algorithm are also shown in Figure~\ref{fig:intro:imagesegmentation}. 
\begin{figure}[th]
    \centering
    \begin{subfigure}{0.3\textwidth}  
    \includegraphics[width=0.9\textwidth]{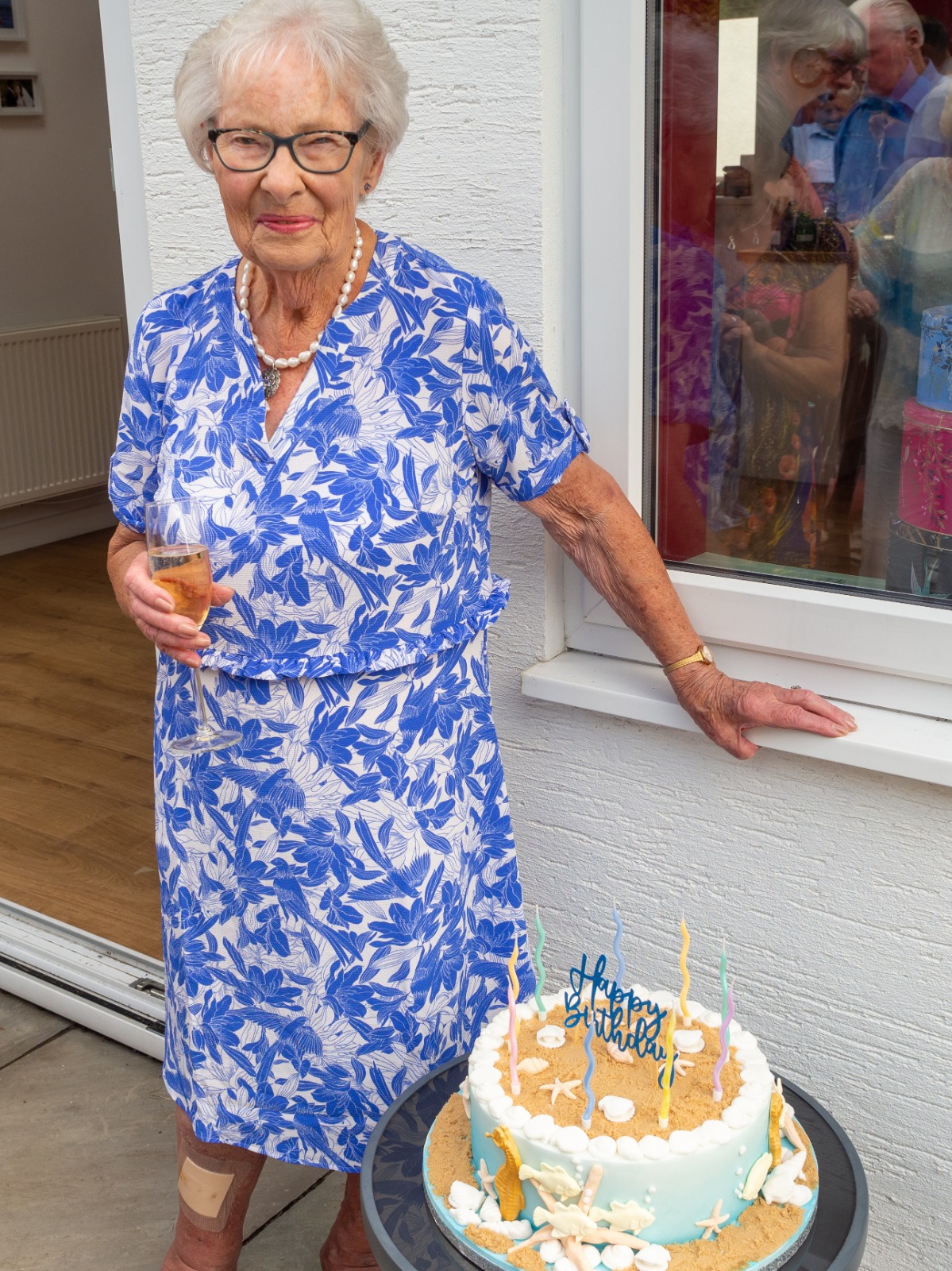}
    \caption{Original Image}
    \end{subfigure}
    \hspace{0.15\textwidth}
    \begin{subfigure}{0.3\textwidth} 
    \includegraphics[width=0.9\textwidth]{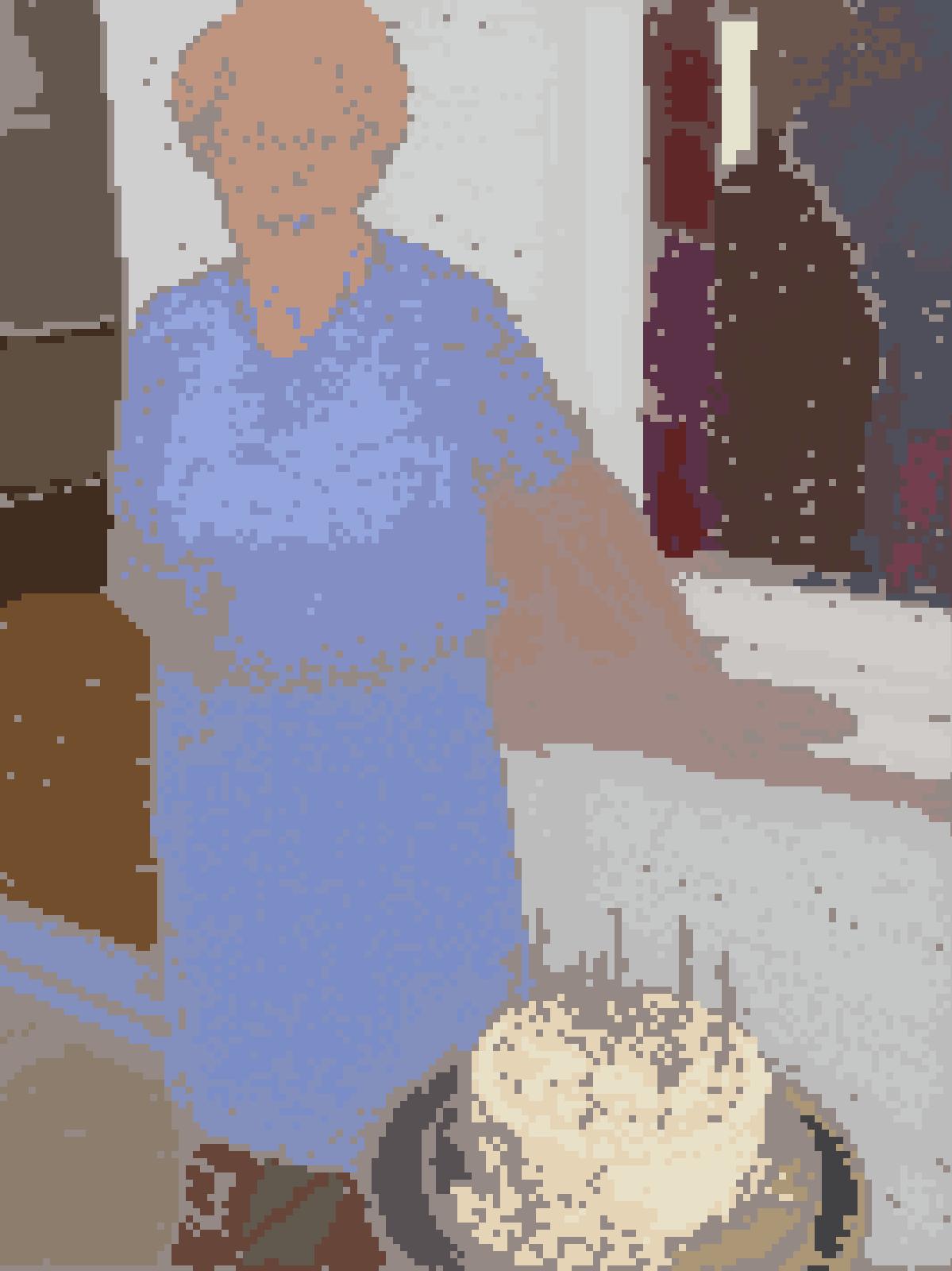}
    \caption{Segmented Image}
    \end{subfigure}
    \caption[Image segmentation based on graph clustering]{An example of image segmentation using the clustering algorithm in Chapter~\ref{chap:meta}.}
    \label{fig:intro:imagesegmentation}
\end{figure}

Finally, we will consider a completely different example dataset in which we are interested in learning the structure of clusters in a graph.
The \emph{interstate dispute dataset} records every military dispute between pairs of countries
during 1816--2010, including the level of hostility resulting from the dispute and the number of casualties if applicable.
This dataset already has a natural graph structure: for a given time period, we can construct a graph from the data by representing each country with a vertex and adding an edge between each pair of countries weighted according to the severity of their military disputes.
Notice that since the edges in this graph represent a `negative' relationship between countries, we would expect that a cluster of allied countries would not contain many edges inside the cluster.
On the other hand, two different clusters might be densely connected to each other when there are two groups of opposing countries.
If we wish to develop an algorithm for finding the clusters in this graph, clearly we must take the high-level structure into account and most existing graph clustering algorithms will not perform well.
In Chapter~\ref{chap:local}, we will develop an algorithm capable of learning the clusters in this dataset which correspond to the underlying geopolitical relationships between countries.
 Figure~\ref{fig:intro:disputegraph} demonstrates that from this dataset, our algorithm is able to identify changes in geopolitics over time, including the changing roles of Russia, Japan, and Eastern Europe during the 20th century.
\begin{figure}[th]
    \centering
    \begin{subfigure}{0.48\textwidth} 
    \includegraphics[width=\textwidth]{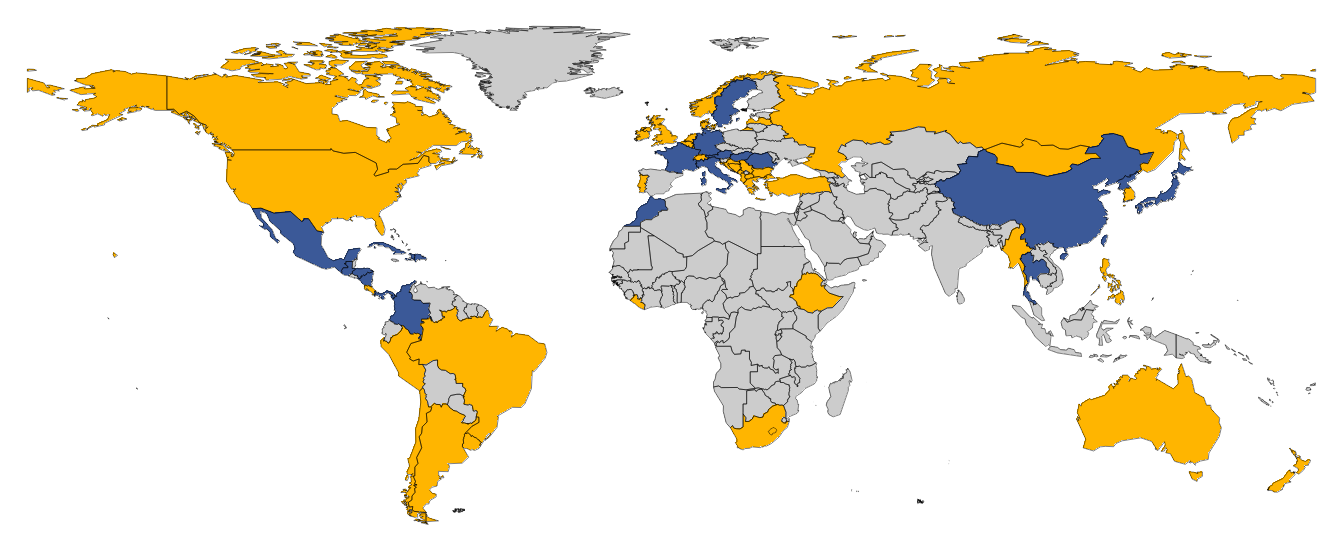}
    \caption{1900-1950}
    \end{subfigure}
    \hfill
    \begin{subfigure}{0.48\textwidth} 
    \includegraphics[width=\textwidth]{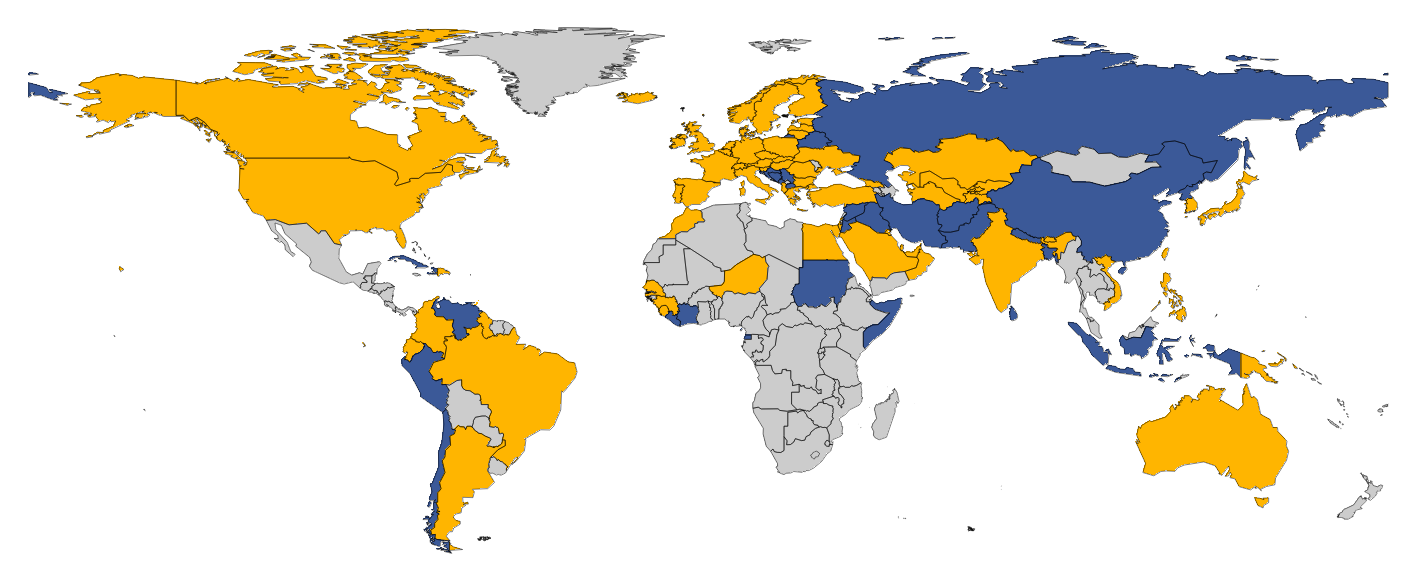}
    \caption{1990-2010}
    \end{subfigure}
    \caption[Output of our algorithm on the interstate dispute dataset]{ Clusters found by the algorithm from Chapter~\ref{chap:local} on the interstate dispute dataset.
    For any given time period and using the USA as a `seed', the algorithm returns two densely connected clusters representing groups of allied countries.
    The results are well explained by geopolitics, and capture the changing political picture during the 20th century.
    \label{fig:intro:disputegraph}}
\end{figure}

Although our highlighted examples come from different domains  and the method used to construct a graph is different in each case,
their common theme is that the clusters in the generated graphs have a significant high-level structure.
In this thesis we will see that these structures can be learned efficiently through our newly developed algorithms and that the cluster structure can be used to significantly improve the performance of spectral clustering.
Throughout the thesis, we will touch on a wide variety of techniques from the field of algorithmic spectral graph theory, which studies the matrix representations of graphs and applies the algebraic properties of these matrices to design efficient algorithms.
Along the way, we will examine the classical spectral clustering algorithm (Chapters~\ref{chap:tight}~and~\ref{chap:meta}), random walks and local algorithms (Chapter~\ref{chap:local}), and the recently emerging spectral hypergraph theory based on non-linear heat diffusion operators (Chapter~\ref{chap:hyper}).
The newly developed techniques are all extensively evaluated on real-world datasets including those introduced above, and the results demonstrate a wide applicability for learning the structure of clusters in natural datasets.
The remainder of the thesis is organised as follows.

\begin{itemize}
    \item Chapter~\ref{chap:prelim} introduces the basic mathematical concepts and notation which is used throughout the thesis, including a brief introduction to linear algebra, graphs, and hypergraphs. Additionally, this chapter introduces the field of spectral graph theory, on which most of the results in this thesis are based.
    \item Chapter~\ref{chap:related} describes the previous work which is most related to the methods and results presented in later chapters.
    This includes a detailed description of the history of spectral graph partitioning, techniques for local graph clustering, and spectral hypergraph theory.
    This chapter forms an important reference in order to understand the technical material of the thesis.
    \item Chapter~\ref{chap:tight} studies spectral clustering, and presents a tighter analysis on its performance. This result explains why spectral clustering works under a much weaker and more natural condition than the ones previously studied in the literature, and helps to close the gap between the theoretical guarantees of the spectral clustering algorithm and its excellent empirical performance.
    This chapter is largely based on Sections~3~and~4 of Macgregor and Sun~\cite{macgregorTighterAnalysisSpectral2022}.
    \item Chapter~\ref{chap:meta} builds on the theoretical guarantees of the previous chapter and shows that, when the clusters of the underlying graph have certain structures, spectral clustering with fewer than $k$ eigenvectors is able to produce better output than classical spectral clustering in which $k$ eigenvectors are employed, where $k$ is the number of clusters.
    This chapter includes experiments on the MNIST dataset and the BSDS dataset for image segmentation, and demonstrates improved clustering performance over previous methods.
    This chapter is based on Sections~5~and~6 of Macgregor and Sun~\cite{macgregorTighterAnalysisSpectral2022}.
    \item Chapter~\ref{chap:local} considers efficient learning of the structure of clusters with local algorithms, whose runtime depends only on the size of the target clusters and is independent of the underlying input graph.
    The experimental part of the chapter focuses on graphs like the interstate dispute dataset, in which clusters are \emph{densely} connected to each other, and demonstrates that certain structures of clusters can be learned efficiently in the local setting.
    This chapter is based on Macgregor and Sun~\cite{macgregorLocalAlgorithmsFinding2021}.
    \item Finally, Chapter~\ref{chap:hyper} studies the problem of learning densely connected clusters in hypergraphs.
    The developed algorithm is based on a new heat diffusion process, whose analysis extends a sequence of recent work on the spectral theory of hypergraphs.
    It allows the structure of clusters to be learned in datasets modelling higher-order relations of objects and can be applied to efficiently analyse many complex datasets occurring in practice.
    This chapter is based on Macgregor and Sun~\cite{macgregorFindingBipartiteComponents2021}. 
\end{itemize}

\begin{samepage}
\chapter{Preliminaries} \label{chap:prelim}
This chapter introduces the basic technical concepts which form the prerequisites for understanding the material in this thesis.
Although most of the concepts listed here are standard, everything is included so that the thesis is as self-contained as possible.
The chapter is divided into six sections.
Section~\ref{sec:prelim:notation} defines some basic mathematical notation;
Section~\ref{sec:prelim:la} gives a brief introduction to concepts from linear algebra; 
Section~\ref{sec:prelim:graphs} introduces graphs and spectral graph theory;
Section~\ref{sec:prelim:hypergraphs} defines hypergraphs;
Section~\ref{sec:prelim:algorithms} defines asymptotic notation used in the analysis of algorithms;
finally, Section~\ref{sec:prelim:metrics} defines the metrics used to evaluate clustering algorithms.

\section{Basic Notation} \label{sec:prelim:notation}
Let us first define some of the mathematical notation used throughout the thesis. 
The set of real numbers is given by $\R$, the set of integers is $\Z$, and the set of complex numbers is $\C$.
Occasionally, $\R_{\geq 0}$ or $\Z_{\geq 0}$ are used to restrict these sets to have non-negative values.
For any $n \in \Z_{\geq 0}$, the notation $[n]$ is used to mean $\{1, 2, \ldots, n\}$.
For some logical proposition $X$, the expression $\iverson{X}$ is equal to $1$ if $X$ is true, or $0$ otherwise.
For example, we say that $\iverson{a=b}$ is equal to $1$ if $a=b$, and it is equal to $0$ otherwise.
For any operators $f, g: \R^n \rightarrow \R^n$, we define $f \circ g:\R^n \rightarrow \R^n$ by  $\left(f \circ g\right)(v) \triangleq f(g(v))$ for any $v$.
Given two sets $\seta$ and $\setb$, the set difference is expressed as $\seta \setminus \setb$.
The \firstdef{symmetric difference} between the sets $A$ and $B$ is given by $\seta \symmetricdiff \setb \triangleq \left(\seta \setminus \setb\right) \union \left(\setb \setminus \seta\right)$.
When $\seta$ is a subset of some universe $\set{U}$,   the \firstdef{complement} of $\seta$ is given by $\setcomplement{\seta} \triangleq \set{U} \setminus \seta$.
For any set $\seta$, we   use $2^\seta$ to denote the \firstdef{power set} of $\seta$; that is, $2^\seta$ is the set consisting of all possible subsets of $\seta$.

\section{Linear Algebra}\label{sec:prelim:la}
The subject of this thesis lies at the intersection of linear algebra and graph theory.
\end{samepage}
As such, let us review some basic concepts in linear algebra.
A vector $\vecx \in \R^n$ is written as
$(\vecx(1), \vecx(2), \ldots, \vecx(n))^\transpose$.
Notice that we use $\vecx(i)$ to refer to the $i$-th element in the vector $\vecx$.
Subscripts, such as $\vecx_i$, are used to distinguish between different vectors.
For a vector $\vecx \in \R^n$ and a set $\seta \subseteq [n]$, let $\vecx(\seta) = \sum_{i \in A} \vecx(i)$.
For any vectors $\vecx, \vecy \in \R^n$, we write $\vecx \preceq \vecy$ if it holds for all $v$ that $\vecx(v) \leq \vecy(v)$.
We will use $\vec{1}_n \in \R^n$ to represent the vector with $\vec{1}_n(i) = 1$ for all $i \in [n]$ and $\vec{0}_n$ to represent the vector with $\vec{0}_n(i) = 0$.
For any vector $\vecx$, we define the \firstdef{support} of $\vecx$ to be $\supp(\vecx) = \{u : \vecx(u) \neq 0\}$.
Given any two vectors $\vecv, \vecu \in \R^n$, the \firstdef{inner product} of $\vecv$ and $\vecu$ is written as $\inner{\vecv}{\vecu}$ or $\vecv^\transpose \vecu$ and defined to be $\sum_{i = 1}^n \vecv(i) \vecu(i)$.
Throughout the thesis we will use the \firstdef{Euclidean norm}, defined for any $\vecv \in \R^n$ as
\[
    \norm{\vecv} = \sqrt{\inner{\vecv}{\vecv}} = \sqrt{\sum_{i = 1}^n \vecv(i)^2}.
\]
Given a set of vectors $\{\vecv_1, \ldots, \vecv_n\}$, their \firstdef{span} is denoted by $\mathrm{span}(\vecv_1, \ldots \vecv_n)$ and is defined to be the set of vectors $\vecx$ which can be written as a linear combination of $\vecv_1, \ldots, \vecv_n$.
That is,  $\vecx\in\mathrm{span}(\vecv_1,\ldots, \vecv_n)$ if and only if
\[
    \vecx = \sum_{i = 1}^n c_i \vecv_i
\]
for some numbers $c_1, \ldots, c_n \in \R$.

Given a matrix $\matm \in \R^{n \times m}$, the expression $\matm(i, j)$ is the element in the $i$-th row  and $j$-th column of the matrix. 
An \firstdef{eigenvector} of a matrix $\matm \in \R^{n \times n}$ is any vector $\vecv \in \R^n$ such that $\vecv \neq \vec{0}_n$ and
\[
\matm \vecv = \lambda \vecv
\]
for some $\lambda \in \R$, and $\lambda$ is called an \firstdef{eigenvalue}.
When discussing some matrix $\matm \in \R^{n \times n}$, we often use $\lambda_1, \ldots, \lambda_n$ to denote the eigenvalues of $\matm$ and $\vecf_1, \ldots \vecf_n$ for the eigenvectors.
Occasionally, we   consider matrices with complex-valued entries.
Recall that for any complex number $x = a + b \cdot \imag$, the \firstdef{complex conjugate} of $x$ is given by $\conjugate{x} = a - b \cdot \imag$.
For any $\matm \in \C^{n \times n}$, the \firstdef{conjugate transpose} of $\matm$ is written as $\matm^\conjtranspose$, and is constructed by taking the transpose of $\matm$ and replacing each entry with its complex conjugate.
A matrix $\matm$ is said to be \firstdef{Hermitian} if $\matm = \matm^\conjtranspose$.
A real matrix $\matm$ is \firstdef{positive semi-definite (PSD)} if and only if it holds that
\[
    \vecx^\transpose \matm \vecx \geq 0
\]
for all $\vecx \in \R^n$.
It is equivalent to say that all of the eigenvalues of $\matm$ are non-negative.

The following theorem, known as the Courant-Fischer Theorem, characterises the eigenvectors and eigenvalues of a symmetric matrix with respect to an optimisation over $\R^n$.
\begin{theorem}[Courant-Fischer Theorem] \label{thm:courantfischer}
Let $\mat{M}$ be a real symmetric matrix with eigenvalues $\lambda_1 \leq \ldots \leq \lambda_n$. Then
\[
    \lambda_k = \min_{\substack{S \subset \R^n \\ \mathrm{dim}(S) = k}} \hspace{0.5em} \max_{\substack{\vecx \in S \\ \vecx \neq \vec{0}_n}} \frac{\vecx^\transpose \mat{M} \vecx}{\vecx^\transpose \vecx},
\]
where the minimisation is taken over all $k$ dimensional subspaces of $\R^n$.
\end{theorem}

\section{Graphs} \label{sec:prelim:graphs}
An \firstdef{undirected graph} $\geqvew$ is defined by a set of vertices $\vertexset$, a set of edges $\edgeset$, and a weight function $\weight: \vertexset \times \vertexset \rightarrow \R_{\geq 0}$.
Every edge $e \in \edgeset$ is an unordered pair $\{v, u\}$ where $v, u \in \vertexset$ and $v \connected u$ is used to mean that $\{v, u\} \in \edgeset$.
The weight function $\weight$ is defined such that $\weight(v, u) > 0$ if $\{v, u\} \in \edgeset$ and $\weight(v, u) = 0$ otherwise.
Sometimes, the weight function $w$ is omitted when discussing a graph and in this case the weight of every edge can be taken to be $1$.
In this thesis, we consider only graphs with non-negative edge weights.
A \firstdef{directed graph} is defined similarly to an undirected graph, but in this case each edge $e \in \edgeset$ is an \emph{ordered} pair $(u, v)$. 
Unless otherwise specified, the graphs defined in this thesis are undirected.
Sometimes, the vertex and edge sets of different graphs are distinguished by subscript, such as $\graphg = (\vertexset_\graphg, \edgeset_\graphg, \weight_\graphg)$ and $\graphh = (\vertexset_\graphh, \edgeset_\graphh, \weight_\graphh)$.

Given some graph $\geqvew$, the number of vertices in a graph is commonly given by $n = \cardinality{\vertexset}$.
The \firstdef{degree} of any vertex $v \in \vertexset$ is defined by $\deg(v) \triangleq \sum_{u \in \vertexset} \weight(v,u)$,
and the set of neighbours of $v$ is $\neighbors(v) \triangleq \{u \in \vertexset : \{v, u\} \in \edgeset \}$.
When $\geqve$ is a directed graph, for any $v \in \vertexset$ we use $\deg_\Out(v)$ and $\deg_\In(v)$ to express the number of edges with $v$ as the tail or the head, respectively. 
For any $\sets \subset \vertexset$, we define
$\vol_\Out(\sets) = \sum_{v \in \sets} \deg_\Out(v)$, and
$\vol_\In(\sets) = \sum_{v \in \sets} \deg_\In(v)$.

For an undirected graph $\geqve$ and any $\sets \subseteq \vertexset$, the \firstdef{volume} of $\sets$ is defined by $\vol(\sets) \triangleq \sum_{v \in \sets} \deg(v)$,
and for two sets $\sets_1, \sets_2 \subseteq \vertexset$ let us define the weight between $\sets_1$ and $\sets_2$ by
\[
    \weight(\sets_1, \sets_2) \triangleq \sum_{v \in \sets_1} \sum_{u \in \sets_2} \weight(v, u).
\]
The \firstdef{conductance} of $\sets \subset \vertexset$ is given by
\[
\cond(\sets) \triangleq \frac{\weight(\sets, \setcomplement{\sets})}{\min\{\vol(\sets),\vol(\setcomplement{\sets})\}}.
\]
 
Notice that if $\cond(\sets) = 0$, then $\sets$ is disconnected from the rest of the graph.
Loosely, if $\cond(\sets)$ is low, then there are only a small number of edges connecting $\sets$ to the rest of the graph and we can say that $\sets$ forms a cluster in the graph. 
The conductance of the graph $\geqve$ is given by
\[
    \cond_\graphg \triangleq \min_{\substack{\sets \subsetneq \vertexset \\ \sets \neq \emptyset}} \cond(\sets).
\]
The \firstdef{boundary} of $\sets \subset \vertexset$ is given by
\[
    \partial(\sets) \triangleq \{ \{u, v\} \in \edgeset : u \in \sets, v \not\in \sets\}.
\]
The \firstdef{indicator vector} of a set $\sets \subseteq \vertexset$ is given by $\indicatorvec_\sets \in \R^n$, where
\[
    \indicatorvec_\sets(v) = \twopartdefowfs{1}{v \in \sets}{0}
\]
If $\sets$ contains only a single vertex $v$, we can write $\indicatorvec_v$ instead of $\indicatorvec_{\{v\}}$.

A graph $\geqve$ is said to be \firstdef{bipartite} if the vertex set can be decomposed into disjoint sets $\setl$ and $\setr$ such that
every edge $e \in \edgeset$ connects a vertex in $\setl$ with a vertex in $\setr$.
That is, for all $e \in \edgeset$, it holds that $e \intersect \setl \neq \emptyset$ and $e \intersect \setr \neq \emptyset$.
A \firstdef{connected component} of $\graphg$ is any maximal set of vertices $\sets \subseteq \vertexset$ such that there is a path from any vertex in $\sets$ to every other vertex in $\sets$.

\subsection{Spectral Graph Theory} \label{sec:prelim:sgt}
Having introduced linear algebra and graph theory, let us now take their intersection.
Spectral graph theory is the study of the eigenvalues and eigenvectors of various matrices associated with a graph.
This section gives a brief introduction to the key concepts in spectral graph theory, and Chapter~\ref{chap:related} includes several specific results which are used in this thesis.
For a comprehensive overview of the field, we refer the reader to Chung's excellent book~\cite{chungSpectralGraphTheory1997}.

Given a graph $\geqvew$ with $n$ vertices, let us label the vertices as $\vertexset = \{v_1, v_2, \ldots, v_n\}$.
Then, the \firstdef{adjacency matrix} of $\graphg$ is given by $\adj \in \R^{n \times n}$ where $\adj(i, j) = \weight(v_i, v_j)$.
The \firstdef{degree matrix} $\degm \in \R^{n \times n}$ is a diagonal matrix whose non-zero entries are $\degm(i, i) = \deg(v_i)$ for $i \in [n]$.
A key object of study in spectral graph theory is the \firstdef{Laplacian matrix} of $\graphg$ which is defined to be
\[
    \lap = \degm - \adj.
\]
The most useful property of this operator is the expression given by its \firstdef{quadratic form}.
Given some vector $\vecx \in \R^n$, the Laplacian quadratic form is given by
\[
    \vecx^\transpose \lap \vecx = \sum_{\{u, v\} \in \edgeset} \weight(u, v) \left(\vecx(u) - \vecx(v)\right)^2.
\]
It is also common to instead consider the \firstdef{normalised Laplacian matrix}, which is defined as
\[
    \lapn = \degmhalfneg \lap \degmhalfneg = \identity - \degmhalfneg \adj \degmhalfneg.
\]
We also use the \firstdef{signless Laplacian operator} defined as
\[
\signlap \triangleq \degm + \adj,
\]
and its normalised counterpart 
\[
\signlapn \triangleq \degmhalfneg \signlap \degmhalfneg.
\]

Many results in spectral graph theory focus on the eigenvectors and eigenvalues of the normalised Laplacian matrix.
The $n$ eigenvalues are often given in order of increasing value, and are represented by $\lambda_1 \leq \lambda_2 \leq \ldots \leq \lambda_n$.
The corresponding eigenvectors are often denoted by $\vecf_1, \ldots, \vecf_n$.
Notice that the Laplacian matrix uniquely defines the original graph, and so the eigenvectors and eigenvalues of the Laplacian fully describe the graph.
 A recurring theme in 
spectral graph theory is that, in many cases, we can tell a lot about a graph from only a small number of eigenvalues and eigenvectors.  
The following basic results in spectral graph theory are not difficult to prove, and can  be found in~\cite{chungSpectralGraphTheory1997}.
\begin{enumerate}
    \item The eigenvalues of $\lapn$ all lie between $0$ and $2$.
    \item The smallest eigenvalue of $\lapn$ is always equal to $0$.
    \item If the graph has $k$ connected components, then the smallest eigenvalue of $\lapn$ has multiplicity $k$.
    \item The largest eigenvalue of $\lapn$ is equal to $2$ if and only if the graph has a bipartite connected component.
\end{enumerate}

\begin{figure}[t]
    \centering
    \includegraphics[width=0.27\textheight]{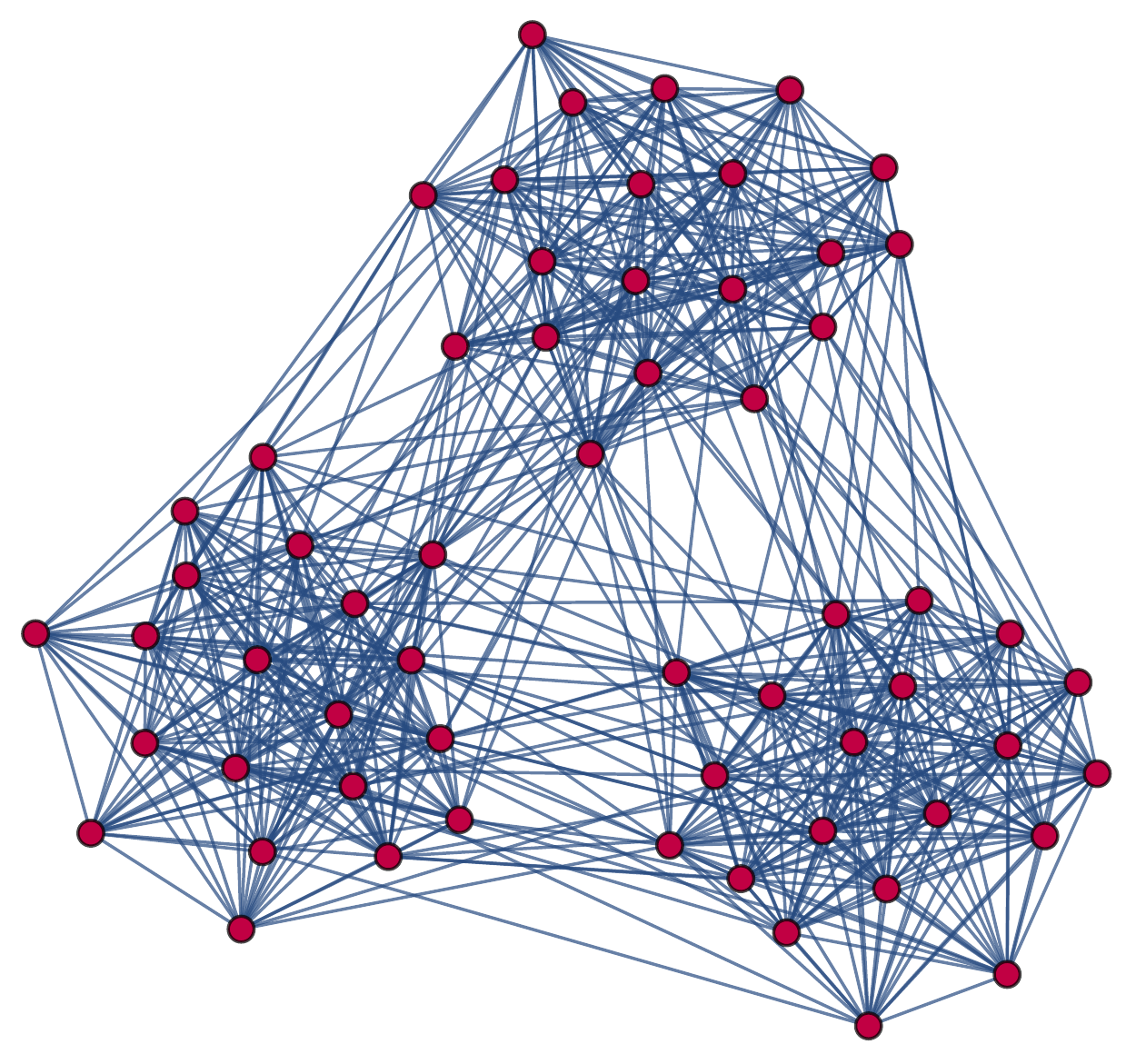}
    \caption[A random graph from the stochastic block model]{
    A random graph sampled from $\mathrm{SBM}(60, 3, 0.9, 0.05)$.
    }
    \label{fig:sbmExample}
\end{figure}

\subsection{The Stochastic Block Model} \label{sec:prelim:sbm}
It is useful to evaluate the performance of graph clustering algorithms
on graphs with a known ground-truth clustering.
A popular model for generating such random graphs is the \firstdef{stochastic block model (SBM)}, and we use several variants of this model throughout the thesis.
For $n, k \in \Z_{\geq 0}$ and $p, q \in [0, 1]$ the model $\mathrm{SBM}(n, k, p, q)$ defines a random graph on $n$ vertices in which the vertices are split evenly between $k$ clusters $\sets_1, \ldots, \sets_k$.
Every pair of vertices $u$ and $v$ are connected with probability
\begin{itemize}
    \item  $p$ if they are in the same cluster, and
    \item $q$ otherwise.
\end{itemize}
Figure~\ref{fig:sbmExample} shows an example graph drawn from the stochastic block model.
When $p > q$, the stochastic block model typically generates graphs with a ground-truth cluster structure which makes it a useful model for evaluating the performance of graph clustering algorithms.
For the interested reader, there is an extensive literature on several variants of the stochastic block model~\cite{abbeCommunityDetectionStochastic2017, kanadeGlobalLocalInformation2016}.

\section{Hypergraphs} \label{sec:prelim:hypergraphs}
A \firstdef{hypergraph} $\graphh = (\vertexset, \edgeset, \weight)$ is defined by a vertex set $\vertexset$, a set of edges $\edgeset \subseteq 2^\vertexset$ and a weight function $w: \edgeset \rightarrow \R_{\geq 0}$.
The difference between a graph and a hypergraph is   that each edge in a graph connects exactly two vertices while edges in a hypergraph can contain any number of vertices.
Note that this means a graph is a special case of a hypergraph.
Figure~\ref{fig:hypergraphExample} shows an example hypergraph.
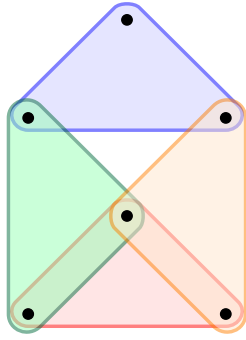
\begin{figure}[t]
    \centering
    \scalebox{1.3}{\begin{tikzpicture}
	\begin{pgfonlayer}{nodelayer}
		\node [style=none] (24) at (-7.75, 4.25) {};
		\node [style=none] (25) at (-8.25, 4.25) {};
		\node [style=none] (26) at (-10.25, 2.25) {};
		\node [style=none] (27) at (-10, 1.75) {};
		\node [style=none] (28) at (-5.75, 2.25) {};
		\node [style=none] (29) at (-6, 1.75) {};
		\node [style=smallblack] (0) at (-8, 4) {};
		\node [style=smallblack] (1) at (-10, 2) {};
		\node [style=smallblack] (2) at (-6, 2) {};
		\node [style=smallblack] (30) at (-8, 0) {};
		\node [style=smallblack] (31) at (-10, -2) {};
		\node [style=smallblack] (32) at (-6, -2) {};
		\node [style=none] (33) at (-7.75, 0.25) {};
		\node [style=none] (34) at (-8.25, 0.25) {};
		\node [style=none] (35) at (-10.25, -1.75) {};
		\node [style=none] (36) at (-10, -2.25) {};
		\node [style=none] (37) at (-5.75, -1.75) {};
		\node [style=none] (38) at (-6, -2.25) {};
		\node [style=none] (39) at (-7.75, -0.25) {};
		\node [style=none] (40) at (-7.75, 0.25) {};
		\node [style=none] (41) at (-9.75, 2.25) {};
		\node [style=none] (42) at (-10.4, 2) {};
		\node [style=none] (43) at (-9.75, -2.25) {};
		\node [style=none] (44) at (-10.4, -2) {};
		\node [style=none] (45) at (-8.25, 0.25) {};
		\node [style=none] (46) at (-8.25, -0.25) {};
		\node [style=none] (47) at (-6.25, -2.25) {};
		\node [style=none] (48) at (-5.6, -2) {};
		\node [style=none] (49) at (-6.25, 2.25) {};
		\node [style=none] (50) at (-5.6, 2) {};
	\end{pgfonlayer}
	\begin{pgfonlayer}{edgelayer}
		\draw [style=blue filled] (25.center)
			 to (26.center)
			 to [in=180, out=-135, looseness=1.25] (27.center)
			 to (29.center)
			 to [in=-45, out=0, looseness=1.25] (28.center)
			 to (24.center)
			 to [bend right=45, looseness=0.75] cycle;
		\draw [style=red filled] (34.center)
			 to (35.center)
			 to [in=180, out=-135, looseness=1.25] (36.center)
			 to (38.center)
			 to [in=-45, out=0, looseness=1.25] (37.center)
			 to (33.center)
			 to [bend right=45, looseness=0.75] cycle;
		\draw [style=green filled] (39.center)
			 to [bend right=45, looseness=0.75] (40.center)
			 to (41.center)
			 to [in=90, out=135, looseness=1.25] (42.center)
			 to (44.center)
			 to [in=-135, out=-90, looseness=1.25] (43.center)
			 to cycle;
		\draw [style=orange filled] (45.center)
			 to [bend right=45, looseness=0.75] (46.center)
			 to (47.center)
			 to [in=-90, out=-45, looseness=1.25] (48.center)
			 to (50.center)
			 to [in=45, out=90, looseness=1.25] (49.center)
			 to cycle;
	\end{pgfonlayer}
\end{tikzpicture}}
    \caption[An example hypergraph]{
    An example of a $3$-uniform hypergraph with $6$ vertices and $4$ edges
    }
    \label{fig:hypergraphExample}
\end{figure}
Let $\graphh=(\vertexset, \edgeset, \weight)$ be a hypergraph with $n=\cardinality{\vertexset}$ vertices.
For any vertex $v\in \vertexset$, the \firstdef{degree} of $v$ is defined by $\deg(v) \triangleq \sum_{e \in \edgeset} \weight(e)\cdot \iverson{v\in e}$, recalling that $\iverson{X}=1$ if event $X$ holds and $\iverson{X}=0 $ otherwise.
The \firstdef{rank} of edge $e\in \edgeset$ is defined to be the total number of vertices in $e$ and is written as $\rank(e)$.
A hypergraph is said to be \firstdef{$\bm{r}$-uniform} if every edge $e \in \edgeset$ has rank $r$.

\section{Asymptotic Notation} \label{sec:prelim:algorithms}
When studying the running time of algorithms, we are often interested in how the running time grows as a function of the size of the algorithm's input.
A useful tool to this end is \firstdef{asymptotic notation}.
Suppose we are given two functions $f, g: \R \rightarrow \R$.
Then, we say that
\[
    f(x) = \bigo{g(x)}
\]
if there exist constants $c, t \in \R$ such that for all $x \geq t$, it holds that $f(x) \leq c \cdot g(x)$.
For example, if $f(n)$ gives the running time of some algorithm for an input of size $n$, then $f(n) = \bigo{n}$ tells us that for sufficiently large inputs, the running time of the algorithm is at most linear in $n$.
On the other hand,
\[
    f(x) = \bigomega{g(x)}
\]
tells us that there exist constants $c, t \in \R$ such that for all $x \geq t$, it holds that $f(x) \geq c \cdot g(x)$.
In other words, $f(x) = \bigomega{g(x)}$ if and only if $g(x) = \bigo{f(x)}$.
Finally, we can say that
\[
    f(x) = \bigtheta{g(x)}
\]
if $f(x) = \bigo{g(x)}$ and $f(x) = \bigomega{g(x)}$.
When discussing the running time of algorithms, it is common to use $\poly(n)$ to represent some polynomial in $n$ and $\polylog(n)$ for some  polynomial in $\log(n)$.
Intuitively, $\bigo{\cdot}$ gives an upper-bound on some function, and hides a constant factor.
We can take this one step further and use $\bigotilde{\cdot}$ to hide constants and logarithmic factors.
For example, suppose $f(x) = 100 \cdot x \log(x)$. Then, we could say that $f(x) = \bigo{x \log(x)}$ or $f(x) = \bigotilde{x}$.

Occasionally, we will use asymptotic notation in order to hide a fixed constant when we are not discussing
the asymptotic behaviour of a function.
For example, $x = \Omega(y)$ means that there exists a universal constant $c$ such that $x \geq c y$.
This use of asymptotic notation is common throughout theoretical computer science~\cite{leeMultiwaySpectralPartitioning2014, pengPartitioningWellClusteredGraphs2017}.

\section{Clustering Metrics} \label{sec:prelim:metrics}
When evaluating the performance of a clustering algorithm, we would like to measure the similarity between the ground truth clusters and the clusters returned by the algorithm.
The \firstdef{Rand Index}~\cite{randObjectiveCriteriaEvaluation1971}, also referred to as the \firstdef{clustering accuracy}, is a common metric used for this purpose.

Suppose we are given a dataset with $n$ data points and $k$ clusters.
The ground truth clusters are given by $\setc_1, \ldots, \setc_k$ and the clusters returned by our algorithm are $\sets_1, \ldots, \sets_k$.
Let $\sigma_{\mathrm{C}}: [n] \xrightarrow[]{} [k]$ be a function which maps each data point to its ground truth cluster.
That is, if $j \in \setc_i$ then $\sigma_{\mathrm{C}}(j) = i$.
Let $\sigma_{\mathrm{S}}: [n] \xrightarrow[]{} [k]$ be the equivalent function for the clusters $\sets_1, \ldots, \sets_k$.
Then, let
\[
    \mathrm{TP} \triangleq \cardinality{\{(i, j) \in [n] \times [n] : i \neq j \land \sigma_{\mathrm{C}}(i) = \sigma_{\mathrm{C}}(j) \land \sigma_{\mathrm{S}}(i) = \sigma_{\mathrm{S}}(j) \}}
\]
be the number of pairs of data points which are in the same ground truth cluster and the same cluster returned by the algorithm.
Similarly, let 
\[
    \mathrm{TN} \triangleq \cardinality{\{(i, j) \in [n] \times [n] : i \neq j \land \sigma_{\mathrm{C}}(i) \neq \sigma_{\mathrm{C}}(j) \land \sigma_{\mathrm{S}}(i) \neq \sigma_{\mathrm{S}}(j) \}}
\]
be the number of pairs of data points which the algorithm correctly identifies as belonging to different clusters.
Then, the Rand Index is
\[
    \mathrm{RI} \triangleq \frac{\mathrm{TP} + \mathrm{TN}}{\binom{n}{2}}.
\]
That is, the Rand Index measures the proportion of pairs of data points which the algorithm correctly classifies as belonging to the same cluster or different clusters.
While this is a natural definition, notice that even a baseline algorithm which assigns data points to clusters at random will achieve a non-zero Rand Index.
To correct for this, the \firstdef{Adjusted Rand Index}~\cite{gatesImpactRandomModels2017} normalises the Rand Index by the expected value for a random assignment.
Further defining
\[
    \mathrm{FP} \triangleq \cardinality{\{(i, j) \in [n] \times [n] : i \neq j \land \sigma_{\mathrm{C}}(i) \neq \sigma_{\mathrm{C}}(j) \land \sigma_{\mathrm{S}}(i) = \sigma_{\mathrm{S}}(j) \}}
\]
and 
\[
    \mathrm{FN} \triangleq \cardinality{\{(i, j) \in [n] \times [n] : i \neq j \land \sigma_{\mathrm{C}}(i) = \sigma_{\mathrm{C}}(j) \land \sigma_{\mathrm{S}}(i) \neq \sigma_{\mathrm{S}}(j) \}},
\]
the Adjusted Rand Index is given by
\[
    \mathrm{ARI} \triangleq \frac{2 \left( \mathrm{TP} \cdot \mathrm{TN} - \mathrm{FP} \cdot \mathrm{FN} \right) }{ ( \mathrm{TP} + \mathrm{FN}) (\mathrm{TN} + \mathrm{FN}) + (\mathrm{TN} + \mathrm{FP}) (\mathrm{TP} + \mathrm{FP})}.
\]
In contrast with the Rand Index, which gives a value between $0$ and $1$, the Adjusted Rand Index can give negative values.
A random assignment of data points to clusters will give an expected Adjusted Rand Index of $0$.

\chapter{Related Work} \label{chap:related}
The goal of this chapter is to provide a broad overview of the techniques related to the algorithms developed later in the thesis,  and to serve as an introduction to spectral graph partitioning, local graph clustering, and spectral hypergraph theory.
As well as placing the novel work of this thesis in the context of the related work in the field, this chapter also introduces several technical concepts which are built on in the later chapters.
The structure of the chapter is given below.
\begin{itemize}
    \item Section~\ref{sec:related:graph_partitioning} gives an introduction to spectral graph partitioning, and describes several key points in the development of the field which underpins this thesis.
    \item Section~\ref{sec:related:structured} describes some recently developed techniques for finding clusters in graphs with specific structures of clusters.
    \item Section~\ref{sec:related:local} introduces local graph clustering and describes two key algorithms based on the personalised \pagerank\ and evolving set process, respectively.
    \item Section~\ref{sec:related:hypergraph_spectral} gives a brief overview of spectral hypergraph theory, with a focus on the hypergraph Laplacian operator recently developed by Chan~\etal~\cite{chanSpectralPropertiesHypergraph2018}.
\end{itemize}

\section{Spectral Graph Partitioning} \label{sec:related:graph_partitioning}
Spectral graph partitioning is the use of the eigenvectors of graph matrices to help with finding partitions (or clusters) in graphs.
The roots of this idea stretch back to the 1970s, when Donath and Hoffman~\cite{donathAlgorithmsPartitioningGraphs1972} first proposed to use the eigenvectors of the adjacency matrix for partitioning, although they gave no theoretical justification for this choice.

\subsection{Fiedler Vector}
In 1973, Fiedler~\cite{fiedlerAlgebraicConnectivityGraphs1973} proved a number of results relating the second smallest eigenvalue of the graph Laplacian to the combinatorial connectivity of the graph.
This gave the first theoretical justification for using eigenvectors in graph partitioning.
Many authors refer to the eigenvector corresponding to the second smallest eigenvalue of the Laplacian as the \firstdef{Fiedler vector} of the graph.
The corresponding eigenvalue is sometimes referred to as the \firstdef{algebraic connectivity} of the graph.

\subsection{Cheeger Inequality}
In parallel to this, Cheeger~\cite{cheegerLowerBoundSmallest1970} studied the connectivity of Riemannian manifolds, and proved an inequality relating the second smallest eigenvalue of the Laplacian operator to the connectivity of the manifold.
The significance of Cheeger's result for graph partitioning was demonstrated several years later, when Dodziuk~\cite{dodziukDifferenceEquationsIsoperimetric1984} and Alon~\cite{alonEigenvaluesExpanders1986} independently proved (slightly weaker) versions of Cheeger's inequality in the graph setting.
Chung~\cite{chungLaplaciansGraphsCheeger1996} later tightened the result and proved the following theorem, which is commonly referred to as the \firstdef{Cheeger inequality} in the context of spectral graph theory.

\begin{theorem} [Cheeger Inequality] \label{thm:cheeger}
Let $\graphg$ be a graph and $\lapn_\graphg$ be its normalised Laplacian matrix with eigenvalues $0 = \lambda_1 \leq \lambda_2 \leq \ldots \leq \lambda_n$.
Then,
\[
    \frac{\lambda_2}{2} \leq \cond_\graphg \leq \sqrt{2 \lambda_2},
\]
where $\cond_\graphg$ is the conductance of $G$.
\end{theorem}

The theorem relates the second smallest eigenvalue of the graph Laplacian to the conductance of the graph.
While the theorem itself already provides significant evidence that the Laplacian spectrum may be helpful when partitioning graphs, its proof gives even more.
The proof of the upper bound of Theorem~\ref{thm:cheeger} is constructive: it presents an algorithm which, given a graph $\geqve$, finds a set $\sets \subset \vertexset$ such that $\cond(\sets) \leq \sqrt{2 \lambda_2}$.
The algorithm is formally described in Algorithm~\ref{alg:sweepset}, and for any given graph the partition returned by this algorithm is often called the \firstdef{Cheeger cut} of the given graph.

\begin{algorithm}[b]
\caption[The sweep set procedure: \algsweepset$(\graphg)$]{\algsweepset$(\graphg)$ \label{alg:sweepset}}
Compute the second smallest eigenvalue $\lambda_2$ of $\lapn_\graphg$, and its eigenvector $\vecf_2$ \\
 Sort the vertices $v_1, \ldots, v_n \in \vertexset_\graphg$ such that $\vecf_2(v_1) / \sqrt{\deg(v_1)} \leq \ldots \leq \vecf_2(v_n) / \sqrt{\deg(v_n)} $ \\
Let $i = \argmin_{i \in [n-1]} \cond_\graphg(\{v_1, \ldots, v_i\})$ \\
Return $\{v_1, \ldots, v_i\}$
\end{algorithm}

It is worth noting that both sides of the Cheeger inequality are tight up to a constant, as demonstrated in Examples~\ref{ex:cheeg_lower}~and~\ref{ex:cheeg_upper}.

\vspace{1em}

\noindent
\begin{minipage}{0.65\textwidth}
\begin{example} [Tightness of Cheeger Inequality lower bound]\label{ex:cheeg_lower}
Let $\graphg$ be the $d$-dimensional hypercube graph.
Then,
\[
    \lambda_2 = \frac{2}{d}
\]
and
\[
    \cond(\graphg) = \frac{2^{d-1}}{d \cdot  2^{d-1}} = \frac{1}{d} = \frac{\lambda_2}{2},
\]
which demonstrates that the lower bound is tight.
\end{example}
\end{minipage}
\begin{minipage}{0.3\textwidth}
\vspace{0.5em}
\centering
\captionsetup{type=figure}
        \scalebox{0.9}{\begin{tikzpicture}
	\begin{pgfonlayer}{nodelayer}
		\node [style=med black] (0) at (-1, 2) {};
		\node [style=med black] (1) at (1, 2) {};
		\node [style=med black] (2) at (-1, 0) {};
		\node [style=med black] (3) at (1, 0) {};
		\node [style=med black] (4) at (2.5, 4) {};
		\node [style=med black] (5) at (4.5, 4) {};
		\node [style=med black] (6) at (2.5, 2) {};
		\node [style=med black] (7) at (4.5, 2) {};
		\node [style=none] (8) at (1.875, 5) {};
		\node [style=none] (9) at (1.875, -0.75) {};
	\end{pgfonlayer}
	\begin{pgfonlayer}{edgelayer}
		\draw [style=undirected] (0) to (4);
		\draw [style=undirected] (1) to (5);
		\draw [style=undirected] (2) to (6);
		\draw [style=undirected] (3) to (7);
		\draw [style=undirected] (1) to (3);
		\draw [style=undirected] (0) to (2);
		\draw [style=undirected] (2) to (3);
		\draw [style=undirected] (0) to (1);
		\draw [style=undirected] (4) to (6);
		\draw [style=undirected] (6) to (7);
		\draw [style=undirected] (4) to (5);
		\draw [style=undirected] (5) to (7);
		\draw [style=red dashed] (8.center) to (9.center);
	\end{pgfonlayer}
\end{tikzpicture}
}
        
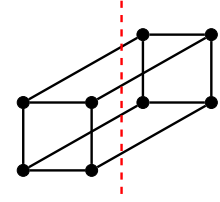
\captionof{figure}[A hypercube graph and the cut with minimum conductance]{
        A hypercube graph and the cut with minimum conductance.
        }
\end{minipage}

\vspace{1em}

\noindent
\begin{minipage}{0.65\textwidth}
\begin{example} [Tightness of Cheeger Inequality upper bound]\label{ex:cheeg_upper}
Let $\graphg$ be the cycle graph on $n$ vertices.
Then,
\[
    \lambda_2 = 2 - 2 \cos\left(\frac{2 \pi}{n}\right) = \bigo{\frac{1}{n^2}},
\]
using the fact that $1 - \cos(x) = 2 \sin^2(x/2)$ and $\sin(x) \leq x$.
Moreover, we have that
\[
    \cond_\graphg = \frac{2}{n} = \bigomega{\sqrt{\lambda_2}},
\]
which shows that the Cheeger inequality is tight up to a constant.
\end{example}
\end{minipage}
\begin{minipage}{0.3\textwidth}
\centering
\captionsetup{type=figure}
\scalebox{0.9}{\begin{tikzpicture}
	\begin{pgfonlayer}{nodelayer}
		\node [style=med black] (1) at (0, 2.25) {};
		\node [style=med black] (2) at (2, 1.25) {};
		\node [style=med black] (3) at (2, -1) {};
		\node [style=med black] (4) at (0, -2.25) {};
		\node [style=med black] (5) at (-2, -1.25) {};
		\node [style=med black] (6) at (-2, 1) {};
		\node [style=none] (8) at (-1.5, 2.5) {};
		\node [style=none] (9) at (1.5, -2.5) {};
	\end{pgfonlayer}
	\begin{pgfonlayer}{edgelayer}
		\draw [style=undirected] (1) to (2);
		\draw [style=undirected] (2) to (3);
		\draw [style=undirected] (3) to (4);
		\draw [style=undirected] (4) to (5);
		\draw [style=undirected] (5) to (6);
		\draw [style=red dashed] (8.center) to (9.center);
		\draw [style=undirected] (1) to (6);
	\end{pgfonlayer}
\end{tikzpicture}
}

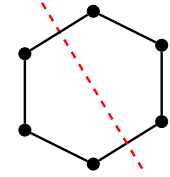
\captionof{figure}[A cycle graph and the cut with minimum conductance]{
A cycle graph and the cut with minimum conductance.}
\end{minipage}

\vspace{1em}

Finally, Kwok~\etal~\cite{kwokImprovedCheegerInequality2013a} gave the following improved version of the Cheeger inequality, which shows that $\lambda_2$ better approximates the conductance $\cond_\graphg$ when there is a large gap between $\lambda_2$ and $\lambda_3$. 
\begin{theorem}[Improved Cheeger Inequality]
Let $\graphg$ be a graph, and $\lapn_\graphg$ be its normalised Laplacian matrix with eigenvalues $0 = \lambda_1 \leq \lambda_2 \leq \ldots \leq \lambda_n$.
Then, it holds for any $2 \leq k \leq n$ that 
\[
    \cond_\graphg \leq \bigo{k} \frac{\lambda_2}{\sqrt{\lambda_k}},
\]
where $\cond_\graphg$ is the conductance of $G$.
\end{theorem}

\subsection{Higher-order Cheeger Inequality}
The Cheeger inequality clearly justifies the use of the second eigenvector of the graph Laplacian for partitioning a graph into two parts.
 This raises a natural question of whether the other eigenvalues give any information for partitioning the graph into more than two parts.
Lee \etal~\cite{leeMultiwaySpectralPartitioning2014} considered the \firstdef{$\bm{k}$-way expansion constant} of a graph, which generalises the conductance of a graph to a partitioning with more than two clusters.
\begin{definition}[$k$-way expansion]
Given a graph $\graphg$, the $k$-way expansion is defined as
\[
    \kcondg(k) \triangleq \min_{\mathrm{partition}~\sets_1, \ldots, \sets_k} \max_{i \in [k]} \cond_\graphg(\sets_i),
\]
 where the minimum is taken over all sets $\sets_1, \ldots, \sets_k \subset \vertexsetg$ where $\sets_i \intersect \sets_j = \emptyset$ for all $i \neq j$ and $\bigcup_{i = 1}^k \sets_i = \vertexsetg$.
\end{definition}
It is clear that $\kcondg(2) = \cond_\graphg$, and so this generalises the definition of graph conductance.
Lee \etal~\cite{leeMultiwaySpectralPartitioning2014} proved the following theorem, which relates the $k$-way expansion of a graph to the $k$-th smallest eigenvalue of the normalised Laplacian.

\begin{theorem}[Higher-order Cheeger Inequality] \label{thm:higher_order_cheeger}
Let $\geqve$ be a graph, and $\lapn_\graphg$ be its normalised Laplacian matrix with eigenvalues $0 = \lambda_1 \leq \lambda_2 \leq \ldots \leq \lambda_n$.
Then,
\[
    \frac{\lambda_k}{2} \leq \kcondg(k) \leq \bigo{k^3} \sqrt{\lambda_k},
\]
where $\kcondg(k)$ is the $k$-way expansion constant of $\graphg$.
\end{theorem}
Once again, the proof of this theorem is constructive, and gives an algorithm for finding clusters $\sets_1, \ldots, \sets_k$ satisfying the inequalities.
The dependence on $k$ in the upper bound of Theorem~\ref{thm:higher_order_cheeger} is needed, but it is not known if the inequality is tight~\cite{leeMultiwaySpectralPartitioning2014}.

\subsection{Spectral Clustering} \label{sec:spectralclusteringworks}
The spectral clustering algorithm was proposed over a decade before the higher-order Cheeger inequality~\cite{ngSpectralClusteringAnalysis2001},  and is widely used in practice due to its simplicity and empirical effectiveness.
 There are several variants of this algorithm since
it is possible to use the adjacency matrix, Laplacian matrix, or normalised Laplacian matrix for constructing the spectral embedding~\cite{vonluxburgTutorialSpectralClustering2007}.
However, they are all composed of the same three steps:
\begin{enumerate}
    \item Compute $k$ eigenvectors of some graph matrix;
    \item Embed the vertices of the graph into $\R^k$ according to the values of the eigenvectors;
    \item Apply a simple clustering algorithm, such as $k$-means, to obtain the final clusters.
\end{enumerate}
See Algorithm~\ref{alg:sc} for the formal description of spectral clustering.
\begin{algorithm}[hb]
\caption[Classical spectral clustering: \algspectralcluster$(\graphg, k)$]{\algspectralcluster$(\graphg, k)$ \label{alg:sc}}
Construct the normalised graph Laplacian matrix $\lapn_\graphg$\\
Compute eigenvectors $\vecf_1, \ldots, \vecf_k$ of $\lapn_\graphg$ corresponding to the $k$ smallest eigenvalues \\
Embed each $v \in \vertexsetg$ to the point $(\vecf_1(v), \ldots, \vecf_k(v))^\transpose$ \\
Apply $k$-means to the embedded vertices; partition $\vertexsetg$ based on the output
\end{algorithm}

Although spectral clustering has been used in practice for several decades,
 the theoretical justification for its good performance has seen some recent progress~\cite{pengPartitioningWellClusteredGraphs2017} and deserves further study.
There are two technical approaches which are common when analysing spectral clustering.
The first is to see that the algorithm solves a relaxation of a natural optimisation problem.
For more details on this perspective, the reader can refer to the excellent tutorial of Von Luxburg~\cite{vonluxburgTutorialSpectralClustering2007}.
The second comes from considering the structure of clusters in the eigenspace of the graph matrices.
Ng~\etal~\cite{ngSpectralClusteringAnalysis2001} introduced this idea with the following observations, which are illustrated in Figure~\ref{fig:scExample}.
\begin{enumerate}
    \item If the clusters in the input graph $G$ are disconnected, then the normalised Laplacian matrix $\lapn_\graphg$ is block diagonal.
    Therefore, the first $k$ eigenvectors of $\lapn_\graphg$ are the indicator vectors of the connected components.
    \item If the clusters are then connected by a small number of edges, then the Laplacian matrix doesn't change much, and accordingly the eigenvalues and eigenvectors of the Laplacian do not change too much.
    \item As such, if a graph possesses a clear cluster structure, then the clusters are well separated in the spectral embedding.
\end{enumerate}

\begin{figure}[t]
    \centering
    \begin{tikzpicture}
	\begin{pgfonlayer}{nodelayer}
		\node [style=smallblue] (0) at (-7.425, 1.65) {};
		\node [style=smallblue] (1) at (-6.65, 2.1) {};
		\node [style=smallblue] (2) at (-6, 1.575) {};
		\node [style=smallblue] (3) at (-6.3, 0.75) {};
		\node [style=smallblue] (4) at (-7.25, 0.775) {};
		\node [style=none] (25) at (-0.5, 0) {$\xrightarrow{}$};
		\node [style=none] (26) at (4, 0) {$\begin{bmatrix}\mat{B}_1 & \mat{0} & \mat{0} \\ \mat{0} & \mat{B}_2  & \mat{0} \\ \mat{0} & \mat{0} & \mat{B}_3 \end{bmatrix}$};
		\node [style=none] (27) at (8.75, 0) {$\xrightarrow{}$};
		\node [style=none] (28) at (-5, -3.5) {Graph $\graphg$};
		\node [style=none] (29) at (4, -3.5) {Normalised Laplacian $\lapn$};
		\node [style=none] (30) at (13.25, -2.25) {};
		\node [style=none] (31) at (13.25, 2.25) {};
		\node [style=none] (32) at (11, 0) {};
		\node [style=none] (33) at (15.5, 0) {};
		\node [style=med blue] (34) at (11.75, 1.25) {};
		\node [style=med red] (35) at (14.75, 1.25) {};
		\node [style=med green] (36) at (13.25, -1.5) {};
		\node [style=none] (37) at (13.25, -3.5) {Spectral Embedding};
		\node [style=smallred] (38) at (-3.75, 1.65) {};
		\node [style=smallred] (39) at (-2.975, 2.1) {};
		\node [style=smallred] (40) at (-2.325, 1.575) {};
		\node [style=smallred] (41) at (-2.625, 0.75) {};
		\node [style=smallred] (42) at (-3.575, 0.775) {};
		\node [style=small green] (43) at (-5.6, -1.1) {};
		\node [style=small green] (44) at (-4.825, -0.65) {};
		\node [style=small green] (45) at (-4.175, -1.175) {};
		\node [style=small green] (46) at (-4.475, -2) {};
		\node [style=small green] (47) at (-5.425, -1.975) {};
		\node [style=smallblue] (48) at (-7.5, -5.85) {};
		\node [style=smallblue] (49) at (-6.725, -5.4) {};
		\node [style=smallblue] (50) at (-6.075, -5.925) {};
		\node [style=smallblue] (51) at (-6.375, -6.75) {};
		\node [style=smallblue] (52) at (-7.325, -6.725) {};
		\node [style=none] (53) at (-0.575, -7.75) {$\xrightarrow{}$};
		\node [style=none] (54) at (3.925, -7.75) {$\begin{bmatrix}\mat{B}_1 & \mat{C}_{12} & \mat{C}_{13} \\ \mat{C}_{12}^\transpose & \mat{B}_2  & \mat{C}_{23} \\ \mat{C}_{13}^\transpose & \mat{C}_{23}^\transpose & \mat{B}_3 \end{bmatrix}$};
		\node [style=none] (55) at (8.675, -7.75) {$\xrightarrow{}$};
		\node [style=none] (56) at (-5.075, -11.25) {Graph $\graphg$};
		\node [style=none] (57) at (3.925, -11.25) {Normalised Laplacian $\lapn$};
		\node [style=none] (58) at (13.175, -10) {};
		\node [style=none] (59) at (13.175, -5.5) {};
		\node [style=none] (60) at (10.925, -7.75) {};
		\node [style=none] (61) at (15.425, -7.75) {};
		\node [style=none] (65) at (13.175, -11.25) {Spectral Embedding};
		\node [style=smallred] (66) at (-3.825, -6.1) {};
		\node [style=smallred] (67) at (-3.05, -5.65) {};
		\node [style=smallred] (68) at (-2.4, -6.175) {};
		\node [style=smallred] (69) at (-2.7, -7) {};
		\node [style=smallred] (70) at (-3.65, -6.975) {};
		\node [style=small green] (71) at (-5.675, -8.85) {};
		\node [style=small green] (72) at (-4.9, -8.4) {};
		\node [style=small green] (73) at (-4.25, -8.925) {};
		\node [style=small green] (74) at (-4.55, -9.75) {};
		\node [style=small green] (75) at (-5.5, -9.725) {};
		\node [style=none] (76) at (-9, -4.5) {};
		\node [style=none] (77) at (17, -4.5) {};
		\node [style=smallblue] (78) at (11.525, -6.225) {};
		\node [style=smallblue] (79) at (12.025, -5.85) {};
		\node [style=smallblue] (80) at (12.25, -6.625) {};
		\node [style=smallblue] (81) at (12, -6.35) {};
		\node [style=smallblue] (82) at (11.725, -6.8) {};
		\node [style=smallred] (83) at (14, -6.225) {};
		\node [style=smallred] (84) at (14.55, -5.925) {};
		\node [style=smallred] (85) at (14.95, -6.275) {};
		\node [style=smallred] (86) at (14.475, -6.35) {};
		\node [style=smallred] (87) at (14.3, -6.725) {};
		\node [style=small green] (88) at (12.85, -8.75) {};
		\node [style=small green] (89) at (13.3, -8.55) {};
		\node [style=small green] (90) at (13.7, -8.775) {};
		\node [style=small green] (91) at (13.425, -9.075) {};
		\node [style=small green] (92) at (13, -9.225) {};
	\end{pgfonlayer}
	\begin{pgfonlayer}{edgelayer}
		\draw [style=undirected] (0) to (1);
		\draw [style=undirected] (1) to (2);
		\draw [style=undirected] (2) to (3);
		\draw [style=undirected] (3) to (4);
		\draw [style=undirected] (4) to (0);
		\draw [style=undirected] (4) to (1);
		\draw [style=undirected] (0) to (2);
		\draw [style=undirected] (3) to (1);
		\draw [style=undirected] (0) to (3);
		\draw [style=directed] (30.center) to (31.center);
		\draw [style=directed] (32.center) to (33.center);
		\draw [style=undirected] (38) to (39);
		\draw [style=undirected] (39) to (40);
		\draw [style=undirected] (40) to (41);
		\draw [style=undirected] (41) to (42);
		\draw [style=undirected] (42) to (38);
		\draw [style=undirected] (42) to (39);
		\draw [style=undirected] (38) to (40);
		\draw [style=undirected] (41) to (39);
		\draw [style=undirected] (38) to (41);
		\draw [style=undirected] (43) to (44);
		\draw [style=undirected] (44) to (45);
		\draw [style=undirected] (45) to (46);
		\draw [style=undirected] (46) to (47);
		\draw [style=undirected] (47) to (43);
		\draw [style=undirected] (47) to (44);
		\draw [style=undirected] (43) to (45);
		\draw [style=undirected] (46) to (44);
		\draw [style=undirected] (43) to (46);
		\draw [style=undirected] (48) to (49);
		\draw [style=undirected] (49) to (50);
		\draw [style=undirected] (50) to (51);
		\draw [style=undirected] (51) to (52);
		\draw [style=undirected] (52) to (48);
		\draw [style=undirected] (52) to (49);
		\draw [style=undirected] (48) to (50);
		\draw [style=undirected] (51) to (49);
		\draw [style=undirected] (48) to (51);
		\draw [style=directed] (58.center) to (59.center);
		\draw [style=directed] (60.center) to (61.center);
		\draw [style=undirected] (66) to (67);
		\draw [style=undirected] (67) to (68);
		\draw [style=undirected] (68) to (69);
		\draw [style=undirected] (69) to (70);
		\draw [style=undirected] (70) to (66);
		\draw [style=undirected] (70) to (67);
		\draw [style=undirected] (66) to (68);
		\draw [style=undirected] (69) to (67);
		\draw [style=undirected] (66) to (69);
		\draw [style=undirected] (71) to (72);
		\draw [style=undirected] (72) to (73);
		\draw [style=undirected] (73) to (74);
		\draw [style=undirected] (74) to (75);
		\draw [style=undirected] (75) to (71);
		\draw [style=undirected] (75) to (72);
		\draw [style=undirected] (71) to (73);
		\draw [style=undirected] (74) to (72);
		\draw [style=undirected] (71) to (74);
		\draw [style=undirected] (76.center) to (77.center);
		\draw [style=undirected] (51) to (70);
		\draw [style=undirected] (51) to (72);
		\draw [style=undirected] (52) to (71);
		\draw [style=undirected] (72) to (70);
		\draw [style=undirected] (73) to (69);
		\draw [style=undirected] (50) to (66);
	\end{pgfonlayer}
\end{tikzpicture}
    \caption[Why does spectral clustering work?]{
    An illustration of the intuition behind spectral clustering.
    If the clusters are disconnected, then the Laplacian is block diagonal and the spectral embedding perfectly separates the clusters.
    If a small number of edges are added between the clusters, then the embedding doesn't change `too much'.
    }
    \label{fig:scExample}
\end{figure}
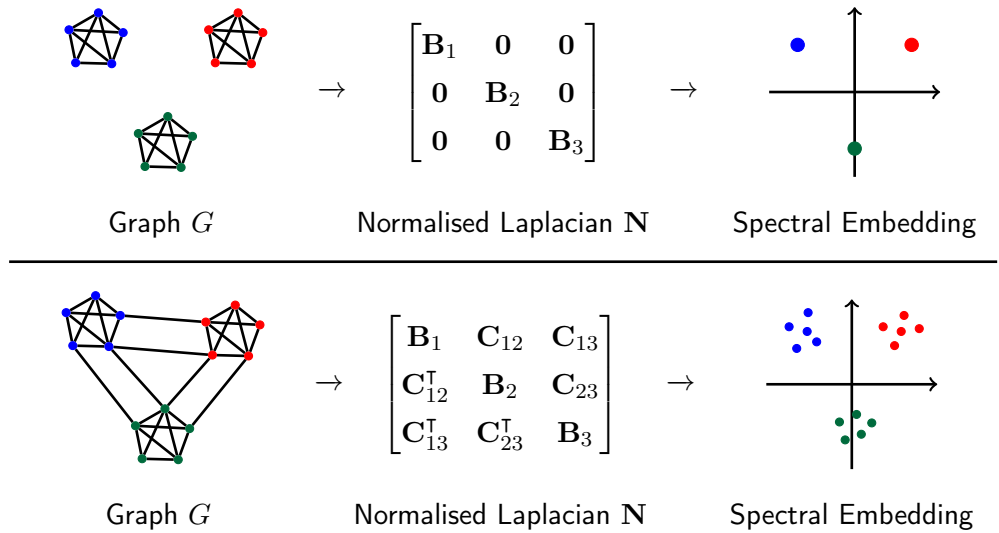

Peng~\etal~\cite{pengPartitioningWellClusteredGraphs2017} built on this intuition and proved that, not only are the clusters well separated in the spectral embedding, but the output of the $k$-means algorithm is guaranteed to correctly cluster most of the vertices.
To quantify the `clusterability' of a given graph $G$, they introduced the parameter $\Upsilon(k)$ defined as follows:

\begin{definition}
For a graph $\graphg$, let $\lambda_1 \leq \lambda_2 \leq \ldots \leq \lambda_n$ be the eigenvalues of the normalised Laplacian matrix $\lapn_\graphg$. Then,
\[
    \Upsilon(k) \triangleq \frac{\lambda_{k+1}}{\kcond(k)},
\]
where $\kcond(k)$ is the $k$-way expansion constant.
\end{definition}
 To understand the motivation behind this definition, suppose that a graph $\geqve$ contains \emph{exactly} $k$ clusters with low conductance.
That is,
\begin{enumerate}
    \item there is a partition of $\vertexset$ into $k$ sets $\sets_1, \ldots, \sets_k$ such that $\max_{i \in [k]} \cond(\sets_i)$ is low, and
    \item for \emph{any} partition of $\vertexset$ into $k + 1$ sets $\sets_1, \ldots, \sets_{k+1}$, there is some $i \in [k+1]$ such that $\cond(\sets_i)$ is high.
\end{enumerate}
Then, by the definition of $\kcond(k)$, the denominator of $\Upsilon(k)$ is low, and by the higher-order Cheeger inequality (Theorem~\ref{thm:higher_order_cheeger}), the numerator is high.
Thus, a large value of $\Upsilon(k)$ indicates that there are exactly $k$ low-conductance clusters in the graph.
The definition of $\Upsilon(k)$ also closely resembles the widely-used eigen-gap heuristic for spectral clustering~\cite{ngSpectralClusteringAnalysis2001, vonluxburgTutorialSpectralClustering2007}, which suggests that spectral clustering works when the value of $\abs{\lambda_{k+1} - \lambda_k}$ is much larger than $\abs{\lambda_k - \lambda_{k-1}}$.

Peng~\etal~\cite{pengPartitioningWellClusteredGraphs2017} prove that for graphs with a high value of $\Upsilon(k)$, the eigenvectors of the Laplacian are close to the indicator vectors of the optimal clusters.
Let $\sets_1, \ldots, \sets_k \subset \vertexset$ be the optimal clusters which achieve the $k$-way expansion $\kcond(k)$, then the \firstdef{normalised indicator vector} of each cluster $\sets_i$ is defined as
\[
    \barg_i = \frac{\degmhalf \indicatorvec_{\sets_i}}{\norm{\degmhalf \indicatorvec_{\sets_i}}}.
\]
Their result is given in the following theorem.
\begin{theorem}[\cite{pengPartitioningWellClusteredGraphs2017}, Theorem 1.1] \label{thm:pengStructure}
Let $\graphg$ be a graph with $\Upsilon(k) = \bigomega{k^2}$, and let $\vecf_1, \ldots, \vecf_n$ be the eigenvectors of $\lapn_\graphg$ corresponding to eigenvalues $\lambda_1 \leq \ldots \leq \lambda_n$. Then,
\begin{enumerate}
    \item for every $\barg_i$, there is a linear combination of $\{\vecf_i\}_{i = 1}^k$ called $\hatf_i$ such that $\norm{\barg_i - \hatf_i} \leq 1 / \Upsilon(k)$;
    \item for every $\vecf_i$, there is a linear combination of $\{\barg_i\}_{i = 1}^k$, called $\hatg_i$, such that $\norm{\vecf_i - \hatg_i}^2 \leq 1.1 k / \Upsilon(k)$.
\end{enumerate}
\end{theorem}
 Peng~\etal~\cite{pengPartitioningWellClusteredGraphs2017} used this to prove an upper bound on the number of vertices misclassified by the $k$-means clustering step of the spectral clustering algorithm.
\begin{theorem}[\cite{pengPartitioningWellClusteredGraphs2017}, Theorem 1.2]
Let $\graphg$ be a graph with $\Upsilon(k) = \bigomega{k^3}$. Let $\seta_1, \ldots, \seta_k$ be the clusters found by spectral clustering, and let the optimal correspondent of $\seta_i$ be $\sets_i$. Then, it holds for any $i\in[k]$ that
\[
    \frac{\vol(\seta_i \triangle \sets_i)}{\vol(\sets_i)} = \bigo{k^3 / \Upsilon(k)}.
\]
\end{theorem}
 
Kolev and Melhorn~\cite{kolevNoteSpectralClustering2016} and Mizutani~\cite{mizutaniImprovedAnalysisSpectral2021} developed similar results under the weaker assumptions of $\Upsilon(k) = \Omega(k^2)$ and $\Upsilon(k) = \Omega(k)$ respectively, although their results are not directly comparable since they use a slightly different definition of the gap $\Upsilon(k)$.

\section{Clustering Structured Graphs} \label{sec:related:structured}
Spectral clustering gives an excellent starting point for the subject of graph partitioning;
however, the classical analysis is usually based only on the assumption that the clusters are sparsely connected to each other and does not take into account the high-level structure of clusters.
As demonstrated in Chapter~\ref{chap:intro}, graphs created from real-world data often exhibit a high-level structure, and this section describes techniques for learning clusters in such graphs.

\subsection{Partitioning Almost-Bipartite Graphs}
Let us start by considering the case of graphs exhibiting \firstdef{heterophily}.
That is, graphs in which vertices have a tendency to share connections with those in a \emph{different} cluster to their own.\footnote{For completeness, the opposite of heterophily is \firstdef{homophily}, in which vertices are more likely to share edges with other vertices in the same cluster.}
An example of such a graph is the inter-state dispute graph introduced in Chapter~\ref{chap:intro}.
More generally, graphs in which edges represent `negative' relationships tend to exhibit heterophily.

A bipartite graph is a `perfect' example of a two-cluster graph with heterophily.
Of course, it is not difficult to separate a bipartite graph into its two clusters, but what if the graph is only \emph{almost} bipartite?
To answer this, Bauer and Jost~\cite{bauerBipartiteNeighborhoodGraphs2013} and Trevisan~\cite{trevisanMaxCutSmallest2012} introduced the \firstdef{bipartiteness} of two sets $\setl, \setr \subset \vertexsetg$, which measures how close the sets $\setl$ and $\setr$ are to forming a bipartite connected component.
\begin{definition} [Bipartiteness] \label{def:bipartiteness}
Given a graph $\geqve$ and two clusters $\setl, \setr \subset \vertexset$, the bipartiteness of $\setl$ and $\setr$ is
\[
\bipart(\setl, \setr) \triangleq 1 - \frac{2 \weight(\setl, \setr)}{\vol(\lur)}.
\]
The bipartiteness of the graph $\graphg$ is given by
\[
    \bipart_\graphg \triangleq \min_{\substack{\setl, \setr \subset \vertexsetg \\ \setl \intersect \setr = \emptyset}} \bipart(\setl, \setr).
\]
\end{definition}
Notice that $\bipart(\setl, \setr) = 0$ if $\setl$ and $\setr$ form a bipartite and connected component of $\graphg$.
Given a graph $\graphg$, let $0 = \lambdaeigs \leq 2$ be the eigenvalues of the normalised Laplacian $\lapn_\graphg$.
It is a classical result in spectral graph theory that $\lambda_n = 2$ if and only if $\bipart_\graphg = 0$, and so a natural question is whether there is an analogue of the Cheeger inequality which connects the bipartiteness of a graph with the largest eigenvalue $\lambda_n$.
Bauer and Jost~\cite{bauerBipartiteNeighborhoodGraphs2013} and Trevisan~\cite{trevisanMaxCutSmallest2012} independently showed that this is indeed the case with the following theorem.
\begin{theorem} \label{thm:trevisansinequality}
Given a graph $\graphg$, let $\lambdaeigs$ be the eigenvalues of the normalised Laplacian $\lapn_\graphg$. Then,
\[
    \frac{2 - \lambda_n}{2} \leq \bipart_\graphg \leq \sqrt{2 (2 - \lambda_n)},
\]
where $\bipart_\graphg$ is the bipartiteness of the graph.
\end{theorem}
Trevisan used this result to develop a novel approximation algorithm for the \textsc{MAX-CUT} problem.
Soto~\cite{sotoImprovedAnalysisMaxCut2015} improved Trevisan's analysis and showed that the algorithm achieves an approximation ratio of $0.614247$.
Liu~\cite{liuMultiwayDualCheeger2015} generalised Theorem~\ref{thm:trevisansinequality} and proved a higher-order version of this inequality which links the top eigenvalues of a graph Laplacian to the bipartiteness of multiple clusters in the graph.

The problem of learning the clusters in graphs exhibiting heterophily has also been studied from a supervised and semi-supervised learning perspective, such as in \cite{mooreActiveLearningNode2011, peiGeomGCNGeometricGraph2019, zhuHomophilyGraphNeural2020}.
These techniques are not covered in detail here since the focus of this thesis is on algorithms for unsupervised learning.

\subsection{Learning Clusters from Edge Directions}
In the preceding sections, we considered algorithms for learning clusters which are defined by
sparse~(or dense) connections. 
In this section, we consider directed graphs in which the clusters are defined by the \emph{directions} of their connecting edges.

 Cucuringu~\etal~\cite{cucuringuHermitianMatricesClustering2020} studied this problem and introduced the so-called \firstdef{cut imbalance} between two clusters $\sets_1, \sets_2 \subset \vertexset$, which measures the imbalance of edge directions between the clusters.
They proposed to use the Hermitian adjacency matrix of a directed graph to find clusters with a large cut imbalance. 
The Hermitian adjacency matrix of a directed graph $\graphg$ is the matrix $\adjh \in \C^{n \times n}$ where $\adjh(u, v) = \conjugate{\adjh(v, u)} = \imag$ if $\graphg$ contains an edge from $u$ to $v$ and $\adjh(u, v) = \adjh(v, u) = 0$ otherwise.
They develop a spectral clustering algorithm based on the complex-valued eigenvectors of $\adjh$.
For graphs generated from a directed stochastic block model, their algorithm provably recovers the ground truth clusters when the cut imbalance between clusters is sufficiently large.

Laenen and Sun~\cite{laenenHigherOrderSpectralClustering2020} further studied spectral clustering with the Hermitian adjacency matrix, and proved that if the clusters form a directed path structure, then they can be recovered using a \emph{single} complex-valued eigenvector.
In Chapter~\ref{chap:meta}, we   see that this can be generalised to undirected graphs and arbitrary structures of clusters.
Moreover, we   prove that when the clusters have a clearly defined structure, spectral clustering with fewer than $k$ eigenvectors performs better than spectral clustering with $k$ eigenvectors.

\section{Local Graph Clustering} \label{sec:related:local}
Let us now turn our attention to the problem of \firstdef{local graph clustering}, which was first introduced by Spielman and Teng~\cite{spielmanLocalClusteringAlgorithm2013}.
In the classical local clustering problem, we are given as input a large graph $\geqve$ and some starting vertex $v \in \vertexset$.
Our aim is to find a low-conductance cluster $\sets \subset \vertexset$ which contains the vertex $v$.
Moreover, the algorithm should have running time which depends on $\vol(\sets)$ and is independent of $\vol(\vertexset)$.
This sublinear running time makes local clustering a crucially important technique for studying massive datasets occurring in real-world applications.
This section introduces two techniques which form the basis for several local clustering algorithms.

\subsection{\pagerank\ Based Algorithm}
The first key local clustering technique that we   consider is known as the personalised \pagerank.
With its roots in the well-known \pagerank\ algorithm for ranking web pages~\cite{pagePageRankCitationRanking1999}, personalised \pagerank\ can be used to find vertices in a graph close to some starting vertex.
This technique has proved useful for many graph learning applications~\cite{andersenLocalGraphPartitioning2006, andersenLocalPartitioningDirected2007, arrigoNonbacktrackingPagerank2019, yinLocalHigherorderGraph2017, takaiHypergraphClusteringBased2020}.

\paragraph{Random walks.}
Before coming to the personalised \pagerank, we first discuss random walks on graphs.
In a graph $\geqve$, a random walk is a sequence of vertices $x_1, x_2, \ldots, x_n \in \vertexset$ where each $x_t$ is chosen uniformly at random from the neighbours of $x_{t-1}$.
See Figure~\ref{fig:randwalkExample} for an illustration.

\begin{figure}[t]
    \centering
    \begin{tikzpicture}
	\begin{pgfonlayer}{nodelayer}
		\node [style=med blue] (0) at (0.25, -1.5) {};
		\node [style=med blue] (1) at (0.25, 1) {};
		\node [style=med blue] (2) at (-2.25, 1) {};
		\node [style=med black] (3) at (-1, 3) {};
		\node [style=med red] (4) at (-2.25, -1.5) {};
		\node [style=med blue] (5) at (8.75, -1.5) {};
		\node [style=med red] (6) at (8.75, 1) {};
		\node [style=med blue] (7) at (6.25, 1) {};
		\node [style=med blue] (8) at (7.5, 3) {};
		\node [style=med blue] (9) at (6.25, -1.5) {};
		\node [style=none] (10) at (3.5, 0.75) {$\rightarrow$};
		\node [style=none] (11) at (-3, -1.5) {$v_1$};
		\node [style=none] (12) at (-3, 1) {$v_3$};
		\node [style=none] (13) at (-1.75, 3) {$v_5$};
		\node [style=none] (14) at (1, 1) {$v_4$};
		\node [style=none] (15) at (1, -1.5) {$v_2$};
		\node [style=med black] (16) at (16.75, -1.5) {};
		\node [style=med blue] (17) at (16.75, 1) {};
		\node [style=med blue] (18) at (14.25, 1) {};
		\node [style=med red] (19) at (15.5, 3) {};
		\node [style=med black] (20) at (14.25, -1.5) {};
		\node [style=none] (21) at (11.5, 0.75) {$\rightarrow$};
		\node [style=none] (22) at (-1, -3) {$x_1 = v_1$};
		\node [style=none] (23) at (7.5, -3) {$x_2 = v_4$};
		\node [style=none] (24) at (15.5, -3) {$x_3 = v_5$};
	\end{pgfonlayer}
	\begin{pgfonlayer}{edgelayer}
		\draw [style=undirected] (4) to (0);
		\draw [style=undirected] (0) to (1);
		\draw [style=undirected] (1) to (4);
		\draw [style=undirected] (4) to (2);
		\draw [style=undirected] (2) to (1);
		\draw [style=undirected] (3) to (1);
		\draw [style=undirected] (2) to (3);
		\draw [style=undirected] (9) to (5);
		\draw [style=undirected] (5) to (6);
		\draw [style=undirected] (6) to (9);
		\draw [style=undirected] (9) to (7);
		\draw [style=undirected] (7) to (6);
		\draw [style=undirected] (8) to (6);
		\draw [style=undirected] (7) to (8);
		\draw [style=undirected] (20) to (16);
		\draw [style=undirected] (16) to (17);
		\draw [style=undirected] (17) to (20);
		\draw [style=undirected] (20) to (18);
		\draw [style=undirected] (18) to (17);
		\draw [style=undirected] (19) to (17);
		\draw [style=undirected] (18) to (19);
	\end{pgfonlayer}
\end{tikzpicture}
    \caption[A random walk]{
    An example random walk starting with $x_1 = v_1$. At each time step, the current position is shown by the red vertex, and the next vertex of the walk is chosen randomly from the neighbours highlighted in blue.
    }
    \label{fig:randwalkExample}
\end{figure}
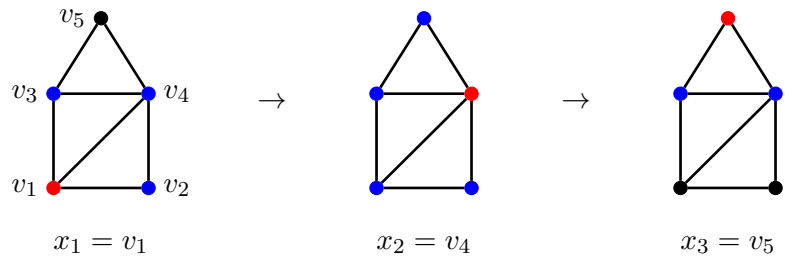

We can represent this process algebraically by defining the \firstdef{random walk matrix} to be
\[
    \walkm \triangleq \adj \degm^{-1}.
\]
Let $\vecp_t \in \R^n$ be the probability distribution at time $t$, where $\vecp_t(v)$ is the probability that $x_t = v$.
 Define $\vecp_0 \in \R^n$ to be some starting distribution such that $\sum_{i = 1}^n \vecp(i) = 1$.
 It is common to use $\vecp_0 = \indicatorvec_v$ for some vertex $v \in \vertexset$.
Then, we can write $\vecp_{t + 1}$ as
\[
    \vecp_{t + 1}  = \walkm \vecp_t.
\]
In order to ensure that the random walk converges, we instead use the \firstdef{lazy random walk} defined by the matrix
\[
    \lazywalkm \triangleq \frac{1}{2} \left(\identity + \walkm \right).
\]
This corresponds to a random walk process in which with probability $1/2$, the next vertex $x_{t + 1}$ is equal to $x_t$, and with probability $1/2$ the next vertex $x_{t+1}$ is chosen uniformly at random from the neighbours of $x_t$.
If the graph is connected, this random walk process always converges to the stationary distribution denoted by $\statdist$.~\footnote{Notice that this is not guaranteed when using $\walkm$. For example, if $\graphg$ is bipartite such that every edge connects $\setl \subset \vertexset$ to $\setr \subset \vertexset$, then $x_t$ always alternates between $\setl$ and $\setr$.}
In particular, by setting
\[
    \statdist(v) = \frac{\deg(v)}{\vol(\vertexset)}
\]
for every $v\in \vertexset$, it is easy to verify that
\[
    \lazywalkm \statdist = \statdist.
\]
Many algorithms for local graph clustering are based on the fact that a random walk starting inside some low conductance cluster is likely to stay inside the cluster.
We can study this through the \firstdef{escape probability}, which is the probability that a lazy random walk starting at $v$ leaves a set $\seta$ within the first $T$ steps and is defined as
\[
    \esc(v, T, \seta) \triangleq \p \left[\bigcup_{j = 0}^T x_j \not\in \seta  \middle\vert x_0 = v \right].
\]
Spielman and Teng~\cite{spielmanLocalClusteringAlgorithm2013} gave the following result which connects the escape probability to the conductance of the set $\sets$.
\begin{proposition}[\cite{spielmanLocalClusteringAlgorithm2013}, Proposition 2.5] \label{prop:escape_prob}
Let $G=(V,E)$ be a graph, and $A\subset V$. Then, there is a set $\seta_T \subseteq \seta$ with $\vol(\seta_T) \geq \frac{1}{2}\vol(\seta)$, such that it holds for any $x \in \seta_T$ that
\[ 
    \mathrm{esc}(x, T, \seta) \leq T \cond(\seta).
\]
\end{proposition}

\paragraph{Personalised \pagerank.}
In the local clustering problem, we are interested in finding a cluster in the graph $\graphg$ which is close to some starting vertex $v$.
In order to find such a cluster, let us consider a new random walk process called the \firstdef{teleporting random walk}. 
We first define
some starting distribution $\vecs \in \R^n$, and a \firstdef{teleport probability} $\alpha$.
Then, at time $t$,
\begin{enumerate}
    \item with probability $\alpha$, the next vertex $x_{t+1}$ is chosen according to the distribution $\vecs$;
    \item with probability $1 - \alpha$, the vertex $x_{t+1}$ is chosen according to the lazy random walk from $x_t$.
\end{enumerate}
If we take $\vecs$ to be the indicator vector of some vertex $v$, then at each step the random walk `teleports' to $v$ with probability $\alpha$.
We can represent this walk with the following update rule:
\[
    \vecp_{t+1} = \alpha \vecs + \left(1 - \alpha\right) \lazywalkm \vecp_t.
\]
The \firstdef{personalised \pagerank\ vector} $\ppr(\alpha, \vecs)$ is the stationary distribution of this random walk for some specified $\alpha$ and $\vecs$.
That is, one can view $\ppr: \R \times \R^n \rightarrow \R^n$ as a function such that, for any $\alpha$ and $\vecs$, $\ppr(\alpha, \vecs)$ is the unique solution to the equation
\[
    \ppr(\alpha, \vecs) = \alpha \vecs + (1 - \alpha) \lazywalkm \ppr(\alpha, \vecs).
\]
Andersen et al.~\cite{andersenLocalGraphPartitioning2006} showed that the personalised \pagerank\ vector can also be written as 
\[
    \ppr(\alpha,  \vecs) = \alpha \sum_{t = 0}^\infty (1 - \alpha)^t \lazywalkm^t \vecs.
\]
Therefore, we can also study $\ppr(\alpha, \vecs)$ through the following random process: pick some integer $t\in\mathbb{Z}_{\geq 0}$ with probability $\alpha(1-\alpha)^t$, and perform a $t$-step lazy random walk, where the starting vertex of the random walk is picked according to $\vecs$. Then, $\ppr(\alpha, \vecs)$ describes the probability of reaching each vertex in this process. 
Intuitively, we would expect the vertices which are `close' to the vertices in the starting distribution $\vecs$ to have a higher value in the $\ppr(\alpha, \vecs)$ vector.
 
\paragraph{Approximate personalised \pagerank.}
Our goal is to use the personalised \pagerank\ vector as the basis of local graph clustering algorithms. 
However, computing a personalised \pagerank\ vector $\ppr(\alpha, \vecs)$ exactly is equivalent to computing the stationary distribution of a Markov chain on the vertex set $\vertexset$, which has a time complexity of $\bigomega{n}$.
This is too slow for any local clustering algorithm,
whose running time is independent of the size of the input graph.

The good news is that, since the probability mass of a personalised \pagerank\ vector is concentrated around some starting vertex, it is possible to compute a good approximation of the \pagerank\ vector in a local way.
Andersen et al.~\cite{andersenLocalGraphPartitioning2006}
introduced the \firstdef{approximate \pagerank\ vector} for this purpose.
\begin{definition}[Approximate \pagerank]
    A vector $\vecp = \apr(\alpha, \vecs, \vecr)$ is called an approximate \pagerank\ vector if 
    $\vecp + \ppr(\alpha,  \vecr) = \ppr(\alpha,  \vecs)$.
    The vector $\vecr$ is called the residual vector.
\end{definition}

\paragraph{The \lovasz-Simonovits curve.}
The \lovasz-Simonovits curve was originally introduced by \lovasz\ and Simonovits~\cite{lovaszMixingRateMarkov1990} to reason about the mixing rate of Markov chains, and has been used extensively in the analysis of local clustering algorithms based on the personalised \pagerank~\cite{andersenLocalGraphPartitioning2006, zhuLocalAlgorithmFinding2013}.
For any vector $\vecp \in \R_{\geq 0}^{n}$, order the vertices such that $$\frac{\vecp(v_1)}{\deg(v_1)} \geq \frac{\vecp(v_2)}{\deg(v_2)} \geq \ldots \geq \frac{\vecp(v_n)}{\deg(v_n)},$$ 
and define the \firstdef{sweep sets} of $\vecp$ to be
$
    \pjsweep = \{v_1, \ldots, v_j\}
$ for $1\leq j\leq n$.
Then, the \firstdef{\lovasz-Simonovits curve}, denoted by $\vecp[x]$ for $x \in [0, \vol(\vertexset)]$, is defined by the points
\[
    \vecp[\vol(\pjsweep)] \triangleq \vecp(\pjsweep),
\]
and is linear between those points for consecutive $j$.  
Lov\'asz and Simonovits also showed that
\begin{equation} \label{eq:deflscurve}
    \lscurve{x} = \max_{\substack{\vecw \in [0, 1]^{n} \\ \sum_{u\in \vertexset } \vecw(u) \deg(u) = x}} \sum_{u \in \vertexset} \vecw(u) \vecp(u),
\end{equation}
which leads immediately to the following proposition.
\begin{proposition} \label{prop:lsleq}
    For any $\vecp \in \R_{\geq 0}^{n}$ and any $\sets \subseteq \vertexset$, it holds that 
   $
        \vecp(\sets) \leq \lscurve{\vol(\sets)}.
    $
\end{proposition}
\begin{proof} 
    We have that
 $\vecp(\sets)  =  \sum_{u \in \vertexset} \indicatorvec_{\sets}(u) \vecp(u).$
  Hence, the statement follows by (\ref{eq:deflscurve}).
\end{proof}
It is worth taking  time to understand this definition and gain some of the intuition behind the \lovasz-Simonovits curve.
First of all, we have that $\vecp[0] = 0$ and $\vecp[\vol(\vertexset)] = 1$.
Then, for some $x \in [0, \vol(V)]$, the value of $\vecp[x]$ is the maximum probability mass in $\vecp$ on any vertex set with volume $x$.
Notice that for the stationary distribution $\statdist$, we have
\[
    \statdist[x] = \frac{x}{\vol(\vertexset)}.
\]
Andersen \etal~\cite{andersenLocalGraphPartitioning2006} used the \lovasz-Simonovits curve to prove that the sweep sets of an approximate \pagerank\ vector can be used to find a low-conductance cluster close to some starting vertex.
Their proof is built around the following two results.
\begin{enumerate}
    \item If there is a set $\sets \subset \vertexset$ with low conductance, then for most of the vertices $v \in \sets$, the \lovasz-Simonovits curve of the approximate \pagerank\ vector $\vecp = \apr(\alpha, \indicatorvec_v, \vecr)$ deviates from that of the stationary distribution. That is, $\vecp[x] - \statdist[x]$ is large for some $x \in [0, \vol(\vertexset)]$.
    This corresponds to our intuition that a random walk is likely to stay inside a cluster with low conductance.
    \item If $\vecp[x] - \statdist[x]$ is large for some $x \in [0, \vol(\vertexset)]$, then there is a sweep set of $\vecp$ with low conductance.
\end{enumerate}
These two results lead to the following simple local graph clustering algorithm, which takes as input a graph $\graphg$, a seed vertex $v \in \vertexset$, and a parameter $\alpha$.
\begin{enumerate}
    \item Compute vectors $\vecp$ and $\vecr$ satisfying $\vecp = \apr(\alpha, \indicatorvec_v, \vecr)$.
    \item Compute $j = \argmin_{j \in [\abs{\supp(p)}]} \cond(\pjsweep)$ and return $\pjsweep$.
\end{enumerate}
Andersen~\etal~\cite{andersenLocalGraphPartitioning2006} also provided an algorithm for computing the approximate \pagerank\ locally, which is used in the first step of the clustering algorithm. 

\subsection{Evolving Set Process} \label{sec:esp}
 The \firstdef{evolving set process~(ESP)} is another key technique for designing local clustering algorithms~\cite{andersenFindingDenseSubgraphs2009, andersenAlmostOptimalLocal2016}.
Given a graph $\geqve$, the ESP is a Markov chain whose states are sets of vertices $\sets_i \subseteq \vertexset$.
Given a state $\sets_i$, the next state $\sets_{i+1}$ is determined by the following process: (1) choose $t \in [0, 1]$ uniformly at random; (2) let $\sets_{i + 1} = \{v \in \vertexset | \indicatorvec_{\sets_i}^\transpose \lazywalkm \indicatorvec_v \geq t\}$,
 where $\indicatorvec_{\sets}$ is the indicator vector defined in Section~\ref{sec:prelim:graphs}.
To understand this update process, notice that $\indicatorvec_{\sets_i}^\transpose \lazywalkm \indicatorvec_v$ is the probability that a one-step random walk starting at $v$ ends inside $\sets_i$.
Since the ESP is based on a \emph{lazy} random walk,  we have $\sets_i \subseteq \sets_{i+1}$ if $t \leq 1/2$, and $\sets_{i+1} \subseteq \sets_i$ otherwise.

If the graph $\graphg$ is connected, then the ESP has two absorbing states: $\emptyset$ and $\vertexset$.
Given that the transition kernel of the evolving set process is given by $\mathbf{K}(\sets, \sets') = \p[\sets_{i+1} = \sets' | \sets_i = \sets]$, the \firstdef{volume-biased ESP} is defined by the transition kernel
\[
    \mathbf{\widehat K}(\sets, \sets') = \frac{\vol(\sets')}{\vol(\sets)} \mathbf{K}(\sets, \sets'),
\]
and is guaranteed to absorb in the state $\vertexset$ rather than $\emptyset$.
Andersen and Peres~\cite{andersenFindingSparseCuts2009} gave a local algorithm for undirected graph clustering using the volume-biased ESP.
 As part of this, they described an algorithm, $\alggeneratesample(\graphg, u, T)$, which samples $\sets_T$ from the volume-biased ESP with $\sets_0 = \{u\}$ and has running time proportional to the volume of the output set. 

Andersen and Peres~\cite{andersenFindingSparseCuts2009} gave two further results which are useful in the analysis of the evolving set process.
The first tells us that an evolving set process running for sufficiently many steps is very likely to have at least one state with low conductance.\footnote{Note that they use a slightly different definition of conductance: $\cond(\sets) = \weight(\sets, \setcomplement{\sets}) / \vol(\sets)$. The difference in the denominator does not impact our use of this theorem, since we   guarantee that the volume of $\sets$ is small.}
\begin{proposition}[\cite{andersenFindingSparseCuts2009}, Corollary 1] \label{prop:esp_cond}
    Let $\sets_0, \sets_1, \ldots, \sets_T$ be sets drawn from the volume-biased evolving set process beginning at $\sets_0$. Then, it holds  for any constant $c \geq 0$ that 
    \[
        \p\left[\min_{j \leq T} \cond(\sets_j) \leq \sqrt{\frac{4 c}{T} \ln (\vol(\vertexset))}\right] \geq 1 - \frac{1}{c}.
    \]
\end{proposition}
Secondly, they showed that the overlap between the sets generated by the evolving set process and some target cluster $\seta$ is closely linked to the escape probability of the lazy random walk.
    \begin{proposition}[\cite{andersenFindingSparseCuts2009}, Lemma 2] \label{prop:esp_vol}
        For any vertex $x$, let $\sets_0, \sets_1, \ldots, \sets_T$ be sampled from a volume-biased ESP starting from $\sets_0 = \{x\}$.
        Then, it holds for any set $\seta \subseteq \vertexset$ and $\lambda > 0$ that 
        \[
            \p\left[\max_{t \leq T} \frac{\vol(\sets_t \setminus \seta)}{\vol(\sets_t)} > \lambda \esc(x, T, \seta)\right] < \frac{1}{\lambda}.
        \]
    \end{proposition}
    
\subsection{Other Local Clustering Techniques}
Chung~\cite{chungLocalGraphPartitioning2009} and Kloster and Gleich~\cite{klosterHeatKernelBased2014} developed local clustering algorithms based on the heat-kernel diffusion process.
Orecchia and Allen-Zhu~\cite{orecchiaFlowbasedAlgorithmsLocal2014} developed an algorithm based on network flows which can be used to find a low-conductance cluster close to some seed set.
Li and Peng~\cite{liDetectingCharacterizingSmall2013} developed a local algorithm for finding densely connected clusters in graphs, based on truncated random walks.

\section{Spectral Hypergraph Theory} \label{sec:related:hypergraph_spectral}
Given the usefulness of hypergraphs for modelling higher-order relationships between data
and the effectiveness of spectral graph theory for learning the structure of graphs, it is natural to look for an extension of these techniques to hypergraphs.
Several authors have proposed different ways to generalise the Laplacian matrix from graphs to hypergraphs~\cite{chungLaplacianHypergraph1993, heinTotalVariationHypergraphslearning2013, jostHypergraphLaplaceOperators2019, liSubmodularHypergraphsPLaplacians2018}.
 This thesis focuses on the Laplacian operator proposed by Chan~\etal~\cite{chanSpectralPropertiesHypergraph2018} which
has been applied for learning clusters in hypergraphs~\cite{takaiHypergraphClusteringBased2020}.

Given some hypergraph $\graphh$ on $n$ vertices, Chan~\etal~\cite{chanSpectralPropertiesHypergraph2018} defined a new \emph{non-linear} operator $\lap_\graphh: \R^n \rightarrow \R^n$ which describes a heat diffusion process on hypergraphs: for any vector $\vecf_t \in \R^n$, the diffusion proceeds according to the equation 
\[
    \dfdt = - \lap_\graphh \vecf_t.
\]
Among many properties of $\lap_\graphh$, they proved that for any vector $\vecf$, the value of $\lap_\graphh \vecf$ can be computed in polynomial time.
 Moreover, the heat diffusion process is guaranteed to converge to an eigenvector of $\lap_\graphh$, and the corresponding eigenvalue satisfies an analogue of the Cheeger inequality for the \firstdef{hypergraph conductance}. 
 
\begin{definition}[Hypergraph conductance]
Given a hypergraph $\graphh = (\vertexset, \edgeset, \weight)$ and any set $\sets \subset \vertexset$, the conductance of $\sets$ is
\[
    \cond_\graphh(\sets) = \frac{\weight(\sets, \setcomplement{\sets})}{\vol(\sets)},
\]
where $\weight(\sets, \setcomplement{\sets}) = \sum_{e \in \edgeset} \weight(e) \cdot \iverson{e \intersect \sets \neq \emptyset \land e \intersect \setcomplement{\sets} \neq \emptyset}$ and $\vol(\sets) = \sum_{v \in \sets} \deg(v)$.
\end{definition}

\begin{theorem}[\cite{chanSpectralPropertiesHypergraph2018}, Theorem 3.3]
For any hypergraph $\graphh$, let $\lambda_2$ be the second smallest eigenvalue of $\lap_\graphh$. Then,
\[
    \frac{\lambda_2}{2} \leq \cond_\graphh \leq 2 \sqrt{\lambda_2},
\]
where $\cond_\graphh$ is the hypergraph conductance.
\end{theorem}
Chan and Liang~\cite{chanGeneralizingHypergraphLaplacian2020} generalised this result by considering a slightly different diffusion process, and Chan~\etal~\cite{chanDiffusionOperatorSpectral2019} generalised the hypergraph Laplacian operator to directed hypergraphs.
 Takai~\etal~\cite{takaiHypergraphClusteringBased2020} showed that the hypergraph Laplacian operator can be used to compute the personalised \pagerank\ vector in a hypergraph.
They further used this to develop a local clustering algorithm for hypergraphs.
Finally, Yoshida~\cite{yoshidaCheegerInequalitiesSubmodular2019} described a Laplacian operator for hypergraphs with \emph{submodular cut functions}, and posed the question of whether there is a hypergraph operator which satisfies an inequality like Trevisan's bipartiteness inequality given in Theorem~\ref{thm:trevisansinequality}.
We give an affirmative answer to Yoshida's question in Chapter~\ref{chap:hyper} and demonstrate that the new hypergraph signless Laplacian operator can be used to find densely connected clusters in hypergraphs.

\chapter{A Tighter Analysis of Spectral Clustering} \label{chap:tight}
Spectral clustering, as described in Algorithm~\ref{alg:sc}, has been widely applied over the past two decades due to its simplicity and excellent empirical performance~\cite{ngSpectralClusteringAnalysis2001}. 
In this chapter we theoretically analyse the spectral clustering algorithm, and develop a tighter guarantee for well-clustered graphs than the ones previously given in the literature.
Informally, we analyse the performance guarantee of spectral clustering under a simple assumption on the input graph.
While all the previous work~(e.g., \cite{leeMultiwaySpectralPartitioning2014,kolevNoteSpectralClustering2016,mizutaniImprovedAnalysisSpectral2021,pengPartitioningWellClusteredGraphs2017}) on the same problem suggests that the assumption on the input graph must depend on $k$, we show that the performance of spectral clustering can be rigorously analysed under a general condition independent of $k$.
This is the first result of its kind,  and could have further applications in graph clustering.

Given a graph $\geqve$, the goal of graph clustering is to partition $V$ into $k$ clusters $S_1,\ldots, S_k$ according to some objective function.
Recall that we can measure the quality of a cluster $\sets_i$ by its conductance
\[
    \cond(\sets_i) = \frac{\weight(\sets_i, \setcomplement{\sets_i})}{\min\{\vol(\sets_i), \vol(\setcomplement{\sets_i})\}},
\]
and that the $k$-way expansion of $\graphg$ is a measure of the quality of the optimal partitioning:
\[
    \kcond(k) = \min_{\mathrm{partition}~\sets_1, \ldots, \sets_k} \max_{i \in [k]} \cond(\sets_i).
\]
Let $\lapn$ be the normalised Laplacian matrix of $\graphg$, and  $\lambda_1 \leq \lambda_2, \ldots, \leq \lambda_n$ be the eigenvalues of $\lapn$.
In this chapter, we study spectral clustering on $\graphg$ in terms of the function $\Upsilon(k)$ defined by 
\begin{equation}\label{eq:defineupsilon}
    \Upsilon(k) \triangleq \frac{\lambda_{k+1}}{\rho(k)},
\end{equation}
which was introduced in Section~\ref{sec:spectralclusteringworks}.
The first result in this chapter improves on the structure theorem of Peng~\etal~\cite{pengPartitioningWellClusteredGraphs2017} given in Theorem~\ref{thm:pengStructure}.

\begin{mainresult} [See Theorem~\ref{thm:struc1} for the formal statement]
Given any graph $\graphg$, let $\sets_1, \ldots, \sets_k$ be the clusters which achieve $\kcond(k)$.
The normalised indicator vectors of $\sets_1, \ldots, \sets_k$ are well-approximated by the $k$ eigenvectors of the graph Laplacian $\lapn$ corresponding to $\lambda_1, \ldots, \lambda_k$.
The approximation guarantee is parameterised by $\Upsilon(k)$, and is stronger for larger values of $\Upsilon(k)$.
\end{mainresult}

In the remainder of the chapter, we consider spectral clustering for graphs with clusters of \firstdef{almost-balanced} size.

\begin{definition} (Almost-balanced) \label{def:almostBalanced}
    Let $\graphg$ be a graph with $k$ clusters $S_1,\ldots, S_k$.
    We say that the clusters are almost-balanced if
    \[(1/2) \cdot \vol(\vertexsetg)/k \leq \vol(\sets_i) \leq 2 \cdot \vol(\vertexsetg)/k\]
    for all $i \in[k]$.
\end{definition}

The second result of this chapter is a guarantee on the performance of spectral clustering under the assumption that $\Upsilon(k) \geq C$ for some constant $C$.
\begin{mainresult} [See Theorem~\ref{thm:sc_guarantee} for the formal statement]
Let $G=(V,E)$ be a graph, and $S_1,\ldots, S_k$ the $k$  clusters of almost-balanced size that achieves $\rho(k)$. Given $G$ as the input, let the output of spectral clustering be $A_1,\ldots, A_k$ and the 'optimal' correspondent of $A_i$ be $S_i$ for any $i\in [k]$. Then, $\sum_{i=1}^k \vol(A_i\triangle S_i) \leq C\cdot \vol(V)/\Upsilon(k)$ for some constant $C$.
\end{mainresult}

Notice that some condition on $\Upsilon(k)$  is needed 
to ensure that an input graph $\graphg$ has $k$ well-defined clusters, so that misclassified vertices can be formally defined.
Taking this into account, our result is non-trivial as long as 
$\Upsilon(k)$ is lower bounded by some constant\footnote{Note that we can take any constant approximation in Definition~\ref{def:almostBalanced} with a different corresponding constant in the main result.}.
This significantly improves most of the previous analysis of graph clustering algorithms, which make stronger assumptions on the input graphs.
For example, Peng et al.~\cite{pengPartitioningWellClusteredGraphs2017} assumes that $\Upsilon(k)=\bigomega{k^3}$, Mizutani~\cite{mizutaniImprovedAnalysisSpectral2021} assumes that $\Upsilon(k) =\bigomega{k}$, the algorithm presented by Oveis Gharan and Trevisan~\cite{gharanPartitioningExpanders2014} assumes that $\lambda_{k+1} = \bigomega{\mathrm{poly}(k)\lambda^{1/4}_k}$, and the one presented by Dey et al.~\cite{deySpectralConcentrationGreedy2019} further assumes some condition with respect to  $k$, $\lambda_k$, and the maximum degree of $\graphg$.
While these assumptions require at least a linear dependency on $k$, making it difficult for the instances with a large value of $k$ to satisfy, the result in this chapter suggests that the performance of spectral clustering can be rigorously analysed for these graphs.
In particular, compared with previous work, our result better justifies the widely used eigen-gap heuristic for spectral clustering~\cite{ngSpectralClusteringAnalysis2001,vonluxburgTutorialSpectralClustering2007}.
This heuristic suggests that spectral clustering works when the value of $|\lambda_{k+1} -\lambda_k|$ is much larger than $|\lambda_{k} -\lambda_{k-1}|$, and in practice the ratio between the two gaps is usually a constant rather than some function of $k$.

\section{Stronger Structure Theorem}
Let $\{\sets_i\}_{i=1}^k$ be any optimal partition that achieves $\rho(k)$.
 Recall that $\indicatorvec_{\sets_i}$ is the indicator vector of $\sets_i$, and that $\barg_i$ is the corresponding normalised indicator vector, which is formally defined in Section~\ref{sec:spectralclusteringworks}.
Additionally, let $\vecf_1, \ldots \vecf_n$ be the eigenvectors of $\lapn$ corresponding to the eigenvalues $\lambda_1, \ldots, \lambda_n$.
One of the basic results in spectral graph theory states that $\graphg$ consists of at least $k$ connected components if and only if $\lambda_i=0$ for any $i \in [k]$, and
$\mathrm{span}\left(\vecf_1, \ldots, \vecf_k\right) = \mathrm{span}\left(\barg_1, \ldots, \barg_k\right)$~\cite{chungSpectralGraphTheory1997}.
Hence,  one would expect that, when $\graphg$ consists of $k$ densely connected components~(clusters) connected by sparse cuts, the bottom eigenvectors $\{\vecf_i\}_{i=1}^k$ of $\lapn$ are close to $\{\barg_i\}_{i=1}^k$.

In their work, Peng et~al.~\cite{pengPartitioningWellClusteredGraphs2017} assume   $\Upsilon (k) =\Omega(k^2)$, and  prove that the space spanned by $\{\vecf_i\}_{i=1}^k$ and the one spanned by $\{\barg_i\}_{i=1}^k$ are indeed close to each other. 
Specifically, they show that
\begin{enumerate}
    \item  for every $\barg_i$, there is some linear combination of $\{\vecf_i\}^k_{i=1}$, denoted by $\hatf_i$, such that $\|\barg_i - \hatf_i\|^2\leq 1/\Upsilon(k)$; and
    \item  for every $\vecf_i$ there is some linear combination of $\{\barg_i\}_{i=1}^k$, denoted by $\hatg_i$, such that $\|\vecf_i - \hatg_i \|^2\leq 1.1 k/\Upsilon(k)$.
\end{enumerate}
See Theorem~\ref{thm:pengStructure} for the formal statement.
In essence, their structure theorem gives a quantitative explanation on why spectral methods work for graph clustering
when there is a clear cluster-structure in $\graphg$ characterised by $\Upsilon(k)$.
As it holds for graphs with clusters of different sizes and edge densities, this structure theorem has been shown to be a powerful tool in analysing clustering algorithms, and inspired many subsequent works~(e.g.,~\cite{czumajTestingClusterStructure2015, chenCommunicationoptimalDistributedClustering2016, kolevNoteSpectralClustering2016, kloumannBlockModelsPersonalized2017, louisPlantedModelsKway2019, mizutaniImprovedAnalysisSpectral2021,pengRobustClusteringOracle2020,pengAverageSensitivitySpectral2020,sunDistributedGraphClustering2019}).  

We show that a stronger statement of the original structure theorem holds under a much weaker assumption, and this result is summarised as follows:

\begin{theorem}[The Stronger Structure Theorem] \label{thm:struc1}
The following statements hold:
\begin{enumerate}
    \item For any $i\in[k]$, there is a vector $\hatf_i\in\mathbb{R}^n$ which is  
    a linear combination of $\vecf_1,\ldots, \vecf_k$, such that $\|\barg_i - \hatf_i\|^2 \leq 1/\Upsilon(k)$.
    \item There are vectors $\hatg_1,\ldots, \hatg_k$, each of which is a linear combination of $\barg_1,\ldots, \barg_k$, such that $\sum_{i = 1}^k \norm{\vecf_i - \hatg_i}^2 \leq k /\Upsilon(k)$.
\end{enumerate}
\end{theorem}

To examine the significance of Theorem~\ref{thm:struc1}, notice that the two statements hold for any $\Upsilon(k)$, while the original structure theorem relies on the assumption that $\Upsilon(k)=\Omega(k^2)$.
Since $\Upsilon(k)=\Omega(k^2)$ is a strong and even questionable assumption when $k$ is large, obtaining these statements for general $\Upsilon(k)$ is important.
Secondly, the second statement of Theorem~\ref{thm:struc1} significantly improves Theorem~\ref{thm:pengStructure}.
Specifically, instead of stating $\|\vecf_i - \hatg_i \|^2\leq 1.1 k/\Upsilon(k)$ for any $i\in [k]$, the second statement shows that $\sum_{i = 1}^k \norm{\vecf_i - \hatg_i}^2 \leq k /\Upsilon(k)$;
hence, it holds in expectation that $\norm{\vecf_i - \hatg_i}^2 \leq 1 /\Upsilon(k)$, the upper bound of which matches the first statement.
This implies that the vectors $\vecf_1,\ldots, \vecf_k$ and $\barg_1,\ldots, \barg_k$ can be linearly approximated by each other  with  \emph{roughly the same} approximation guarantee. 
Thirdly, rather than employing the machinery from matrix analysis used by Peng et al.~\cite{pengPartitioningWellClusteredGraphs2017}
to prove the original theorem, our new proof is simple and purely linear-algebraic.
Therefore, both the stronger theorem and its much simplified proof are significant, and could have further applications in graph clustering and related problems. 

\begin{proof} [Proof of Theorem~\ref{thm:struc1}]
Let $\hatf_i = \sum_{j=1}^k \langle \barg_i, \vecf_j\rangle \vecf_j $, and we write $\barg_i$ as a linear combination of the vectors $\vecf_1, \ldots, \vecf_n$ by 
   $
        \barg_i = \sum_{j = 1}^n \langle \barg_i, \vecf_j \rangle \vecf_j$.
   Since $\hatf_i$ is a projection of $\barg_i$, we have that $\barg_i - \hatf_i$ is perpendicular to $\hatf_i$ and 
    \begin{align*}
        \norm{\barg_i - \hatf_i}^2 & = \norm{\barg_i}^2 - \norm{\hatf_i}^2  = \left(\sum_{j = 1}^n \langle \barg_i, \vecf_j \rangle^2 \right) - \left(\sum_{j = 1}^{k} \langle \barg_i, \vecf_j \rangle^2 \right)  = \sum_{j = k + 1}^n \langle \barg_i, \vecf_j \rangle^2.
    \end{align*}
     Now, let us consider the quadratic form
    \begin{align}
        \barg_i^\transpose \lapn \barg_i & = \left(\sum_{j = 1}^n \langle \barg_i, \vecf_j \rangle \vecf_j^\transpose \right) \lapn \left(\sum_{j = 1}^n \langle \barg_i, \vecf_j \rangle \vecf_j\right) \nonumber \\
        & = \sum_{j = 1}^n \langle \barg_i, \vecf_j \rangle^2 \lambda_j  \geq \lambda_{k + 1} \norm{\barg_i - \hatf_i}^2, \label{eq:lbquad}
    \end{align}
    where the last inequality follows by the fact that $\lambda_i\geq 0$ holds for any $1\leq i\leq n$. This gives us that  
    \begin{align}
        \barg_i^\transpose \lapn \barg_i & = \sum_{(u, v) \in \edgeset} \weight(u, v) \left(\frac{\barg_i(u)}{\sqrt{\deg(u)}} - \frac{\barg_i(v)}{\sqrt{\deg(v)}}\right)^2 \nonumber \\
        & = \sum_{(u, v) \in \edgeset} \weight(u, v) \left(\frac{\indicatorvec_i(u)}{\sqrt{\vol(\sets_i)}} - \frac{\indicatorvec_i(v)}{\sqrt{\vol(\sets_i)}}\right)^2 \nonumber \\
        & = \frac{\weight(\sets_i, \setcomplement{\sets_i})}{\vol(\sets_i)} \leq \rho(k).  \label{eq:upquad}
    \end{align}
    Combining \eqref{eq:lbquad} with \eqref{eq:upquad}, we have that
    \[
    \norm{\barg_i - \hatf_i}^2 \leq \frac{\barg_i^\transpose \lapn \barg_i}{\lambda_{k+1}} \leq \frac{\rho(k)}{\lambda_{k+1}} \leq \frac{1}{\Upsilon (k)},
    \]
    which proves the first statement of the theorem.

Now we prove the second statement.
We define for any $1\leq i \leq k$ that 
 $\hatg_i = \sum_{j=1}^k \langle \vecf_i, \barg_j\rangle \barg_j$ and have that, since $\vecf_i - \hatg_i$ is perpendicular to $\hatg_i$,
    \begin{align*}
        \sum_{i = 1}^k \norm{\vecf_i - \hatg_i}^2
        & = \sum_{i = 1}^k \left( \norm{\vecf_i}^2 - \norm{\hatg_i}^2 \right) \\
        & = \sum_{i = 1}^k \left( 1 - \sum_{j = 1}^k \langle \barg_j, \vecf_i \rangle^2 \right) \\
        & = k - \sum_{i = 1}^k \sum_{j = 1}^k \langle \barg_j, \vecf_i \rangle^2   \\
        & = \sum_{j = 1}^k \left( 1 - \sum_{i = 1}^k \langle \barg_j, \vecf_i \rangle^2 \right) \\
        & = \sum_{j = 1}^k \left( \norm{\barg_j}^2 - \norm{\hatf_j}^2 \right) \\
        & = \sum_{j = 1}^k \norm{\barg_j - \hatf_j}^2 \\
        & \leq \sum_{j = 1}^k \frac{1}{\Upsilon(k)} \\
        & = \frac{k}{\Upsilon(k)},
    \end{align*}
    where the inequality follows by the first statement of  Theorem~\ref{thm:struc1}.
\end{proof}

\section{Analysis of Spectral Clustering\label{sec:analysis1}}
For any input graph $\geqve$ and  $k\in[n]$, we study the following spectral clustering algorithm: 
\begin{enumerate}
    \item compute the eigenvectors $\vecf_1,\ldots \vecf_k$ of $\lapn$, and embed each $u \in \vertexset$ to the point $F(u) \in \R^k$ defined by
    \begin{equation}\label{eq:embedding}
     F(u) \triangleq \frac{1}{\sqrt{\deg(u)}} \left( \vecf_1(u),\ldots, \vecf_k(u)\right)^\transpose;
    \end{equation}
    \item apply $k$-means on the embedded points $\{ F(u)\}_{u\in \vertexset}$;
    \item partition $\vertexset$ into $k$ clusters based on the output of $k$-means.
\end{enumerate}
 
In this rest of this section, we prove the following theorem, where $\APT$ is the approximation ratio of the $k$-means algorithm used in spectral clustering.
Recall that we can take $\APT$ to be some small constant~\cite{cohen-addadOnlineKmeansClustering2021, kumarSimpleLinearTime2004}. 
\begin{theorem}\label{thm:sc_guarantee}
Let $\geqve$ be a graph with $k$ clusters $S_1,\ldots, S_k$ of almost-balanced size, and
$\Upsilon(k) \geq 2176 (1 + \APT)$.
Let $\{\seta_i\}_{i=1}^k$ be the output of spectral clustering and assume, without loss of generality, that the optimal correspondent of $\seta_i$ is $\sets_i$. Then, it holds that
\[
        \sum_{i = 1}^k \vol\left(\seta_i \triangle \sets_i \right) \leq 2176 (1 + \APT) \frac{\vol(\vertexset)}{\Upsilon(k)}.
\] 
\end{theorem}
The most significant feature of Theorem~\ref{thm:sc_guarantee} is that the lower bound of $\Upsilon(k)$ is independent of the number of clusters $k$.
The theorem gives a non-trivial guarantee on the performance of spectral clustering so long as the spectral gap, as measured by $\Upsilon(k)$, is at least a constant.
This improves on the previous analysis of graph clustering algorithms which make stronger assumptions on the input graphs~\cite{pengPartitioningWellClusteredGraphs2017, mizutaniImprovedAnalysisSpectral2021, leeMultiwaySpectralPartitioning2014}.

\subsection{Properties of Spectral Embedding}
Let us first study the properties of the spectral embedding defined in \eqref{eq:embedding}, and in the next subsection we use these properties to prove Theorem~\ref{thm:sc_guarantee}.
For every cluster $\sets_i$, define the vector $\peye \in \R^k$ by 
\[
\peye(j) = \frac{1}{\sqrt{\vol(\sets_i) }} \langle \vecf_j, \barg_i \rangle,
\]
and we can view these $\{\peye\}_{i=1}^k$ as the approximate centres of the embedded points from the optimal clusters $\{\sets_i\}_{i=1}^k$.
We prove that the total $k$-means cost of the embedded points can be upper bounded as follows:

\begin{lemma} \label{lem:total_cost}
    It holds that
    \[
        \sum_{i = 1}^k \sum_{u \in \sets_i} \deg(u) \norm{F(u) - \peye}^2 \leq \frac{k}{\Upsilon(k)}.
    \]
\end{lemma}
\begin{proof}
    We have
    \begin{align*}
        \sum_{i = 1}^k \sum_{u \in \sets_i} \deg(u) \norm{F(u) - \peye}^2 & = \sum_{i = 1}^k \sum_{u \in \sets_i} \deg(u) \left[ \sum_{j = 1}^k \left(\frac{\vecf_j(u)}{\sqrt{\deg(u)}} - \frac{\langle \barg_i, \vecf_j\rangle}{\sqrt{\vol(\sets_i)}}\right)^2 \right] \\
        & = \sum_{i = 1}^k \sum_{u \in \sets_i} \sum_{j = 1}^k \left(\vecf_j(u) - \langle \barg_i, \vecf_j\rangle \barg_i(u)\right)^2 \\
        & = \sum_{i = 1}^k \sum_{u \in \sets_i} \sum_{j = 1}^k \left(\vecf_j(u) - \hatg_j(u)\right)^2 \\
        & = \sum_{j = 1}^k \norm{\vecf_j - \hatg_j}^2 \leq \frac{k}{\Upsilon(k)},
    \end{align*}
    where we use that for  $u \in \sets_x$ it holds that  $
        \hatg_i(u) = \sum_{j = 1}^k \langle \vecf_i, \barg_j\rangle \barg_j(u) = \langle \vecf_i, \barg_x\rangle \barg_x(u)$, and the final inequality follows by the second statement of Theorem~\ref{thm:struc1}.
\end{proof}
The importance of Lemma~\ref{lem:total_cost} is that, although the optimal centres for $k$-means are unknown, the existence of $\{\peye\}_{i=1}^k$ is sufficient to  show  that the cost of an optimal $k$-means clustering on $\{F(u)\}_{u\in \vertexset}$ is at most $k/\Upsilon(k)$.
Since one can always use an $O(1)$-approximate $k$-means algorithm for spectral clustering~(e.g.,~\cite{kanungoLocalSearchApproximation2004,kumarSimpleLinearTime2004}), the cost of the output of $k$-means on $\{F(u)\}_{u\in \vertexset}$ is $O\left(k/\Upsilon(k)\right)$.
Next, we show that the
length of $\peye$ is approximately equal to
$1/\vol(\sets_i)$,
which is used in our later analysis.

\begin{lemma} \label{lem:pnorm}
    It holds for  any $i \in [k]$ that 
    \[
        \frac{1}{\vol(\sets_i)} \left(1 - \frac{1}{\Upsilon(k)}\right) \leq \norm{\vecp^{(i)}}^2 \leq \frac{1}{\vol(\sets_i)}.
    \]
\end{lemma}

\begin{proof} 
    By definition, we have
    \begin{align*}
        \vol(\sets_i) \norm{\vecp^{(i)}}^2 & = \sum_{j = 1}^k \langle \vecf_j, \barg_i \rangle^2  = \norm{\hatf_i}^2  = 1 - \norm{\hatf_i - \barg_i}^2  \geq 1 - \frac{1}{\Upsilon(k)}, 
    \end{align*}
    where the inequality follows by Theorem~\ref{thm:struc1}. The other direction of the inequality follows similarly.
\end{proof}

In the remainder of this subsection, we prove a sequence of lemmas showing that any pair of $\peye$ and $\pj$ for $i\neq j$ are well separated.
Moreover, their distance is essentially independent of $k$ and $\Upsilon(k)$, as long as $\Upsilon(k) \geq 20$.

\begin{lemma} \label{lem:normp_diff}
It holds for any different $i,j\in[k]$ that 
\[
    \left\|\sqrt{\vol(\sets_i)}\cdot \vecp^{(i)} - \sqrt{\vol(\sets_j)}\cdot \vecp^{(j)}\right\|^2 \geq 2 - \frac{8}{\Upsilon(k)}.
\]
\end{lemma}

\begin{proof}
    By the definitions of $\peye$ and $\hatf_i$, we have
    \begin{align*}
    \lefteqn{\left\|\sqrt{\vol(\sets_i)}\cdot  \vecp^{(i)} - \sqrt{\vol(\sets_j)}\cdot  \vecp^{(j)}\right\|^2} \\
    & = \sum_{x = 1}^k \left(\langle \vecf_x, \barg_i \rangle - \langle \vecf_x, \barg_j \rangle \right)^2 \\
    & = \left( \sum_{x = 1}^k \langle \vecf_x, \barg_i\rangle^2 \right) + \left(\sum_{x = 1}^k \langle \vecf_x, \barg_j \rangle^2 \right) - 2 \sum_{x = 1}^k \langle \vecf_x, \barg_i\rangle \langle \vecf_x, \barg_j \rangle \\
    & \geq \norm{\hatf_i}^2 + \norm{\hatf_j}^2 - 2 \abs{\hatf_i^\transpose \hatf_j}.
    \end{align*}
   Then, considering only the final term, we have
   \begingroup
   \allowdisplaybreaks
   \begin{align*}
    \abs{\hatf_i^\transpose \hatf_j} & = \abs{(\barg_i + \hatf_i - \barg_i)^\transpose (\barg_j + \hatf_j - \barg_j)} \\
    & = \abs{\langle \barg_i, \hatf_j - \barg_j \rangle + \langle \barg_j, \hatf_i - \barg_i \rangle + \langle \hatf_i - \barg_i, \hatf_j - \barg_j\rangle} \\
    & = \abs{\langle \hatf_i - \barg_i, \hatf_j - \barg_j \rangle + \langle \hatf_j - \barg_j, \hatf_i - \barg_i \rangle + \langle \hatf_i - \barg_i, \hatf_j - \barg_j\rangle} 
     \leq \frac{3}{\Upsilon(k)},
   \end{align*}
   \endgroup
   where we use the fact that $\inner{\barg_i}{\barg_j} = 0$ for $i \neq j$ and $\norm{\hatf_i - \barg_i}^2 \leq 1 / \Upsilon(k)$.
   Then, since $\norm{\hatf_i}^2 \geq 1 - 1 / \Upsilon(k)$, we have
    \begin{align*}
    \left\|\sqrt{\vol(\sets_i)}\cdot  \vecp^{(i)} - \sqrt{\vol(\sets_j)}\cdot  \vecp^{(j)}\right\|^2
    & \geq 2 \left(1 - \frac{1}{\Upsilon(k)}\right) - \frac{6}{\Upsilon(k)} \\
    & = 2 - \frac{8}{\Upsilon(k)}. \qedhere
    \end{align*}
\end{proof}

\begin{lemma} \label{lem:normp_diff2}
It holds for any different $i,j \in[k]$ that  
\[
        \left\|\frac{\vecp^{(i)}}{\norm{\peye}} - \frac{\pj}{\norm{\pj}}\right\|^2 \geq 2 - \frac{20}{\Upsilon(k)}.
    \]
\end{lemma} 

\begin{proof}
    Assume without loss of generality that $\sqrt{\vol(\sets_i)} \norm{\peye} \leq \sqrt{\vol(\sets_j)} \norm{\pj}$.
    Let $\veca_i = \sqrt{\vol(\sets_i)} \peye$ and $\veca_j = \sqrt{\vol(\sets_j)} \pj$ and notice that $\norm{\veca_i} \leq \norm{\veca_j} \leq 1$ by Lemma~\ref{lem:pnorm}.
    Then, since $\veca_i$ and $(\norm{\veca_i}/\norm{\veca_j}) \veca_j$ are scaled versions of $\peye$ and $\pj$ with smaller norm, we have
    \begin{align*}
        \left\|\frac{\vecp^{(i)}}{\norm{\peye}} - \frac{\pj}{\norm{\pj}}\right\| & \geq \left\|\veca_i - \frac{\norm{\veca_i}}{\norm{\veca_j}} \veca_j\right\| \\
        & \geq \norm{\veca_i - \veca_j} - \left(\norm{\veca_j} - \norm{\veca_i}\right) \\
        & \geq \sqrt{2 - \frac{8}{\Upsilon(k)}}  - \left(\sqrt{\vol(\sets_j)}\cdot \left\|\pj\right\| - \sqrt{\vol(\sets_i)} \left\|\peye\right\|\right) \\
        & \geq \sqrt{2}\left(1 - \frac{4}{\Upsilon(k)}\right) + \sqrt{1 - \frac{1}{\Upsilon(k)}} - 1 \\
        & \geq \sqrt{2} - \frac{4\sqrt{2}}{\Upsilon(k)} - \frac{1}{\Upsilon(k)} \geq \sqrt{2} - \frac{7}{\Upsilon(k)},
    \end{align*}
    where the second inequality follows by the triangle inequality, and the third and fourth use Lemma~\ref{lem:normp_diff} and Lemma~\ref{lem:pnorm}. We also use the fact that it holds for $x \leq 1$ that $\sqrt{1 - x} \geq 1 - x$.
    This gives
    \begin{align*}
        \left\|\frac{\vecp^{(i)}}{\norm{\peye}} - \frac{\pj}{\norm{\pj}}\right\|^2 & \geq 2 - \frac{14 \sqrt{2}}{\Upsilon(k)}  \geq 2 - \frac{20}{\Upsilon(k)},
    \end{align*}
    which proves the lemma.
\end{proof}

\begin{lemma} \label{lem:p_diff}
    It holds for any different $i, j \in [k]$ that 
    \[
        \norm{\peye - \pj}^2 \geq \frac{1}{\min\{\vol(\sets_i), \vol(\sets_j)\}} \left(\frac{1}{2} - \frac{8}{\Upsilon(k)} \right). 
    \]
\end{lemma}
\begin{proof}
    Assume without loss of generality that $\norm{\peye} \geq \norm{\pj}$. Then, let $\norm{\pj} = \alpha \norm{\peye}$ for some $\alpha \in [0, 1]$.
    By Lemma~\ref{lem:pnorm}, it holds that
    \[
        \left\|\peye\right\|^2 \geq \frac{1}{\min\{\vol(\sets_i), \vol(\sets_j)\}} \left(1 - \frac{1}{\Upsilon(k)}\right).
    \]
    Additionally, notice that by the proof of Lemma~\ref{lem:normp_diff2},  
    \[
        \left\langle \frac{\peye}{\norm{\peye}}, \frac{\pj}{\norm{\pj}} \right\rangle \leq \sqrt{2} - \frac{1}{2} \norm{\frac{\peye}{\norm{\peye}} - \frac{\pj}{\norm{\pj}}} \leq \frac{\sqrt{2}}{2} + \frac{7}{2 \Upsilon(k)},
    \]
    where we use the fact that if $x^2 + y^2 = 1$, then $x + y \leq \sqrt{2}$.
    One can understand the equation above by considering the right-angled triangle with one edge given by $\peye / \norm{\peye}$ and another edge given by $(\peye / \norm{\peye}) . (\pj / \norm{\pj})$.
    Then,
    \begin{align*}
        \norm{\peye - \pj}^2 & = \norm{\peye}^2 + \norm{\pj}^2 - 2 \left\langle \frac{\peye}{\norm{\peye}}, \frac{\pj}{\norm{\pj}} \right\rangle \norm{\peye} \norm{\pj} \\
        & \geq (1 + \alpha) \norm{\peye}^2 - \left(\sqrt{2} + \frac{7}{\Upsilon(k)}\right) \alpha \norm{\peye}^2 \\
        & \geq \left(1 - (\sqrt{2} - 1) \alpha  - \frac{7}{\Upsilon(k)}\right) \norm{\peye}^2 \\
        & \geq  \frac{1}{\min\{\vol(\sets_i), \vol(\sets_j)\}}  \left(\frac{1}{2} - \frac{7}{\Upsilon(k)} \right) \left(1 - \frac{1}{\Upsilon(k)}\right) \\
        & \geq  \frac{1}{\min\{\vol(\sets_i), \vol(\sets_j)\}} \left(\frac{1}{2} - \frac{8}{\Upsilon(k)}\right),
    \end{align*}
    which completes the proof.
\end{proof}

It is worth noting that, despite the similarity in their formulation, the technical lemmas presented here are stronger than the ones in Peng et al.~\cite{pengPartitioningWellClusteredGraphs2017}.
These results are obtained through the stronger structure theorem~(Theorem~\ref{thm:struc1}), and are crucial for us to prove Theorem~\ref{thm:sc_guarantee}.

\subsection{Bounding the Misclassified Vertices}
Now we prove
Theorem~\ref{thm:sc_guarantee},
and see why a mild condition like
$\Upsilon(k) \geq C$ for some constant $C \in \N$
suffices for spectral clustering to perform well in practice. 
Let   $\{\seta_i\}_{i=1}^k$ be the output of spectral clustering,  and we denote  the centre of the embedded points $\{F(u)\}_{u\in \seta_i}$ for any $\seta_i$  by $\vecc_i$.
As the starting point of the analysis, we   show that every $\vecc_i$ is close to its `optimal' correspondent $\vecp^{(\sigma(i))}$ for some $\sigma(i)\in[k]$. That is, the actual centre of embedded points from every $\seta_i$ is close to the approximate centre of the embedded points from some optimal $\sets_i$.
To formalise this,  we define the function $\sigma:[k]\rightarrow[k]$ by
\begin{equation}\label{eq:defsigma}
    \sigma(i) = \argmin_{j \in [k]} \norm{\pj - \vecc_{i}};
\end{equation}
that is, cluster $\seta_i$ should correspond to $\sets_{\sigma(i)}$ in which the value of $\|\vecp^{(\sigma(i))} - \vecc_i\|$ is the lowest among all the distances between $\vecc_i$ and all of the $\vecp^{(j)}$ for $j\in [k]$.
However, one needs to be cautious as \eqref{eq:defsigma} doesn't necessarily define a permutation, and there might exist different $i,i'\in[k]$ such that both of $\seta_i$ and $\seta_{i'}$ map to the same $\sets_{\sigma(i)}$. 
Taking this into account, for any fixed $\sigma:[k]\rightarrow[k]$ and $i\in[k]$, we  further define $\setm_{\sigma, i}$ by
\begin{equation}\label{eq:defmset}
    \setm_{\sigma, i} \triangleq \bigcup_{j : \sigma(j) = i} \seta_j.
\end{equation}
The following lemma shows that, when mapping every output  $\seta_i$ to $\sets_{\sigma(i)}$, the total ratio of misclassified volume with respect to each cluster can be upper bounded:
\begin{lemma} \label{lem:cost_lower_bound}
Let $\{\seta_i\}_{i=1}^k$ be the output of spectral clustering, and $\sigma$ and  $\setm_{\sigma,i}$ be defined as  in \eqref{eq:defsigma} and \eqref{eq:defmset}.
If $\Upsilon(k) \geq 32$, then
\[
        \sum_{i = 1}^k \frac{\vol(\setm_{\sigma, i} \triangle \sets_i)}{\vol(\sets_i)} \leq 64 (1 + \APT) \frac{k}{\Upsilon(k)}.
    \]
 \end{lemma}
 \begin{proof}
    Let us define $\setb_{ij} = \seta_i \cap \sets_j$ to be the vertices in $\seta_i$ which belong to the ground-truth cluster $\sets_j$.
    Then, we have that
    \begin{align}
        \sum_{i = 1}^k \frac{\vol(\setm_{\sigma, i} \triangle \sets_i)}{\vol(\sets_i)} & = \sum_{i = 1}^k \sum_{\substack{j = 1 \\ j \neq \sigma(i)}}^k \vol(\setb_{ij}) \left(\frac{1}{\vol(\sets_{\sigma(i)})} + \frac{1}{\vol(\sets_j)} \right) \nonumber \\
        & \leq 2 \sum_{i = 1}^k \sum_{\substack{j = 1 \\ j \neq \sigma(i)}}^k \frac{\vol(\setb_{ij})}{\min\{\vol(\sets_{\sigma(i)}), \vol(\sets_j)\}}, \label{eq:up_sym_ratio}
    \end{align}
and that 
\begingroup
\allowdisplaybreaks
\begin{align*}
        \lefteqn{\mathrm{COST}(\seta_1, \ldots \seta_k)}\\
        & = \sum_{i = 1}^k \sum_{u \in \seta_i} \deg(u) \norm{F(u) - \vecc_i}^2 \\
        & \geq \sum_{i = 1}^k \sum_{\substack{1\leq j \leq  k \\ j \neq \sigma(i)}}   \sum_{u \in \setb_{ij}} \deg(u) \norm{F(u) - \vecc_i}^2 \\
        & \geq \sum_{i = 1}^k \sum_{\substack{1\leq j \leq k  \\ j \neq \sigma(i)}}  \sum_{u \in \setb_{ij}} \deg(u) \left(\frac{\norm{\pj - \vecc_i}^2}{2} - \norm{\pj - F(u)}^2 \right) \\
        & \geq \sum_{i = 1}^k \sum_{\substack{1\leq j \leq k \\ j \neq \sigma(i)}}  \sum_{u \in \setb_{ij}} \frac{\deg(u) \norm{\pj - \vecp^{(\sigma(i))}}^2}{8} - \sum_{i = 1}^k \sum_{\substack{1\leq j \leq k \\ j \neq \sigma(i)}}  \sum_{u \in \setb_{ij}} \deg(u) \norm{\pj - F(u)}^2 \\
        & \geq \sum_{i = 1}^k \sum_{\substack{1\leq j\leq k \\ j \neq \sigma(i)}}  \vol(\setb_{ij}) \frac{\norm{\pj - \vecp^{(\sigma(i))}}^2}{8} - \sum_{i = 1}^k \sum_{u \in \sets_i} \deg(u) \norm{\peye - F(u)}^2 \\
        & \geq \sum_{i = 1}^k \sum_{\substack{1\leq j\leq k \\ j \neq \sigma(i)}}   \frac{\vol(\setb_{ij})}{8 \min\{\vol(\sets_{\sigma(i)}), \vol(\sets_j)\}}\left(\frac{1}{2} - \frac{8}{\Upsilon(k)}\right) - \sum_{i = 1}^k \sum_{u \in \sets_i} \deg(u) \norm{\peye - F(u)}^2 \\
        & \geq \frac{1}{16} \cdot\left( \sum_{i = 1}^k \frac{\vol(\setm_{\sigma, i} \triangle \sets_i)}{\vol(\sets_i)}\right) \left(\frac{1}{2} - \frac{8}{\Upsilon(k)}\right) - \frac{k}{\Upsilon(k)},
    \end{align*}
    \endgroup
    where the second inequality follows by the inequality
    \[
    \norm{\veca - \vecb}^2 \geq \frac{\norm{\vecb - \vecc}^2}{2} - \norm{\veca - \vecc}^2,
    \]
    the third inequality follows since $\vecc_i$ is closer to $\vecp^{(\sigma(i))}$ than $\vecp^{(j)}$, the fifth one follows from Lemma~\ref{lem:p_diff}, and the last one follows by \eqref{eq:up_sym_ratio}.
    
   On the other side, since $\mathrm{COST}(\seta_1,\ldots, \seta_k) \leq \APT\cdot k/\Upsilon(k)$, we have that 
    \[
    \frac{1}{16} \cdot\left( \sum_{i = 1}^k \frac{\vol(\setm_{\sigma, i} \triangle \sets_i)}{\vol(\sets_i)}\right) \left(\frac{1}{2} - \frac{8}{\Upsilon(k)}\right) - \frac{k}{\Upsilon(k)} \leq \APT\cdot\frac{k}{\Upsilon(k)}.
    \] This implies that 
    \begin{align*}
            \sum_{i = 1}^k \frac{\vol(\setm_{\sigma, i} \triangle \sets_i)}{\vol(\sets_i)}   & \leq 16\cdot   (1+\APT)\cdot \frac{k}{\Upsilon(k)}\cdot  \left(\frac{1}{2} - \frac{8}{\Upsilon(k)}\right)^{-1}  \leq 64\cdot   (1+\APT)\cdot \frac{k}{\Upsilon(k)},
    \end{align*}
    where the last inequality holds by the assumption that $\Upsilon(k)\geq 32$. 
\end{proof}
 
 It remains to  study the case in which $\sigma$ isn't a permutation. Notice that, if this occurs,  there is some $i\in[k]$ such that $\setm_{\sigma,i}=\emptyset$, and different values of $x,y\in[k]$ such that $\sigma(x)=\sigma(y)=j$ for some $j\neq i$.
 Based on this, we construct the function $\sigma':[k]\rightarrow [k]$ from $\sigma$ based on the following procedure: 
\begin{itemize}
    \item set $\sigma'(z)=i$ if $z=x$;
    \item set $\sigma'(z)=\sigma(z)$ for any other $z\in [k]\setminus \{x\}$.
\end{itemize}
Notice that one can construct $\sigma'$ in this way as long as $\sigma$ isn't a permutation, and this constructed $\sigma'$ reduces the number of $\setm_{\sigma,i}$ which are equal to $\emptyset$ by one.
We show that one only needs to construct such $\sigma'$ at most $64 (1 + \APT) (k/\Upsilon(k))$ times
to obtain the final permutation called $\sigma^{\star}$, and it holds for $\sigma^{\star}$ that 
\[
        \sum_{i = 1}^k \frac{\vol(\setm_{\sigma^{\star}, i} \triangle \sets_i)}{\vol(\sets_i)} \leq 1088 \cdot (1 + \APT) \cdot \frac{k}{\Upsilon(k)}.
    \]
Combining this with the fact that $\vol(\sets_i) = \Theta(\vol(\vertexsetg)/k)$ for any $i\in[k]$ proves  Theorem~\ref{thm:sc_guarantee}.

\begin{proof}[Proof of Theorem~\ref{thm:sc_guarantee}]
\newcommand{\voldiffi}{\vol(\setm_{\sigma, i} \triangle \sets_i)}
\newcommand{\voldiffip}{\vol(\setm_{\sigma', i} \triangle \sets_i)}
\newcommand{\voldiffjp}{\vol(\setm_{\sigma', j} \triangle \sets_j)}
\newcommand{\voldiffj}{\vol(\setm_{\sigma, j} \triangle \sets_j)}
    By Lemma~\ref{lem:total_cost}, we have
    \[
        \mathsf{COST}(\sets_1, \ldots, \sets_k) \leq  \frac{k}{\Upsilon(k)}.
        \]
    Combining this with the fact that one can apply an approximate $k$-means clustering algorithm with approximation ratio $\APT$ for spectral clustering,  we have that
    \[
        \mathsf{COST}(\seta_1, \ldots, \seta_k) \leq \APT\cdot\frac{k}{\Upsilon(k)}.
        \]
    Then, let $\sigma: [1, k] \rightarrow [1, k]$ be the function which assigns each cluster $\seta_i$ to some ground-truth cluster $\sets_{\sigma(i)}$ such that
    $
        \sigma(i) = \argmin_{j \in [k]} \norm{\pj - \vecc_i}$.
    Then, it holds by Lemma~\ref{lem:cost_lower_bound} that
    \begin{equation} \label{eq:miscl}
        \epsilon(\sigma) \triangleq \sum_{i = 1}^k \frac{\vol(\setm_{\sigma, i} \triangle \sets_i)}{\vol(\sets_i)} \leq 64 \cdot (1 + \APT)\cdot \frac{k}{\Upsilon(k)}.
    \end{equation}
    Now, assume that $\sigma$ is not a permutation, and we'll apply the following procedure inductively to construct a permutation from  $\sigma$. Since $\sigma$ isn't a permutation, there is $i\in[k]$ such that $M_{\sigma, i } = \emptyset$. Therefore, there are different values of $x,y\in[k]$ such that $\sigma(x)=\sigma(y)=j$ for some $j\neq i$.  Based on this, we construct the function $\sigma':[k]\rightarrow[k]$ such that $\sigma'(z)=i$ if $z=x$, and $\sigma'(z)=\sigma(z)$ for any other $z\in [k]\setminus \{x\}$. Notice that we can construct $\sigma'$ in this way as long as $\sigma$ isn't a permutation.
    By the definition of $\epsilon(\cdot)$ and function $\sigma'$, the difference between $\epsilon(\sigma')$ and $\epsilon(\sigma)$ can be written as
   \begin{align} 
      \lefteqn{\epsilon(\sigma') - \epsilon(\sigma)}\nonumber\\
       & = \underbrace{ \left(\frac{\voldiffip}{\vol(\sets_i)} - \frac{\voldiffi}{\vol(\sets_i)}\right)}_{\alpha} + \underbrace{\left(\frac{\voldiffjp}{\vol(\sets_j)} - \frac{\voldiffj}{\vol(\sets_j)}\right)}_{\beta}.\label{eq:sigmadiff}
   \end{align} 
   Let us consider $4$ cases based on the sign of $\alpha,\beta$ defined above.
   In each case, we   bound the cost introduced by the change from $\sigma$ to $\sigma'$, and then consider the total cost introduced throughout the entire procedure of constructing a permutation.

   \textbf{Case~1:  $\alpha<0, \beta<0$.}
   In this case, it is clear that $\epsilon(\sigma') - \epsilon(\sigma) \leq 0$, and hence the total introduced  cost 
   is at most $0$.
   
   \textbf{Case~2: $\alpha>0, \beta<0$.}
   In this case, we have 
   \begingroup
   \allowdisplaybreaks
   \begin{align*}
       \epsilon(\sigma') - \epsilon(\sigma) & \leq
       \begin{multlined}[t]
        \frac{1}{\min(\vol(\sets_i), \vol(\sets_j))} \Big( \voldiffip - \voldiffi \qquad \\
        + \abs{\voldiffjp - \voldiffj} \Big)
        \end{multlined}\\
       & =
       \begin{multlined}[t]
       \frac{1}{\min(\vol(\sets_i), \vol(\sets_j))} \Big( \voldiffip - \voldiffi \qquad \\
       + \voldiffj - \voldiffjp \Big)
       \end{multlined}\\
       & =
       \begin{multlined}[t]
       \frac{1}{\min(\vol(\sets_i), \vol(\sets_j))} \big( \vol(\seta_x \setminus \sets_i) - \vol(\seta_x \intersect \sets_i) \qquad \\
       + \vol(\seta_x \setminus \sets_j) - \vol(\seta_x \intersect \sets_j) \big)
       \end{multlined}\\
       & \leq \frac{2 \cdot \vol(\seta_x \setminus \sets_j)}{\min(\vol(\sets_i), \vol(\sets_j))} \\
       & \leq \frac{8 \cdot \vol(\seta_x \setminus \sets_j)}{\vol(\sets_j)},
   \end{align*}
   \endgroup
   where the last inequality follows by the  fact that the clusters are almost balanced.
   Since each set $\seta_x$ is moved at most once in order to construct a permutation, the total introduced cost due to this case  is at most
   \begin{align*}
       \sum_{j = 1}^k \sum_{\seta_x \in \setm_{\sigma, j}} \frac{8 \cdot \vol(\seta_x \setminus \sets_j)}{\vol(\sets_j)} \leq 8 \cdot \sum_{j = 1}^k \frac{\voldiffj}{\vol(\sets_j)} \leq 512\cdot \left(1 + \APT\right)\cdot  \frac{k}{\Upsilon(k)}.
   \end{align*}
   
   \textbf{Case~3: $\alpha>0, \beta>0$.}
   In this case, we have
   \begin{align*}
       \epsilon(\sigma') - \epsilon(\sigma) & \leq
       \begin{multlined}[t]
       \frac{1}{\min(\vol(\sets_i), \vol(\sets_j))} \big( \voldiffip - \voldiffi \qquad \\
       + \voldiffjp - \voldiffj \big)
       \end{multlined} \\
        & =
        \begin{multlined}[t]
        \frac{1}{\min(\vol(\sets_i), \vol(\sets_j))} \big( \vol(\seta_x \setminus \sets_i) - \vol(\seta_x \cap \sets_i) \qquad \\
         + \vol(\seta_x \cap \sets_j) -  \vol(\seta_x \setminus \sets_j) \big) 
        \end{multlined} \\
        & \leq  \frac{1}{\min(\vol(\sets_i), \vol(\sets_j))} \left(2\cdot  \vol(\sets_x \cap \sets_j) \right) \\
        & \leq \frac{2\cdot  \vol(\sets_j)}{\min(\vol(\sets_i), \vol(\sets_j))} \\
        & \leq 8,
   \end{align*}
   where the last inequality follows by  the fact that the clusters are almost balanced.
   We will consider the total introduced cost due to this case and Case~4 together, and   so let's first examine Case~4.
   
   \textbf{Case~4: $\alpha<0,\beta>0$.}
   In this case, we have
   \begin{align*}
       \epsilon(\sigma') - \epsilon(\sigma) & \leq \frac{1}{\vol(\sets_j)} \left(\voldiffjp - \voldiffj\right) \\
       & \leq \frac{1}{\vol(\sets_j)} \left( \vol(\seta_x \intersect \sets_j) - \vol(\seta_x \setminus \sets_j)\right)\\
       & \leq \frac{\vol(\sets_j)}{\vol(\sets_j)} \\
       & = 1.
   \end{align*}
   Now, let us bound the total number of times we need to construct $\sigma'$ in order to obtain a permutation.
   For any $i$  with $\setm_{\sigma,i}=\emptyset$,  we have
    \[
        \frac{\vol(\setm_{\sigma, i} \triangle \sets_i)}{\vol(\sets_i)} = \frac{\vol(\sets_i)}{\vol(\sets_i)} = 1,
    \]
    so the total number of required iterations is upper bounded  by
    \[
        \cardinality{\{i : \setm_{\sigma,i}=\emptyset \}} \leq \sum_{i = 1}^k \frac{\vol(\setm_{\sigma, i} \triangle \sets_i)}{\vol(\sets_i)} \leq 64\cdot (1 + \APT)\cdot  \frac{k}{\Upsilon(k)}.
    \]
As such, the total introduced cost due to Cases~3 and 4 is at most  
    \[
        8 \cdot 64\cdot (1 + \APT)\cdot  \frac{k}{\Upsilon(k)} = 512\cdot (1 + \APT) \cdot \frac{k}{\Upsilon(k)}.
    \]
    Putting everything together, we have that  
    \[
        \epsilon(\sigma^\star) \leq \epsilon(\sigma) + 1024\cdot (1 + \APT)\cdot \frac{k}{\Upsilon(k)} \leq 1088\cdot (1 + \APT)\cdot  \frac{k}{\Upsilon(k)}.
    \]
    This implies that
    \[
        \sum_{i = 1}^k \vol(\setm_{\sigma^\star, i} \triangle \sets_i) \leq 2176\cdot (1 + \APT)\cdot  \frac{\vol(\vertexsetg)}{\Upsilon(k)},
    \]
    and completes the proof.
\end{proof}

It is worth noting that this method of upper bounding the ratio of misclassified vertices is very different from the ones used in previous references, e.g.,~\cite{deySpectralConcentrationGreedy2019,mizutaniImprovedAnalysisSpectral2021,pengPartitioningWellClusteredGraphs2017}. In particular, instead of examining all the possible mappings between $\{\seta_i\}_{i=1}^k$ and $\{\sets_i\}_{i=1}^k$, 
in this proof we directly work with some specifically defined function $\sigma$, and constructed our desired mapping $\sigma^{\star}$ from $\sigma$. This is another key for us to obtain stronger results than the previous work.

\section{Future Work}
Notice that the analysis presented in this chapter requires that the optimal clusters have almost-balanced size, which is the only technical assumption of our result. This raises a natural question on the necessity of such a condition.
\begin{openquestion}
Is it possible for spectral clustering to theoretically achieve the same approximation guarantee as Theorem~\ref{thm:sc_guarantee} when the optimal clusters of an input graph have arbitrary sizes? Alternatively, is there a hard instance showing that our presented analysis is tight up to some constant factor?
\end{openquestion}
Secondly, it would be very interesting to see whether
one can theoretically analyse the performance of spectral clustering   with respect to $\lambda_{k+1}/\lambda_{k}$, instead of $\Upsilon(k)$.
\begin{openquestion}
Is it possible to upper bound the number~(or volume) of misclassified vertices from  spectral clustering with respect to the ratio between $\lambda_{k+1}$ and $\lambda_k$?
\end{openquestion}
Notice that, in contrast to $\Upsilon(k)$, the value of $\lambda_{k+1}/\lambda_k$ can be  computed in polynomial time; therefore, an affirmative answer to the question will have more practical impact than our current presented result.
In addition, such an answer would further justify the eigen-gap heuristic~\cite{vonluxburgTutorialSpectralClustering2007}, helping to close
the gap between the theoretical analysis of spectral clustering and its empirical performance.

\chapter{Spectral Clustering with Meta-Graphs} \label{chap:meta}
In this chapter we   build on the analysis of the last chapter and develop a variant of spectral clustering which employs fewer than $k$ eigenvectors to find $k$ clusters.
We prove the surprising result that,
when the structure among the optimal clusters in an input graph satisfies certain conditions,
spectral clustering with fewer eigenvectors is able to produce better results than classical spectral clustering.
This presents the opportunity to significantly speed up the runtime of spectral clustering in practice, since fewer eigenvectors are used to construct the embedding.

The starting point of our new analysis is the following observation: in many graphs arising from real-world data, the structure of clusters is not symmetric.
For example, let us consider the well-known MNIST dataset of handwritten digits~\cite{hullDatabaseHandwrittenText1994}.
This dataset consists of 70,000 images of handwritten digits 0 to 9.
Each image has $28 \times 28$ pixels, and each pixel has a greyscale value between 0 and 255.
Figure~\ref{fig:mnist_example} shows some example images from the MNIST dataset.

\begin{figure}[b]
    \centering
    \includegraphics[width=0.35\textwidth]{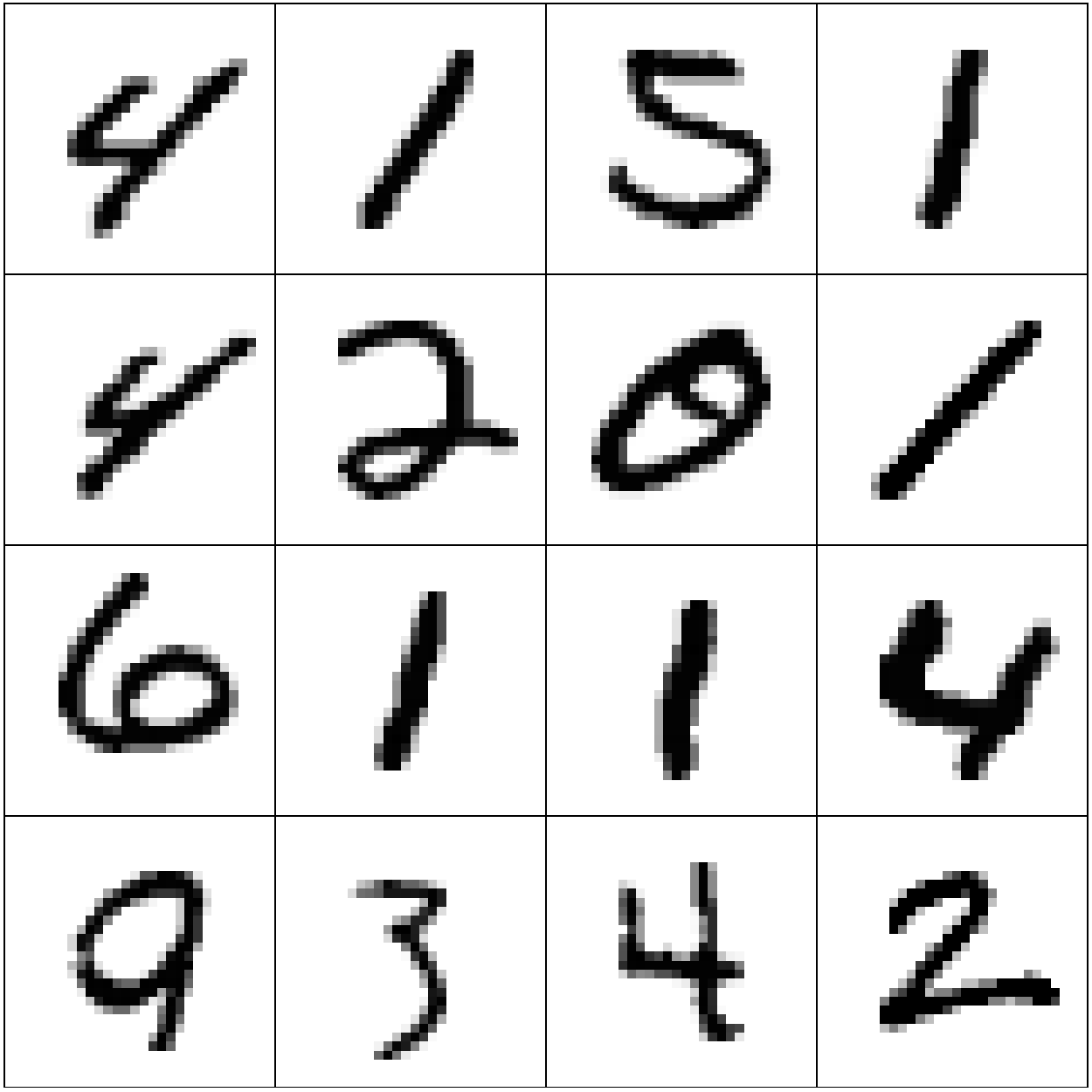}
    \caption[Example images from the MNIST dataset]{Example images from the MNIST dataset.}
    \label{fig:mnist_example}
\end{figure}

Then, we can represent each image in the dataset by a vector $\vecx_i \in \R^{28 \times 28}$ and construct the \firstdef{$\bm{k}$ nearest neighbours graph}, $\geqve$, in the following way:
\begin{itemize}
    \item for each data point $\vecx_i$ in the dataset, there is a corresponding vertex $v_i \in \vertexset$;
    \item for each $v_i$ there is an undirected edge to the $k$ vertices representing the data points with the smallest distances to $\vecx_i$.
\end{itemize}
Figure~\ref{fig:mnist_knn} illustrates the $k$ nearest neighbour graph for a randomly chosen subset of the data.
\begin{figure}[t]
    \centering
    \includegraphics[width=0.5\textwidth]{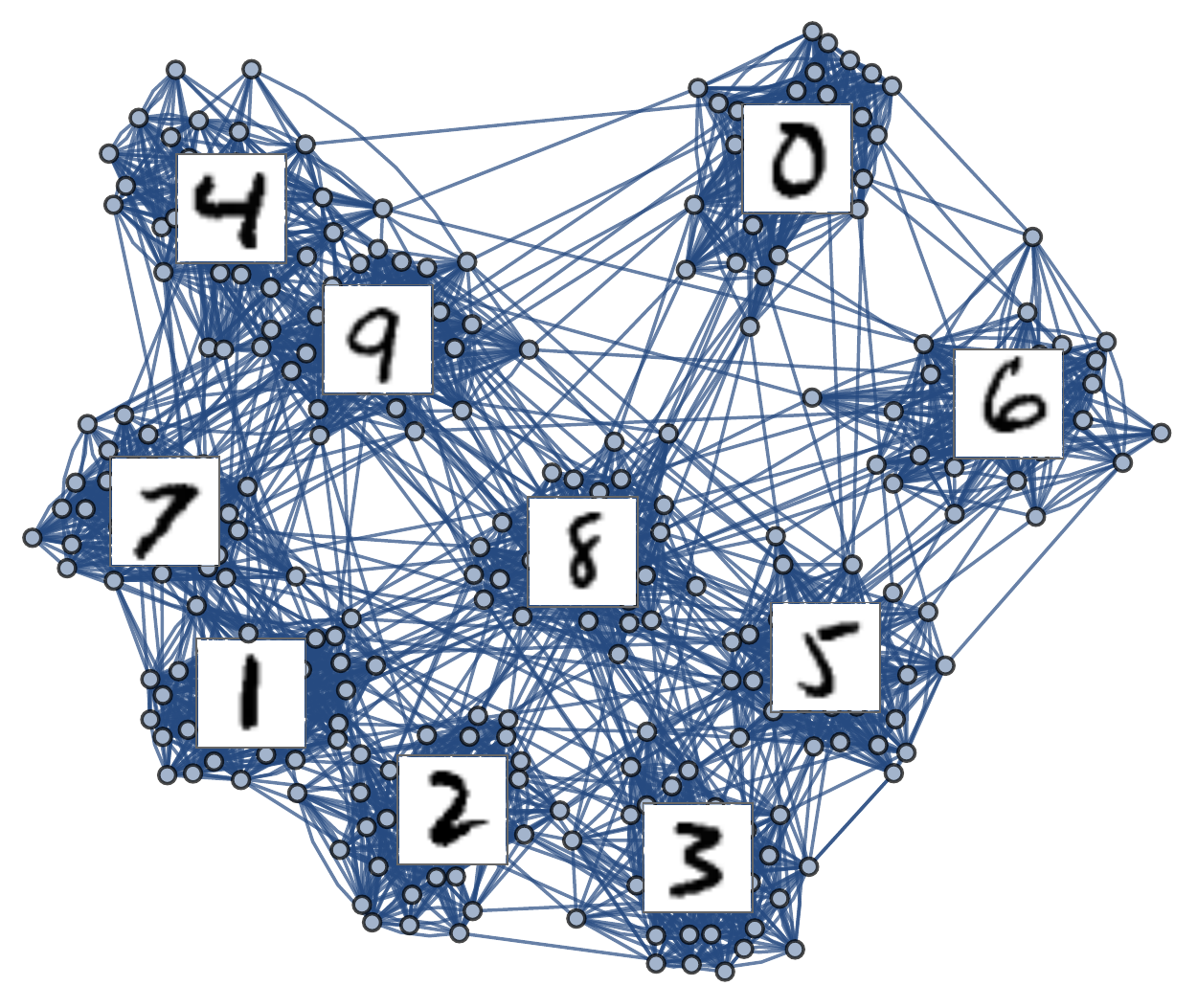}
    \caption[The $k$ nearest neighbour graph of the MNIST dataset]{The $k$ nearest neighbour graph of a random subset of the MNIST dataset, for $k = 3$. The graph has a cluster corresponding to each digit.}
    \label{fig:mnist_knn}
\end{figure}
Notice that some of the clusters in this graph are more similar than others.
For example, the cluster representing the digit $4$ has more connections with the cluster representing $9$ than the one representing $6$.
The central result of this chapter is that this \emph{asymmetry} can be used to improve the performance of spectral clustering, and   allow us to recover all $10$ clusters while using fewer than $10$ eigenvectors of the graph Laplacian.
 
\begin{mainresult}[See Theorem~\ref{thm:sc_meta-graph} for the formal statement]
Given a graph $\graphg$, let $\sets_1, \ldots, \sets_k$ be the clusters which achieve $\kcond(k)$ and suppose that $\sets_1, \ldots, \sets_k$ are almost-balanced.
Additionally, suppose that the clusters $\sets_1, \ldots, \sets_k$ exhibit some significant high-level structure.
Then, spectral clustering with fewer than $k$ eigenvectors   performs better than spectral clustering with $k$ eigenvectors.
\end{mainresult}

The significance of this result is further demonstrated by experiments on the BSDS image segmentation dataset.
It is well known that spectral methods are an effective technique for image segmentation~\cite{shiNormalizedCutsImage2000, tungEnablingScalableSpectral2010, zhangImageSegmentationBased2021}, and in this chapter we see that spectral clustering is effective for image segmentation, and that using fewer eigenvectors often results in a better segmentation.
Figure~\ref{fig:bsds_results_intro} shows examples in which, in order to find $6$ and $45$ clusters, spectral clustering with $3$ and $7$ eigenvectors produce better results than the ones with $6$ and $45$ eigenvectors.
Section~\ref{sec:metaExperiments} describes the full experimental details.

\begin{figure*}[t]
    \centering
    \begin{subfigure}{0.3\textwidth} 
    \includegraphics[width=\textwidth]{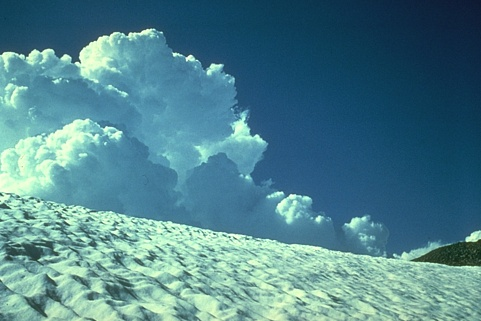}
    \caption{Original Image}
    \end{subfigure}
    \hspace{1em}
    \begin{subfigure}{0.3\textwidth} 
    \includegraphics[width=\textwidth]{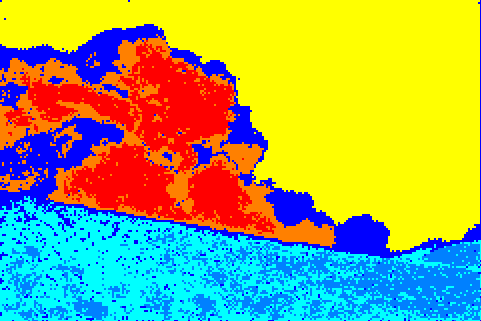}
    \caption{$6$ clusters with $3$ vectors}
    \end{subfigure}
    \hspace{1em}
    \begin{subfigure}{0.3\textwidth} 
    \includegraphics[width=\textwidth]{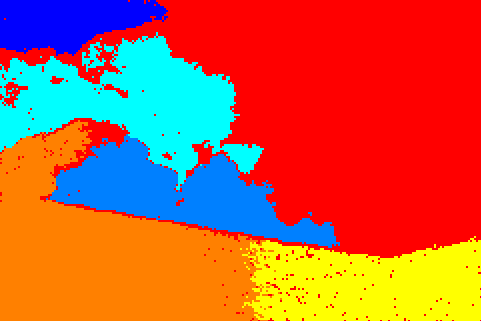}
    \caption{$6$ clusters with $6$ vectors}
    \end{subfigure}
    \par\bigskip
    \begin{subfigure}{0.3\textwidth} 
    \includegraphics[width=\textwidth]{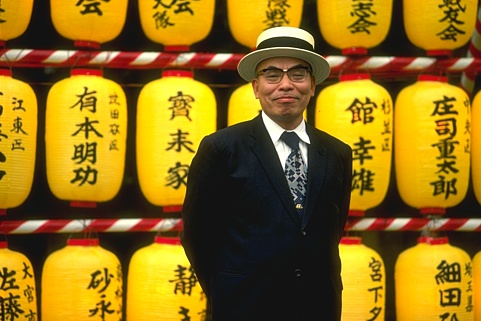}
    \caption{Original Image}
    \end{subfigure}
    \hspace{1em}
    \begin{subfigure}{0.3\textwidth} 
    \includegraphics[width=\textwidth]{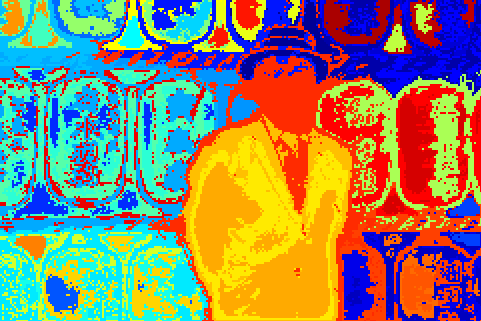}
    \caption{$45$ clusters with $7$ vectors}
    \end{subfigure}
    \hspace{1em}
    \begin{subfigure}{0.3\textwidth} 
    \includegraphics[width=\textwidth]{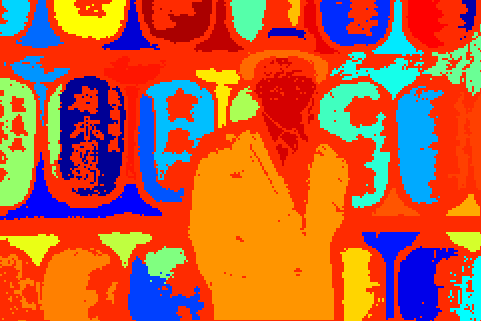}
    \caption{$45$ clusters with $45$ vectors}
    \end{subfigure}
    \caption[Examples of image segmentation using spectral clustering]{Examples of image segmentation using spectral clustering; the original images are from the BSDS dataset.  The Rand Index of segmentation (b) is $0.83$, while (c) has Rand Index $0.78$. Segmentation (e) has Rand Index $0.92$, and (f) has Rand Index $0.80$. Hence, it's clear that spectral clustering with fewer than $k$ eigenvectors suffices to produce comparable or better output.} 
    \label{fig:bsds_results_intro}
\end{figure*}

The result in this chapter is related to previous work on
efficient spectral algorithms to find cluster-structures.
While it's known that flow and path structures of clusters in directed graphs can be uncovered with complex-valued Hermitian  matrices~\cite{cucuringuHermitianMatricesClustering2020, laenenHigherOrderSpectralClustering2020}, our result shows that one can apply  real-valued Laplacians of undirected graphs, and find more general patterns of clusters.
 
\section{Encoding the Cluster-Structure into  Meta-Graphs}\label{sec:structure++}
Suppose that $\{\sets_i\}_{i=1}^k$ is a partition of $\vertexsetg$ that corresponds to the $k$-way expansion $\rho(k)$ for an input graph $\graphg$. We define the  matrix $\adj_{\graphm}\in\R^{k\times k}$ by
\[
    \adj_\graphm(i, j) = \twopartdef{\weight(\sets_i, \sets_j)}{i \neq j,}{2 \weight(\sets_i, \sets_j)}{i = j}
\]
and, taking $\adj_\graphm$ to be the adjacency matrix, this 
defines a graph $\graphm = (\vertexset_\graphm, \edgeset_\graphm, \weight_\graphm)$ which we   refer to as the \firstdef{meta-graph} of the clusters.
Informally, the meta-graph is constructed by viewing every cluster $\sets_i$ as a `giant vertex', and captures the structure of the optimal clusters.
Then, define the normalised adjacency matrix of $\graphm$ by 
$\adjn_\graphm \triangleq \degmhalfneg_\graphm \adj_\graphm \degmhalfneg_\graphm$,
and the normalised Laplacian matrix of $\graphm$ by $\lapn_\graphm \triangleq \identity - \adjn_\graphm$.
Let the eigenvalues of $\lapn_\graphm$ be $\gamma_1\leq\ldots\leq\gamma_k$, and   $\vecg_i\in\R^k$ be the eigenvector corresponding to $\gamma_i$ for any $i\in[k]$.

The starting point of our novel approach is to look at the structure of this meta-graph $\graphm$, and study how the spectral information of $\lapn_\graphm\in\R^{k\times k}$ is encoded in the bottom eigenvectors of $\lapn_\graphg$.
To achieve this, 
 for any    $ \ell\in [k]$ and vertex $i\in \vertexset_\graphm$, let
\begin{equation}\label{eq:defxi}
\bar{\vecx}^{(i)} \triangleq \left( \vecg_1(i),\ldots, \vecg_{\ell}(i) \right)^\transpose;
\end{equation}
notice that  $\bar{\vecx}^{(i)}\in\R^{\ell}$ defines the spectral embedding of  $i\in \vertexset_\graphm$ with the bottom $\ell$ eigenvectors of $\lapn_\graphm$.

\begin{definition}[$(\theta,\ell)$-distinguishable graph]\label{def:distinguishable} For any  $\graphm=(\vertexset_\graphm, \edgeset_\graphm, \weight_\graphm)$ with  $k$ vertices, $\ell\in[k]$, and $\theta\in\R_{\geq 0}$, we say that $\graphm$ is  $(\theta,\ell)$-distinguishable if 
\begin{itemize}
    \item it holds for any $i\in[k]$ that  $\norm{\barx^{(i)}}^2 \geq \theta$, and 
    \item it holds for any different $i,j\in[k]$  that  $
    \left\|\frac{\barx^{(i)}}{\norm{\barx^{(i)}}} - \frac{\barx^{(j)}}{\norm{\barx^{(j)}}}\right\|^2\geq \theta$.
\end{itemize}
\end{definition}
In other words, graph $\graphm$ is $(\theta,\ell)$-distinguishable if (i) every embedded point $\bar{\vecx}^{(i)}$ has squared length at least $\theta$, and (ii) any pair of different embedded points with normalisation are separated by
 a distance
of at least $\theta$.
By definition, it is easy to see that, if $\graphm$ is $(\theta,\ell)$-distinguishable for some large value of $\theta$, then the embedded points $\{\bar{\vecx}^{(i)} \}_{i\in \vertexset_\graphm}$ can be easily separated even if $\ell< k$.
Examples~\ref{ex:cycle} and \ref{ex:grid} below demonstrate that this is indeed the case and, since the meta-graph $\graphm$ is constructed from  $\{\sets_i\}_{i=1}^k$, this well-separation property for $\{\bar{\vecx}^{(i)} \}_{i\in \vertexset_\graphm}$ usually implies 
that the clusters $\{\sets_i\}_{i = 1}^k$ are also well separated when the vertices are mapped to the points $\{F(u)\}_{u \in \vertexset_\graphg}$, in which
\begin{equation}\label{eq:embeddingell}
   F(u)\triangleq \frac{1}{\sqrt{\deg(u)}}\cdot  \left(\vecf_1(u), \ldots \vecf_{\ell}(u)\right)^\transpose.
\end{equation}

\begin{example} \label{ex:cycle}
 Suppose the meta-graph is $C_6$, the cycle on $6$ vertices.
 Figure~\ref{fig:cycle_example}(\subref{subfig:cycle_embedding}) shows that the vertices of $C_6$ are well separated by the second and third eigenvectors of $\lapn_{C_6}$.\footnote{Notice that the first eigenvector is always the constant vector and gives no useful information. This is why we visualise the second and third eigenvectors only.}
 Since the minimum distance between any pair of vertices in this embedding is $2/3$, we say that $C_6$ is $\left(2/3, 3\right)$-distinguishable.
 Figure~\ref{fig:cycle_example}(\subref{subfig:cycle_sbm_embedding}) shows that, when using $\vecf_2, \vecf_3$ of $\lapn_G$ to embed the vertices of a $600$-vertex graph with a cyclical cluster pattern, the embedded points closely match the ones from the meta-graph.~\footnote{In Examples~\ref{ex:cycle}~and~\ref{ex:grid}, the graphs $G$ are generated according to a generalised stochastic block model based on the defined meta-graph. The model is formally described in Section~\ref{sec:metaExperiments}.}
\end{example}
 
\begin{figure}[ht]
\centering
\begin{subfigure}{0.45\textwidth}
        \includegraphics[width=\textwidth]{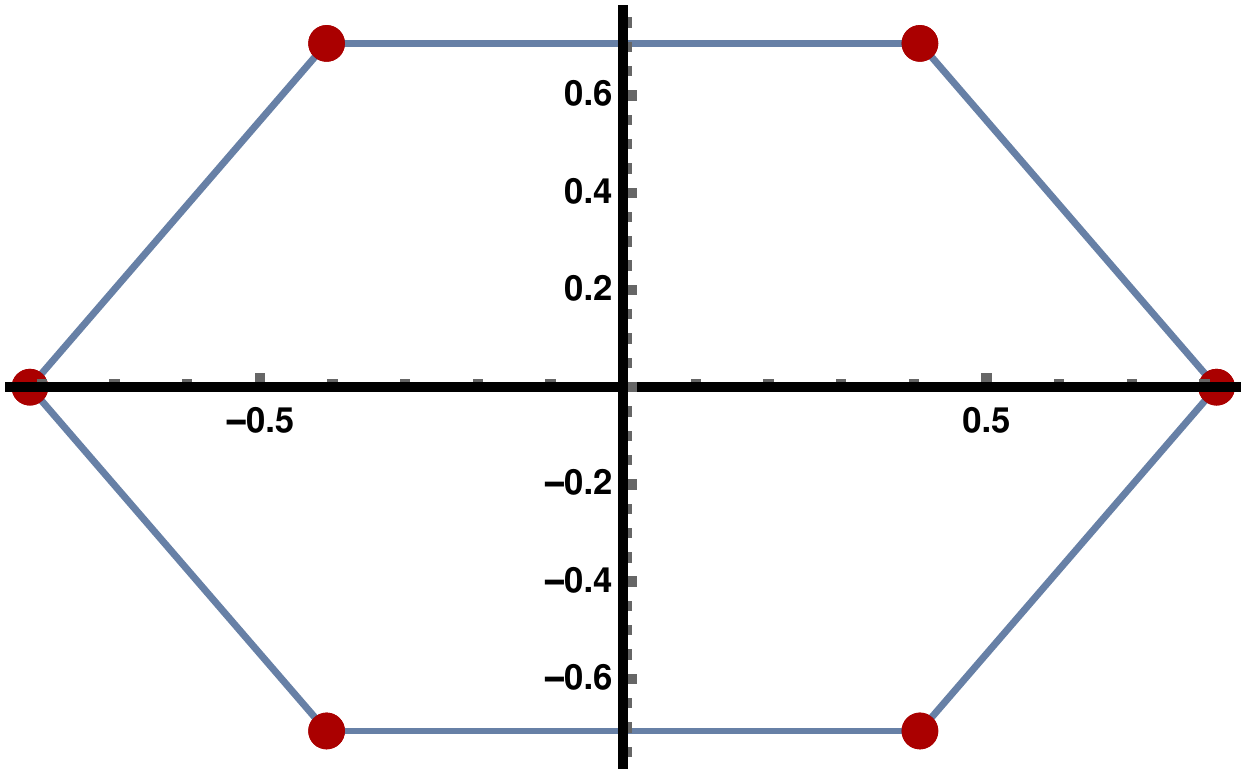}%
    \caption{
    Spectral embedding of $C_6$ with $\vecg_2$ and $\vecg_3$.
    \label{subfig:cycle_embedding}   
    }
\end{subfigure}
\hspace{1em}
\begin{subfigure}{0.45\textwidth}
    \includegraphics[width=\textwidth]{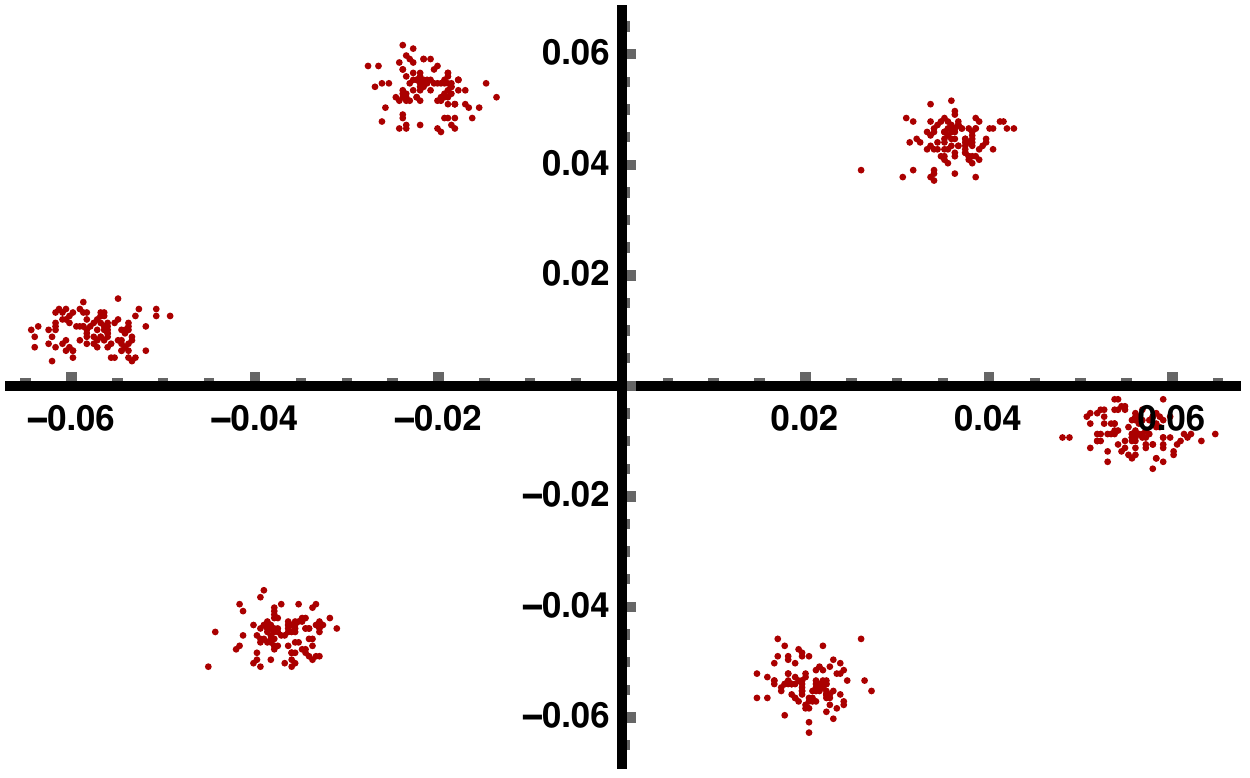}%
    \caption{Spectral embedding of $\graphg$ with $\vecf_2$ and $\vecf_3$.
    \label{subfig:cycle_sbm_embedding}
    }
\end{subfigure}
    \caption[Spectral embedding of a graph whose clusters exhibit a cyclical structure]{The spectral embedding of a large graph $\graphg$ whose clusters exhibit a cyclical structure closely matches the embedding of the meta-graph $C_6$.
    \label{fig:cycle_example}
    }
\end{figure}
 
\begin{example} \label{ex:grid}
Suppose the meta-graph is $P_{4, 4}$, which is the $4 \times 4$ grid graph.
Figure~\ref{fig:grid_example}(\subref{subfig:grid_embedding}) shows that the vertices are well separated using the second and third eigenvectors of $\lapn_{P_{4, 4}}$.
The minimum distance between any pair of vertices in this embedding is roughly $0.1$, and so $P_{4, 4}$ is $(0.1, 3)$-distinguishable.
 Figure~\ref{fig:grid_example}(\subref{subfig:grid_sbm_embedding}) demonstrates that,
 when using $\vecf_2, \vecf_3$ of $\lapn_\graphg$ to construct the embedding, the embedded points    closely match the ones from the meta-graph.
\end{example}

\begin{figure}[ht]
\centering
\begin{subfigure}{0.45\textwidth}
        \includegraphics[width=\textwidth]{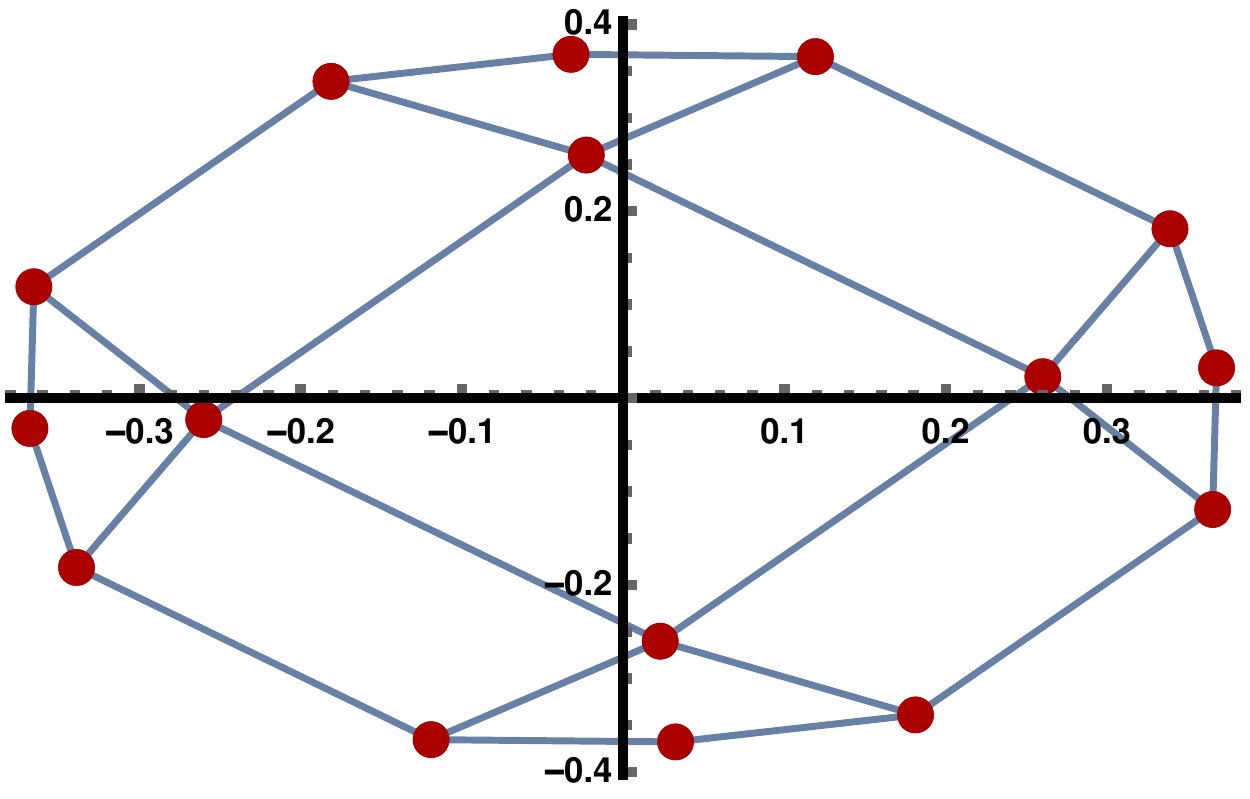}%
        \caption{
        Embedding of $P_{4, 4}$ with $\vecg_2$ and $\vecg_3$.
        \label{subfig:grid_embedding}
        }
\end{subfigure}
\hspace{1em}
\begin{subfigure}{0.45\textwidth}
    \includegraphics[width=\textwidth]{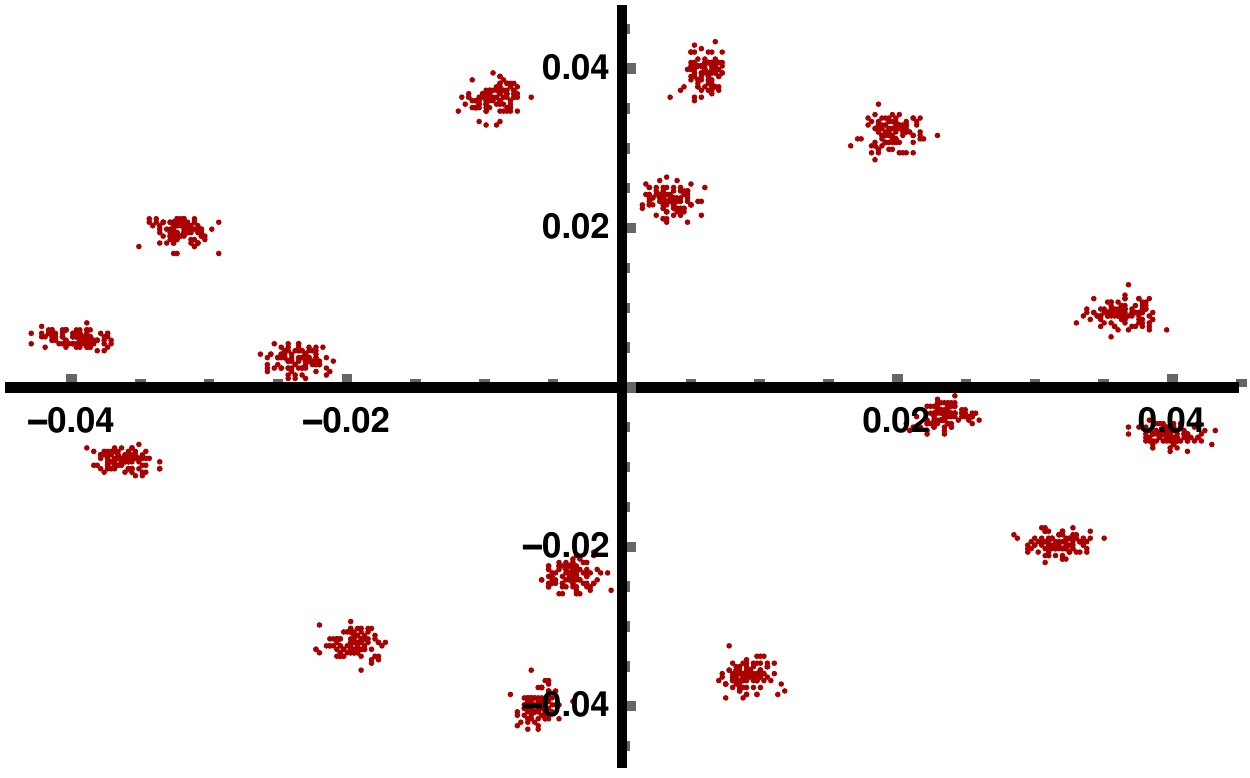}
    \caption{Embedding of $\graphg$ with $\vecf_2$ and $\vecf_3$.
    \label{subfig:grid_sbm_embedding}    
    }
\end{subfigure}
    \caption[Spectral embedding of a graph whose clusters exhibit a grid structure]{
    The spectral embedding of a large $1600$-vertex graph $\graphg$ whose clusters exhibit a grid structure closely matches the embedding of the meta-graph $P_{4, 4}$.
    \label{fig:grid_example}
    }
\end{figure}
 
From these examples, it is clear to see that   there is a close
connection between the embedding $\{\bar{\vecx}^{(i)}\}$ defined in \eqref{eq:defxi} and the embedding $\{F(u)\}$ defined in \eqref{eq:embeddingell}. To formally analyse this connection, we  develop a new structure theorem with meta-graphs. 
 
We define vectors $\barg_1, \ldots, \barg_k \in \R^n$ which represent the eigenvectors of $\lapn_\graphm$ `blown-up' to $n$ dimensions.
Formally, we define $\barg_i$ such that
\[
    \barg_i = \sum_{j = 1}^{k} \frac{\degmhalf \indicatorvec_j}{\norm{\degmhalf \indicatorvec_j}}\cdot \vecg_i(j),
\]
where $\indicatorvec_j \in \R^n$ is the indicator vector of cluster $\sets_j$ defined in Section~\ref{sec:prelim:graphs}.
By definition, it holds for any $u\in \sets_j$ that 
\[
    \barg_i(u) = \sqrt{\frac{\deg(u)}{\vol(\sets_j)}}\cdot \vecg_i(j).
\]
The following lemma shows  that $\{\bar{\vecg}_i\}_{i=1}^k$ form an orthonormal basis.

\begin{lemma} \label{lem:barg_ortho}
The following statements hold:
\begin{enumerate}
    \item it holds for any $i\in[k]$ that  $\norm{\barg_i} = 1$;
    \item it holds for any different  $i, j\in [k]$ that  $\langle \barg_i, \barg_j\rangle = 0$.
\end{enumerate}
\end{lemma}
\begin{proof}
By definition, we have that 
\begin{align*}
        \norm{\barg_i}^2 = \barg_i^\transpose \barg_i & = \sum_{j = 1}^k \sum_{u \in \sets_j} \barg_i(u)^2 \\
        & = \sum_{j=1}^k \sum_{u \in \sets_j} \left(\frac{\sqrt{\deg(u)}}{\sqrt{\vol(\sets_j)}} \cdot  \vecg_i(j) \right)^2  \\
        & = \sum_{j = 1}^k \vecg_i(j)^2 = \norm{\vecg_i}^2 = 1,
    \end{align*}
    which proves the first statement. 

    To prove the second statement, we have for any $i \neq j$ that 
    \begin{align*}
        \langle \barg_i, \barg_j\rangle & = \sum_{x = 1}^{k} \sum_{u \in \sets_x} \barg_i(u) \barg_j(u) \\
        & = \sum_{x = 1}^{k} \sum_{u \in \sets_x} \frac{\deg(u)}{\vol(\sets_x)}\cdot \vecg_i(x) \vecg_j(x) \\
        & = \sum_{x = 1}^k \vecg_i(x) \vecg_j(x) \\
        & = \vecg_i^\transpose \vecg_j = 0,
    \end{align*}
    which completes the proof.
\end{proof}

The following lemma presents an
  important relationship between the eigenvalues of $\lapn_\graphm$ and those of $\lapn_\graphg$.

\begin{lemma} \label{lem:lambda_leq_gamma2}
 It holds for any $i \in [k]$ that  $\lambda_i \leq \gamma_i$.
\end{lemma}

\begin{proof}
    Notice that we have for any $j \leq k$ that 
    \begin{align*}
        \barg_j^\transpose \lapn_\graphg \barg_j & = \sum_{(u, v) \in \edgesetg} \weight_\graphg(u, v) \left(\frac{\barg_j(u)}{\sqrt{\deg(u)}} - \frac{\barg_j(v)}{\sqrt{\deg(v)}}\right)^2 \\
        & = \sum_{x = 1}^{k - 1} \sum_{y = x + 1}^k \sum_{a \in \sets_x} \sum_{b \in \sets_y} \weight(a, b) \left(\frac{\barg_j(a)}{\sqrt{\deg(a)}} - \frac{\barg_j(b)}{\sqrt{\deg(b)}} \right)^2 \\
        & = \sum_{x = 1}^{k - 1} \sum_{y = x + 1}^k \weight(\sets_x, \sets_y) \left(\frac{\vecg_j(x)}{\sqrt{\vol(\sets_x)}} - \frac{\vecg_j(y)}{\sqrt{\vol(\sets_y)}} \right)^2 \\
        &  = \vecg_j^\transpose \lapn_\graphm \vecg_j = \gamma_j. 
    \end{align*}
    By Lemma~\ref{lem:barg_ortho}, we have an $i$-dimensional subspace $\sets_i \subset \R^n$ such that
    \[
        \max_{\vecx \in \sets_i} \frac{\vecx^\transpose \lapn_\graphg \vecx}{\vecx^\transpose \vecx} = \gamma_i,
    \]
    from which the lemma follows by the Courant-Fischer theorem (Theorem~\ref{thm:courantfischer}).
\end{proof}

Next, similar to the function $\Upsilon(k)$
defined in \eqref{eq:defineupsilon}, for any input graph $\geqvewg$ and   $(\theta, \ell)$-distinguishable meta-graph $\graphm$, we define the function $\Psi(\ell)$ 
by
\[
        \Psi(\ell) \triangleq \sum_{i = 1}^{\ell} \frac{\gamma_i}{\lambda_{\ell + 1}}.
    \]
Notice that we have by the higher-order Cheeger inequality that $\gamma_i/2 \leq  \rho_\graphm(i)$ holds for any $i\in[\ell]$, and $\rho_\graphm(i)\leq \rho_\graphg(k)$ by the construction of matrix $\adj_\graphm$. Hence, one can view $\Psi(\ell)$ as a refined version of $l / \Upsilon(k)$.

We   now see that the vectors $\vecf_1, \ldots, \vecf_{\ell}$ and $\barg_1, \ldots, \barg_{\ell}$ are well approximated by each other.
In order to show this, we define for any $ i\in [\ell]$ the vectors
\[
    \hatf_i = \sum_{j = 1}^{\ell} \langle \barg_i, \vecf_j \rangle \vecf_j~\mbox{\ \ \ and\ \ \ }~\hatg_i = \sum_{j = 1}^{\ell} \langle \vecf_i, \barg_j \rangle \barg_j,
\]
and can prove the following structure theorem with meta-graphs.

\begin{theorem}[The Structure Theorem with Meta-Graphs] \label{thm:structure++} 
The following statements hold:
\begin{enumerate}
    \item it holds for  any $i \in [\ell]$ that $\norm{\barg_i - \hatf_i}^2 \leq \gamma_i / \lambda_{\ell + 1}$;
    \item it holds for any $ \ell\in [k]$ that 
    \[
        \sum_{i = 1}^{\ell} \norm{\vecf_i - \hatg_i}^2 \leq \Psi(\ell).
    \]
\end{enumerate}

\end{theorem}
\begin{proof}
For the first statement, we    write $\barg_i$ as a linear combination of the vectors $\vecf_1, \ldots, \vecf_n$, i.e., 
    $
        \barg_i = \sum_{j = 1}^n \langle \barg_i, \vecf_j \rangle \vecf_j$.
    Since  $\hatf_i$ is a projection of $\barg_i$, we have that $\barg_i - \hatf_i$ is perpendicular to $\hatf_i$, and that
    \begin{align*}
        \norm{\barg_i - \hatf_i}^2 & = \norm{\barg_i}^2 - \norm{\hatf_i}^2 \\
        & = \left(\sum_{j = 1}^n \langle \barg_i, \vecf_j \rangle^2 \right) - \left(\sum_{j = 1}^{\ell} \langle \barg_i, \vecf_j \rangle^2 \right) \\
        & = \sum_{j = \ell + 1}^n \langle \barg_i, \vecf_j \rangle^2.
    \end{align*}
    Now, we study the quadratic form $\barg_i^\transpose \lapn_\graphg \barg_i$ and have that 
    \begin{align*}
        \barg_i^\transpose \lapn_\graphg \barg_i & = \left(\sum_{j = 1}^n \langle \barg_i, \vecf_j \rangle \vecf_j^\transpose \right) \lapn_\graphg \left(\sum_{j = 1}^n \langle \barg_i, \vecf_j \rangle \vecf_j\right) \\
        & = \sum_{j = 1}^n \langle \barg_i, \vecf_j \rangle^2 \lambda_j \\
        & \geq \lambda_{\ell + 1} \norm{\barg_i - \hatf_i}^2.
    \end{align*}
    By the proof of Lemma~\ref{lem:lambda_leq_gamma2}, we have that $\barg_i^\transpose \lapn_\graphg \barg_i \leq \gamma_i$, from which the first statement follows.
   
    Now we prove the second statement. 
    We define the vectors  $\barg_{k+1},\ldots, \barg_n$   to be an arbitrary orthonormal basis of the space orthogonal to the space spanned by 
    $\barg_1, \ldots, \barg_k$. Then, we can write any $\vecf_i$ as
    $
        \vecf_i = \sum_{j = 1}^n \langle \vecf_i, \barg_j \rangle \barg_j$,
    and   have that 
    \begin{align*}
        \sum_{i = 1}^{\ell} \norm{\vecf_i - \hatg_i}^2 & = \sum_{i = 1}^{\ell} \left(\norm{\vecf_i}^2 - \norm{\hatg_i}^2 \right) \\
        & = \ell - \sum_{i = 1}^{\ell} \sum_{j = 1}^{\ell} \langle \vecf_i, \barg_j \rangle^2 \\
        & = \sum_{j = 1}^{\ell} \left( 1 - \sum_{i = 1}^{\ell} \langle \barg_j, \vecf_i \rangle^2 \right) \\
        & = \sum_{j = 1}^{\ell} \left( \norm{\barg_j}^2 - \norm{\hatf_j}^2 \right) \\
        & = \sum_{j = 1}^{\ell} \norm{\barg_j - \hatf_j}^2 \leq \sum_{j = 1}^{\ell} \frac{\gamma_j}{\lambda_{{\ell} + 1}},
    \end{align*}
    where the final inequality follows by
    the first statement of the theorem.
\end{proof}

\section{Spectral Clustering with Fewer Eigenvectors \label{sec:embedding+}}
In this section, we   analyse spectral clustering with fewer eigenvectors. The algorithm is essentially the same as the standard spectral clustering described in Section~\ref{sec:analysis1}, with the only difference that every $u\in \vertexsetg$ is embedded into a point in $\R^{\ell}$ by the mapping
defined in
\eqref{eq:embeddingell}.
The analysis follows from the one in Section~\ref{sec:analysis1} at a very high level.
However, since we require that $\{F(u)\}_{u \in \vertexsetg}$ are well separated in $\R^{\ell}$ for some $\ell < k$, the proof is more involved. 
For any  $i \in [k]$, we  define the approximate centre $\vecp^{(i)} \in \R^{\ell}$ of every cluster $\sets_i$ by
\[
    \vecp^{(i)}(j) = \frac{1}{\sqrt{\vol(\sets_i)}}\cdot \sum_{x = 1}^{\ell} \langle \vecf_j, \barg_x \rangle\cdot  \vecg_x(i),
\]
and prove that the total $k$-means cost for the points $\{F(u)\}_{u\in \vertexsetg}$ can be upper bounded with respect to $\Psi(\ell)$.

\begin{lemma} \label{lem:cost_bound}
It holds that 
    \[
        \sum_{i = 1}^k \sum_{u \in \sets_i} \deg(u) \norm{F(u) - \peye}^2 \leq \Psi(\ell).
    \]
\end{lemma}
\begin{proof}
By definition, it holds that 
    \begin{align*}
      \lefteqn{  \sum_{i = 1}^k \sum_{u \in S_i} \deg(u) \norm{F(u) - \peye}^2 }\\
      & = \sum_{i = 1}^k \sum_{u \in \sets_i} \deg(u) \left[ \sum_{j = 1}^{\ell} \left( \frac{\vecf_j(u)}{\sqrt{\deg(u)}} - \left(\sum_{x = 1}^{\ell} \langle \vecf_j, \barg_x \rangle \frac{\vecg_x(i)}{\sqrt{\vol(\sets_i)}}  \right) \right)^2 \right] \\
        & = \sum_{i = 1}^k \sum_{u \in \sets_i} \sum_{j = 1}^{\ell} \left( \vecf_j(u) - \left(\sum_{x = 1}^{\ell} \langle \vecf_j, \barg_x \rangle \barg_x(u) \right) \right)^2 \\
        & = \sum_{i = 1}^k \sum_{u \in \sets_i} \sum_{j = 1}^{\ell} \left( \vecf_j(u) - \hatg_j(u) \right)^2 \\
        & = \sum_{j = 1}^{\ell} \norm{\vecf_j - \hatg_j}^2 \leq \Psi(\ell),
    \end{align*}
    where the final inequality follows from the second statement of Theorem~\ref{thm:structure++}.
\end{proof}

We now prove a sequence of lemmas establishing that the distance between different $\peye$ and $\pj$ can be bounded with respect to $\theta$ and $\Psi(\ell)$.

\begin{lemma}\label{lem:p_norm}
It holds for  $i \in [k]$ that 
    \[
       \left(1 -  \frac{4 \sqrt{\Psi(\ell)}}{\theta}\right) \frac{\norm{\barx^{(i)}}^2}{\vol(\sets_i)}  \leq \norm{\vecp^{(i)}}^2 \leq  \frac{\norm{\barx^{(i)}}^2}{\vol(\sets_i)} \left(1 +  \frac{2 \sqrt{\Psi(\ell)}}{\theta}\right).
    \]
\end{lemma}
\begin{proof}
It holds by definition that 
    \begingroup
    \allowdisplaybreaks
    \begin{align}
        \lefteqn{\vol(\sets_i)\cdot  \norm{\vecp^{(i)}}^2 } \nonumber \\
        & = \sum_{x = 1}^{\ell} \left( \sum_{y = 1}^{\ell } \langle \vecf_x, \barg_y \rangle \vecg_y(i)\right)^2  \nonumber \\
        & = \sum_{x = 1}^{\ell} \sum_{y = 1}^{\ell} \sum_{z = 1}^{\ell} \langle \vecf_x, \barg_y\rangle \langle \vecf_x, \barg_z \rangle   \vecg_y(i) \vecg_z(i) \nonumber \\
        & = \sum_{x = 1}^{\ell} \sum_{y = 1}^{\ell} \langle \vecf_x, \barg_y\rangle^2 \vecg_y(i)^2 + \sum_{x = 1}^{\ell} \sum_{y = 1}^{\ell} \sum_{\substack{z = 1\\z \neq y}}^{\ell} \langle \vecf_x, \barg_y\rangle \langle \vecf_x, \barg_z\rangle \vecg_y(i) \vecg_z(i) \nonumber \\
         & = \sum_{x = 1}^{\ell} \sum_{y = 1}^{\ell} \langle \vecf_x, \barg_y\rangle^2 \vecg_y(i)^2 + \sum_{y = 1}^{\ell} \sum_{\substack{z = 1\\z \neq y}}^{\ell}  \vecg_y(i) \vecg_z(i) \cdot \left( \hatf_y^\transpose \hatf_z\right).  \label{eq:separate}
    \end{align}
    \endgroup
    We study the two terms of \eqref{eq:separate} separately. For the second term, we have that 
    \begingroup
    \allowdisplaybreaks
    \begin{align*}
        \lefteqn{\sum_{y = 1}^{\ell} \sum_{\substack{z = 1\\z \neq y}}^{\ell}  \vecg_y(i) \vecg_z(i) \cdot \left( \hatf_y^\transpose \hatf_z\right)} \\
        & \leq \sum_{y = 1}^{\ell} \abs{\vecg_y(i)} \sum_{\substack{z = 1 \\ z \neq y}}^{\ell} \abs{\vecg_z(i)} \abs{\hatf_y^\transpose \hatf_z} \\
        & = \sum_{y = 1}^{\ell} \abs{\vecg_y(i)} \sum_{\substack{z = 1 \\ z \neq y}}^{\ell} \abs{\vecg_z(i)} \abs{ \left(\barg_y + \hatf_y - \barg_y\right)^\transpose \left(\barg_z + \hatf_z - \barg_z\right) } \\
        & = \sum_{y = 1}^{\ell} \abs{\vecg_y(i)} \sum_{\substack{z = 1 \\ z \neq y}}^{\ell} \abs{\vecg_z(i)} \abs{ \langle \hatf_y - \barg_y, \barg_z \rangle + \langle \hatf_z - \barg_z, \barg_y \rangle + \langle \hatf_y - \barg_y, \hatf_z - \barg_z \rangle } \\
        & =  \sum_{y = 1}^{\ell} \abs{\vecg_y(i)} \sum_{\substack{z = 1 \\ z \neq y}}^{\ell} \abs{\vecg_z(i)} \abs{ \langle \hatf_y - \barg_y, \barg_z \rangle } \\
        & \leq  \sqrt{\left(\sum_{y = 1}^{\ell} \vecg_y(i)^2\right) \sum_{y = 1}^{\ell} \left(\sum_{\substack{1\leq z\leq \ell \\ z \neq y}} \abs{\vecg_z(i)} \abs{ \langle \hatf_y - \barg_y, \barg_z \rangle }\right)^2 } \\
        & \leq   \sqrt{\sum_{y = 1}^{\ell} \left(\sum_{\substack{1\leq z\leq \ell \\ z \neq y}} \vecg_z(i)^2 \right) \left( \sum_{\substack{1\leq  z\leq \ell  \\ z \neq y}} \langle \hatf_y - \barg_y, \barg_z \rangle^2 \right) } \\
        & \leq   \sqrt{\sum_{y = 1}^{\ell} \sum_{\substack{1\leq z\leq \ell \\ z \neq y}} \langle \hatf_y - \barg_y, \barg_z \rangle^2 } \\
        & \leq  \sqrt{\sum_{y = 1}^{\ell} \norm{\hatf_y - \barg_y}^2 }  \leq   \sqrt{\Psi(\ell)},
    \end{align*}
    \endgroup
    where we used the fact that $\sum_{y = 1}^k \vecg_y(i)^2 = 1$ for all $i\in[k]$. Therefore, we have that
    \begin{align*}
        \vol(\sets_i) \norm{\vecp^{(i)}}^2 & \leq \sum_{y = 1}^{\ell} \left(\sum_{x = 1}^{\ell} \langle \vecf_x, \barg_y\rangle^2 \right) \vecg_y(i)^2 +   \sqrt{\Psi (\ell)} \\
        & \leq \sum_{y = 1}^{\ell} \vecg_y(i)^2 +   \sqrt{\Psi (\ell)} \\
        & \leq \norm{\barx^{(i)}}^2 + \sqrt{\Psi(\ell)} \leq \norm{\barx^{(i)}}^2 \left(1 + \frac{2 \sqrt{\Psi(\ell)}}{\theta}\right).
    \end{align*}

    On the other hand, we have that 
    \begin{align*}
        \vol(\sets_i) \norm{\vecp^{(i)}}^2 & \geq \sum_{y = 1}^{\ell} \left(\sum_{x = 1}^{\ell} \langle \vecf_x, \barg_y\rangle^2 \right) \vecg_y(i)^2 - 2 \sqrt{\Psi(\ell)} \\
        & \geq \sum_{y = 1}^{\ell} \norm{\hatf_y}^2 \vecg_y(i)^2 - 2 \sqrt{\Psi(\ell)} \\
        & \geq \left(1 - \Psi(\ell)\right) \norm{\barx^{(i)}}^2 - 2 \sqrt{\Psi(\ell)}\\
        & \geq \norm{\barx^{(i)}}^2 \left(1 - \frac{4 \sqrt{\Psi(\ell)}}{\theta}\right), 
    \end{align*}
    where the last inequality holds by the fact that  $\theta \leq \left\|\barx^{(i)} \right\|^2 \leq 1$ and $\Psi(\ell)<1$.   Hence, the statement holds.
\end{proof}

\begin{lemma} \label{lem:pnorm_diff}
    It holds for any different $i,j\in[k]$ that 
    \[
        \norm{\frac{\sqrt{\vol(\sets_i)} }{\norm{\barx^{(i)}}}\cdot \vecp^{(i)} - \frac{\sqrt{\vol(\sets_j)} }{\norm{\barx^{(j)}}} \cdot \vecp^{(j)} }^2 \geq \theta - 3 \sqrt{\Psi(\ell)}.
    \]
\end{lemma}

\begin{proof}
By definition, it holds that 
    \begin{alignat*}{2}
    \lefteqn{\norm{\frac{\sqrt{\vol(\sets_i)} }{\norm{\barxi}}\cdot \peye - \frac{\sqrt{\vol(\sets_j)} }{\norm{\barxj}}\cdot \pj }^2 }\\
    & =  \sum_{t = 1}^{\ell} \left(\sum_{y = 1}^{\ell} \langle \vecf_t, \barg_y \rangle \left(\frac{\vecg_y(i)}{\norm{\barxi}} - \frac{\vecg_y(j)}{\norm{\barxj}}\right)\right)^2 \\
    & =  \sum_{t = 1}^{\ell} \sum_{y = 1}^{\ell} \langle \vecf_t, \barg_y\rangle^2 \left(\frac{\vecg_y(i)}{\norm{\barxi}} - \frac{\vecg_y(j)}{\norm{\barxj}} \right)^2\\
    & \qquad\qquad  +  \sum_{t = 1}^{\ell} \sum_{y = 1}^{\ell} \sum_{\substack{1\leq z \leq \ell \\ z \neq y}} \langle \vecf_t, \barg_y\rangle \langle \vecf_t, \barg_z \rangle \left(\frac{\vecg_y(i)}{\norm{\barxi}} - \frac{\vecg_y(j)}{\norm{\barxj}} \right) \left(\frac{\vecg_z(i)}{\norm{\barxi}} - \frac{\vecg_z(j)}{\norm{\barxj}} \right).
    \end{alignat*}
    We upper bound  the second term by 
    \newcommand{\gnormij}[1]{\frac{\vecg_{#1}(i)}{\norm{\barx^{(i)}}} - \frac{\vecg_{#1}(j)}{\norm{\barx^{(j)}}}}
    \begingroup
    \allowdisplaybreaks
    \begin{align*}
        & \sum_{y = 1}^{\ell} \abs{\frac{\vecg_y(i)}{\norm{\barx^{(i)}}} - \frac{\vecg_y(j)}{\norm{\barx^{(j)}}}} \sum_{\substack{1\leq z \leq \ell \\ z \neq y}} \abs{\frac{\vecg_z(i)}{\norm{\barx^{(i)}}} - \frac{\vecg_z(j)}{\norm{\barx^{(j)}}}} \sum_{t = 1}^{\ell } \abs{\langle \vecf_t, \barg_y \rangle \langle \vecf_t, \barg_x \rangle} \\
        = & \sum_{y = 1}^{\ell} \abs{\frac{\vecg_y(i)}{\norm{\barx^{(i)}}} - \frac{\vecg_y(j)}{\norm{\barx^{(j)}}}} \sum_{\substack{1\leq z\leq \ell  \\ z \neq y}} \abs{\frac{\vecg_z(i)}{\norm{\barx^{(i)}}} - \frac{\vecg_z(j)}{\norm{\barx^{(j)}}}} \abs{\hatf_y^\transpose \hatf_z} \\
        = & \sum_{y = 1}^{\ell} \abs{\frac{\vecg_y(i)}{\norm{\barx^{(i)}}} - \frac{\vecg_y(j)}{\norm{\barx^{(j)}}}} \sum_{\substack{1\leq z \leq \ell \\ z \neq y}} \abs{\frac{\vecg_z(i)}{\norm{\barx^{(i)}}} - \frac{\vecg_z(j)}{\norm{\barx^{(j)}}}} \abs{ \left(\barg_y + \hatf_y - \barg_y\right)^\transpose \left(\barg_z + \hatf_z - \barg_z\right) } \\
       = &   \sum_{y = 1}^{\ell} \abs{\frac{\vecg_y(i)}{\norm{\barx^{(i)}}} - \frac{\vecg_y(j)}{\norm{\barx^{(j)}}}} \sum_{\substack{1\leq z \leq \ell \\ z \neq y}} \abs{\frac{\vecg_z(i)}{\norm{\barx^{(i)}}} - \frac{\vecg_z(j)}{\norm{\barx^{(j)}}}} \abs{ \langle \hatf_y - \barg_y, \barg_z \rangle } \\
         \leq &   \sqrt{\left(\sum_{y = 1}^{\ell} \left(\gnormij{y}\right)^2\right) \sum_{y = 1}^{\ell} \left(\sum_{\substack{1\leq z \leq \ell \\ z \neq y}} \abs{\gnormij{z}} \abs{ \langle \hatf_y - \barg_y, \barg_z \rangle }\right)^2 } \\
         \leq &  \sqrt{2 \sum_{y = 1}^{\ell} \left(\sum_{\substack{1\leq z \leq \ell \\ z \neq y}} \left(\gnormij{z}\right)^2 \right) \left( \sum_{\substack{1\leq z \leq \ell \\ z \neq y}} \langle \hatf_y - \barg_y, \barg_z \rangle^2 \right) } \\
         \leq &  2 \sqrt{\sum_{y = 1}^{\ell} \sum_{\substack{1\leq z \leq \ell \\ z \neq y}} \langle \hatf_y - \barg_y, \barg_z \rangle^2 } \\
          \leq &  2 \sqrt{\sum_{y = 1}^{\ell} \norm{\hatf_y - \barg_y}^2 } \\
         \leq &  2 \sqrt{\Psi (\ell)},
    \end{align*}
    \endgroup
    from which we   conclude that
    \begin{align*}
    \norm{\frac{\sqrt{\vol(\sets_i)} }{\norm{\barxi}}\cdot \peye - \frac{\sqrt{\vol(\sets_j)} }{\norm{\barxj}}\cdot \pj}^2 & \geq  \left(1 - \Psi(\ell) \right) \theta - 2 \sqrt{\Psi(\ell)} \\
    & \geq  \theta - 3 \sqrt{\Psi(\ell)}. \qedhere
    \end{align*}  
\end{proof}

\begin{lemma}\label{lem:norm_phat_diff}
It holds for any different  $i,j\in[k]$ that  
    \[
        \norm{\frac{\peye}{\norm{\peye}} - \frac{\pj}{\norm{\pj}}}^2 \geq \frac{\theta}{4} - 8 \sqrt{\frac{\Psi}{\theta}}.
    \]
\end{lemma}
\begin{proof}
To follow the proof, it may help to refer to the illustration in Figure~\ref{fig:proof_fig}.
    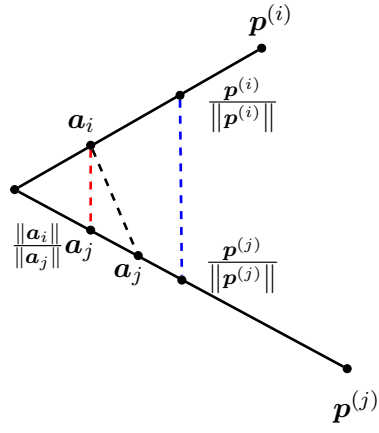
\begin{figure}[t]
        \centering
        \begin{tikzpicture}
	\begin{pgfonlayer}{nodelayer}
		\node [style=smallblack] (4) at (0, 0) {};
		\node [style=smallblack] (7) at (8.75, -4.75) {};
		\node [style=smallblack] (8) at (6.5, 3.75) {};
		\node [style=none] (9) at (6.75, 4.5) {$\peye$};
		\node [style=none] (10) at (9, -5.75) {$\pj$};
		\node [style=smallblack] (11) at (4.35, 2.5) {};
		\node [style=smallblack] (12) at (4.4, -2.4) {};
		\node [style=none] (13) at (6, 2.25) {$\frac{\peye}{\norm{\peye}}$};
		\node [style=none] (15) at (6, -2) {$\frac{\pj}{\norm{\pj}}$};
		\node [style=smallblack] (16) at (2, 1.175) {};
		\node [style=smallblack] (17) at (3.25, -1.75) {};
		\node [style=none] (18) at (1.75, 1.75) {$\veca_i$};
		\node [style=none] (19) at (3, -2.25) {$\veca_j$};
		\node [style=smallblack] (20) at (2, -1.075) {};
		\node [style=none] (21) at (1, -1.5) {$\frac{\norm{\veca_i}}{\norm{\veca_j}}\veca_j$};
	\end{pgfonlayer}
	\begin{pgfonlayer}{edgelayer}
		\draw [style=undirected] (4) to (8);
		\draw [style=undirected] (4) to (7);
		\draw [style=blue dashed] (11) to (12);
		\draw [style=red dashed] (16) to (20);
		\draw [style=undirected dashed] (16) to (17);
	\end{pgfonlayer}
\end{tikzpicture}
        \caption[Illustration of the proof of Lemma~\ref{lem:norm_phat_diff}]{
        Illustration of the proof of Lemma~\ref{lem:norm_phat_diff}.
        Our goal is to give a lower bound on the length of $\left(\frac{\peye}{\norm{\peye}} - \frac{\pj}{\norm{\pj}}\right)$, which is the blue dashed line in the figure.
        We instead calculate a lower bound on the length of $\left( \veca_i - \frac{\norm{\veca_i}}{\norm{\veca_j}} \veca_j \right)$, which is the red dashed line, and use the fact that by construction $\norm{\veca_i} \leq 1$ and $\norm{\veca_j} \leq 1$.}
        \label{fig:proof_fig}
    \end{figure}
   
    We set the parameter $\epsilon = 4 \sqrt{\Psi} / \theta$, and define 
    \[
        \veca_i = \frac{\sqrt{\vol(\sets_i)} }{\left(1 + \epsilon\right) \norm{\barxi}}\cdot \peye, \qquad \hspace{3em} \veca_j = \frac{\sqrt{\vol(\sets_j)} }{\left(1 + \epsilon\right) \norm{\barxj}}\cdot \pj.
    \]
    By the definition of $\epsilon$ and  Lemma~\ref{lem:p_norm}, it holds that 
    $
        \norm{\veca_i} \leq 1$, and $
        \norm{\veca_j} \leq  1$.
    We can also assume without loss of generality that $\norm{\veca_i} \leq \norm{\veca_j}$.
    Then, as illustrated in Figure~\ref{fig:proof_fig}, we have
    \[
        \norm{\frac{\peye}{\norm{\peye}} - \frac{\pj}{\norm{\pj}}}^2 \geq \norm{\veca_i - \frac{\norm{\veca_i}}{\norm{\veca_j}} \veca_j}^2,
    \]
    and so it suffices to lower bound the right-hand side of the inequality above.
    By the triangle inequality, we have
    \begin{align*}
        \norm{\veca_i - \frac{\norm{\veca_i}}{\norm{\veca_j}} \veca_j} & \geq \norm{\veca_i - \veca_j} - \norm{\veca_j - \frac{\norm{\veca_i}}{\norm{\veca_j}} \veca_j} \\
        & = \frac{1}{1+\epsilon} \norm{\frac{\sqrt{\vol(\sets_i)} }{\norm{\barxi}}\cdot \peye - \frac{\sqrt{\vol(\sets_j)} }{\norm{\barxj}} \cdot \pj} - \left(\norm{\veca_j} - \norm{\veca_i}\right).
    \end{align*}
    Now, we have that 
    \begin{align*}
        \norm{\veca_j} - \norm{\veca_i} & = \frac{\sqrt{\vol(\sets_j)}}{(1 + \epsilon) \norm{\barxj}} \cdot \norm{\pj} - \frac{\sqrt{\vol(\sets_i)}}{(1 + \epsilon) \norm{\barxi}}\cdot   \norm{\peye} \\
        & \leq 1 - \frac{1 - \epsilon}{1 + \epsilon} = \frac{2 \epsilon}{1 + \epsilon},
    \end{align*}
    and have by Lemma~\ref{lem:pnorm_diff} that   
    \begin{align*}
        \norm{\frac{\sqrt{\vol(\sets_i)} }{\norm{\barxi}}\cdot \peye - \frac{\sqrt{\vol(\sets_j)} }{\norm{\barxj}} \cdot \pj} & \geq \sqrt{\theta - 3 \sqrt{\Psi(\ell)}} \geq \sqrt{\theta} - \sqrt{2 \epsilon} \geq \sqrt{\theta} - 2 \epsilon,
    \end{align*}
    since $\epsilon = 4 \sqrt{\Psi} / \theta < 1$ by the assumption on $\Psi$.
    This gives us
    \begin{align*}
        \norm{\veca_i - \frac{\norm{\veca_i}}{\norm{\veca_j}} \veca_j} & \geq \frac{\sqrt{\theta} - 2 \epsilon}{1 + \epsilon} - \frac{2\epsilon}{1 + \epsilon}  \geq \frac{1}{2} \left(\sqrt{\theta} - 4 \epsilon\right).
    \end{align*}
    Finally, we have that
    \begin{align*}
        \norm{\frac{\peye}{\norm{\peye}} - \frac{\pj}{\norm{\pj}}}^2 & \geq \frac{1}{4} \left(\sqrt{\theta} - 16 \frac{\sqrt{\Psi(\ell)}}{\theta}\right)^2 \geq \frac{\theta}{4} - 8 \sqrt{\frac{\Psi(\ell)}{\theta}},
    \end{align*}
    which completes the proof.
\end{proof}

\begin{lemma}\label{lem:pij_distance}
   It holds for different $i, j\in[k]$ that
    \[
        \norm{\peye - \pj}^2 \geq \frac{\theta^2 - 20 \sqrt{\theta\cdot \Psi(\ell)}}{16 \min\left\{\vol(\sets_i), \vol(\sets_j)\right\}}.
    \]
\end{lemma}
\begin{proof}
We assume without loss of generality that $\norm{\peye}^2 \geq \norm{\pj}^2$. Then, by Lemma~\ref{lem:p_norm} and the fact that $\|\barx^{(i)}\|^2\geq\theta$ holds for any $i\in[k]$, 
we have
    \[
        \norm{\peye}^2 \geq \left(1 - \frac{4 \sqrt{\Psi(\ell)}}{\theta}\right)\cdot \frac{\norm{\barxi}^2}{\vol(\sets_i)}
    \]
    and
    \[
        \norm{\pj}^2 \geq \left(1 - \frac{4 \sqrt{\Psi(\ell)}}{\theta}\right)\cdot \frac{\norm{\barxj}^2}{\vol(\sets_j)},
    \]
    which implies that
    \[
        \norm{\peye}^2 \geq \frac{\theta - 4 \sqrt{\Psi(\ell)}}{\min\left\{\vol(\sets_i), \vol(\sets_j)\right\}}.
    \]
    Now, we proceed by case distinction.
   
    Case 1: $\norm{\peye} \geq 4 \norm{\pj}$. In this case, we have
    \[
        \norm{\peye - \pj} \geq \norm{\peye} - \norm{\pj} \geq \frac{3}{4} \norm{\peye},
    \]
    and
    \begin{align*}
        \norm{\peye - \pj}^2 & \geq \frac{9}{16}\cdot \frac{\theta - 4 \sqrt{\Psi(\ell)}}{\min\left\{\vol(\sets_i), \vol(\sets_j)\right\}} \\
        & \geq \frac{\theta \left(\theta - 20 \sqrt{\Psi(\ell)/\theta}\right)}{16 \min\left\{\vol(\sets_i), \vol(\sets_j)\right\}} \\
        & = \frac{\theta^2 - 20 \sqrt{\theta\cdot  \Psi(\ell)}}{16 \min\left\{\vol(\sets_i), \vol(\sets_j)\right\}},
    \end{align*}
    since $\theta < 1$.  
   
    Case 2: $\norm{\pj} = \alpha \norm{\peye}$ for some $\alpha \in \left(\frac{1}{4}, 1\right]$.
    By Lemma~\ref{lem:norm_phat_diff}, we have
    \[
        \left\langle \frac{\peye}{\norm{\peye}}, \frac{\pj}{\norm{\pj}} \right\rangle \leq 1 - \frac{1}{2} \left(\frac{\theta}{4} - 8 \sqrt{\frac{\Psi}{\theta}}\right) \leq 1 - \frac{\theta}{8} + 2 \sqrt{\frac{\Psi}{\theta}}.
    \]
    Then, it holds that 
    \begingroup
    \allowdisplaybreaks
    \begin{align*}
        \lefteqn{\norm{\peye - \pj}^2}\\
        & = \norm{\peye}^2 + \norm{\pj}^2 - 2 \left\langle \frac{\peye}{\norm{\peye}}, \frac{\pj}{\norm{\pj}} \right\rangle \norm{\peye} \norm{\pj}   \\
        & \geq \left(1 + \alpha^2\right) \norm{\peye}^2 - 2 \left( 1 - \frac{\theta}{8} + 2 \sqrt{\frac{\Psi(\ell)}{\theta}} \right) \alpha \norm{\peye}^2 \\
        & \geq \left(1 + \alpha^2 - 2\alpha + \frac{\theta}{4} \alpha - 4 \sqrt{\frac{\Psi(\ell)}{\theta}} \alpha\right) \norm{\peye}^2 \\
        & \geq \left(\frac{\theta}{4} - 4 \sqrt{\frac{\Psi(\ell)}{\theta}}\right)\cdot \alpha\cdot  \frac{\theta - 4 \sqrt{\Psi(\ell)}}{\min\left\{\vol(\sets_i), \vol(\sets_j)\right\}} \\
        & \geq \left(\frac{\theta}{16} - \sqrt{\frac{\Psi(\ell)}{\theta}} \right) \left(\theta - 4 \sqrt{\Psi(\ell)} \right)\cdot  \frac{1}{\min\left\{\vol(\sets_i), \vol(\sets_j)\right\}} \\
        & \geq \left(\frac{\theta^2}{16} -\frac{5}{4} \sqrt{\theta \Psi(\ell) } \right) \cdot \frac{1}{\min\left\{\vol(\sets_i), \vol(\sets_j)\right\}} \\
        & = \frac{\theta^2 - 20 \sqrt{\theta \Psi(\ell) }}{16 \min\left\{\vol(\sets_i), \vol(\sets_j)\right\}}
        \end{align*}
        \endgroup
        which completes the proof.
\end{proof}

It is important to recognise that the lower bound in Lemma~\ref{lem:pij_distance} implies a condition on $\theta$ and $\Psi(\ell)$ under which $\vecp^{(i)}$ and $\vecp^{(j)}$ are well separated.
With this, we analyse the performance of spectral clustering when fewer eigenvector are employed to construct the embedding and show that it works when the optimal clusters present a noticeable pattern.

\begin{lemma} \label{lem:cost_lower_bound_2}
Let $A_1,\ldots A_k$ be the output of spectral clustering with $\ell$ eigenvectors, and $\sigma$ and  $\setm_{\sigma,i}$ be defined as in \eqref{eq:defsigma} and \eqref{eq:defmset}.
If $\Psi(\ell) \leq \theta^3 / 40^2$, then
\[
        \sum_{i = 1}^k \frac{\vol(\setm_{\sigma, i} \triangle \sets_i)}{\vol(\sets_i)} \leq 64 (1 + \APT) \frac{\Psi(\ell)}{\theta^2}.
    \]
\end{lemma}
\begin{proof} 
    Let us define $\setb_{ij} = \seta_i \cap \sets_j$ to be the vertices in $\seta_i$ which belong to the ground-truth cluster $\sets_j$.
    Then, we have that
    \begin{align}
         \sum_{i = 1}^k \frac{\vol(\setm_{\sigma, i} \triangle \sets_i)}{\vol(\sets_i)}  
        & = \sum_{i = 1}^k \sum_{\substack{j = 1 \\ j \neq \sigma(i)}}^k \vol(\setb_{ij}) \left(\frac{1}{\vol(\sets_{\sigma(i)})} + \frac{1}{\vol(\sets_j)} \right) \nonumber \\
        & \leq 2 \sum_{i = 1}^k \sum_{\substack{j = 1 \\ j \neq \sigma(i)}}^k \frac{\vol(\setb_{ij})}{\min\{\vol(\sets_{\sigma(i)}), \vol(\sets_j)\}}, \label{eq:up_sym_ratio2}
    \end{align}
and that 
\begingroup
\allowdisplaybreaks
\begin{align*}
        \mathrm{COST}(\seta_1, \ldots \seta_k) & = \sum_{i = 1}^k \sum_{u \in \seta_i} \deg(u) \norm{F(u) - \vecc_i}^2 \\
        & \geq \sum_{i = 1}^k \sum_{\substack{1\leq j \leq  k \\ j \neq \sigma(i)}}   \sum_{u \in \setb_{ij}} \deg(u) \norm{F(u) - \vecc_i}^2 \\
        & \geq \sum_{i = 1}^k \sum_{\substack{1\leq j \leq k  \\ j \neq \sigma(i)}}  \sum_{u \in \setb_{ij}} \deg(u) \left(\frac{\norm{\pj - \vecc_i}^2}{2} - \norm{\pj - F(u)}^2 \right) \\
        & \geq
        \begin{multlined}[t]
        \sum_{i = 1}^k \sum_{\substack{1\leq j \leq k \\ j \neq \sigma(i)}}  \sum_{u \in \setb_{ij}} \frac{\deg(u) \norm{\pj - \vecp^{(\sigma(i))}}^2}{8} \qquad \\
        - \sum_{i = 1}^k \sum_{\substack{1\leq j \leq k \\ j \neq i}}  \sum_{u \in \setb_{ij}} \deg(u) \norm{\pj - F(u)}^2
        \end{multlined} \\
        & \geq \sum_{i = 1}^k \sum_{\substack{1\leq j\leq k \\ j \neq \sigma(i)}}  \vol(\setb_{ij}) \frac{\norm{\pj - \vecp^{(\sigma(i))}}^2}{8} - \sum_{i = 1}^k \sum_{u \in \sets_i} \deg(u) \norm{\peye - F(u)}^2 \\
        & \geq
        \begin{multlined}[t]
        \sum_{i = 1}^k \sum_{\substack{1\leq j\leq k \\ j \neq \sigma(i)}}   \frac{\vol(\setb_{ij})}{16\cdot  \min\{\vol(\sets_{\sigma(i)}), \vol(\sets_j)\}}\left( \theta^2 - 20\sqrt{\theta\cdot \Psi(\ell)}\right) \qquad \\
        - \sum_{i = 1}^k \sum_{u \in \sets_i} \deg(u) \norm{\peye - F(u)}^2
        \end{multlined} \\
        & \geq \frac{1}{32} \cdot\left( \sum_{i = 1}^k \frac{\vol(\setm_{\sigma, i} \triangle \sets_i)}{\vol(\sets_i)}\right) \left( \theta^2 - 20\sqrt{\theta\cdot \Psi(\ell)}\right) - \Psi(\ell),
    \end{align*}
    \endgroup
    where the second inequality follows by the inequality $\norm{\veca - \vecb}^2 \geq \frac{\norm{\vecb - \vecc}^2}{2} - \norm{\veca - \vecc}^2$, the third inequality follows since $\vecc_i$ is closer to $\vecp^{(\sigma(i))}$ than $\vecp^{(j)}$, the fifth one follows from Lemma~\ref{lem:pij_distance}, and the last one follows by \eqref{eq:up_sym_ratio2}.
    
    On the other hand, since $\mathrm{COST}(\seta_1,\ldots, \seta_k) \leq \APT\cdot  \Psi(\ell)$ by Lemma~\ref{lem:cost_bound}, we have that 
    \begin{align*}
         \sum_{i = 1}^k \frac{\vol(\setm_{\sigma, i} \triangle \sets_i)}{\vol(\sets_i)} & \leq 32\cdot (1+\APT) \cdot \Psi(\ell)\cdot  \left( \theta^2 - 20\sqrt{\theta\cdot \Psi(\ell)}\right)^{-1} \\
         & \leq 64\cdot   (1+\APT) \cdot \Psi(\ell),
    \end{align*}
    where the last inequality follows by the assumption that $\Psi(\ell) \leq \theta^3 /1600$. Therefore, the statement follows.
\end{proof}

Combining this with other technical ingredients, including the new technique for constructing the desired mapping $\sigma^*$ described in Section~\ref{sec:analysis1}, we obtain the performance guarantee of the designed algorithm, which is summarised as follows:

\begin{theorem} \label{thm:sc_meta-graph}
Let $\geqve$ be a graph with $k$ clusters $S_1,\ldots, S_k$ of almost balanced size,
and a $(\theta, \ell)$-distinguishable meta-graph that satisfies 
$\Psi(\ell) \leq \left(2176 (1 + \APT)\right)^{-1} \theta^3$.
Let $A_1,\ldots, A_k$ be the output of spectral clustering with $\ell$ eigenvectors, and without loss of generality let the  optimal correspondent of $\seta_i$ be $\sets_i$. Then, it holds that
\[
        \sum_{i = 1}^k \vol\left(\seta_i \triangle \sets_{i}\right) \leq 2176 (1 + \APT) \frac{\Psi(\ell) \cdot \vol(\vertexset)}{k \cdot \theta^2}.
    \]
\end{theorem}
\begin{proof}
This result can be obtained by using the same technique as the one used in the proof of Theorem~\ref{thm:sc_guarantee}, but applying Lemma~\ref{lem:cost_bound} instead of Lemma~\ref{lem:total_cost} and
Lemma~\ref{lem:cost_lower_bound_2} instead of Lemma~\ref{lem:cost_lower_bound} in the analysis.
\end{proof}
 Notice that if we take $\ell = k$, then we have that $\theta = 1$ and $\Psi(\ell) \leq k / \Upsilon(k)$ which makes the guarantee in Theorem~\ref{thm:sc_meta-graph} the same as the one in Theorem~\ref{thm:sc_guarantee}.
 However, if the meta-graph corresponding to the optimal clusters is $(\theta, \ell)$-distinguishable for large $\theta$ and $\ell \ll k$, then it holds that $\Psi(\ell) \ll k / \Upsilon(k)$ and Theorem~\ref{thm:sc_meta-graph} gives a stronger guarantee than the 
 one from Theorem~\ref{thm:sc_guarantee}.

\section{Experimental Results} \label{sec:metaExperiments}
In this section we empirically evaluate the performance of spectral clustering for finding $k$ clusters while using fewer than $k$ eigenvectors.
Our results on synthetic data demonstrate that, for graphs with a clear pattern of clusters,
spectral clustering with fewer than $k$ eigenvectors performs better.
This is further confirmed on real-world datasets including BSDS, MNIST, and USPS. 
The code used to produce all experimental results is available at
\begin{equation*}
\mbox{\url{https://github.com/pmacg/spectral-clustering-meta-graphs}.}
\end{equation*}
We implement spectral clustering in Python, using the \texttt{scipy} library for computing eigenvectors, and the $k$-means algorithm from the \texttt{sklearn} library.
Our experiments on synthetic data are performed on a desktop computer with an Intel(R) Core(TM) i5-8500 CPU @ 3.00GHz processor and 16 GB RAM.
The experiments on the BSDS, MNIST, and USPS datasets are performed on a compute server with 64 AMD EPYC 7302 16-Core Processors.

\subsection{Results on Synthetic Data}
We first study the performance of spectral clustering on random graphs whose clusters exhibit a clear pattern.
Given the parameters $n \in \Z^+$, $0 \leq q \leq p \leq 1$, and some meta-graph $\graphm = (\vertexset_\graphm, \edgeset_\graphm)$ with $k$ vertices, we generate a graph with clusters $\{\sets_i\}_{i = 1}^k$, each of size $n$, as follows.
For each pair of vertices $u \in \sets_i$ and $v \in \sets_j$, we add the edge $(u, v)$ with probability $p$ if $i = j$, and with probability $q$ if $i \neq j$ 
and $(i, j) \in \edgeset_\graphm$.
The metric used for our evaluation is defined by
$
    \frac{1}{n k} \sum_{i = 1}^k \cardinality{\sets_i \cap \seta_i}$,
for the optimal matching between the output $\{\seta_i\}_{i=1}^k$ and the ground truth $\{\sets_i\}_{i=1}^k$.

In our experiments, we fix $n = 1000$, $p = 0.01$, and consider the meta-graphs $C_{10}$ and $P_{4, 4}$, similar to those illustrated in Examples~\ref{ex:cycle} and~\ref{ex:grid}; this results in graphs with 10,000 and 16,000 vertices respectively.
We vary the ratio $p / q$ and the number of eigenvectors used to find the clusters.
Our experimental result, which is reported as the average score over 10 trials and shown in 
Figure~\ref{fig:sbm_results}, clearly shows that spectral clustering with fewer than $k$ eigenvectors performs better. This is particularly the case when $p$ and $q$ are close, which corresponds to the more challenging regime in the model.

\begin{figure}[t]
\centering
\begin{subfigure}{0.45\textwidth}
    \includegraphics[width=\textwidth]{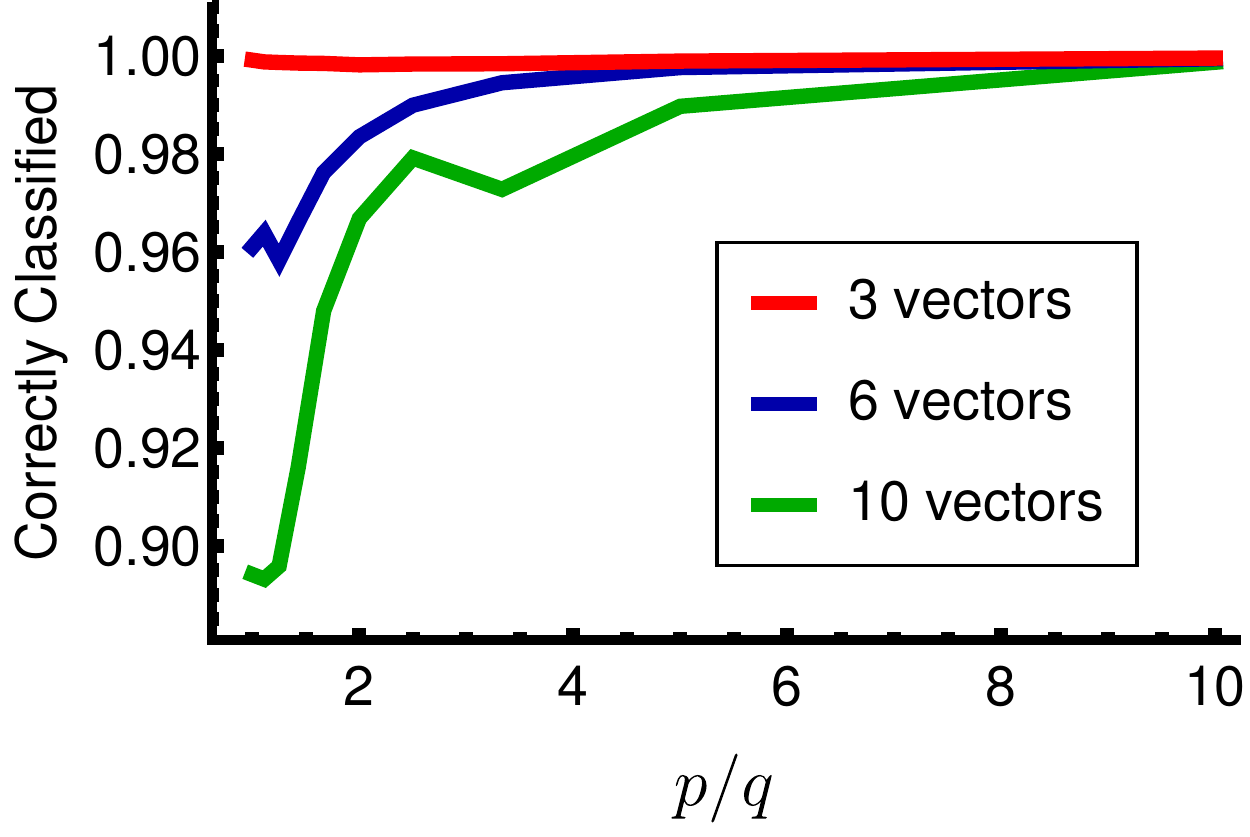}
    \caption{
    Meta-Graph $C_{10}$
    }
\end{subfigure}
\hspace{1em}
\begin{subfigure}{0.45\textwidth}
    \includegraphics[width=\textwidth]{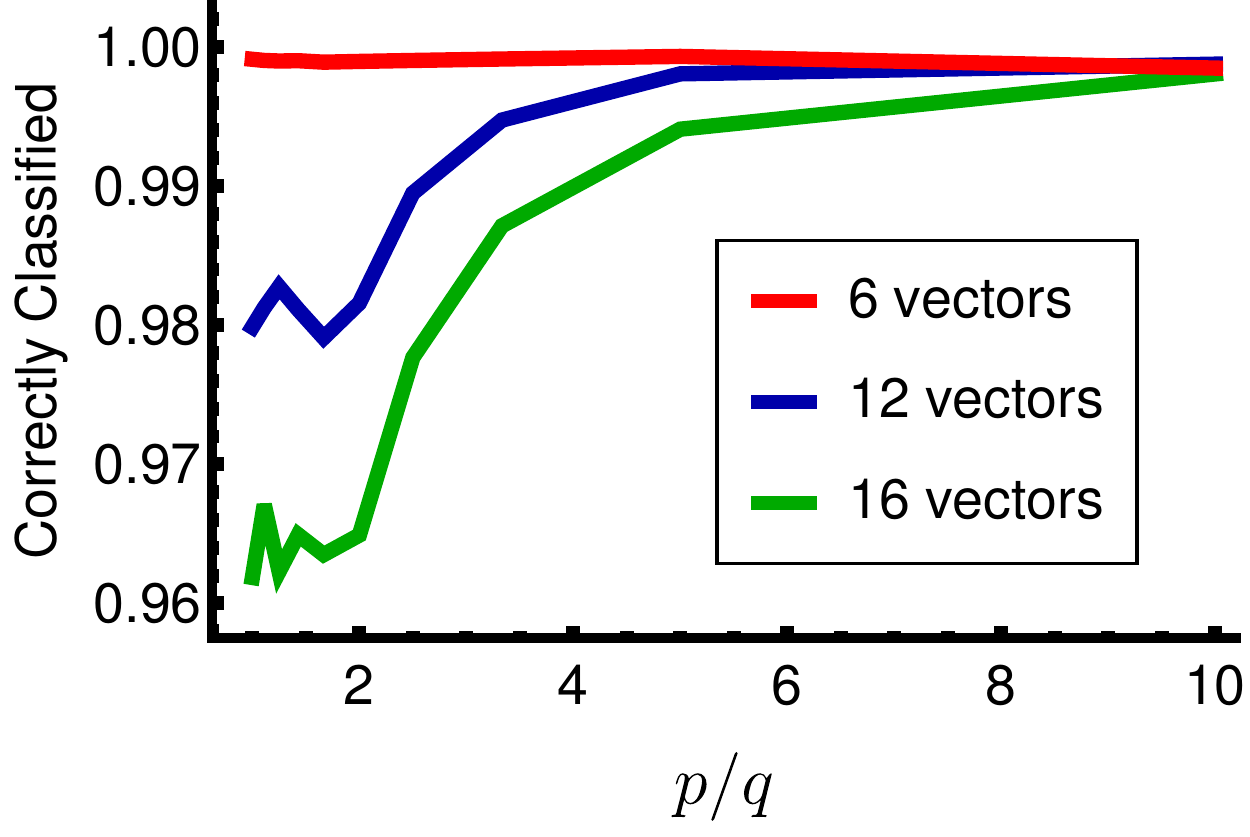}
    \caption{
        Meta-Graph $P_{4, 4}$
    }
\end{subfigure}
    \caption[The performance of spectral clustering with fewer eigenvectors]{
    \label{fig:sbm_results}
    A comparison on the output of spectral clustering with    meta-graphs $C_{10}$ and $P_{4, 4}$, when different number of eigenvectors are used.
    Note that classical spectral clustering uses $10$ and $16$ eigenvectors respectively.
    }
\end{figure}

\subsection{Results on the BSDS Dataset}
In this experiment, we study the performance of spectral clustering for image segmentation when using different numbers of eigenvectors.
We consider the Berkeley Segmentation Data Set (BSDS)~\cite{arbelaezContourDetectionHierarchical2011}, which consists of $500$ images along with their ground-truth segmentations.
For each image, we construct a similarity graph on the pixels and take $k$ to be the number of clusters in the ground-truth segmentation\footnote{
The BSDS dataset provides several human-generated ground truth segmentations for each image.
Since there are different numbers of ground truth clusterings associated with each image, in our experiments we take the target number of clusters for a given image to be the
one closest to the median.
}.
Given a particular image in the dataset, we first downsample the image to have at most 20,000 pixels.
Then, we represent each pixel by the point $(r, g, b, x, y)^\transpose \in \R^5$, where $r, g, b \in [255]$ are the RGB values of the pixel and $x$ and $y$ are the coordinates of the pixel in the downsampled image.
We construct the similarity graph by taking each pixel to be a vertex in the graph, and for every pair of pixels $\vecu, \vecv \in \R^5$, we add an edge with weight $\exp(- \norm{\vecu - \vecv}^2 / 2 \sigma^2)$ where $\sigma = 20$.
Then we apply spectral clustering, varying the number of eigenvectors used.
We evaluate each segmentation produced with spectral clustering using the Rand Index~\cite{randObjectiveCriteriaEvaluation1971} as implemented in the benchmarking code provided along with the BSDS dataset.
For each image, this computes the average Rand Index across all of the provided ground-truth segmentations for the image.
Figure~\ref{fig:bsds_results_intro} shows two images from the dataset along with the segmentations produced with spectral clustering, and additional examples are given in Figures~\ref{fig:bsds_app_0}~and~\ref{fig:bsds_app_2}.
These examples demonstrate that spectral clustering with fewer eigenvectors performs better.

We conduct the experiments on the entire BSDS dataset, and the average Rand Index of the algorithm's  output is reported in Table~\ref{tab:bsds_results}: it is clear that spectral clustering with $k/2$ eigenvectors consistently out-performs the one with $k$ eigenvectors. We further notice that, on $89\%$ of the images across the entire dataset,  using fewer than $k$ eigenvectors gives a better result than using $k$ eigenvectors.

\begin{table}[t]
    \caption[Experimental results for image segmentation with spectral clustering]{The average Rand Index across the BSDS dataset for different numbers of eigenvectors.
    \textsc{Optimal}
    refers to the algorithm which runs spectral clustering with   $\ell$  eigenvectors for all possible $\ell\in[k]$ and returns the output with the highest Rand Index.}
    \label{tab:bsds_results}
    \centering
    \begin{tabular}{cc}
    \toprule
        Number of Eigenvectors & Average Rand Index  \\
        \midrule
        $k$ & 0.71 \\
        $k / 2$ & 0.74 \\
        \textsc{Optimal} & 0.76 \\
    \bottomrule
    \end{tabular}
\end{table}

\subsection{Results on the MNIST and USPS Datasets}
We further demonstrate the real-world applicability of our results on the MNIST and USPS datasets~\cite{hullDatabaseHandwrittenText1994, lecunGradientbasedLearningApplied1998}, which consist of images of hand-written digits.
In both the MNIST and USPS datasets, each image is represented as an array of greyscale pixels with values between 0 and 255. 
The MNIST dataset has 60,000 images with dimensions $28 \times 28$ and the USPS dataset has 7,291 images with dimensions $16 \times 16$.
In each dataset, we consider each image to be a single data point in $\R^{(d^2)}$ where $d$ is the dimension of the images and construct the $k$ nearest neighbour graph for $k = 3$.
For the MNIST dataset, this gives a graph with 60,000 vertices and 138,563 edges, and for the USPS dataset this gives a graph with 7,291 vertices and 16,715 edges.
We use spectral clustering to partition the graphs into $10$ clusters.
We measure the similarity between the found clusters and the ground truth using the Adjusted Rand Index~\cite{gatesImpactRandomModels2017} and accuracy~\cite{randObjectiveCriteriaEvaluation1971},
and plot the results in Figure~\ref{fig:mnist_usps}.
Our experiments show that spectral clustering with just 7 eigenvectors gives the best performance on both datasets.

\begin{figure}[t]
\centering
\begin{subfigure}{0.45\textwidth}
    \includegraphics[width=\textwidth]{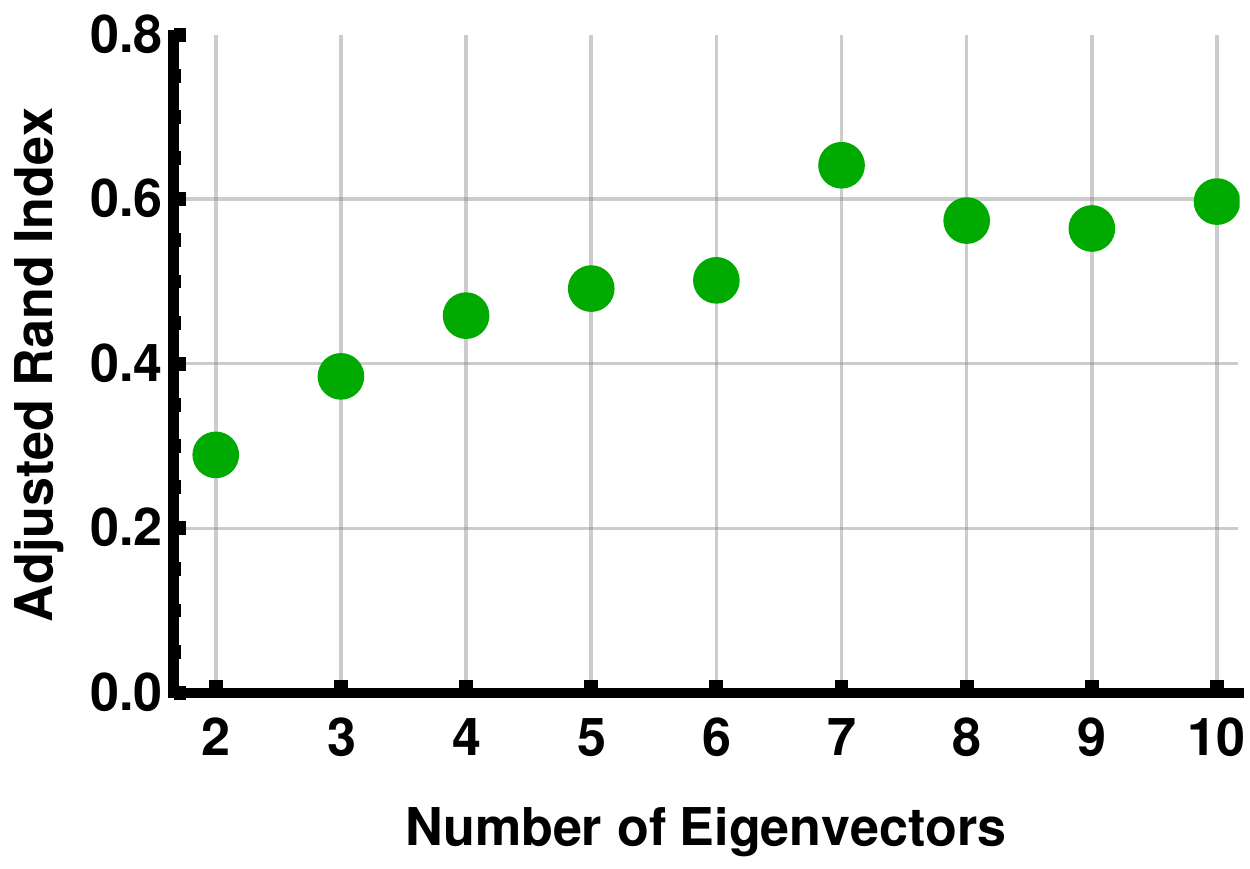}
    \caption{}
\end{subfigure}
\hspace{1em}
\begin{subfigure}{0.45\textwidth}
    \includegraphics[width=\textwidth]{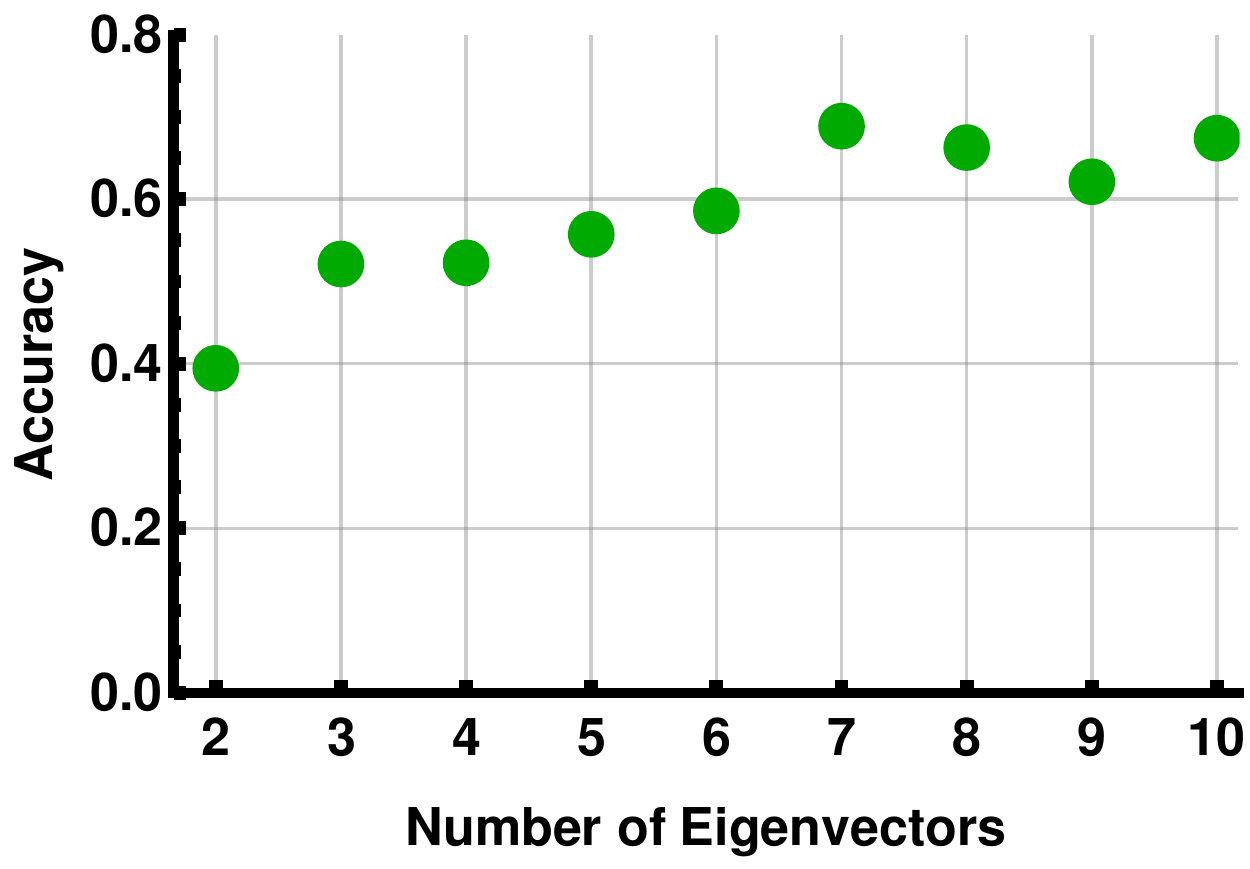}
    \caption{}
\end{subfigure}
\par\bigskip
\begin{subfigure}{0.45\textwidth}
    \includegraphics[width=\textwidth]{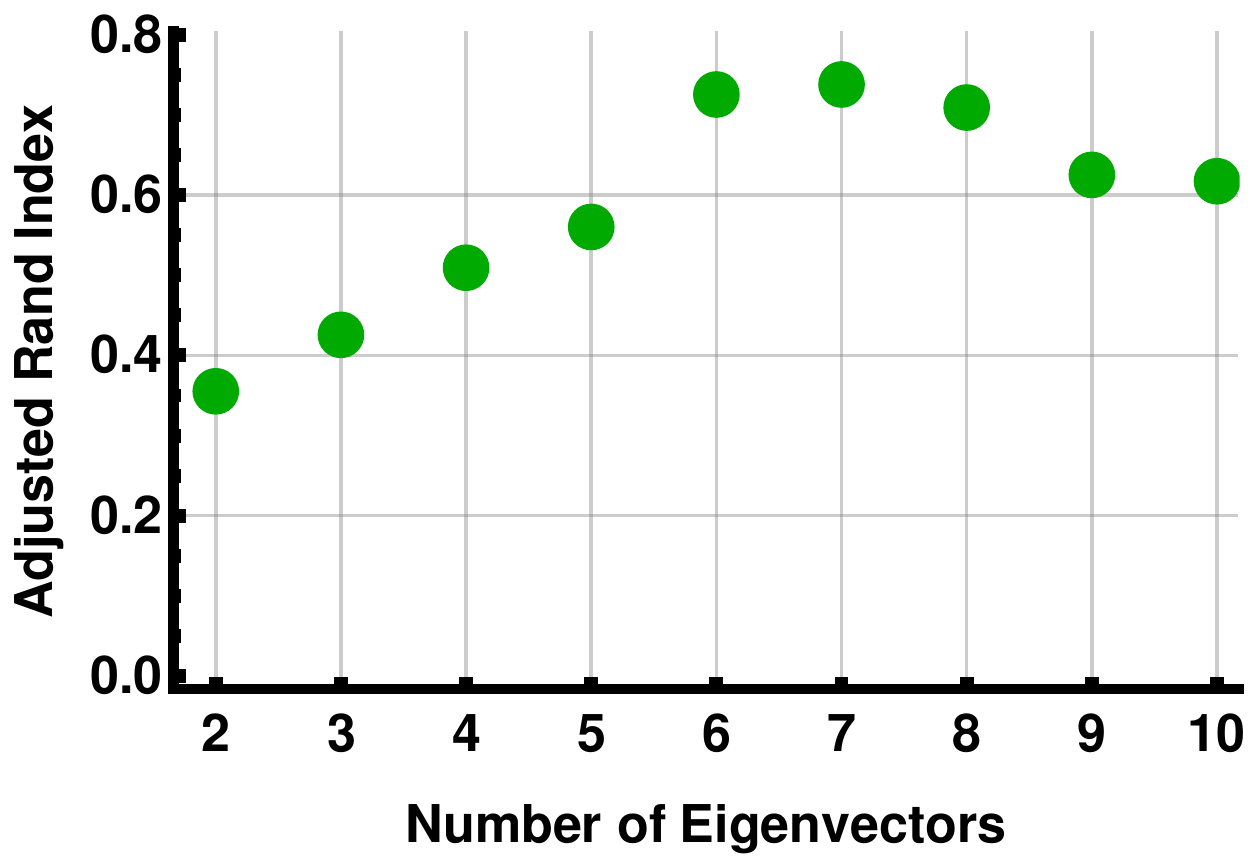}
    \caption{}
\end{subfigure}
\hspace{1em}
\begin{subfigure}{0.45\textwidth}
    \includegraphics[width=\textwidth]{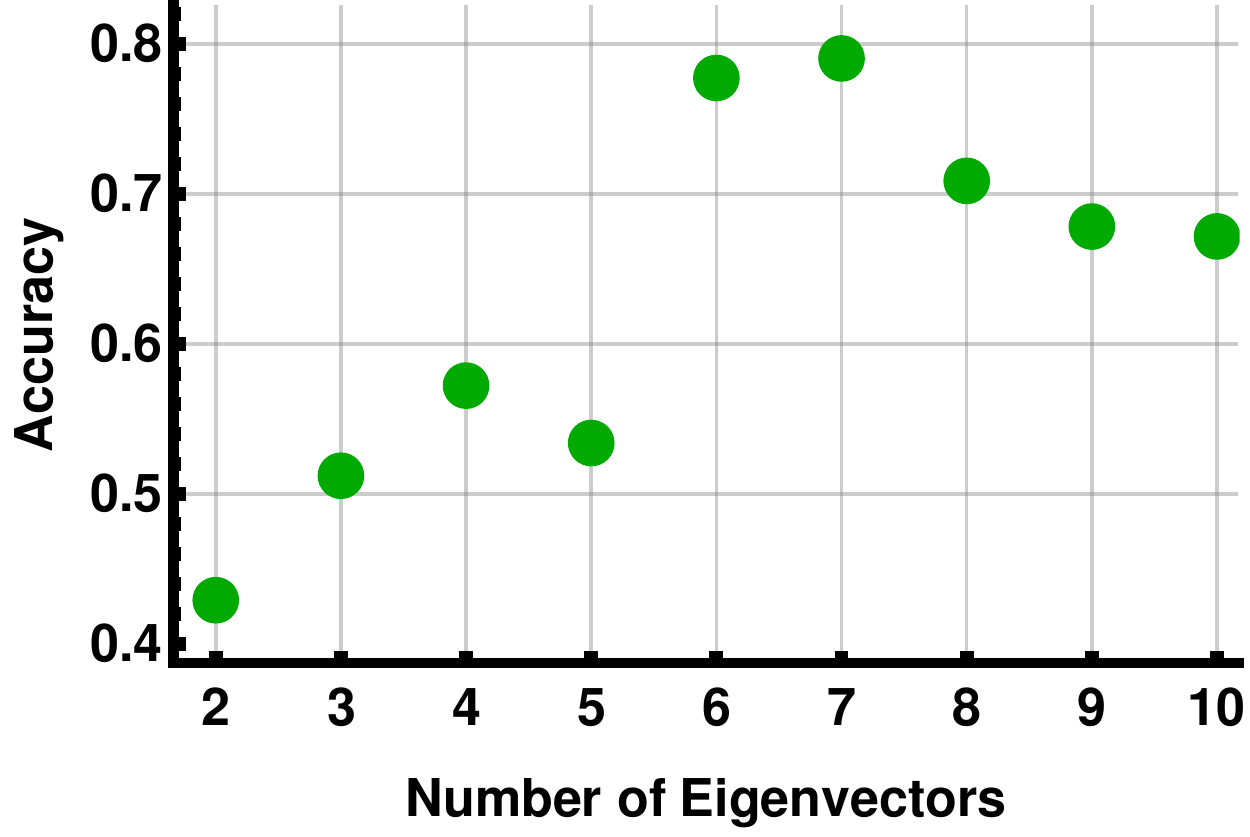}
    \caption{}
\end{subfigure}
    \caption[Experimental results on the MNIST and USPS datasets]{
    \label{fig:mnist_usps}
    Experimental results on the MNIST and USPS datasets.
    \textbf{(a)} and \textbf{(b)}: Adjusted Rand Index and accuracy of spectral clustering on the MNIST dataset;
    \textbf{(c)} and \textbf{(d)}: Adjusted Rand Index and accuracy of spectral clustering on the USPS dataset.
    These experiments show that spectral clustering  with 7 eigenvectors gives the best partition of the input into 10 clusters.
    }
\end{figure}

\section{Future Work}
This work is merely the starting point in the study of spectral clustering with fewer than $k$ eigenvectors, and leaves many open questions. For example, although one can enumerate the number of used eigenvectors from $1$ to $k$ and take the
 best clustering result with respect to some objective function~(such as graph conductance or the normalised cut value),
it would be interesting to know whether the `optimal' number of eigenvectors
can be computed directly.
\begin{openquestion}
Can the `optimal' number of eigenvectors for spectral clustering be computed directly from the Laplacian spectrum?
\end{openquestion}
Additionally, more empirical studies could give further insight into the structure of the meta-graphs of real-world datasets.
\begin{openquestion}
What are the typical meta-graph structures appearing in real-world datasets?
Do they share any common properties that can be exploited to improve spectral clustering in practice?
\end{openquestion}

\begin{figure*}[ht]
    \centering
    \begin{subfigure}[t]{0.25\textwidth}
        \includegraphics[width=\textwidth]{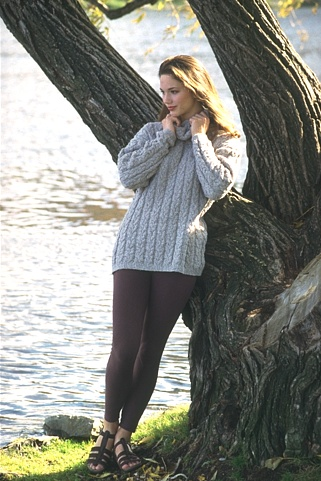}
        \caption{Original Image}
    \end{subfigure}
    \hspace{1.5em}
    \begin{subfigure}[t]{0.25\textwidth}
        \includegraphics[width=\textwidth]{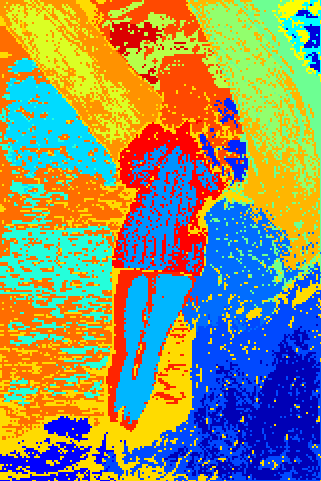}
        \caption{Segmentation into $24$ clusters with $8$ eigenvectors; Rand Index $0.82$.}
    \end{subfigure}
    \hspace{1.5em}
    \begin{subfigure}[t]{0.25\textwidth}
        \includegraphics[width=\textwidth]{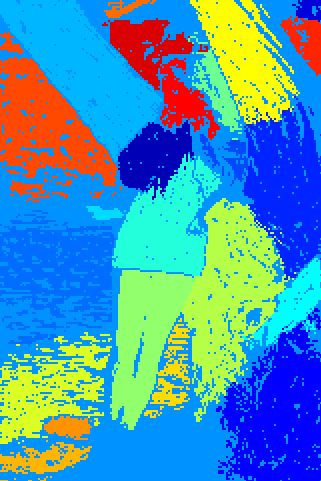}
        \caption{Segmentation into $24$ clusters with $24$ eigenvectors; Rand Index $0.77$.}
    \end{subfigure}
    \par\bigskip
    \begin{subfigure}[t]{0.25\textwidth}
        \includegraphics[width=\textwidth]{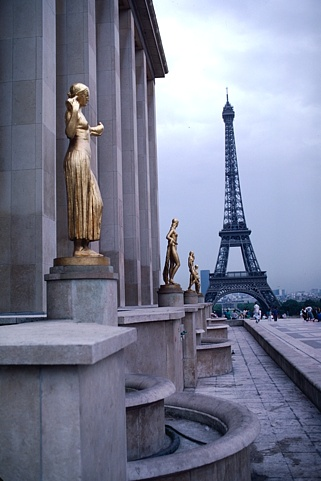}
        \caption{Original Image}
    \end{subfigure}
    \hspace{1.5em}
    \begin{subfigure}[t]{0.25\textwidth}
        \includegraphics[width=\textwidth]{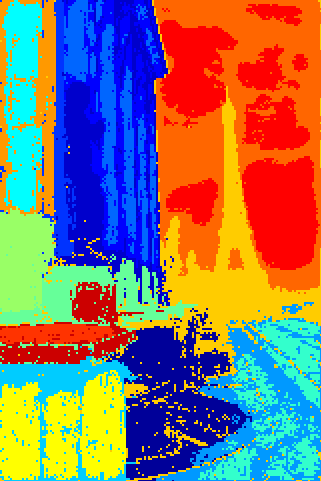}
        \caption{Segmentation into $18$ clusters with $6$ eigenvectors; Rand Index $0.77$.}
    \end{subfigure}
    \hspace{1.5em}
    \begin{subfigure}[t]{0.25\textwidth}
        \includegraphics[width=\textwidth]{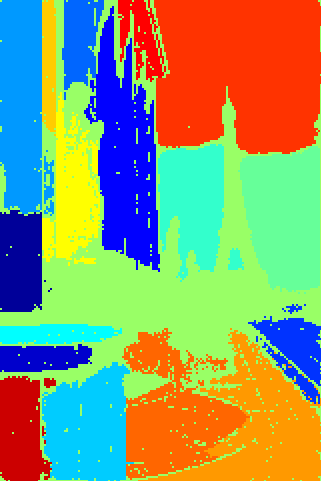}
        \caption{Segmentation into $18$ clusters with $18$ eigenvectors; Rand Index $0.74$.}
    \end{subfigure}
    \caption[Additonal examples from the BSDS dataset]{Examples of the segmentations produced with spectral clustering on the BSDS dataset.} 
    \label{fig:bsds_app_0}
\end{figure*}

\begin{figure*}[t]
    \centering
    \begin{subfigure}[t]{0.3\textwidth}
      \includegraphics[width=\textwidth]{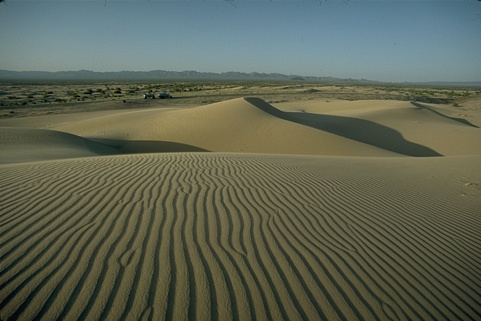}
      \caption{Original Image}
    \end{subfigure}
    \hspace{1em}
    \begin{subfigure}[t]{0.3\textwidth}
      \includegraphics[width=\textwidth]{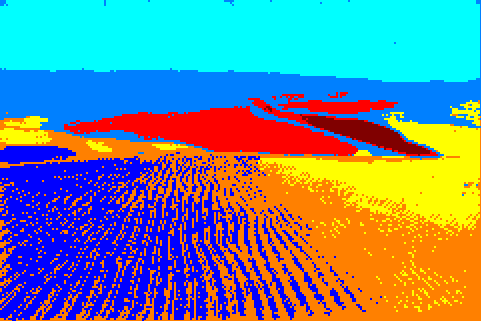}
      \caption{Segmentation into $7$ clusters with $3$ eigenvectors; Rand Index $0.76$.}
    \end{subfigure}
    \hspace{1em}
    \begin{subfigure}[t]{0.3\textwidth}
      \includegraphics[width=\textwidth]{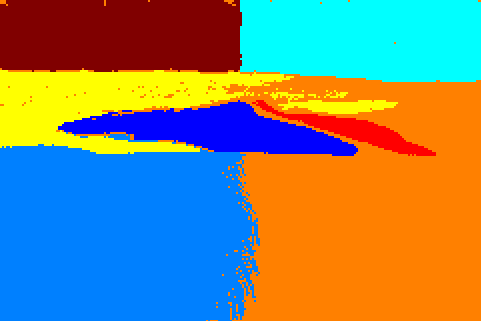}
      \caption{Segmentation into $7$ clusters with $7$ eigenvectors; Rand Index $0.74$.}
    \end{subfigure}
    \par\bigskip
    \begin{subfigure}[t]{0.3\textwidth}
      \includegraphics[width=\textwidth]{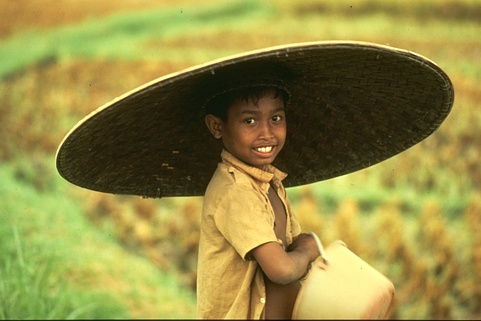}
      \caption{Original Image}
    \end{subfigure}
    \hspace{1em}
    \begin{subfigure}[t]{0.3\textwidth}
      \includegraphics[width=\textwidth]{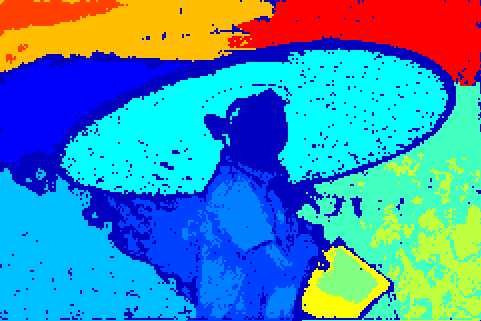}
      \caption{Segmentation into $13$ clusters with $8$ eigenvectors; Rand Index $0.80$.}
    \end{subfigure}
    \hspace{1em}
    \begin{subfigure}[t]{0.3\textwidth}
      \includegraphics[width=\textwidth]{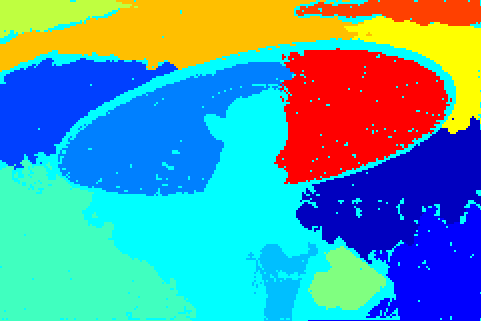}
      \caption{Segmentation into $13$ clusters with $13$ eigenvectors; Rand Index $0.77$.}
    \end{subfigure}
    \par\bigskip
    \begin{subfigure}[t]{0.3\textwidth}
      \includegraphics[width=\textwidth]{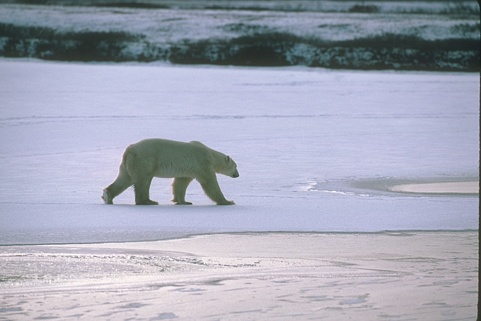}
      \caption{Original Image}
    \end{subfigure}
    \hspace{1em}
    \begin{subfigure}[t]{0.3\textwidth}
      \includegraphics[width=\textwidth]{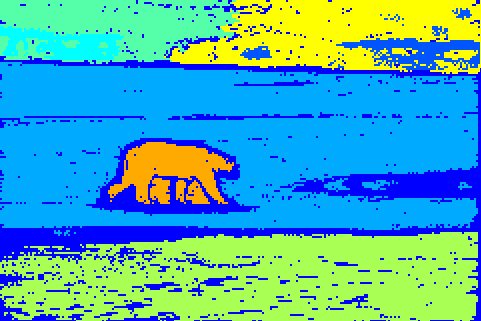}
      \caption{Segmentation into $8$ clusters with $5$ eigenvectors; Rand Index $0.86$.}
    \end{subfigure}
    \hspace{1em}
    \begin{subfigure}[t]{0.3\textwidth}
      \includegraphics[width=\textwidth]{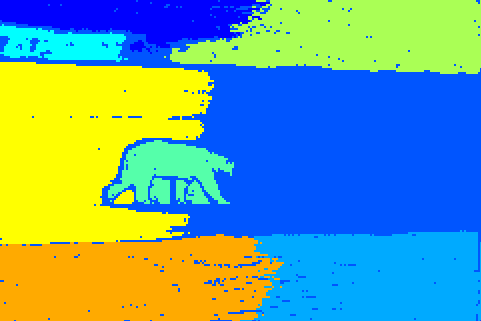}
      \caption{Segmentation into $8$ clusters with $8$ eigenvectors; Rand Index $0.79$.}
    \end{subfigure}
    \par\bigskip
    \begin{subfigure}[t]{0.3\textwidth}
      \includegraphics[width=\textwidth]{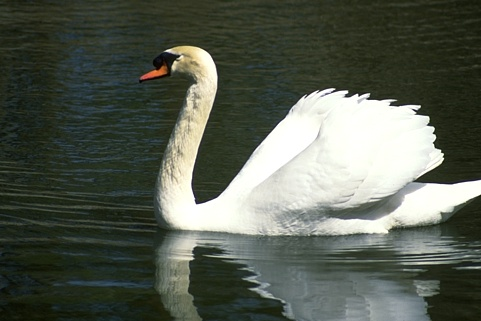}
      \caption{Original Image}
    \end{subfigure}
    \hspace{1em}
    \begin{subfigure}[t]{0.3\textwidth}
      \includegraphics[width=\textwidth]{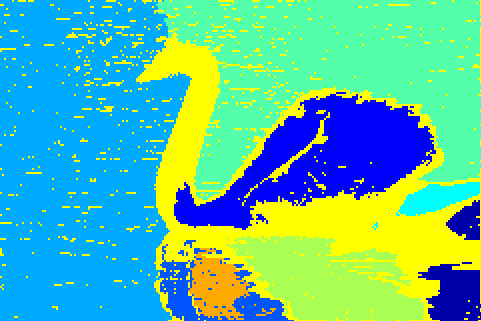}
      \caption{Segmentation into $9$ clusters with $7$ eigenvectors; Rand Index $0.69$.}
    \end{subfigure}
    \hspace{1em}
    \begin{subfigure}[t]{0.3\textwidth}
      \includegraphics[width=\textwidth]{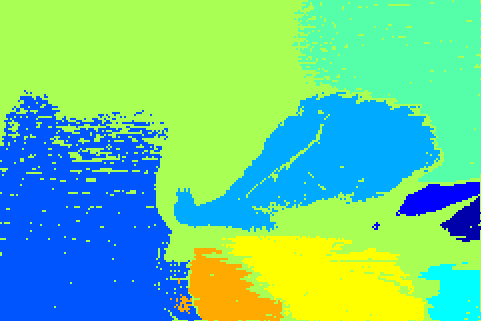}
      \caption{Segmentation into $9$ clusters with $9$ eigenvectors; Rand Index $0.61$.}
    \end{subfigure}
    \caption[Additional examples from the BSDS dataset]{Examples of the segmentations produced with spectral clustering on the BSDS dataset.} 
    \label{fig:bsds_app_2}
\end{figure*}
\newcommand{\simplify}[1]{\sigma \circ \left(#1\right)}
\newcommand{\sigmap}{(\sigma\circ \vecp)}

\chapter{Finding Densely Connected Clusters Locally} \label{chap:local}
In this chapter we turn our attention from learning the arbitrary structure of clusters to the problem of finding clusters with a specific known structure.
Moreover, we would like to solve the problem \emph{locally} so that we can learn the structure of a graph close to some particular vertex of interest.

To begin, let us recall the definition of a \emph{local clustering algorithm}.
Given an arbitrary vertex $u$ of graph $\graphg=(\vertexset,\edgeset)$ as input, a local graph clustering algorithm finds some low-conductance set $\sets \subset \vertexset$ containing $u$, while the algorithm runs in time proportional to the size of the target cluster and independent of the size of the graph $\graphg$.
 Because of the increasing size of available datasets, which makes centralised computation too expensive, local graph clustering has become an important learning technique for analysing a number of large-scale graphs and has been applied to solve many other learning and combinatorial optimisation problems~\cite{zhuLocalAlgorithmFinding2013, andersenLocalAlgorithmFinding2010,andersenAlmostOptimalLocal2016,fountoulakisPNormFlowDiffusion2020,liuStronglyLocalPnormcut2020,takaiHypergraphClusteringBased2020,  wangCapacityReleasingDiffusion2017, yinLocalHigherorderGraph2017}.
 In this chapter, we study local graph clustering for learning the structure of clusters that are defined by their inter-connections,
 and develop local algorithms to achieve this objective in both undirected and directed graphs.
 The design and analysis of the algorithms in this chapter build on the \emph{personalised \pagerank} and \emph{evolving set process} techniques for local clustering which are described in Section~\ref{sec:related:local}.

 The first result of this chapter is a local algorithm for finding densely connected clusters
in an undirected graph $\geqve$.
\begin{mainresult}[See Theorem~\ref{thm:main_thm} for the formal statement]
Given an undirected graph $\geqve$ and a seed vertex $u$, our local algorithm finds \emph{two} clusters $\setl, \setr$ around $u$, which are densely connected to each other and are loosely connected to $\vertexset \setminus (\lur)$.
\end{mainresult}
The design of the algorithm is based on a new reduction that allows us to relate the connections between $\setl, \setr$ and $\vertexset \setminus (\lur)$ to a \emph{single} cluster in the resulting graph $\graphh$, and a generalised analysis of \pagerank-based algorithms for local graph clustering.
The significance of the designed algorithm is demonstrated by experimental results on the Interstate Dispute Network from 1816 to 2010.
By connecting two vertices~(countries) with an undirected edge
weighted according to the severity of their military disputes
and using the USA as the seed vertex,
the algorithm recovers two groups of countries that tend to have conflicts with each other, and shows how the two groups evolve over time.
In particular, as shown in Figure~\ref{fig:intro_undirected}, the algorithm not only identifies the changing roles of Russia, Japan, and eastern Europe in line with 20th century geopolitics, but also the reunification of east and west Germany around 1990.

\begin{figure*}[t]
\centering
    \begin{subfigure}{0.48\textwidth}  
    \includegraphics[width=\textwidth]{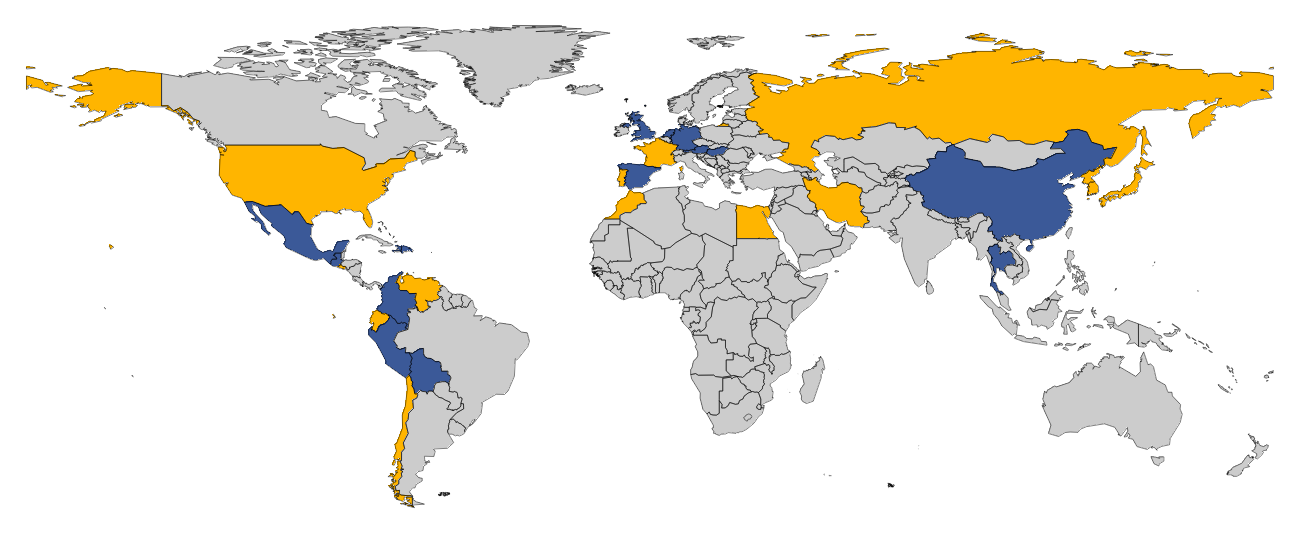}
    \caption{1816-1900}
    \end{subfigure}
    \hfill
    \begin{subfigure}{0.48\textwidth} 
    \includegraphics[width=\textwidth]{figures/localDense/1900-1950_yellow_blue.png}
    \caption{1900-1950}
    \end{subfigure}
    \begin{subfigure}{0.48\textwidth} 
    \includegraphics[width=\textwidth]{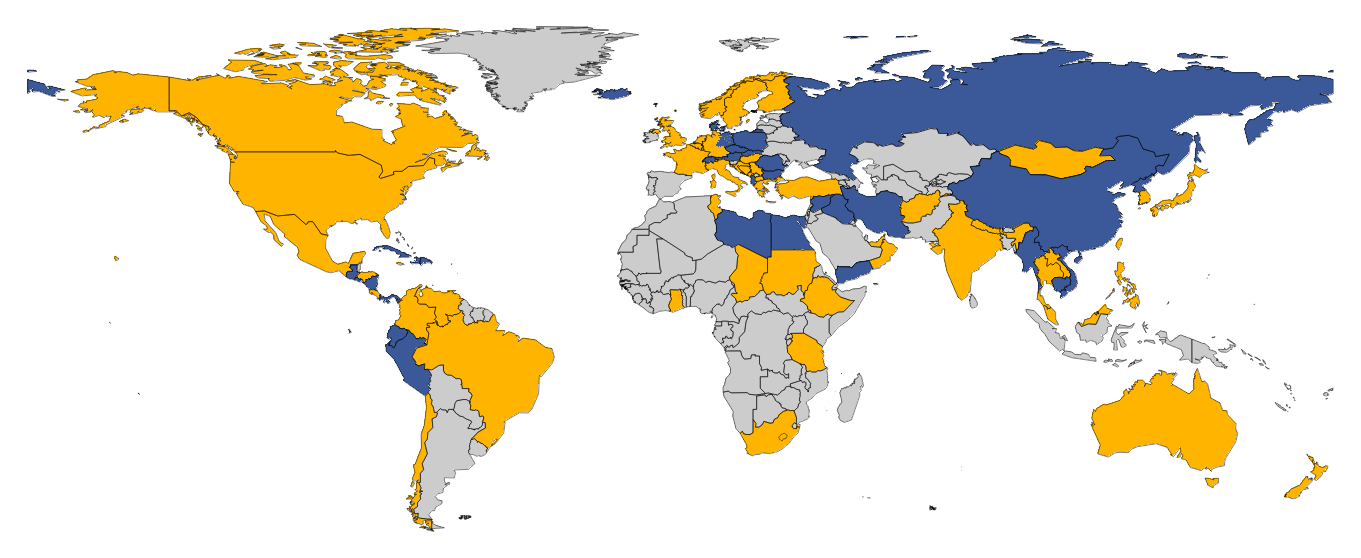}
    \caption{1950-1990}
    \end{subfigure}
    \hfill
    \begin{subfigure}{0.48\textwidth} 
    \includegraphics[width=\textwidth]{figures/localDense/1990-2010_yellow_blue.png}
    \caption{1990-2010}
    \end{subfigure}
    \caption[Example output of the \alglocbipartdc\ algorithm]{
    Clusters found by \alglocbipartdc\ on the interstate dispute network, using the USA as the seed vertex. In each case, countries in the yellow cluster tend to enter conflicts with countries in the blue cluster and vice versa.
    \label{fig:intro_undirected}}
\end{figure*}

 The second result of this chapter is a local algorithm for finding densely connected clusters
in a \emph{directed graph}.
\begin{mainresult}[See Theorem~\ref{thm:directedresult} for the formal statement]
Given a directed graph $\geqve$ and a seed vertex $u$, our local algorithm finds two clusters $\setl$ and $\setr$, such that (i) there are many edges \emph{from $\setl$ to $\setr$}, and (ii) $\lur$ is loosely connected to $\vertexset \setminus (\lur)$.
\end{mainresult}
The design of this algorithm is based on the following two techniques: (1) a
new reduction that allows us to relate the edge weight from $\setl$ to $\setr$, as well as the edge connections between $\lur$ and $\vertexset \setminus (\lur)$, to a \emph{single} vertex set in the resulting \emph{undirected} graph $\graphh$; (2) a refined analysis of the ESP-based algorithm for local graph clustering. 
 We demonstrate that the algorithm is able to recover local densely connected clusters in the US migration dataset, in which two vertex sets $\setl$ and $\setr$ defined as above could represent a higher-order migration trend.
 In particular, as shown in  Figure~\ref{fig:intro_directed}, by using different counties as starting vertices, the algorithm uncovers refined and more localised migration patterns than the previous work on the same dataset~\cite{cucuringuHermitianMatricesClustering2020}.
 This algorithm is the first local clustering algorithm that achieves a similar goal on directed graphs.
 
\begin{figure*}[t]
\centering
    \begin{subfigure}{0.48\textwidth} 
    \includegraphics[width=\textwidth]{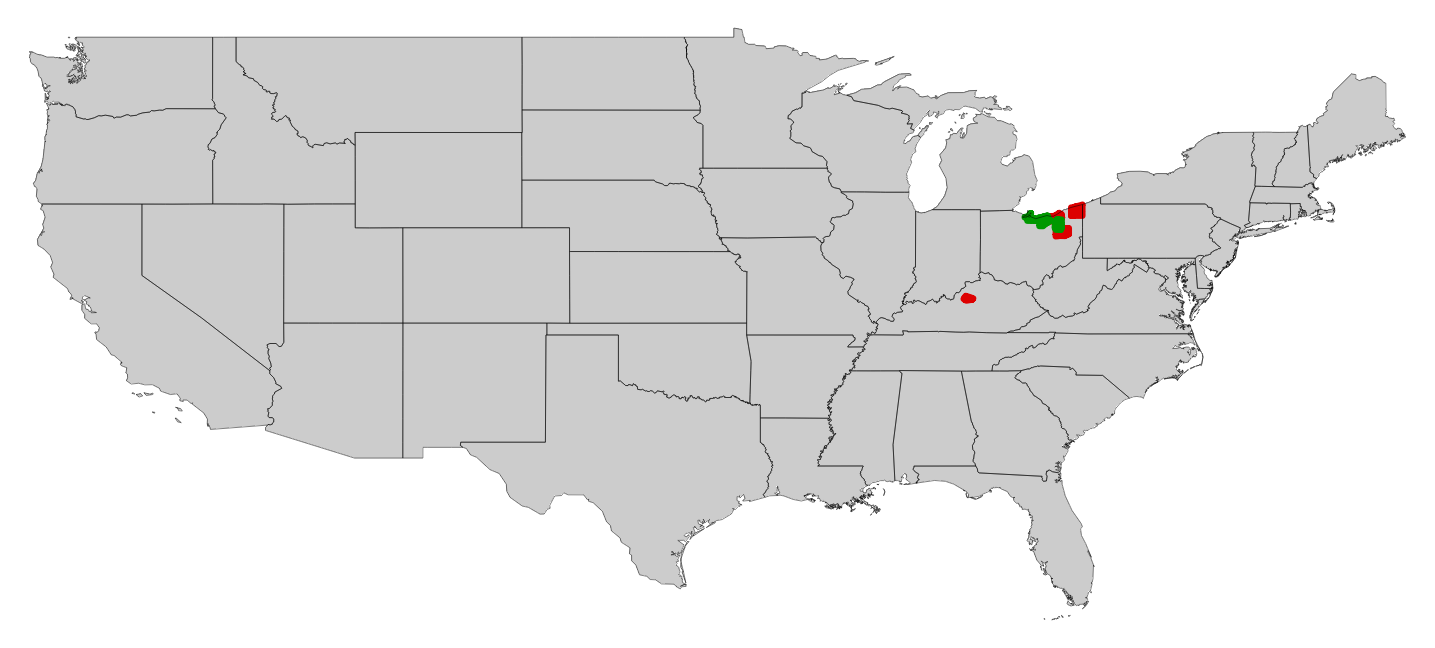}
    \caption{Ohio Seed}
    \end{subfigure}
    \hfill
    \begin{subfigure}{0.48\textwidth} 
    \includegraphics[width=\textwidth]{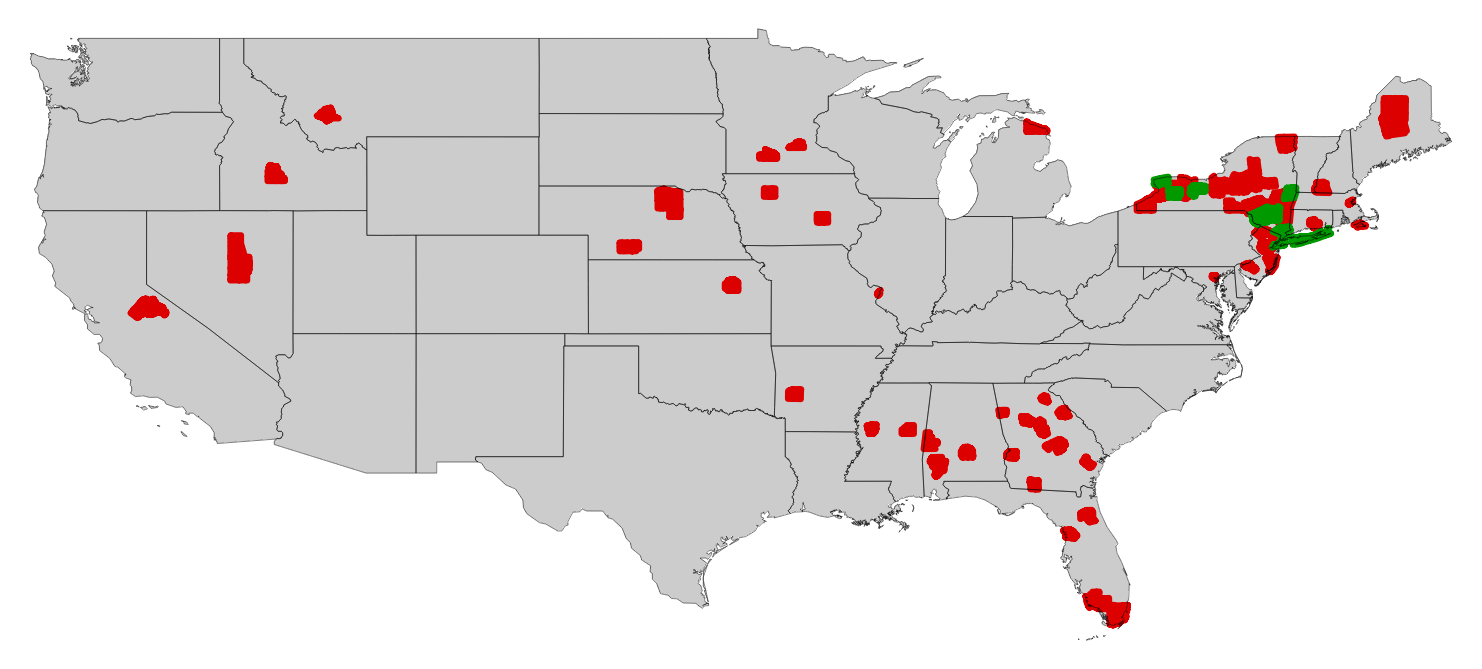}
    \caption{New York Seed}
    \end{subfigure}
    \begin{subfigure}{0.48\textwidth} 
    \includegraphics[width=\textwidth]{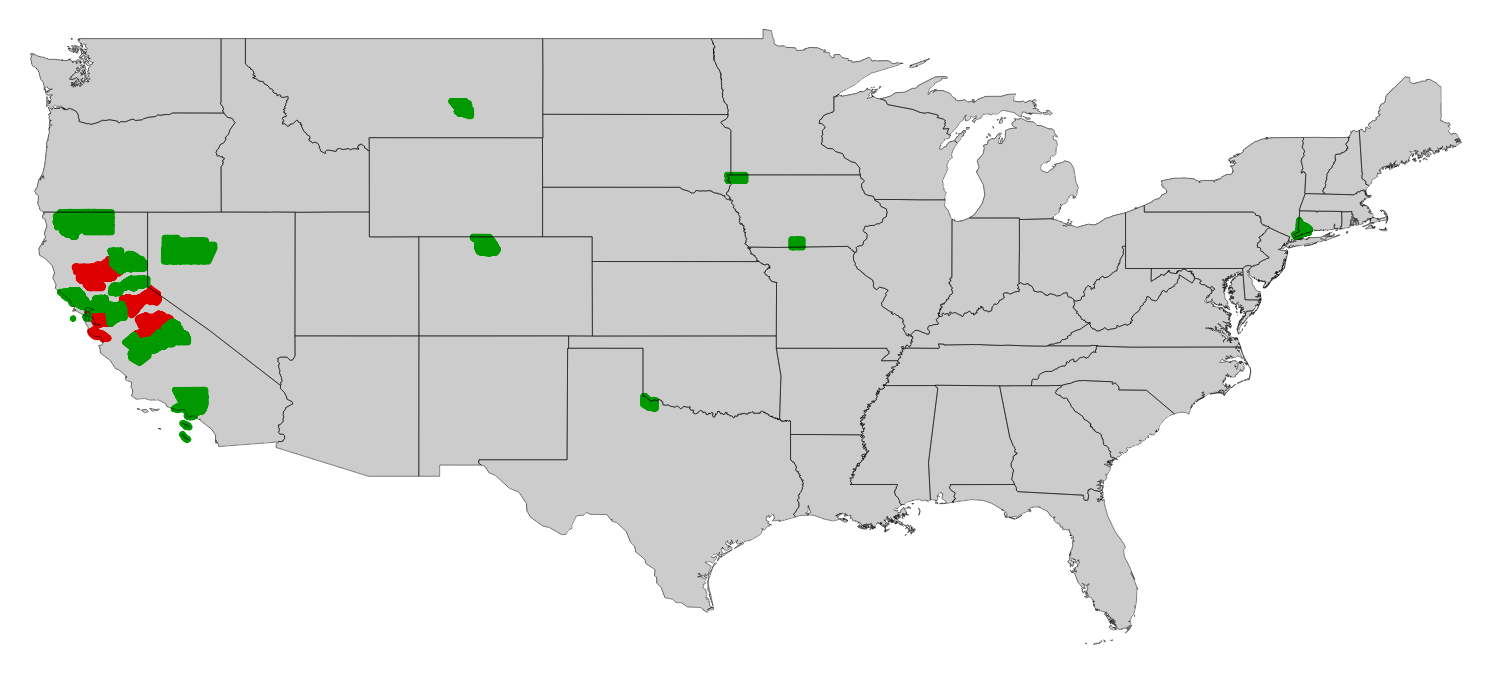}
    \caption{California Seed}
    \end{subfigure}
    \hfill
    \begin{subfigure}{0.48\textwidth} 
    \includegraphics[width=\textwidth]{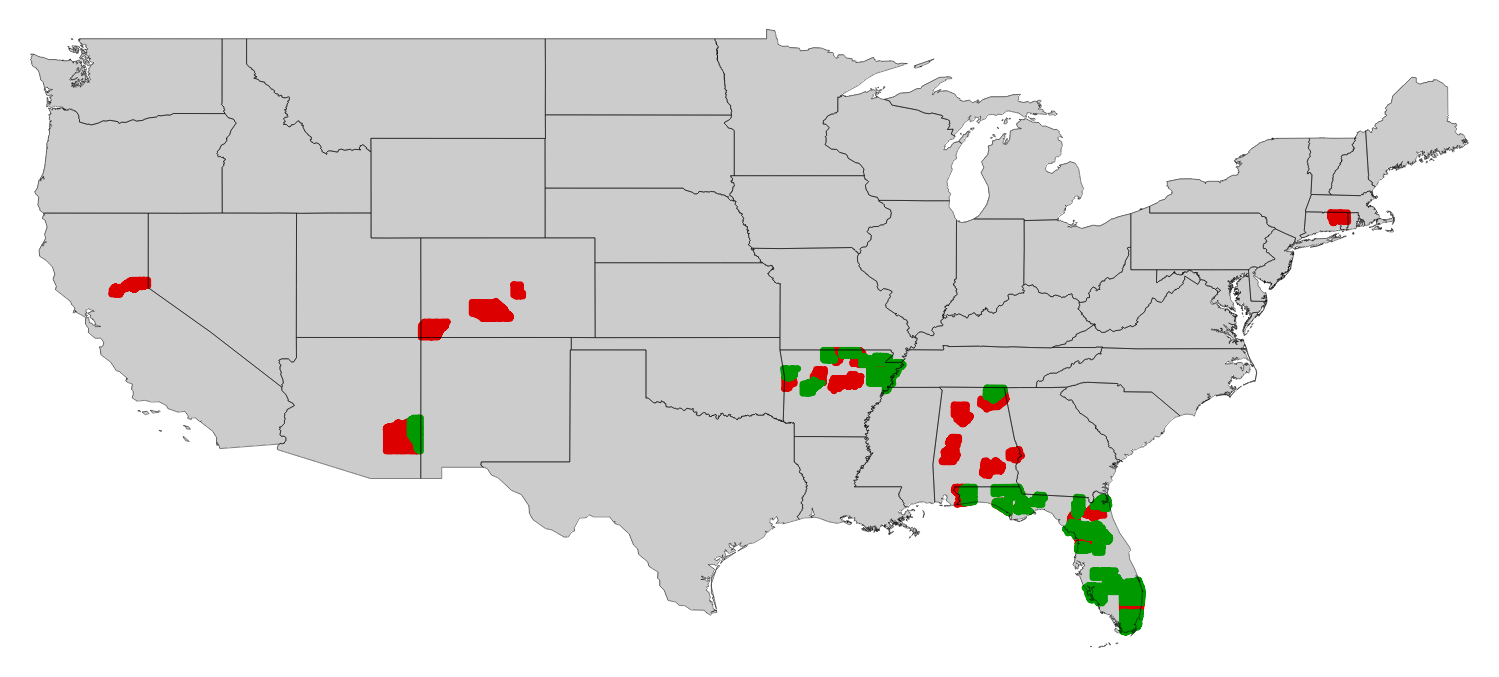}
    \caption{Florida Seed}
    \end{subfigure}
    \caption[Example output of the \algevocutdirected\ algorithm]{
    Clusters found by \algevocutdirected\ in the US migration network, in which four different counties are used as starting vertices of the local algorithm.
    In each case, people tend to migrate from the red counties to the green counties.} \label{fig:intro_directed}
\end{figure*}
 
\section{Algorithm for Undirected Graphs\label{sec:undirected}}
We first develop and analyse a local algorithm for finding two clusters in an undirected graph with a dense cut between them.
To formalise this notion, for any undirected graph $\geqve$ and disjoint $\setl, \setr \subset \vertexset$, we
 recall the definition of bipartiteness given in Section~\ref{sec:related:structured}:
\[
\bipart(\setl, \setr) \triangleq 1 - \frac{2 e(\setl, \setr)}{\vol(\lur)}.
\]
Notice that a low $\bipart(\setl, \setr)$ value means that there is a dense cut between $\setl$ and $\setr$, and there is a sparse cut between 
$\lur$ and $\vertexset \setminus (\lur)$.
In particular, $\bipart(\setl, \setr)=0$ implies that $(\setl, \setr)$ forms a bipartite and connected component of $\graphg$.
We   describe a local algorithm for finding almost-bipartite sets $\setl$ and $\setr$ with a low value of $\bipart(\setl, \setr)$.
 
\subsection{Reduction by Double Cover} \label{sec:reduction}
The design of most local algorithms for finding a target set $\sets \subset \vertexset$ of low conductance is based on analysing the behaviour of random walks starting from vertices in $\sets$. In particular, when the conductance $\cond_\graphg(\sets)$ is low, a random walk starting from most vertices in $\sets$ leaves $\sets$ with low probability.
However, for our setting, the target is a pair of sets $\setl, \setr$ with many connections between them. As such, a random walk starting in either $\setl$ or $\setr$ is very likely to leave the starting set. To address this,
 we   use a novel technique based on the double cover of $\graphg$
to reduce the problem of finding two sets of high conductance to the problem of finding one of \emph{low} conductance.
Formally, for any undirected graph $\graphg=(\vertexset_\graphg, \edgeset_\graphg)$, its \firstdef{double cover} is the  bipartite graph $\graphh=(\vertexset_\graphh, \edgeset_\graphh)$ defined as follows: (1) every vertex $v \in \vertexset_\graphg$ has two corresponding vertices $v_1, v_2\in \vertexset_\graphh$; (2) for every edge $\{u,v\}\in \edgeset_\graphg$, there are edges $\{u_1, v_2\}$  and $\{u_2, v_1\}$ in $\edgeset_\graphh$. See Figure~\ref{fig:dc} for an illustration.

\begin{figure}[t]
\centering
\scalebox{1}{\begin{tikzpicture}
	\begin{pgfonlayer}{nodelayer}
		\node [style=smallCircle] (0) at (0.75, 2) {$a_1$};
		\node [style=smallCircle] (2) at (0.75, -2) {$a_2$};
		\node [style=smallCircle] (8) at (-10.25, 2) {$a$};
		\node [style=smallCircle] (9) at (-6.75, 2) {$b$};
		\node [style=smallCircle] (10) at (-10.25, -2) {$c$};
		\node [style=smallCircle] (11) at (-6.75, -2) {$d$};
		\node [style=none] (12) at (-3, 0) {$\Longrightarrow$};
		\node [style=smallCircle] (21) at (4.25, 2) {$b_1$};
		\node [style=smallCircle] (22) at (4.25, -2) {$b_2$};
		\node [style=smallCircle] (23) at (7.75, 2) {$c_1$};
		\node [style=smallCircle] (24) at (7.75, -2) {$c_2$};
		\node [style=smallCircle] (25) at (11.25, 2) {$d_1$};
		\node [style=smallCircle] (26) at (11.25, -2) {$d_2$};
	\end{pgfonlayer}
	\begin{pgfonlayer}{edgelayer}
		\draw [style=undirected] (8) to (9);
		\draw [style=undirected blue] (9) to (11);
		\draw [style=undirected red] (8) to (11);
		\draw [style=undirected green] (8) to (10);
		\draw [style=undirected] (0) to (22);
		\draw [style=undirected] (2) to (21);
		\draw [style=undirected green] (0) to (24);
		\draw [style=undirected green] (2) to (23);
		\draw [style=undirected red] (0) to (26);
		\draw [style=undirected red] (2) to (25);
		\draw [style=undirected blue] (21) to (26);
		\draw [style=undirected blue] (22) to (25);
	\end{pgfonlayer}
\end{tikzpicture}}
\caption[Construction of the double cover]{\small{An example of the construction of the double cover\label{fig:dc}}}
\end{figure}
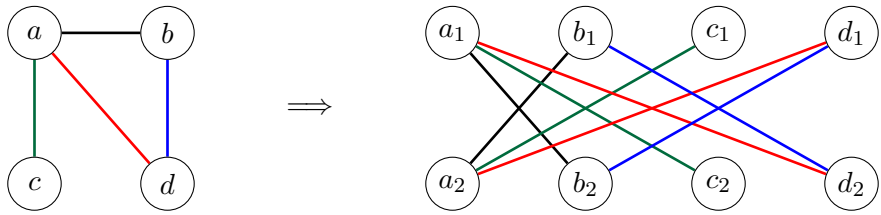

Now  we show a tight connection between the value of $\bipart(\setl,\setr)$ for any disjoint  sets  $\setl, \setr \subset \vertexset_\graphg$ and
the conductance of
a \emph{single} set in the double cover of $\graphg$.
To this end, for any $\sets \subset \vertexset_\graphg$, we define $\sets_1\subset \vertexset_\graphh$ and $\sets_2\subset \vertexset_\graphh$ by $\sets_1 \triangleq \{v_1~|~ v\in \sets\}$ and $\sets_2\triangleq \{v_2~|~ v\in \sets\}$.
We formalise the connection in the following lemma.
 
\begin{lemma} \label{lem:cond_bip}
Let $\graphg$ be an undirected graph, $\sets \subset \vertexset$ with partitioning $(\setl, \setr)$, and $\graphh$ be the double cover of $\graphg$.
Then, it holds that
$\cond_\graphh(\setl_1 \union \setr_2) = \bipart_\graphg(\setl, \setr)$.
\end{lemma}
\begin{proof}
    Let $\sets' = \setl_1 \cup \setr_2$, and by definition we have  that $
        \vol_\graphg(\sets) = \vol_\graphh(\sets')
    $. On the other hand, it holds that
    \[
    \vol_\graphh(\sets') = \weight_\graphh(\sets', \vertexset \setminus \sets') + 2\weight_\graphh(\sets', \sets') = \weight_\graphh(\sets', \vertexset \setminus \sets') + 2 \weight_\graphh(\setl_1, \setr_2).
    \]
    This implies that 
    \[
    \cond_\graphh(\sets') = \frac{\weight_\graphh(\sets', \vertexset \setminus \sets')}{\vol_\graphh(\sets')} = \frac{\vol_\graphg(\sets) - 2\weight_\graphh(\setl_1, \setr_2)}{\vol_\graphg(\sets)} = 1 - \frac{2\weight_\graphg(\setl,\setr)}{\vol_\graphg(\sets)} = \bipart_\graphg(\setl, \setr),
    \]
    which proves the statement.
\end{proof}

Next we look at the other direction of this correspondence.
Specifically, given any $\sets \subset \vertexset_\graphh$ in the double cover of a graph $\graphg$, we would like to find two disjoint sets $\setl \subset \vertexset_\graphg$ and $\setr \subset \vertexset_\graphg$ such that $\bipart_\graphg(\setl, \setr) = \cond_\graphh(\sets)$.
However, such a connection does not hold in general.
To overcome this, we restrict our attention to those subsets of $\vertexset_\graphh$ which can be unambiguously interpreted as two disjoint sets in $\vertexset_\graphg$.

\begin{definition} (Simple set) \label{def:simple}
We call $\sets \subset \vertexset_\graphh$ \firstdef{simple} if $\abs{\{v_1,v_2\}\intersect \sets}\leq 1$ holds for all $v\in \vertexset_\graphg$.
\end{definition}

Given this definition, the following corollary follows directly from Lemma~\ref{lem:cond_bip}.

\begin{corollary}\label{lem:ReductionForSimpleset}
    For any  simple set $\sets \subset \vertexset_\graphh$, let  $\setl = \{u : u_1 \in \sets\}$ and $\setr = \{u : u_2 \in \sets\}$. Then,   
    $
        \bipart_\graphg(\setl, \setr) = \cond_\graphh(\sets).
    $
\end{corollary}

\subsection{Design of the Algorithm} \label{sec:algdesc}
So far we have shown that the problem of finding densely connected sets 
$\setl, \setr \subset \vertexset_\graphg$ can be reduced to finding $\sets \subset \vertexset_\graphh$ of low conductance in the double cover $\graphh$, and this reduction raises the natural question of whether existing local algorithms can be directly employed to find $\setl$ and $\setr$ in $\graphg$.
However, this is not the case: even though a set $\sets \subset \vertexset_\graphh$ returned by most local algorithms is guaranteed to have low conductance, vertices of $\graphg$ could be included in $\sets$ twice, and as such $\sets$ doesn't necessarily give us disjoint sets $\setl, \setr \subset \vertexset_\graphg$ with low value of $\bipart(\setl,\setr)$.
See Figure~\ref{fig:nonsimple} for an illustration.

\begin{figure}[t]
\centering
\begin{tikzpicture}
	\begin{pgfonlayer}{nodelayer}
		\node [style=none] (11) at (-4.5, 2) {\huge{$\cdots$}};
		\node [style=smallCircle] (21) at (-1.5, 2) {$a_1$};
		\node [style=none] (1) at (-3.2, 3) {$S$};
		\node [style=smallCircle] (31) at (1.5, 2) {$b_1$};
		\node [style=smallCircle] (41) at (4.5, 2) {$c_1$};
		\node [style=smallCircle] (51) at (7.5, 2) {$d_1$};
		\node [style=smallCircle] (61) at (10.5, 2) {$e_1$};
		\node [style=smallCircle] (71) at (13.5, 2) {$f_1$};
		\node [style=none] (81) at (16.5, 2) {\huge{$\cdots$}};
		\node [style=none] (12) at (-4.5, -2) {\huge{$\cdots$}};
		\node [style=smallCircle] (22) at (-1.5, -2) {$a_2$};
		\node [style=smallCircle] (32) at (1.5, -2) {$b_2$};
		\node [style=smallCircle] (42) at (4.5, -2) {$c_2$};
		\node [style=smallCircle] (52) at (7.5, -2) {$d_2$};
		\node [style=smallCircle] (62) at (10.5, -2) {$e_2$};
		\node [style=smallCircle] (72) at (13.5, -2) {$f_2$};
		\node [style=none] (82) at (16.5, -2) {\huge{$\cdots$}};
		\node [style=none] (83) at (-1, 3) {};
		\node [style=none] (84) at (-2.5, 1.75) {};
		\node [style=none] (85) at (0.25, -2.25) {};
		\node [style=none] (86) at (1.75, -3) {};
		\node [style=none] (87) at (13, 3) {};
		\node [style=none] (88) at (14.6, 1.675) {};
		\node [style=none] (89) at (11.75, -2.25) {};
		\node [style=none] (90) at (10, -3) {};
	\end{pgfonlayer}
	\begin{pgfonlayer}{edgelayer}
		\draw [fill=blue!20] (85.center)
			 to (84.center)
			 to [in=180, out=120, looseness=1.50] (83.center)
			 to (87.center)
			 to [in=60, out=0, looseness=1.25] (88.center)
			 to (89.center)
			 to [in=0, out=-135] (90.center)
			 to (86.center)
			 to [in=-60, out=180] cycle;
		\draw [style=dashed] (71) to (82.center);
		\draw [style=dashed] (81.center) to (72);
		\draw [style=dashed] (21) to (12.center);
		\draw [style=dashed] (11.center) to (22);
		\draw [style=undirected] (31) to (42);
		\draw [style=undirected] (32) to (41);
		\draw [style=undirected] (41) to (52);
		\draw [style=undirected] (42) to (51);
		\draw [style=undirected] (31) to (52);
		\draw [style=undirected] (32) to (51);
		\draw [style=undirected] (51) to (62);
		\draw [style=undirected] (52) to (61);
		\draw [style=undirected] (61) to (72);
		\draw [style=undirected] (62) to (71);
		\draw [style=undirected] (21) to (32);
		\draw [style=undirected] (22) to (31);
		\draw [style=undirected] (41) to (62);
		\draw [style=undirected] (42) to (61);
	\end{pgfonlayer}
\end{tikzpicture}
\caption[A non-simple set in the double cover which has low conductance]{\small{The indicated vertex set $\sets \subset \vertexset_\graphh$
has low conductance.
However, since $\sets$ contains many pairs of vertices which correspond to the same vertex in $\graphg$, $\sets$ gives us little information for finding disjoint $\setl,\setr \subset \vertexset_\graphg$ with a low $\bipart(\setl,\setr)$ value.}}
\label{fig:nonsimple}
\end{figure}
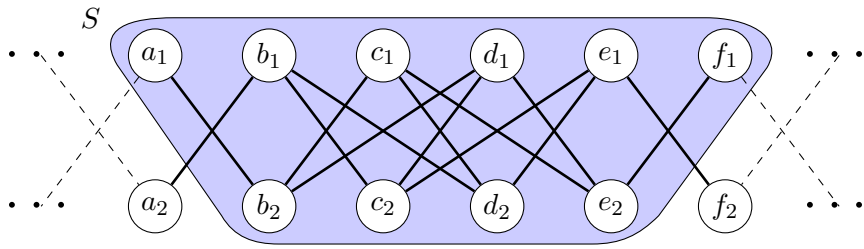

\subsubsection{The Simplify Operator}
Taking the example shown in Figure~\ref{fig:nonsimple} into account,  our objective is to  design a local algorithm for finding some set $\sets \subset \vertexset_\graphh$ of low conductance which is also simple.
To ensure this, we  introduce  the simplify operator, and analyse its properties.
 
\begin{definition}[Simplify operator] Let $\graphh$ be the double cover of $\graphg$, where the two corresponding vertices of any $u\in \vertexset_\graphg$ are defined as $u_1,u_2\in \vertexset_\graphh$. Then, for any   $\vecp \in \R^{2n}_{\geq 0}$,  the simplify operator is a function $\sigma:\R^{2n}_{\geq 0} \rightarrow \R^{2n}_{\geq 0}$  defined by 
\begin{align*}
    	(\sigma\circ \vecp)(u_1) & \triangleq \max (0, \vecp(u_1) - \vecp(u_2)), \\
        (\sigma\circ \vecp)(u_2) & \triangleq \max (0, \vecp(u_2) - \vecp(u_1))
\end{align*}
for  every $u\in \vertexset_\graphg$.
\end{definition}
 
Notice that, for any vector $\vecp$ and any $u\in \vertexset_\graphg$,  at most one of $u_1$ and $u_2$ is in the support of $\sigma \circ \vecp$; hence, the support of $\sigma\circ \vecp$ is always simple. To explain the meaning of $\sigma$, for a vector $\vecp$ one could view $\vecp(u_1)$ as our `confidence' that $u\in \setl$, and   $\vecp(u_2)$ as our `confidence' that $u\in \setr $. Hence, when $\vecp(u_1) \approx \vecp(u_2)$, both $(\sigma \circ \vecp)(u_1)$ and $(\sigma \circ \vecp)(u_2)$ are low which captures the fact that we would not prefer to include $u$ in either $\setl$ or $\setr$. 
On the other hand, when $\vecp(u_1) \gg \vecp(u_2)$, we have $(\sigma \circ \vecp)(u_1) \approx \vecp(u_1)$, which captures our confidence that $u$ should belong to  $\setl$. 
The following lemma summaries some key properties of $\sigma$.
 
\begin{lemma} \label{lem:sigmaproperty}
The  following  holds for the 
$\sigma$-operator:   
\begin{itemize} 
    \item $\sigma\circ(c\cdot \vecp) = c\cdot (\sigma \circ \vecp)$ for $\vecp \in \R_{\geq 0}^{2n}$ and any $c \in\R_{\geq 0}$;
    \item $\sigma\circ(\veca + \vecb) \preceq \sigma\circ \veca + \sigma\circ \vecb$ for $\veca, \vecb \in \R_{\geq 0}^{2n}$;
    \item $\sigma \circ (\lazywalkm \vecp) \preceq  \lazywalkm (\sigma \circ \vecp)$ for $\vecp \in \R_{\geq 0}^{2n}$.
\end{itemize}
\end{lemma}
\begin{proof}
For any $u_1 \in \vertexset_\graphh$, we have by definition that 
\begin{align*}
    \sigma\circ(c\cdot  \vecp)(u_1) & = \max(0, c\cdot \vecp(u_1) - c\cdot  \vecp(u_2)) \\
    & = c\cdot \max(0, \vecp(u_1) - \vecp(u_2)) \\
    & = c \cdot \sigma\circ \vecp(u_1),
\end{align*}
and similarly we have $\sigma\circ(c\cdot \vecp)(u_2) = c \cdot \sigma\circ \vecp(u_2)$. Therefore, the first statement holds.

Now we prove the second statement. Let $u_1, u_2 \in \vertexset_\graphh$ be the vertices that correspond to $u\in \vertexset_\graphg$. We have by definition that 
\begin{align*}
    \sigma\circ(\veca + \vecb)(u_1) & = \max (0, \veca(u_1) + \vecb(u_1) - \veca(u_2) - \vecb(u_2) )  \\ 
   &  = \max (0, \veca(u_1)  - \veca(u_2) + \vecb(u_1) - \vecb(u_2) ) \\
   & \leq \max (0, \veca(u_1)  - \veca(u_2)) + \max (0, \vecb(u_1) - \vecb(u_2) )\\
   &  = \sigma\circ \veca(u_1) + \sigma\circ \vecb(u_1),
\end{align*}
where the last inequality holds by the fact that $\max(0, x+y) \leq \max (0, x) + \max(0,y) $ for any $x,y\in\R$.
For the same reason, we have that 
\[
    \sigma\circ(\veca+\vecb) (u_2) \leq \sigma\circ \veca(u_2) + \sigma\circ \vecb(u_2).
\]
Combining these proves the second statement.

Finally, we   prove the third statement.
Recall that the lazy random walk matrix for a graph is defined to be $\lazywalkm = \frac{1}{2}(\identity + \adj \degm^{-1})$.
By this definition, we have that
\[
(\lazywalkm \vecp)(u_1) = \frac{1}{2}\vecp(u_1) + \frac{1}{2} \sum_{v \in N_\graphg(u)} \frac{\vecp(v_2)}{\deg_\graphh(v_2)},
\]
and
\[
(\lazywalkm \vecp)(u_2) = \frac{1}{2}\vecp(u_2) + \frac{1}{2} \sum_{v \in N_\graphg(u)} \frac{\vecp(v_1)}{\deg_\graphh(v_1)}.
\]
Without loss of generality, we assume that $(\lazywalkm \vecp)(u_1) \geq (\lazywalkm \vecp)(u_2)$. This implies that $(\sigma\circ(\lazywalkm \vecp))(u_2) = 0 \leq (\lazywalkm \sigmap)(u_2)$, hence it suffices to prove that $(\sigma\circ(\lazywalkm \vecp))(u_1) \leq (\lazywalkm \sigmap)(u_1)$. By definition, we have that 
\begin{align}
(\sigma\circ(\lazywalkm \vecp))(u_1)   &  = \max(0, (\lazywalkm \vecp)(u_1) - (\lazywalkm \vecp)(u_2)) \notag \\
& = (\lazywalkm \vecp)(u_1) - (\lazywalkm \vecp)(u_2) \notag\\
& = \frac{1}{2}\left( \vecp(u_1) - \vecp(u_2)\right)  + \frac{1}{2} \sum_{v \in N_\graphg(u)} \left( \frac{\vecp(v_2)}{\deg_\graphh(v_2)} -\frac{\vecp(v_1)}{\deg_\graphh(v_1)} \right) \notag\\
& = \frac{1}{2}\left(\vecp(u_1) - \vecp(u_2)\right)  + \frac{1}{2} \sum_{v \in N_\graphg(u)} \frac{\vecp(v_2) - \vecp(v_1)}{\deg_\graphg(v)}, \label{eq:sigmadif}
\end{align}
where the last equality follows by the fact that $\deg_\graphh(v_1) = \deg_\graphh(v_2) = \deg_\graphg(v)$  for any $v\in \vertexset_\graphg$. To analyse \eqref{eq:sigmadif}, notice that $\vecp(v_1) - \vecp(v_2) =(\sigma\circ \vecp)(v_1) -\sigmap (v_2)$ for any $v\in \vertexset_\graphg$ for the following reasons:
\begin{itemize}
    \item When $\vecp(v_1)\geq \vecp(v_2)$, we have that $
(\sigma\circ \vecp)(v_1) -\sigmap (v_2) = \left(\vecp(v_1) - \vecp(v_2)\right) - 0 = \vecp(v_1) - \vecp(v_2)$;
    \item Otherwise, we have $\vecp(v_1)<\vecp(v_2)$ and $
(\sigma\circ \vecp)(v_1) -\sigmap (v_2) = 0 - (\vecp(v_2) - \vecp(v_1)) = \vecp(v_1) - \vecp(v_2)$.
\end{itemize} 
Therefore, we can rewrite \eqref{eq:sigmadif} as  
\begin{align*}
        (\sigma\circ (\lazywalkm \vecp))(u_1)  
       &  = \frac{1}{2}((\sigma\circ \vecp)(u_1) - (\sigma\circ \vecp)(u_2)) + \frac{1}{2} \sum_{v \in N(u)} \frac{ \sigmap(v_2) - \sigmap(v_1)}{\deg_\graphg(v)}  \\
   & \leq \frac{1}{2}\cdot\sigmap(u_1) + \frac{1}{2}\sum_{v \in N(u)} \frac{ \sigmap(v_2) }{\deg_\graphg(v)} \\
    & = (\lazywalkm \sigmap)(u_1),
    \end{align*}
where the inequality holds by the fact that $\sigmap(u_1)\geq 0$ and $\sigmap(u_2)\geq 0$ for any $u\in \vertexset_\graphg$.
Since the analysis above holds for any $v\in \vertexset_\graphg$, we have that $\sigma\circ(\lazywalkm \vecp)\preceq \lazywalkm \sigmap$, which finishes the proof of the lemma.
\end{proof}

While our analysis is based on all the three properties in Lemma~\ref{lem:sigmaproperty}, the third is of particular importance:  it implies that, if $\vecp$ is the probability distribution of a random walk in $\graphh$, applying $\sigma$ before taking a one-step random walk would never result in lower probability mass than applying $\sigma$ after taking a one-step random walk.
This means that no probability mass would be lost when the $\sigma$-operator is applied between every step of a random walk, in comparison with applying $\sigma$ at the end of an entire random walk process.

\subsubsection{Description of the Algorithm}
Our proposed algorithm is conceptually simple: every vertex $u$ of the input graph $\graphg$ maintains two copies $u_1, u_2$ of itself, and these two `virtual' vertices are used to simulate $u$'s corresponding vertices in the double cover $\graphh$ of $\graphg$.
Then, as the neighbours of $u_1$ and $u_2$ in $\graphh$ are entirely determined by $u$'s neighbours in $\graphg$ and can be constructed locally, a random walk process in $\graphh$ is simulated in $\graphg$.
This   allows us to apply a local algorithm similar to the one by Anderson et al.~\cite{andersenLocalGraphPartitioning2006} on this `virtual' graph $\graphh$.
Finally, since all the required information about $u_1,u_2\in \vertexset_\graphh$ is maintained by $u\in \vertexset_\graphg$, the $\sigma$-operator can be applied locally. 

The formal description of the algorithm is given in Algorithm~\ref{alg:local_max_cut}, which invokes Algorithm~\ref{alg:aprdc} as the key component  to compute the approximate \pagerank\ $\apr_\graphh(\alpha, \indicatorvec_{u_1}, \vecr)$. Specifically, Algorithm~\ref{alg:aprdc} maintains, for every vertex $u \in \vertexset_\graphg$, tuples $(\vecp(u_1), \vecp(u_2))$ and $(\vecr(u_1), \vecr(u_2))$ to keep track of the values of $\vecp$ and $\vecr$ in $\graphg$'s double cover.
For a given vertex $u$, the entries in these tuples are expressed by $\vecp_1(u), \vecp_2(u), \vecr_1(u),$ and $\vecr_2(u)$ respectively.
Every $\algdcpush$ operation~(Algorithm~\ref{alg:dcpush}) preserves the invariant
$
        \vecp + \ppr_\graphh(\alpha, \vecr) = \ppr_\graphh(\alpha, \indicatorvec_{v_1}),
$
which ensures that the final output of Algorithm~\ref{alg:aprdc} is an approximate Pagerank vector.
It is worth noting that, although the presentation of the \algaprdc\ procedure is   similar to the one in Andersen et al.~\cite{andersenLocalGraphPartitioning2006}, in the \algdcpush\ procedure the update of the residual vector $\vecr$ is slightly more involved: specifically, for every vertex $u\in \vertexset_\graphg$, both $\vecr_1(u)$ and $\vecr_2(u)$ are needed in order to update $\vecr_1(u)$ (or $\vecr_2(u)$).
That is one of the reasons that the performance of the algorithm in Andersen et al.~\cite{andersenLocalGraphPartitioning2006} cannot be directly applied for this algorithm, and a more technical analysis, some of which is parallel to theirs, is needed in order to analyse the correctness and performance of the algorithm.

\begin{algorithm}[htb] \SetAlgoLined
\SetKwInOut{Input}{Input}
\SetKwInOut{Output}{Output}
   \Input{A graph $\graphg$, starting vertex $u$, target volume $\gamma$, and target bipartiteness $\beta$}
   \Output{Two disjoint sets $\setl$ and $\setr$}
  Set $\alpha = \frac{\beta^2}{378}$, and  $\epsilon = \frac{1}{20 \gamma}$ \\
Compute $\vecp' = \algaprdc(u, \alpha, \epsilon)$ \\
Compute $\vecp = \sigma \circ \vecp'$ \\
\For{$j \in [1, \abs{\supp(p)}]$}{
    \If{$\cond(\pjsweep) \leq \beta$}{
      Set $\setl = \{u : u_1 \in \pjsweep \}$, and   $\setr = \{u : u_2 \in \pjsweep \}$ \\
    \Return $(\setl, \setr)$
    }
}
   \caption[Find densely connected clusters locally: \alglocbipartdc$(\graphg, u, \gamma, \beta)$]{\alglocbipartdc$(\graphg, u, \gamma, \beta)$ \label{alg:local_max_cut}}
\end{algorithm}

\begin{algorithm}[htb] \SetAlgoLined
\SetKwInOut{Input}{Input}
\SetKwInOut{Output}{Output}
   \Input{Starting vertex $v$, parameters $\alpha$ and $\epsilon$ }
   \Output{Approximate Pagerank vector $\apr_H(\alpha, \indicatorvec_{v_1}, \vecr)$}
   Set $\vecp_1 = \vecp_2 = \vecr_2 = \zerovec$; set $\vecr_1 = \indicatorvec_v$ \\
   \While{$\max_{(u, i) \in \vertexset \times \{1, 2\}} \frac{\vecr_i(u)}{\deg(u)} \geq \epsilon$} {
   Choose any $u$ and $i \in \{1, 2\}$ such that $\frac{\vecr_i(u)}{\deg(u)} \geq \epsilon$ \\
   $(\vecp_1, \vecp_2, \vecr_1, \vecr_2) = \algdcpush(\alpha, (u, i), \vecp_1, \vecp_2, \vecr_1, \vecr_2)$
   }
   \Return $\vecp = \left[\vecp_1, \vecp_2 \right]$
   \caption[Approximate \pagerank\ on the double cover: \algaprdcshort$(v, \alpha, \epsilon)$]{\algaprdc$(v, \alpha, \epsilon)$ \label{alg:aprdc}}
\end{algorithm}

\begin{algorithm}[htb] \SetAlgoLined
\SetKwInOut{Input}{Input}
\SetKwInOut{Output}{Output}
   \Input{$\alpha, (u, i), \vecp_1, \vecp_2, \vecr_1, \vecr_2$}
   \Output{$(\vecp'_1, \vecp'_2, \vecr'_1, \vecr'_2)$}
   Set $(\vecp_1', \vecp_2', \vecr_1', \vecr_2') = (\vecp_1, \vecp_2, \vecr_1, \vecr_2)$ \\
   Set $\vecp'_i(u) = \vecp_i(u) + \alpha \vecr_i(u)$ \\
   Set $\vecr'_i(u) = (1 - \alpha) \frac{\vecr_i(u)}{2}$ \\
   \For{$v\in N_\graphg(u)$}{
   Set $\vecr'_{3 - i}(v) = \vecr_{3 - i}(v) + (1 - \alpha) \frac{\vecr_i(u)}{2 \deg(u)}$
   }
   \Return $(\vecp'_1, \vecp'_2, \vecr'_1, \vecr'_2)$
   \caption[\pagerank\ helper method: \algdcpush$(\alpha, u, i, \vecp_1, \vecp_2, \vecr_1, \vecr_2)$]{\algdcpush$(\alpha, (u, i), \vecp_1, \vecp_2, \vecr_1, \vecr_2)$}
   \label{alg:dcpush}
\end{algorithm}

\subsection{Analysis of \algaprdc}
We begin by considering the \algaprdc\ algorithm,  and   show that it computes an approximate \pagerank\ vector on the double cover of the input graph.
The following lemma shows that \algdcpush\ maintains that throughout the algorithm $\left[\vecp_1^\transpose, \vecp_2^\transpose\right]^\transpose$ is an approximate Pagerank vector on the double cover with residual $\left[\vecr_1^\transpose, \vecr_2^\transpose\right]^\transpose$,
where $\left[\veca^\transpose, \vecb^\transpose\right]^\transpose$ for $\veca \in \R^{n_1}$ and $\vecb \in \R^{n_2}$ is the vector in $\R^{n_1 + n_2}$ obtained by concatenating $\veca$ and $\vecb$.
\begin{lemma} \label{lem:dcpush}
Let $\vecp'_1, \vecp'_2, \vecr'_1$ and $\vecr'_2$ be the result of $\algdcpush((u, i), \vecp_1, \vecp_2, \vecr_1, \vecr_2)$. Let $\vecp = \left[\vecp_1^\transpose, \vecp_2^\transpose \right]^\transpose$, $\vecr = \left[\vecr_1^\transpose, \vecr_2^\transpose \right]^\transpose$, $\vecp' = \left[(\vecp'_1)^\transpose, (\vecp'_2)^\transpose \right]^\transpose$ and $\vecr' = \left[(\vecr'_1)^\transpose, (\vecr'_2)^\transpose \right]^\transpose$.
Then, we have that
$$
\vecp' + \ppr_\graphh(\alpha, \vecr') = \vecp + \ppr_\graphh(\alpha, \vecr).$$
\end{lemma}
\begin{proof}
    After the push operation, we have
    \begin{align*}
        \vecp' & = \vecp + \alpha \vecr(u_i) \indicatorvec_{u_i} \\
        \vecr' & = \vecr - \vecr(u_i) \indicatorvec_{u_i} + (1 - \alpha) \vecr(u_i) \lazywalkm_\graphh \indicatorvec_{u_i}  
    \end{align*}
    where $\lazywalkm_\graphh$ is the lazy random walk matrix on the double cover of the input graph. 
    Then, we have
    \begin{align*}
        \lefteqn{\vecp + \ppr_\graphh(\alpha, \vecr)}\\
        & = \vecp + \ppr_\graphh(\alpha, \vecr - \vecr(u_i) \indicatorvec_{u_i}) + \ppr_\graphh(\alpha, \vecr(u_i) \indicatorvec_{u_i}) \\
        & = \vecp + \ppr_\graphh(\alpha, \vecr - \vecr(u_i) \indicatorvec_{u_i}) + \left[\alpha \vecr(u) \indicatorvec_{u_i} + (1 - \alpha) \ppr_\graphh(\alpha, \vecr(u) \lazywalkm_\graphh \indicatorvec_{u_i}) \right] \\
        & = \left[ \vecp + \alpha \vecr(u_i) \indicatorvec_{u_i}\right] + \ppr_\graphh(\alpha, \left[\vecr - \vecr(u_i) \indicatorvec_{u_i} + (1 - \alpha)\vecr(u_i)\lazywalkm_\graphh \indicatorvec_{u_i}\right]) \\
        & = \vecp' + \ppr_\graphh(\alpha, \vecr').
    \end{align*}
    Here, we use the fact that $\ppr(\alpha, \cdot)$ is linear and 
    \[
        \ppr(\alpha, \vecs) = \alpha \vecs + (1 - \alpha) \ppr(\alpha, \lazywalkm_\graphh \vecs),
    \]
    which follows from the fact that $\ppr(\alpha, \vecs) = \alpha \sum_{t = 0}^\infty (1 - \alpha)^t \lazywalkm_\graphh^t \vecs$.
\end{proof}
Now we can show the correctness and performance guarantee of the \algaprdc\ algorithm.
\begin{lemma} \label{lem:aprdc}
    For a graph $\graphg$  and any $v \in \vertexset_\graphg$, $\algaprdc(v, \alpha, \epsilon)$ has running time $\bigo{\frac{1}{\epsilon \alpha}}$ and computes $\vecp = \apr_\graphh(\alpha, \indicatorvec_{v_1}, \vecr)$ such that $\max_{u \in \vertexset_\graphh} \frac{\vecr(u)}{\deg(u)} < \epsilon$ and $\vol(\supp(\vecp))) \leq \frac{1}{\epsilon \alpha}$, where $\graphh$ is the double cover of $\graphg$.
\end{lemma}
\begin{proof}
    By Lemma~\ref{lem:dcpush}, we have that $\vecp = \apr_\graphh(\alpha, \indicatorvec_{v_1}, \vecr)$ throughout the algorithm and so this clearly holds for the output vector.
    By the stopping condition of the algorithm, it holds that $\max_{u \in \vertexset_\graphh}\frac{\vecr(u)}{\deg(u)} < \epsilon$.
   
    Let $T$ be the total number of push operations performed by $\algaprdc$, and   $d_j$ be the degree of the vertex $u_i$ used in the $j$-th push.
    The amount of probability mass on $u_i$ at time $j$ is at least $\epsilon d_j$ and so the amount of probability mass that moves from $\vecr_i$ to $\vecp_i$ at step $j$ is at least $\alpha \epsilon d_i$.
    Since $\norm{\vecr_1}_1 + \norm{\vecr_2}_1 = 1$ at the beginning of the algorithm, we have that 
    $
        \alpha \epsilon \sum_{j = 1}^T d_j  \leq 1,   
    $
    which implies that 
    \begin{align}
        \sum_{j = 1}^T d_j & \leq \frac{1}{\alpha \epsilon}. \label{eq:pushnum}
    \end{align}
    Now, note that for every vertex in $\supp(\vecp)$, there must be at least one push operation on that vertex. So,
    \[
        \vol(\supp(\vecp)) = \sum_{v \in \supp(\vecp)} \deg(v) \leq \sum_{j = 1}^T d_j \leq \frac{1}{\epsilon \alpha}.
    \]
    Finally, to bound the running time, we can implement the algorithm by maintaining a queue of vertices satisfying $\frac{\vecr(u)}{\deg(u)} \geq \epsilon$, and repeat the push operation on the first item in the queue at every iteration, updating the queue as appropriate.
    The push operation and queue updates can be performed in time proportional to $\deg(u)$ and so the running time follows from \eqref{eq:pushnum}.
\end{proof}

\subsection{Analysis of \alglocbipartdc} \label{sec:alganalysis}
We now come to the analysis of the main algorithm and the core technical result of this chapter.
To prove the correctness of the \alglocbipartdc\ algorithm, we   show two complementary facts which we state informally here: 
\begin{enumerate}
    \item If there is a simple set $\sets \subset \vertexset_\graphh$ with low conductance, then for most $u_1 \in \sets$ the simplified approximate Pagerank vector $\vecp = \sigma \circ \apr(\alpha, \indicatorvec_{u_1}, \vecr)$ would have a lot of probability mass on a small set of vertices.
    \item If $\vecp = \sigma \circ \apr(\alpha, \indicatorvec_{u_1}, \vecr)$ contains a lot of probability mass on some small set of vertices, then there is a sweep set of $\vecp$ with low conductance.
\end{enumerate}
\vspace{-0.2cm}
As we see in Section~\ref{sec:reduction}, there is a direct correspondence between almost-bipartite sets in $\graphg$, and low-conductance and  simple sets in $\graphh$.
This means that the two facts above are exactly what we need to prove that
Algorithm~\ref{alg:local_max_cut}
can find densely connected sets in $\graphg$.
  
We   first show in Lemma~\ref{lem:sapr_escapingmass} how the $\sigma$-operator affects some standard mixing properties of Pagerank vectors in order to establish the first fact promised above.
This lemma relies on the fact that $\sets \subset \vertexset_\graphg$ corresponds to a \emph{simple} set in $\vertexset_\graphh$.
This allows us to apply the $\sigma$-operator to the approximate Pagerank vector $\apr_\graphh(\alpha, \indicatorvec_{u_1}, \vecr)$ while preserving a large probability mass on the target set.
We   make use of the following fact proved by Andersen et al.\ \cite{andersenLocalGraphPartitioning2006} which shows that the approximate \pagerank\ vector has a large probability mass on a set of vertices if they have low conductance.
\begin{proposition}[\cite{andersenLocalGraphPartitioning2006}, Theorem~4] \label{lem:apr_escaping_mass}
    For any vertex set $\sets$ and any constant $\alpha \in [0, 1]$, there is a subset $\sets_{\alpha} \subseteq \sets$ with $\vol(\sets_\alpha) \geq \vol(\sets)/2$ such that for any vertex $v \in \sets_\alpha$, the approximate \pagerank\ vector $\apr(\alpha, \indicatorvec_v, r)$ satisfies
    \[
        \apr(\alpha, \indicatorvec_v, \vecr)(\sets) \geq 1 - \frac{\cond(\sets)}{\alpha} - \vol(\sets) \max_{u \in \vertexset} \frac{\vecr(u)}{\deg(u)}.   
    \]
\end{proposition}
The following lemma generalises this result to show that applying the $\sigma$-operator to the approximate \pagerank\ vector maintains this large probability mass up to a constant factor.
\begin{lemma} \label{lem:sapr_escapingmass}
 For any set $\sets \subset \vertexset_{\graphg}$ with partitioning $(\setl, \setr)$ and any constant  $\alpha \in [0, 1]$, there is a subset $\sets_{\alpha} \subseteq \sets$ with $\vol(\sets_{\alpha}) \geq \vol(\sets) / 2$ such that, for any vertex $v \in \sets_{\alpha}$, the simplified approximate \pagerank\ on the double cover $\vecp = \simplify{\apr_\graphh(\alpha, \indicatorvec_{v_1}, \vecr)}$ satisfies
    \[
    	\vecp(\setl_1 \union \setr_2) \geq 1 - \frac{2 \bipart(\setl, \setr)}{\alpha} - 2 \vol(\sets) \max_{u \in \vertexset} \frac{\vecr(u)}{\deg(u)}.
    \]
\end{lemma}
\begin{proof}
By Lemma \ref{lem:cond_bip} and Proposition~\ref{lem:apr_escaping_mass}, there is a set $\sets_\alpha \subseteq \sete$ with $\vol(\sets_\alpha)\geq \vol(\sets)/2$ such that it holds for $v \in \sets_\alpha$ that  
    \[
    	\apr_\graphh(\alpha, \indicatorvec_v, \vecr)(\setl_1 \cup \setr_2) \geq 1 - \frac{\bipart(\setl, \setr)}{\alpha} - \vol(\sets) \max_{u \in \vertexset} \frac{\vecr(u)}{\deg(u)}.
    \]
    Therefore, by the definition of $\vecp$ we have that
    \begin{align*}
    	\vecp(\setl_1 \cup \setr_2) & \geq \apr(\alpha, \indicatorvec_v, \vecr)(\setl_1 \cup \setr_2) - \apr(\alpha, \indicatorvec_v, \vecr)(\setl_2 \cup \setr_1) \\
    	&\geq \apr(\alpha, \indicatorvec_v, \vecr)(\setl_1 \cup \setr_2) - \apr(\alpha, \indicatorvec_v, \vecr)(\overline{\setl_1 \cup \setr_2}) \\
    	& \geq \apr(\alpha, \indicatorvec_v, \vecr)(\setl_1 \union \setr_2) - (1 - \apr(\alpha, \indicatorvec_v, \vecr)(\setl_1 \union \setr_2)) \\
    	& = 2\cdot\apr(\alpha,\indicatorvec_v, \vecr) (\setl_1\cup \setr_2) - 1\\
    	& \geq 1   - \frac{2 \bipart(\setl, \setr)}{\alpha} - 2 \vol(\sets) \max_{u \in \vertexset} \frac{\vecr(u)}{\deg(u)},
    \end{align*}
    which proves the statement of the lemma.
\end{proof}

 To prove the second fact promised at the beginning of this subsection, we show as an intermediate lemma that the value of $\vecp(u_1)$ can be bounded with respect to its value after taking a step of the random walk: $(\lazywalkm\vecp)(u_1)$. 
In order to prove this, we need to carefully analyse the effect of the $\sigma$-operator on the approximate \pagerank\ vector, and in particular examine the effect of the residual vector $\vecr$.
  
\begin{lemma} \label{lem:sapr_updatestep} Let $\graphg$ be a graph with double cover $\graphh$, and $\apr(\alpha, \vecs, \vecr)$ be the approximate \pagerank\ vector defined with respect to $\graphh$. Then, $\vecp = \simplify{\apr(\alpha, \vecs, \vecr)}$ satisfies that
$ \vecp(u_1)  \leq \alpha \left(\vecs(u_1) + \vecr(u_2)\right) + (1 - \alpha)(\lazywalkm\vecp)(u_1)$, and $ \vecp(u_2)  \leq \alpha \left(\vecs(u_2) + \vecr(u_1)\right) + (1 - \alpha)(\lazywalkm\vecp)(u_2)$
for any $u\in \vertexset_\graphg$.
\end{lemma}
\begin{proof}
Let   $u \in \vertexset_\graphg$ be an arbitrary vertex. We have that 
    \begin{align*}
       \vecp(u_1) & = \sigma \circ \big( \ppr(\alpha, \vecs) - \ppr(\alpha, \vecr)\big)(u_1) \\
        & = \sigma \circ \big(\alpha \vecs + (1 - \alpha)\lazywalkm \ppr(\alpha, \vecs) - \alpha \vecr - (1 - \alpha)\lazywalkm\ppr(\alpha, \vecr) \big)(u_1) \\
        & = \sigma \circ \big(\alpha \vecs + (1 - \alpha)\lazywalkm\apr(\alpha, \vecs, \vecr) - \alpha \vecr\big)(u_1) \\
        & = \max(0, (\alpha \vecs + (1 - \alpha)\lazywalkm\apr(\alpha, \vecs, \vecr) )(u_1) \\
        & \quad \quad {} - (\alpha \vecs + (1 - \alpha)\lazywalkm\apr(\alpha, \vecs, \vecr) )(u_2) + \alpha \vecr(u_2) - \alpha \vecr(u_1)) \\
        & \leq \sigma \circ \big(\alpha \vecs + (1 - \alpha)\lazywalkm\apr(\alpha, \vecs, \vecr) \big)(u_1) + \max(0, \alpha \vecr(u_2) - \alpha \vecr(u_1)) \\
        & \leq \sigma \circ \big(\alpha \vecs + (1 - \alpha)\lazywalkm\apr(\alpha, \vecs, \vecr) \big)(u_1) + \alpha \vecr(u_2) \\
        & \leq \sigma \circ \big(\alpha \vecs \big)(u_1) + (1 - \alpha)\sigma \circ \big( \lazywalkm\apr(\alpha, \vecs, \vecr) \big)(u_1) + \alpha \vecr(u_2) \\
        & \leq \alpha \vecs(u_1) + \alpha \vecr(u_2) + (1 - \alpha) (\lazywalkm\vecp) (u_1),
    \end{align*}
    where the first equality follows by the definition of the approximate \pagerank, the second equality follows by the definition of the \pagerank, and the last two inequalities follow by Lemma~\ref{lem:sigmaproperty}. This proves the first statement, and the second statement follows by symmetry.
\end{proof}
 Notice that applying the $\sigma$-operator for any vertex $u_1$ introduces a new dependency on the value of the residual vector $\vecr$ at $u_2$.
This subtle observation demonstrates the additional complexity introduced by the $\sigma$-operator when compared with previous analysis of \pagerank-based local algorithms~\cite{andersenLocalGraphPartitioning2006}.
 Taking account of the $\sigma$-operator, let us further analyse the  Lov\'asz-Simonovits curve defined by $\vecp$, which is a common technique in the analysis of random  walks on graphs~\cite{lovaszMixingRateMarkov1990}.
 Our goal is to show that if there is a set $\sets$ with a large value of $\vecp(\sets)$, then there must be a sweep set $\pjsweep$ with low conductance.
We   first bound the curve at some point by the conductance of the corresponding sweep set.
\begin{lemma} \label{lem:ls1}
Let $\vecp = \simplify{\apr(\alpha, \vecs, \vecr)}$. It holds for any $j \in [n-1]$ that 
    \begin{align*}
        \lscurve{\vol(\pjsweep)} \leq & \alpha \Big(\lscurve[\vecs]{\vol(\pjsweep)} + \lscurve[\vecr]{\vol(\pjsweep)} \Big) \\
        & {} + (1 - \alpha) \frac{1}{2}\Big( \lscurve{\vol(\pjsweep) - \abs{\partial(\pjsweep)}}\Big) \\
        & {} + (1 - \alpha)\frac{1}{2}\Big( \lscurve{\vol(\pjsweep) + \abs{\partial(\pjsweep)}}\Big).
    \end{align*}
\end{lemma}
Since $\vol(\pjsweep) \pm \abs{\partial(\pjsweep)} = \vol(\pjsweep)(1 \pm \cond(\pjsweep))$, this tells us that when the conductance of the sweep set $\pjsweep$ is high, the Lov\'asz-Simonovits curve around the point $\vol(\pjsweep)$ is close to a straight line.
For the proof of Lemma~\ref{lem:ls1}, we view the undirected graph $\graphh$ instead as a directed graph, and each edge $\{u, v\} \in \edgeset_\graphh$ as a pair of directed edges $(u, v)$ and $(v, u)$. For any edge $(u, v)$ and vector $\vecp \in \R_{\geq 0}^{2n}$, let
\[
    p(u, v) = \frac{\vecp(u)}{\deg(u)}
\]
and for a set of directed edges $\seta$, let
\[
    p(\seta) = \sum_{(u, v) \in \seta} p(u, v).
\]
Now for any set of vertices $\sets \subset \vertexset_\graphh$, let $\In(\sets) = \{(u, v) \in \edgeset_\graphh : v \in \sets\}$ and $\Out(\sets) = \{(u, v) \in \vertexset_\graphh : u \in \sets\}$.
Additionally, for any simple set $\sets \subset \vertexset_\graphh$, we define the \firstdef{simple complement} of the set $\sets$ by
\[
    \widehat{\sets} = \{u_2 : u_1 \in \sets\} \union \{u_1 : u_2 \in \sets\}.
\]
We   need the following property of the Lov\'asz-Simonovitz curve with regard to sets of edges.
\begin{lemma} \label{lem:lsedges}
    For any set of edges $\seta$, it holds that
   $$
        p(\seta) \leq \lscurve{\cardinality{\seta}}.$$
\end{lemma}
\begin{proof}
    For each $u \in \vertexset$, let $x_u = \cardinality{\{(u, v) \in \seta\}}$ be the number of edges in $\seta$ with the tail at $u$.
    Then, we have by definition  that 
    \begin{align*}
        p(\seta) & = \sum_{(u, v) \in \seta} \frac{\vecp(u)}{\deg(u)}  = \sum_{u \in \vertexset} \frac{x_u}{\deg(u)} \vecp(u).
    \end{align*}
    Since $x_u/\deg(u) \in [0, 1]$ and
    \[
        \sum_{u \in \vertexset} \frac{x_u}{\deg(u)} \deg(u) = \sum_{u \in \vertexset} x_u = \cardinality{\seta},
    \]
      the statement follows from   (\ref{eq:deflscurve}).
\end{proof}
We   also make use of the following lemma from \cite{andersenLocalGraphPartitioning2006}.
\begin{lemma}[\cite{andersenLocalGraphPartitioning2006}, Lemma~4]\label{lem:andersenlem4}
    For any distribution $\vecp$ and any set of vertices $\seta$,
    \[
        (\lazywalkm\vecp)(\seta) \leq \frac{1}{2}\cdot 
        \left(p(\In(\seta) \union \Out(\seta)) + p(\In(\seta) \intersect \Out(\seta))\right).
    \]
\end{lemma}
These two facts give us everything we need to prove the behaviour of the Lov\'asz-Simonovits curve described in Lemma~\ref{lem:ls1}.
\begin{proof}[Proof of Lemma~\ref{lem:ls1}]
By definition, we have that 
\begin{align*}
    \lefteqn{\vecp(\sets)} \\
    & = \sum_{u_1\in  \sets} \vecp(u_1) + \sum_{u_2\in \sets} \vecp (u_2) \\
    & \leq\sum_{u_1\in \sets}\Big( \alpha \left(\vecs(u_1) + \vecr(u_2)\right) + (1 - \alpha)(\lazywalkm\vecp)(u_1)\Big) \\
    & \quad \quad {} + \sum_{u_2\in \sets} \Big(\alpha \left(\vecs(u_2) + \vecr(u_1)\right) + (1 - \alpha)(\lazywalkm\vecp)(u_2) \Big) \\
    & = \alpha\cdot \Bigg(\sum_{u_1\in \sets} \vecs(u_1)+ \sum_{u_2\in \sets} \vecs(u_2)  +  \sum_{u_1\in \sets} \vecr(u_2)+ \sum_{u_2\in \sets} \vecr(u_1) \Bigg) \\
    & \quad \quad {} + (1-\alpha)\cdot\Bigg( \sum_{u_1\in \sets} (\lazywalkm\vecp)(u_1) + \sum_{u_2\in \sets} (\lazywalkm\vecp)(u_2) \Bigg)\\
    & = \alpha\cdot \Big( \vecs(\sets) + \vecr\left(\widehat{\sets}\right)\!\Big) + (1-\alpha)\cdot (\lazywalkm\vecp)(\sets)\\
    & \leq \alpha\cdot \Big( \vecs(\sets) + \vecr\left(\widehat{\sets}\right)\!\Big) + (1-\alpha)\cdot \frac{1}{2} \cdot \Big(p(\In(S) \union \Out(S)) + p(\In(S) \intersect\Out(S))\Big),
\end{align*}
where the first inequality follows by Lemma~\ref{lem:sapr_updatestep} and the last one follows by  Lemma~\ref{lem:andersenlem4}.

 Now, recall that $\lscurve{\vol(\pjsweep)} = \vecp(\pjsweep)$ and notice that $\vol\left(\widehat{\pjsweep}\right) = \vol\left(\pjsweep\right)$ holds by the definition of $\widehat{\sets}$. Hence, using Lemma~\ref{lem:lsedges}, we have that
 \begin{align*}
        \lefteqn{ \lscurve{\vol(\pjsweep)}}\\
       & =  \vecp(\pjsweep) \\
        & \leq \alpha\left( \vecs(\pjsweep) + \vecr\left(\widehat{\pjsweep} \right)
         \right)  + (1 - \alpha)\cdot \frac{1}{2}\Big( p\left( \In(\pjsweep) \union \Out(\pjsweep)\right) + p(\In(\pjsweep) \intersect \Out(\pjsweep))\Big) \\
        & \leq  \alpha\Big(\lscurve[\vecs]{\vol(\pjsweep)} + \lscurve[\vecr]{\vol(\pjsweep)}\Big) \\
        & \qquad {} + (1 - \alpha)\cdot \frac{1}{2}\left( \lscurve{\cardinality{\In(\pjsweep) \union \Out(\pjsweep)}} + \lscurve{\cardinality{\In(\pjsweep) \intersect \Out(\pjsweep)}} \right).
    \end{align*}
    It remains to show that
    \[
        \cardinality{\In(\pjsweep) \intersect \Out(\pjsweep)} = \vol(\pjsweep) - \cardinality{\partial(\pjsweep)},
    \]
    and
    \[
        \cardinality{\In(\pjsweep) \union \Out(\pjsweep)} = \vol(\pjsweep) + \cardinality{\partial(\pjsweep)}.
    \]
    The first follows by noticing that $\cardinality{\In(\pjsweep) \intersect \Out(\pjsweep)}$ counts precisely twice the number of edges with both endpoints inside $\pjsweep$.
    The second follows since
    \[
        \cardinality{\In(\pjsweep) \union \Out(\pjsweep)} = \cardinality{\In(\pjsweep) \intersect \Out(\pjsweep)} + \cardinality{\In(\pjsweep) \setminus \Out(\pjsweep)} + \cardinality{\Out(\pjsweep) \setminus \In(\pjsweep)}
    \]
    and both the second and third terms on the right hand side are equal to $\cardinality{\partial(\pjsweep)}$.
\end{proof}
We now show another technical lemma which bounds the \LS\ curve with respect to the conductance of the sweep sets $\pjsweep$.
To understand this lemma,
 recall from the discussion in Section~\ref{sec:related:local} that the \lovasz-Simonovits curve of the stationary distribution is linear: $\statdist[k] = k / \vol(\vertexset_\graphh)$.
This lemma bounds how far the Lov\'asz-Simonovits curve deviates from
 that of the stationary distribution.
In particular, if there is no sweep set with low conductance, then the deviation must be small.
\begin{lemma} \label{lem:lscurve}
     Let $\vecp = \simplify{\apr_\graphh(\alpha, \vecs, \vecr)}$ such that $\max_{u \in \vertexset_\graphh} \frac{\vecr(u)}{\deg(u)} \leq \epsilon$, and  $\phi$ be any constant in $[0, 1]$. If $\cond(\pjsweep) \geq \phi$ for all $j \in [\cardinality{\supp(\vecp)}]$, then
    \[ 
        \vecp[k] - \frac{k}{\vol(\vertexset_\graphh)} \leq \alpha t + \alpha \epsilon k t + \sqrt{\min(k, \vol(\vertexset_\graphh) - k)}\left(1 - \frac{\phi^2}{8}\right)^t
    \]
    for all $k \in [0, \vol(\vertexset_\graphh)]$ and integer $t \geq 0$.
\end{lemma}
\begin{proof} 
For ease of discussion we write  $k_j = \vol(\pjsweep)$ and let $\overline{k_j} = \min(k_j, \vol(\vertexset_\graphh) - k_j)$. For convenience we write
    \[
        f_t(k) = \alpha t + \alpha \epsilon k t + \sqrt{\min(k, \vol(\vertexset_\graphh) - k)}\left(1 - \frac{\phi^2}{8}\right)^t.
    \]
    We prove by induction that for all $t \geq 0$,
    \[
        \vecp[k] - \frac{k}{\vol(\vertexset_\graphh)} \leq f_t(k).
    \]
    For the base case, this   holds for $t = 0$ for any $\phi$. Our proof is by case distinction:
    (1) for any integer $k \in [\vol(\vertexset_\graphh) - 1]$,
    \[
        \vecp[k] - \frac{k}{\vol(\vertexset_\graphh)}\leq 1 \leq \sqrt{\min(k, \vol(\vertexset_\graphh) - k)} \leq f_0(k);
    \]
    (2) 
    for $k = 0$ or $k = \vol(\vertexset_\graphh)$, $\vecp[k] - \frac{k}{\vol(\vertexset_\graphh)} = 0 \leq f_0(k)$. The base case follows, since $f_0$ is concave and $\vecp[k]$ is linear between integer values.
   
    Now for the inductive case, we assume that the claim holds for $t$ and  show that it holds for $t+1$.
    It suffices to show that the following holds for every $j \in [\cardinality{\supp(\vecp)}]$:
    \[
        \vecp[k_j] - \frac{k_j}{\vol(\vertexset_\graphh)} \leq f_{t+1}(k_j).
    \]
    The lemma   follows because this inequality holds trivially for $j = 0$ and $j = n$ and $\vecp[k]$ is linear between $k_j$ for $j \in \{0, n\} \union [1, \cardinality{\supp(\vecp)}]$.
   
   For any $j \in [1, \cardinality{\supp(\vecp)}]$, it holds by Lemma~\ref{lem:ls1} that 
   \begin{align*}
        \lscurve{k_j} & \leq \alpha \left(\lscurve[\vecs]{k_j} + \lscurve[\vecr]{k_j} \right) + (1 - \alpha) \frac{1}{2}\left( \lscurve{k_j - \abs{\partial(\pjsweep)}} + \lscurve{k_j + \abs{\partial(\pjsweep)}} \right) \\
        & \leq \alpha + \alpha \epsilon k_j + \frac{1}{2}\left(\lscurve{k_j - \cond(\pjsweep)\cdot\overline{k_j}} + \lscurve{k_j + \cond(\pjsweep)\cdot \overline{k_j}} \right) \\
        & \leq \alpha + \alpha \epsilon k_j + \frac{1}{2}\left( \lscurve{k_j - \phi \overline{k_j}} + \lscurve{k_j + \phi \overline{k_j}}\right),
    \end{align*}
    where the second inequality uses the fact that $\lscurve[\vecr]{k_j} \leq \epsilon k_j$ and the last inequality uses the concavity of $\lscurve{k}$ and the fact that $\cond(\pjsweep) \geq \phi$. By the induction hypothesis, we have that
    \begin{align*}
        \lefteqn{\lscurve{k_j}}\\
        & \leq \alpha + \alpha \epsilon k_j + \frac{1}{2}\left( f_t(k_j - \phi \overline{k_j}) + \frac{k_j - \phi \overline{k_j}}{\vol(\vertexset_\graphh)} + f_t(k_j + \phi \overline{k_j}) + \frac{k_j + \phi \overline{k_j}}{\vol(\vertexset_\graphh)} \right) \\
          & = \alpha + \alpha \epsilon k_j + \frac{k_j}{\vol(\vertexset_\graphh)} + \frac{1}{2}\left( f_t(k_j + \phi \overline{k_j}) + f_t(k_j - \phi \overline{k_j}) \right) \\
            &  = \alpha + \alpha \epsilon k_j + \alpha t + \alpha \epsilon k_j t +\frac{k_j}{\vol(\vertexset_\graphh)} + \frac{1}{2}\Bigg( \sqrt{\min(k_j - \phi \overline{k_j}, \vol(\vertexset_\graphh) - k_j + \phi \overline{k_j})} \\
            & \qquad {} + \sqrt{\min(k_j + \phi \overline{k_j}, \vol(\vertexset_\graphh) - k_j - \phi \overline{k_j})} \Bigg) \left(1 - \frac{\phi^2}{8}\right)^t \\
        & \leq   \alpha(t + 1) + \alpha \epsilon k_j (t + 1) + \frac{k_j}{\vol(\vertexset_\graphh)}+ \frac{1}{2}\left(\sqrt{\overline{k_j} - \phi \overline{k_j}} + \sqrt{\overline{k_j} + \phi \overline{k_j}}\right) \left(1 - \frac{\phi^2}{8}\right)^t
    \end{align*}
    Now, using the Taylor series of $\sqrt{1 + \phi}$ at $\phi = 0$, we see that for any $x \geq 0$ and $\phi \in [0, 1]$,
    \begin{align*}
        \frac{1}{2}\left( \sqrt{x - \phi x} + \sqrt{x + \phi x}\right) & \leq \frac{\sqrt{x}}{2}\left( \left(1 - \frac{\phi}{2} - \frac{\phi^2}{8} - \frac{\phi^3}{16}\right) + \left(1 + \frac{\phi}{2} - \frac{\phi^2}{8} + \frac{\phi^3}{16}\right)  \right) \\
        & \leq \sqrt{x}\left(1 - \frac{\phi^2}{8}\right).
    \end{align*}
    Applying this with $x = \overline{k_j}$, we have
    \begin{align*}
        \lscurve{k_j} - \frac{k_j}{\vol(\vertexset_\graphh)} & \leq   \alpha (t + 1) + \alpha \epsilon k_j (t+1) + \sqrt{\overline{k_j}}\left(1 - \frac{\phi^2}{8}\right) \left(1 - \frac{\phi^2}{8}\right)^t \\
        & = f_{t+1}(k_j),
    \end{align*}
    which completes the proof.
\end{proof}
We can now use this bound to show the other direction: if the \LS\ curve is particularly large at some point, then there must be a sweep set with low conductance.
\begin{lemma} \label{lem:probimpliescond}
 Let $\graphg$ be a graph with double cover $\graphh$, and let $\vecp = \simplify{\apr_\graphh(\alpha, \vecs, \vecr)}$ such that  $\max_{u \in \vertexset_\graphh} \frac{\vecr(u)}{\deg(u)} \leq \epsilon$.
If there is a set of vertices $\sets \subset \vertexset_\graphh$ and a constant $\delta$ such that
$
    \vecp(\sets) - \frac{\vol(\sets)}{\vol(\vertexset_\graphh)} \geq \delta$,
then there is some $j \in [1, \abs{\supp(\vecp)}]$ such that
\[
    \cond_\graphh(\pjsweep) < 6 \sqrt{ (1 + \epsilon \vol(\sets)) \alpha \ln \left(\frac{4}{\delta}\right)\big/\delta}.\]
\end{lemma}
\begin{proof}
Let $\graphh_c$ be a weighted version of the double cover with the weight of every edge equal to some constant $c$.
    Let $k = \vol_{\graphh_c}(\sets)$ and $\overline{k} = \min(k, \vol(\vertexset_{\graphh_c}) - k)$.
    Notice that the behaviour of random walks and the conductance of vertex sets are unchanged by re-weighting the graph by a constant factor.
    As such, $\vecp = \sigma \circ \apr_{\graphh_c}(\alpha, \vecs, \vecr)$ and we let $\phi = \min_j \cond_{\graphh_c}(\pjsweep)$. Hence, by Lemma~\ref{lem:lscurve} it holds for any $t>0$ that
    \begin{align*}
        \delta & < \lscurve{k} - \frac{k}{\vol(\vertexset_{\graphh_c})}  < \alpha t (1 + \epsilon k) + \sqrt{\overline{k}} \left(1 - \frac{\phi^2}{8}\right)^t.
    \end{align*}
    Letting $v = \vol(\vertexset_{\graphh_c})$ we set
    $t = \left\lceil \frac{8}{\phi^2} \ln \frac{2 \sqrt{v/2}}{\delta} \right\rceil$.
    Now, assuming that
    \begin{equation} \label{eq:deltabound}
        \delta \geq \frac{2}{\sqrt{v / 2}}
    \end{equation}
    we have
\[
        \delta < (1 + \epsilon k) \alpha \left\lceil \frac{8}{\phi^2} \ln \frac{2 \sqrt{v/2}}{\delta}\right\rceil + \frac{\delta}{2 \sqrt{v/2}} \sqrt{\overline{k}} 
        < \frac{9 (1 + \epsilon k) \alpha \ln(v/2)}{\phi^2} + \frac{\delta}{2}.
        \]
    We set $c$ such that
    \[
        v = c\ \vol(\vertexset_\graphh) = \vol(\vertexset_{\graphh_c}) = \left(\frac{4}{\delta}\right)^2,
    \]
    which satisfies \eqref{eq:deltabound} and implies that \[
        \phi< \sqrt{\frac{36 (1 + \epsilon k) \alpha \ln(4 / \delta)}{\delta}}.
        \]
By the definition of $\phi$, there must be some $\pjsweep$ with the desired property.
\end{proof}

We have now shown the two facts promised at the beginning of this subsection.
Putting these together,
if there is a simple set $\sets \subset \vertexset_\graphh$ with low conductance, then we can find a sweep set of $\sigma \circ \apr(\alpha, \vecs, \vecr)$ with low conductance.
By the reduction from almost-bipartite sets in $\graphg$ to low-conductance simple sets in $\graphh$, our target set corresponds to a simple set $\sets \subset \vertexset_\graphh$ which leads to Algorithm~\ref{alg:local_max_cut} for finding almost-bipartite sets. Our result is summarised as follows.
\begin{theorem} \label{thm:main_thm}
 Let $\graphg$ be an $n$-vertex undirected graph, and $\setl, \setr \subset \vertexset_\graphg$ be disjoint sets such that $\bipart(\setl, \setr) \leq \beta$ and $\vol(\lur) \leq \gamma$. Then, there is a set $\setc \subseteq \lur$ with $\vol(C) \geq \vol(\lur)/2$ such that, for any $v\in \setc$,  $\alglocbipartdc\left(\graphg, v, \gamma, \sqrt{7560 \beta}\right)$ returns $(\setl', \setr')$ with $\bipart(\setl', \setr') = \bigo{\sqrt{\beta}}$ and  ${\vol(\setl' \union \setr')} = \bigo{\beta^{-1}\gamma}$.
Moreover, the algorithm has running time $\bigo{\beta^{-1}\gamma \log n}$. 
\end{theorem}
\begin{proof}
We set the parameters of Algorithm~\ref{alg:local_max_cut} to be
\begin{itemize}
    \item $\alpha = \widehat{\beta}^2/378 = 20\beta$, and
    \item $\epsilon=1/(20\gamma)$. 
\end{itemize}
Let  $\vecp = \simplify{\apr_\graphh(\alpha, \indicatorvec_v, \vecr)}$ be the simplified approximate Pagerank vector computed in the algorithm, which satisfies that $\max_{v \in \vertexset_\graphh} \frac{\vecr(v)}{\deg(v)} \leq \epsilon$, and $\setc'=\setl_1\cup \setr_2$ be the target set in the double cover. Hence, by  Lemma~\ref{lem:sapr_escapingmass} we have that
\[
    \vecp(\setc')  \geq 1 - \frac{2 \beta}{\alpha} - 2 \epsilon\cdot  \vol(\setc)  \geq 1 - \frac{1}{10} - \frac{1}{10} = \frac{4}{5}.
\]
Notice that, since $\setc'$ is simple, we have that $\frac{\vol(\setc')}{\vol(\vertexset_\graphh)} \leq \frac{1}{2}$ and
\[
    \vecp(\setc') - \frac{\vol(\setc')}{\vol(\vertexset_\graphh)} \geq \frac{4}{5} - \frac{1}{2} = \frac{3}{10}.
\]
Therefore, by Lemma~\ref{lem:probimpliescond} and choosing $\delta=3/10$, there is some $j \in [1, \cardinality{\supp(\vecp)}]$ such that
\begin{align*}
    \cond(\pjsweep) & < \sqrt{\frac{360}{3}(1 + \epsilon \vol(\setc))\alpha \ln\left(\frac{40}{3}\right)} \leq \sqrt{\left(1 + \frac{1}{20}\right) 7200 \beta}   = \widehat \beta.
\end{align*}
Since such a $\pjsweep$ is a simple set, by Lemma~\ref{lem:cond_bip}, we know that the sets $\setl' = \{u \in \vertexset_\graphg : u_1 \in \pjsweep \}$ and $\setr' = \{u \in \vertexset_\graphg : u_2 \in \pjsweep \}$ satisfy $\bipart(\setl', \setr') \leq \widehat \beta$.
For the volume guarantee, Lemma~\ref{lem:aprdc} gives that
\[
    \vol(\supp(\vecp)) \leq \frac{1}{\epsilon \alpha} = \frac{\gamma}{\beta}.
\]
    Since the sweep set procedure in the algorithm is over the support of $\vecp$, we have that
    $$
        \vol(\setl' \union \setr') \leq \vol(\supp(\vecp)) \leq \beta^{-1} \gamma.$$
    Now, we analyse the running time of the algorithm.
    By Lemma~\ref{lem:aprdc}, computing the approximate Pagerank vector $\vecp'$ takes time
    \[
        \bigo{\frac{1}{\epsilon \alpha}} = \bigo{\beta^{-1}\gamma}.
    \]
    Since $\vol(\supp(\vecp')) = \bigo{\frac{1}{\epsilon \alpha}}$, computing $\simplify{\vecp'}$ can also be done in time $\bigo{\beta^{-1}\gamma}$.
    The cost of checking each conductance in the sweep set loop is $\bigo{\vol(\supp(\vecp)) \log(n)} = \bigo{\beta^{-1} \gamma \log(n)}$. This dominates the running time of the algorithm, which completes the proof of the theorem.
\end{proof}

We remark that the quadratic approximation guarantee in Theorem~\ref{thm:main_thm} matches the state-of-the-art local algorithm for finding a single set with low conductance~\cite{andersenAlmostOptimalLocal2016}.
Furthermore,
this result presents a significant improvement over the previous state-of-the-art by Li and Peng~\cite{liDetectingCharacterizingSmall2013}, whose design is based on an entirely different technique.
For any $\epsilon \in [0, 1/2]$, their algorithm has running time $\bigo{\epsilon^2 \beta^{-2}\gamma^{1 + \epsilon}\log^3\gamma}$ and returns a set with volume $\bigo{\gamma^{1 + \epsilon}}$ and bipartiteness $\bigo{\sqrt{\beta / \epsilon}}$.
In particular, their algorithm requires much higher time complexity in order to guarantee the same bipartiteness $O\left(\sqrt{\beta}\right)$.
  
\section{Algorithm for Directed Graphs} \label{sec:directed}
We now turn our attention to local algorithms for finding densely-connected clusters in directed graphs.
In comparison with undirected graphs, we are interested in finding disjoint $\setl, \setr \subset \vertexset$ of an input directed graph $\geqve$ such that most of the edges adjacent to  $\lur$ are \emph{from $\setl$ to $\setr$}. To formalise this,
we define the \emph{flow ratio} from $\setl$ to $\setr$ as
\[
F(\setl,\setr) \triangleq 1  - \frac{2 \weight(\setl, \setr)}{\vol_\Out(\setl) + \vol_\In(\setr)},
\]
where $\weight(\setl, \setr)$ is the number of directed edges from $\setl$ to $\setr$. Notice that  we take not only edge densities but also edge directions into account:
a low $F(\setl, \setr)$-value tells us that almost all edges with their tail in $\setl$ have their head in $\setr$, and conversely almost all edges with their head in $\setr$ have their tail in $\setl$. One could also see $F(\setl, \setr)$ as a generalisation of $\bipart(\setl, \setr)$.
In particular, if we view an undirected graph as a directed graph by replacing each edge with two directed edges, then
$\bipart(\setl, \setr) = F(\setl, \setr)$. In this section, we   develop a local algorithm for finding such vertex sets in a directed graph, and analyse the algorithm's performance.

\subsection{Reduction by Semi-Double Cover}
Given  a directed graph $\graphg=(\vertexset_\graphg, \edgeset_\graphg)$, we construct its \firstdef{semi-double cover} $\graphh=(\vertexset_\graphh, \edgeset_\graphh)$ as follows:
(1) every vertex $v\in \vertexset_\graphg$ has two corresponding vertices $v_1, v_2\in \vertexset_\graphh$; (2) for every edge $(u,v)\in \edgeset_\graphg$, we add the edge $\{u_1, v_2\}$ in $\edgeset_\graphh$, see Figure~\ref{fig:directedconstruct} for illustration.
\begin{figure}[t]
     \centering
 \scalebox{1}{\begin{tikzpicture}
	\begin{pgfonlayer}{nodelayer}
		\node [style=smallCircle] (13) at (0.75, 2) {$a_1$};
		\node [style=smallCircle] (14) at (0.75, -2) {$a_2$};
		\node [style=smallCircle] (15) at (-10.25, 2) {$a$};
		\node [style=smallCircle] (16) at (-6.75, 2) {$b$};
		\node [style=smallCircle] (17) at (-10.25, -2) {$c$};
		\node [style=smallCircle] (18) at (-6.75, -2) {$d$};
		\node [style=none] (19) at (-3, 0) {$\Longrightarrow$};
		\node [style=smallCircle] (20) at (4.25, 2) {$b_1$};
		\node [style=smallCircle] (21) at (4.25, -2) {$b_2$};
		\node [style=smallCircle] (22) at (7.75, 2) {$c_1$};
		\node [style=smallCircle] (23) at (7.75, -2) {$c_2$};
		\node [style=smallCircle] (24) at (11.25, 2) {$d_1$};
		\node [style=smallCircle] (25) at (11.25, -2) {$d_2$};
	\end{pgfonlayer}
	\begin{pgfonlayer}{edgelayer}
		\draw [style=directed blue] (17) to (15);
		\draw [style=directed] (15) to (16);
		\draw [style=directed red] (16) to (18);
		\draw [style=directed green] (17) to (18);
		\draw [style=directed purple] (16) to (17);
		\draw [style=undirected red] (20) to (25);
		\draw [style=undirected green] (22) to (25);
		\draw [style=undirected purple] (20) to (23);
		\draw [style=undirected] (13) to (21);
		\draw [style=undirected blue] (22) to (14);
	\end{pgfonlayer}
\end{tikzpicture}}
     \caption[An example of a directed graph and its semi-double cover]{An example of a directed graph and its semi-double cover}
     \label{fig:directedconstruct}
\end{figure}
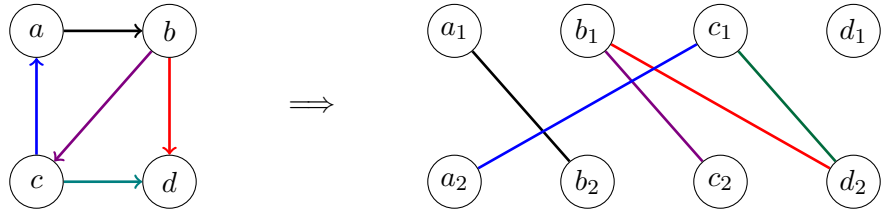

It is worth comparing this reduction with the one for undirected graphs:  
\begin{itemize} 
    \item For undirected graphs, we apply the standard double cover and every undirected edge in $\graphg$ corresponds to two edges in the double cover $\graphh$;
    \item For directed graphs, every directed edge in $\graphg$ corresponds to \emph{one} undirected edge in $\graphh$. This \emph{asymmetry} would allow us to `recover' the direction of any edge in $\graphg$.
\end{itemize}
\vspace{-0.2cm}

We follow the use of $\sets_1, \sets_2$ from Section~\ref{sec:undirected}:  
for any $\sets \subset \vertexset_\graphg$,
we define $\sets_1 \subset \vertexset_\graphh$ and $\sets_2 \subset \vertexset_\graphh$ by $\sets_1\triangleq \{v_1~|~ v\in \sets\}$ and $\sets_2\triangleq \{v_2~|~ v\in \sets\}$. 
The lemma below shows the connection between the value of $F_\graphg(\setl, \setr)$ for any   $\setl, \setr$ and $\cond_\graphh(\setl_1\cup \setr_2)$.
 
\begin{lemma} \label{lem:flow_cond}
Let $\graphg$ be a directed graph with semi-double cover $\graphh$. Then, it holds  for any $\setl, \setr \subset \vertexset_\graphg$
that 
$F_\graphg(\setl, \setr) = \cond_\graphh(\setl_1 \union \setr_2)$.
Similarly, for any simple set $\sets \subset \vertexset_\graphh$, let $\setl = \{u : u_1 \in \sets\}$ and $\setr = \{u : u_2 \in \sets\}$. Then, it holds that $F_\graphg(\setl, \setr) = \cond_\graphh(\sets)$.
\end{lemma}
\begin{proof}
We define  $\sets_\graphh = \setl_1 \union \setr_2$, and have that 
\begin{align*}
\cond_\graphh(\sets_\graphh) & = \frac{\weight(\sets_\graphh, \vertexset_\graphh \setminus \sets_\graphh)}{\vol_\graphh(\sets_\graphh) }\\
& = \frac{\vol_\graphh(\sets_\graphh) - 2\weight_\graphh(\sets_\graphh, \sets_\graphh)}{\vol_\graphh(\sets_\graphh) } \\
& = 1  - \frac{2\weight_\graphh(\sets_\graphh, \sets_\graphh)}{\vol_\graphh(\sets_\graphh)} \\
& =   1 - \frac{2 \weight_\graphg(\setl, \setr)}{\vol_\Out(\setl) + \vol_\In(\setr)} \\
& = F_\graphg(\setl,\setr).
\end{align*}
This proves the first statement of the lemma.  
The second statement of the lemma follows by the same argument.
\end{proof}

\subsection{Design and Analysis of the Algorithm}
The new algorithm is a modification of the algorithm by Andersen and Peres~\cite{andersenFindingSparseCuts2009} described in Section~\ref{sec:esp}.
Given a directed graph $\graphg$ as input, the algorithm simulates the volume-biased ESP on $\graphg$'s semi-double cover $\graphh$.
Notice that the graph $\graphh$ can be constructed locally in the same way as the local construction of the double cover.
However, as the output set $\sets$ of an ESP-based algorithm is not necessarily simple, the algorithm only returns vertices $u\in \vertexset_\graphg$ in which \emph{exactly} one of $u_1$ and $u_2$ are included in $\sets$.
The key procedure for the algorithm is given in Algorithm~\ref{alg:directed_evo_cut}, in which the \alggeneratesample\ procedure is the one described at the end of Section~\ref{sec:esp}.

\begin{algorithm}[ht] \SetAlgoLined
\SetKwInOut{Input}{Input}
\SetKwInOut{Output}{Output}
   \Input{Starting vertex $u$, $i \in \{1, 2\}$,  target flow ratio $\phi$}
   \Output{A pair of sets $\setl, \setr \subset \vertexset_\graphg$}
   Set $T = \lfloor (100 \phi^{\frac{2}{3}})^{-1}\rfloor$ \\
   Compute $\sets = \alggeneratesample_\graphh(u_i, T)$ \\
   Let $\setl = \{u \in \vertexset_\graphg : u_1 \in \sets$ and $u_2 \not \in \sets\}$ \\
   Let $\setr = \{u \in \vertexset_\graphg : u_2 \in \sets$ and $u_1 \not \in \sets\}$  \\
   \Return $\setl$ and $\setr$
   \caption[Find densely connected clusters in directed graphs: \algevocutdirected$(u, i, \phi)$]{\algevocutdirected$(u, i, \phi)$}
   \label{alg:directed_evo_cut}
\end{algorithm}

Notice that in the constructed graph $\graphh$,
$\cond(\setl_1 \union \setr_2) \leq \phi$ does not imply that $\cond(\setl_2 \union \setr_1) \leq \phi$.
Due to this asymmetry, Algorithm~\ref{alg:directed_evo_cut} takes a parameter $i \in \{1, 2\}$ to indicate whether the starting vertex is in $\setl$ or $\setr$.
If it is not known whether $u$ is in $\setl$ or $\setr$,
two copies
can be run in parallel, one with $i = 1$ and the other with $i = 2$.
Once one of them  terminates with the performance guaranteed   in Theorem~\ref{thm:directedresult}, the other
can be terminated.
Hence, we can always assume that it is known whether the starting vertex $u$ is in $\setl$ or $\setr$.

We now analyse the algorithm. Notice that, since  the evolving set process gives us an arbitrary set on the semi-double cover, Algorithm~\ref{alg:directed_evo_cut} converts this into a simple set by 
removing any vertices $u$ where $u_1 \in \sets$ and $u_2 \in \sets$. The following definition allows us to discuss sets which are close to being simple.
\begin{definition}[$\epsilon$-simple set] For any set $\sets \subset \vertexset_\graphh$, let $\setp =\{u_1,u_2: u\in \vertexset_\graphg, u_1\in \sets~\mbox{and}~u_2\in \sets\}$. We call set $\sets$ $\epsilon$-simple if it holds that
   $
        \frac{\vol(\setp)}{\vol(\sets)} \leq \epsilon.
    $
\end{definition}
The notion of $\epsilon$-simple sets measures the ratio of vertices
in which both $u_1$ and $u_2$ are in $\sets$.
In particular, any simple set defined in Definition~\ref{def:simple} is $0$-simple.
We show that, if $\sets \subset \vertexset_\graphh$ is an $\epsilon$-simple set for small $\epsilon$, we can construct a simple set $\sets'$ with low conductance.
 \begin{lemma} \label{lem:lazysimplify}
For  any $\epsilon$-simple set $\sets\subset \vertexset_\graphh$,
let $\setp = \{u_1, u_2 : u \in \vertexset_\graphg, u_1 \in \sets \mbox{ and } u_2 \in \sets\}$
and set $\sets'=\sets\setminus \setp$. 
Then,  
 $
        \cond(\sets') \leq \frac{1}{1 - \epsilon} (\cond(\sets) + \epsilon).
    $
\end{lemma}
\begin{proof} 
By the definition of conductance, we have
    \begin{align*}
        \cond(\sets') & = \frac{\weight(\sets', \vertexset \setminus \sets')}{\vol(\sets')} \\
        & \leq \frac{\weight(\sets, \vertexset \setminus \sets) + \vol(\setp)}{\vol(\sets) - \vol(\setp)} \\
        & \leq \frac{1}{1 - \epsilon} \frac{\weight(\sets, \vertexset \setminus \sets) + \vol(\setp)}{\vol(\sets)} \\
        & \leq \frac{1}{1 - \epsilon} \left( \frac{\weight(\sets, \vertexset \setminus \sets)}{\vol(\sets)} + \epsilon \right) \\
        & = \frac{1}{1 - \epsilon} \left( \cond(\sets) + \epsilon \right),
    \end{align*}
    which proves the lemma.
\end{proof}
Given this, in order to guarantee that $\cond(\sets')$ is low, we need to construct $\sets$ such that $\cond(\sets)$ is low and $\sets$ is $\epsilon$-simple for small $\epsilon$.
The following lemma shows that, by choosing the number of steps of the evolving set process carefully, we can very tightly control the overlap of the output of the evolving set process with the target set $\sets$, which we know to be simple.
    \begin{lemma} \label{lem:esp_tightvol}
        Let $\setc \subset \vertexset_\graphh$ be a set with conductance $\cond(\setc) = \phi$, and let  $T = (100 \phi^{\frac{2}{3}})^{-1}$. There exists a set $\setc_g \subseteq \setc$ with volume at least $\vol(\setc) / 2$ such that for any $x \in \setc_g$, with probability at least $7 / 9$, a sample path $(\sets_0, \sets_1, \ldots, \sets_T)$ from the volume-biased ESP started from $\sets_0 = \{x\}$ satisfies the following:
        \begin{itemize}
            \item $\cond(\sets_t) < 60 \phi^{\frac{1}{3}} \sqrt{\ln(\vol(\vertexset))}$ for some $t \in [0, T]$;
            \item $\vol(\sets_t \intersect \setc) \geq \left(1 - \frac{1}{10}\phi^{\frac{1}{3}}\right) \vol(\sets_t)$ for all $t \in [0, T]$. 
        \end{itemize}
    \end{lemma}
\begin{proof}
    By Proposition~\ref{prop:esp_cond}, with probability at least $(1 - 1/9)$ there is some $t < T$ such that
    \begin{align*}
        \cond(\sets_t) & \leq 3 \sqrt{4 T^{-1} \ln \vol(\vertexset_\graphh)}  = 3 \sqrt{400 \phi^{\frac{2}{3}} \ln \vol(\vertexset_\graphh)}  = 60 \phi^{\frac{1}{3}} \sqrt{\ln (\vol(\vertexset_\graphh))}.
    \end{align*}
    Then, by Proposition~\ref{prop:escape_prob} we have that 
    $
        \mathrm{esc}(x, T, \setc) \leq T \phi = \frac{1}{100} \phi^{1/3}$.
     Hence, by Proposition~\ref{prop:esp_vol}, with probability at least $9/10$,
    \[
    \max_{t\leq T} \frac{\vol(\sets_t\setminus \setc)}{\vol(\sets_t)} \leq 10\cdot \mathrm{esc}(x,T,\setc) \leq \frac{1}{10}\cdot \phi^{1/3}
    \]
    Applying the union bound on these two events proves the statement of the lemma. 
\end{proof}
This allows us to analyse the performance of Algorithm~\ref{alg:directed_evo_cut}, which
is summarised in Theorem~\ref{thm:directedresult}. 
\begin{theorem} \label{thm:directedresult}
Let $\graphg$ be an $n$-vertex directed graph, and $\setl, \setr \subset \vertexset_\graphg$ be disjoint sets such that $F(\setl, \setr) \leq \phi$ and $\vol(\lur)\leq \gamma$. There is a set $\setc \subseteq \lur$ with $\vol(\setc) \geq \vol(\lur)/2$ such that, for any $v\in \setc$ and some $i \in \{1, 2\}$,
$\algevocutdirected(\graphg, v, i, \phi)$ returns $(\setl', \setr')$ such that  $F(\setl', \setr') = \bigo{\phi^{\frac{1}{3}} \log^{\frac{1}{2}} n }$ and  ${\vol(\setl' \union \setr')} =\bigo{\left(1 - \phi^{\frac{1}{3}}\right)^{-1} \gamma}$.
Moreover, the algorithm has running time $\bigo{ \phi^{-\frac{1}{2}}\gamma \log^{\frac{3}{2}} n}$.
\end{theorem}
\begin{proof}
By Lemma~\ref{lem:esp_tightvol}, we know that, with  probability at least $7/9$,  the set $\sets$ computed by Algorithm~\ref{alg:directed_evo_cut} satisfies
    \[
        \vol(\sets \intersect \setc) \geq \left(1 - \frac{1}{10}\phi^{\frac{1}{3}}\right) \vol(\sets).
    \]
    Let $\setp = \{u_1, u_2 : u_1 \in \sets $ and $u_2 \in \sets\}$. Then, since $\setc$ is simple, we have that
    \begin{align*}
        \vol(\setp) & \leq 2\cdot (\vol(\sets) - \vol(\sets \intersect \setc))   \leq 2\cdot \left(1 - \left(1 - \frac{1}{10} \phi^{\frac{1}{3}}\right)\right)\cdot \vol(\sets)  = \frac{1}{5}\cdot  \phi^{\frac{1}{3}} \vol(\sets),
    \end{align*}
    which implies that $\sets$ 
    is $(1/5)\cdot \phi^{\frac{1}{3}}$-simple. 
    Therefore, by letting $\sets' = \sets \setminus \setp$, we have by Lemma~\ref{lem:lazysimplify} that 
    \begin{align*}
        \cond(\sets') & \leq \frac{1}{1 - \epsilon}\cdot  \left(\cond(\sets) + \epsilon\right) 
          \leq \frac{1}{1 - \frac{1}{5}} \left(\cond(\sets) + \frac{1}{5}\phi^{\frac{1}{3}}\right)   \leq \frac{5}{4}\cdot  \cond(\sets) + \frac{1}{4}\cdot  \phi^{\frac{1}{3}}.  
    \end{align*}
    Since $\cond(\sets) = \bigo{\phi^{\frac{1}{3}} \ln^{\frac{1}{2}}(n)}$ by Lemma~\ref{lem:esp_tightvol}, we   have that $\cond(\sets') = \bigo{\phi^{\frac{1}{3}} \ln^{\frac{1}{2}}(n)}$.
    For the volume guarantee, direct calculation shows that 
    \begin{align*}
        \vol(\sets') \leq \vol(\sets)  \leq \frac{\vol(\sets \intersect \setc)}{1 - \frac{1}{10}\phi^{\frac{1}{3}}} \leq \frac{\vol(\setc)}{1 - \phi^{\frac{1}{3}}}.
    \end{align*}
    Finally, the running time of the algorithm is dominated by the \alggeneratesample\ method, which is shown in \cite{andersenFindingSparseCuts2009} to have running time $\bigo{\vol(\setc) \phi^{-\frac{1}{2}} \ln^{\frac{3}{2}}(n)}$.
\end{proof}

Algorithm~\ref{alg:directed_evo_cut} is the first local algorithm for directed graphs that approximates a pair of densely connected clusters,
and demonstrates that finding such a pair appears to be much easier than finding a low-conductance set in a directed graph; in particular, existing local algorithms for finding a low-conductance set require the stationary distribution of the random walk in the directed graph~\cite{andersenLocalPartitioningDirected2007}, the sublinear-time computation of which is unknown~\cite{cohenAlmostlineartimeAlgorithmsMarkov2017}.
However, the knowledge of the stationary distribution is not needed for our algorithm.
 
\section{Experimental Results} \label{sec:localDense:experiments}
In this section we evaluate the performance of the newly proposed algorithms on both synthetic and real-world datasets. 
For undirected graphs, we   compare the performance of \alglocbipartdc\ against the previous state-of-the-art~\cite{liDetectingCharacterizingSmall2013}  through synthetic datasets with various parameters, and apply a real-world dataset to demonstrate the significance of the new algorithm.
For directed graphs, we   compare the performance of \algevocutdirected\ with the state-of-the-art \emph{non-local} algorithm since our presented local algorithm for directed graphs is the first such algorithm in the literature.
All experiments are performed on a Lenovo Yoga 2 Pro with an Intel(R) Core(TM) i7-4510U CPU @ 2.00GHz processor and 8GB of RAM.
The code can be downloaded from
\begin{equation*}
\mbox{\url{https://github.com/pmacg/local-densely-connected-clusters}.}
\end{equation*}

\subsection{Results for Undirected Graphs}
We first compare the performance of \alglocbipartdc\ with the previous state-of-the-art given by Li and Peng \cite{liDetectingCharacterizingSmall2013}, which we can refer to as \alglp, on synthetic graphs generated from the stochastic block model (SBM).
Specifically, we assume that the graph has $k=3$ clusters $\{\setc_j\}_{j=1}^3$, and the number of vertices in each cluster, denoted by $n_1, n_2$ and $n_3$ respectively, satisfy $n_1=n_2=0.1 n_3$. Moreover, any pair of vertices $u\in \setc_i$ and $v\in \setc_j$ is connected with probability $P_{i,j}$.
We assume that $P_{1,1}=P_{2,2}=p_1$, $P_{3,3}=p_2$, $P_{1,2}=q_1$, and $P_{1,3}=P_{2,3}=q_2$.
Throughout the experiments, we maintain the ratios $p_2 = 2 p_1$ and $q_2 = 0.1 p_1$, leaving the parameters $n_1$, $p_1$ and $q_1$ free.
Notice that the different values of $q_1$ and $q_2$ guarantee that $\setc_1$ and $\setc_2$ are the ones optimising the $\beta$-value, which is why the proposed model is slightly more involved than the standard SBM described in Section~\ref{sec:prelim:sbm}.

We   evaluate the quality of the output $(\setl, \setr)$ returned by each algorithm with respect to its $\beta$-value, the Adjusted Rand Index (ARI)~\cite{gatesImpactRandomModels2017}, as well as the ratio of the misclassified vertices defined by $\left(\cardinality{\setl \triangle \setc_1} + \cardinality{\setr \triangle \setc_2}\right)/\left(\cardinality{\setl \union \setc_1} + \cardinality{\setr \union \setc_2}\right)$.
All the reported results are the average performance of each algorithm over $10$ runs, in which a random vertex from $\setc_1 \cup \setc_2$ is chosen as the starting vertex of the algorithm.

\paragraph{Setting the $\epsilon_{LP}$ parameter in the \alglp\ algorithm.}
The \alglp\ algorithm has an additional parameter over \alglocbipartdc, which we refer to as $\epsilon_{LP}$.
This parameter influences the runtime and performance of the algorithm and must be in the range $[0, 1/2]$.
In order to choose a fair value of $\epsilon_{LP}$ for comparison with \alglocbipartdc, we consider several values on graphs with a range of target volumes.

We generate graphs from the SBM such that $p_1 = 1 / n_1$ and $q_1 = 100 p_1$ and vary the size of the target set by varying $n_1$ between $100$ and $10,000$.
For values of $\epsilon_{LP}$ in $\{0.01, 0.02, 0.1, 0.5\}$, Figure~\ref{fig:choosing_eps}(a) shows how the runtime of the \alglp\ algorithm compares to the \alglocbipartdc\ algorithm for a range of target volumes.
In this experiment, the runtime of the \alglocbipartdc\ algorithm lies between the runtimes of the \alglp\ algorithm for $\epsilon_{LP} = 0.01$ and $\epsilon_{LP} = 0.02$.
However, Figure~\ref{fig:choosing_eps}(b) shows that the performance of the \alglp\ algorithm with $\epsilon_{LP} = 0.01$ is significantly worse than the performance of the \alglocbipartdc\ algorithm, and so for a fair comparison we set $\epsilon_{LP} = 0.02$ for the remainder of the experiments.

\begin{figure}[t]
\centering
\begin{subfigure}{0.45\textwidth}
\includegraphics[width=\textwidth]{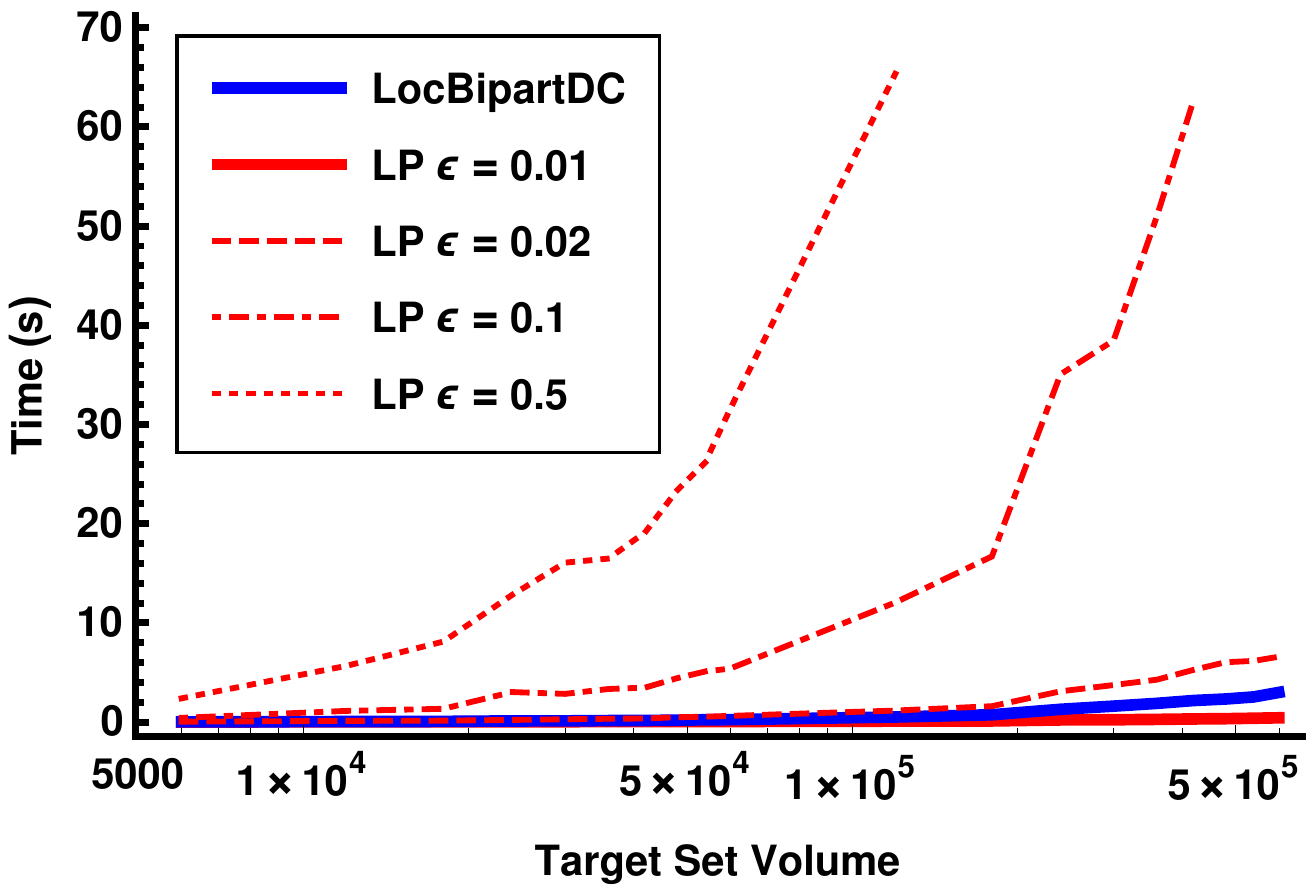}
\caption{}
\end{subfigure}
\hspace{1em}
\begin{subfigure}{0.45\textwidth}
\includegraphics[width=\textwidth]{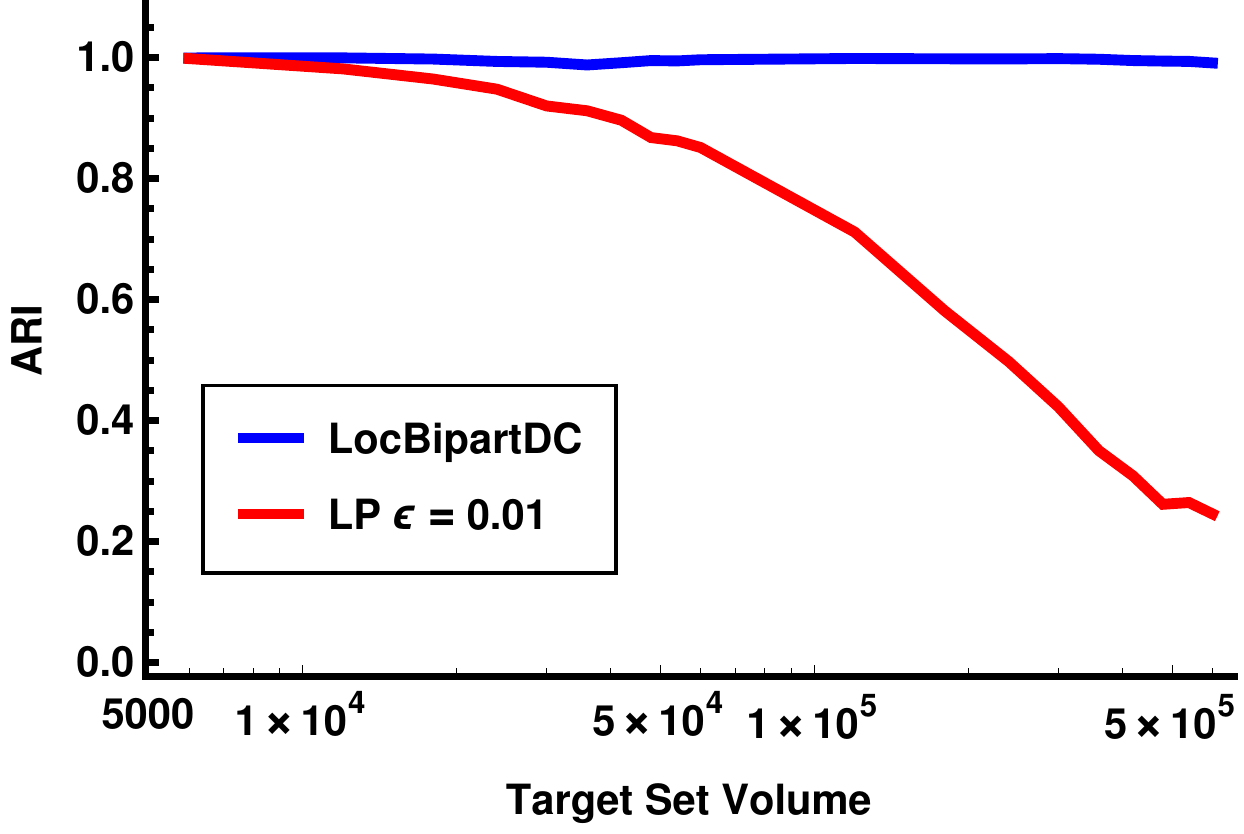}
\caption{}
\end{subfigure}
\caption[Runtime comparison of \alglocbipartdc\ and \alglp]{\textbf{(a)} A comparison of the runtimes of the \alglocbipartdc\ algorithm and the \alglp\ algorithm with various values of $\epsilon_{LP}$.
\textbf{(b)} Setting $\epsilon_{LP} = 0.01$ results in a significantly worse Adjusted Rand Index (ARI) than the \alglocbipartdc\ algorithm.
\label{fig:choosing_eps}}
\end{figure}

\paragraph{Comparison on the synthetic data.}
\begin{table*}[t]
\caption[Comparison between \alglocbipartdc\ and the previous state-of-the art]{\small Full comparison between Algorithm~\ref{alg:local_max_cut}~(\alglbdc) and the previous state-of-the-art~\cite{liDetectingCharacterizingSmall2013} (\alglp). For clarity we report the target bipartiteness $\beta = \bipart(\setc_1, \setc_2)$ and target volume $\gamma = \vol(\setc_1 \union \setc_2)$ along with the SBM parameters.
The final column shows the  ratio between the number of misclassified vertices and the total number of vertices in the target clusters.
}
\label{tab:full_results}
\vskip 0.1in
\centering
\begin{small}
\begin{sc}
\begin{tabular}{cccccc}
\toprule
Input graph parameters & Algo. & Runtime & $\beta$-value & ARI & Misclass.\\
\midrule
$n_1 = 10^3$, $p_1 = 10^{-3}$, $q_1 = 0.018$ & \alglbdc & \textbf{0.09} & \textbf{0.154} & \textbf{0.968} & \textbf{0.073} \\
$\beta \approx 0.1$, $\gamma \approx 4 \cdot 10^4$ & \alglp & 0.146 & 0.202 & 0.909 & 0.138 \\
\midrule
$n_1 = 10^4$, $p_1 = 10^{-4}$, $q_1 = 0.0018$ & \alglbdc & \textbf{0.992} & \textbf{0.215} & \textbf{0.940} & \textbf{0.145} \\
$\beta \approx 0.1$, $\gamma \approx 4 \cdot 10^5$ & \alglp & 1.327 & 0.297 & 0.857 & 0.256 \\
\midrule
$n_1 = 10^5$, $p_1 = 10^{-5}$, $q_1 = 0.00018$ & \alglbdc & \textbf{19.585} & \textbf{0.250} & \textbf{0.950} & \textbf{0.166} \\
$\beta \approx 0.1$, $\gamma \approx 4 \cdot 10^6$ & \alglp & 30.285 & 0.300 & 0.865 & 0.225 \\
\midrule
$n_1 = 10^3$, $p_1 = 4 \cdot 10^{-3}$, $q_1 = 0.012$ & \alglbdc & \textbf{1.249} & \textbf{0.506} & \textbf{0.503} & \textbf{0.763} \\
$\beta \approx 0.4$, $\gamma \approx 4 \cdot 10^4$ & \alglp & 1.329 & 0.597 & 0.445 & 0.785 \\
\bottomrule
\end{tabular}
\end{sc}
\end{small}
\vskip -0.1in
\end{table*}
We now compare the \alglocbipartdc\ and \alglp\ algorithms'
performance on graphs generated from the SBM with different values of $n_1,p_1$ and $q_1$.
As shown in Table~\ref{tab:full_results}, \alglocbipartdc\ not only runs faster, but also produces better clusters with respect to all three metrics.
Secondly, since the clustering task becomes more challenging when the target clusters have higher $\beta$-value,
we compare the algorithms' performance on a sequence of instances with increasing value of $\beta$.
Since $q_1/p_1 = 2(1-\beta)/\beta$, we simply fix the values of  $n_1, p_1$ as $n_1=1,000,p_1=0.001$, and generate graphs with increasing value of $q_1/p_1$; this gives us graphs with monotone values of $\beta$.
As shown in Figure~\ref{fig:ari}(a),
the performance of our algorithm is always better than the previous state-of-the-art.


 Thirdly, notice that both algorithms use some parameters to control the algorithm's runtime and the output's approximation ratio, which are naturally influenced by each other. To study this dependency, we generate graphs according to   $n_1 = 100,000$,  $p_1 = 0.000015$, and $q_1 = 0.00027$ which results in target sets with     $\beta \approx 0.1$ and volume $\gamma \approx 6,000,000$.
 Figure~\ref{fig:ari}(b) shows that, in comparison with the previous state-of-the-art, \alglocbipartdc\ takes much less time to produce output with the same ARI value.
 
\begin{figure}[t]
\centering
    \begin{subfigure}{0.45\textwidth} 
    \includegraphics[width=\textwidth]{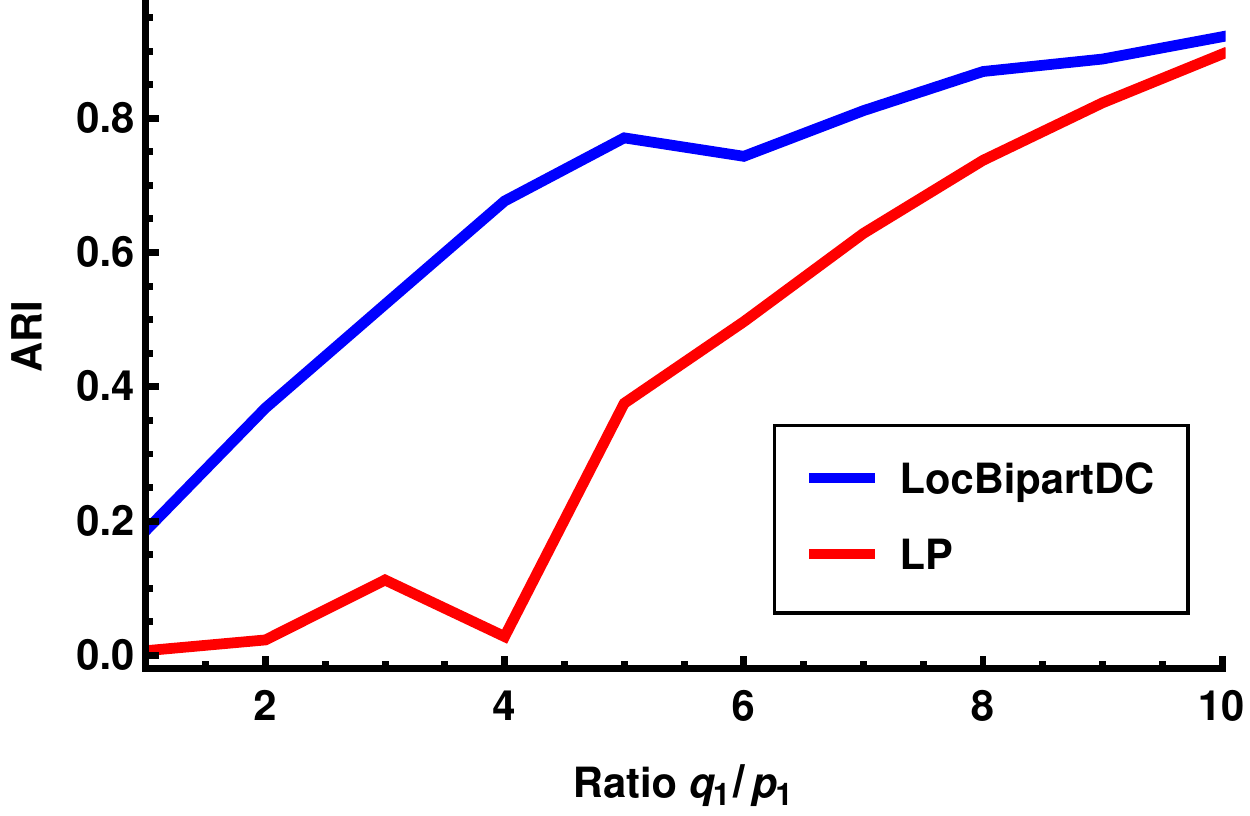}
    \caption{\label{fig:bipvsari}}
    \end{subfigure}
    \hspace{1em}
    \begin{subfigure}{0.45\textwidth} 
    \includegraphics[width=\textwidth]{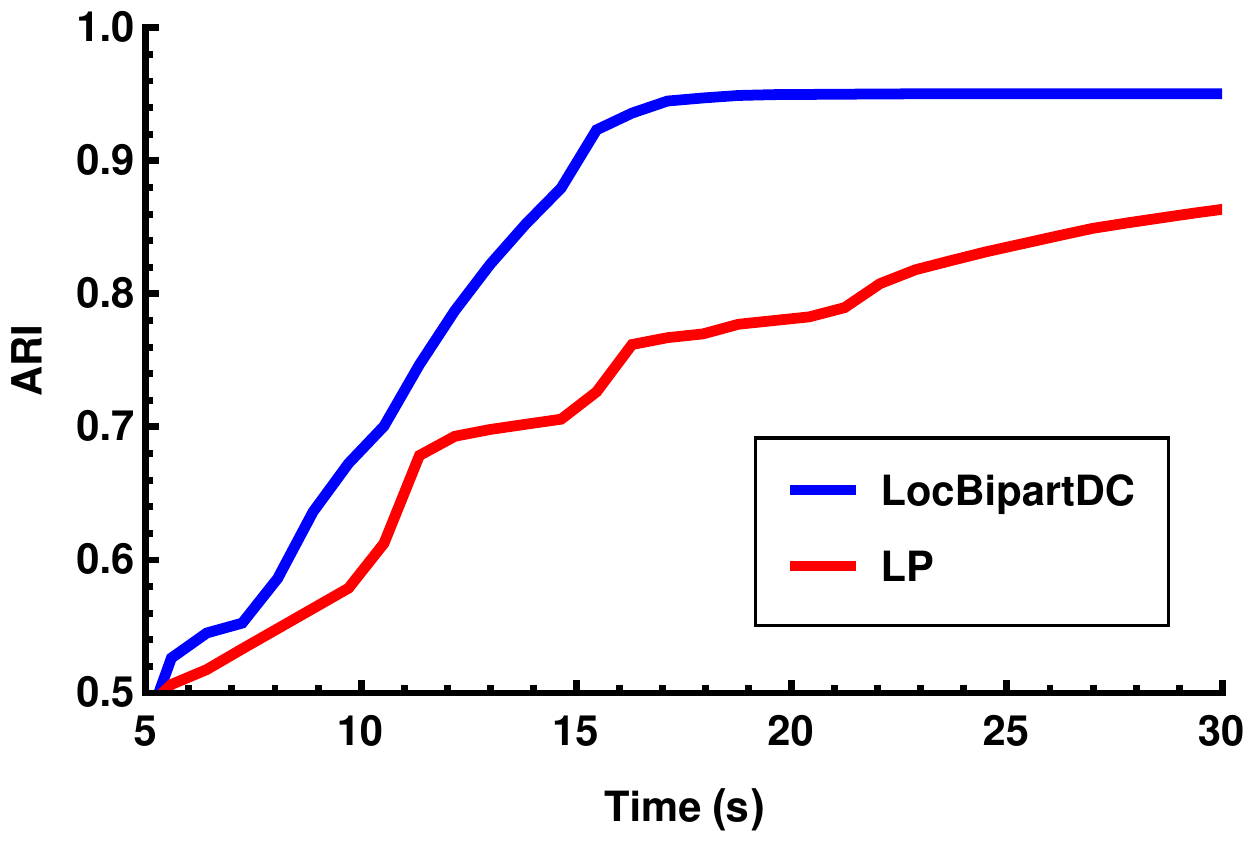}
    \caption{\label{fig:timevsari}}
    \end{subfigure}
    \caption[Experimental results for \alglocbipartdc]{
    \small
         \textbf{(a)} The ARI value of each algorithm when varying the target $\beta$-value. A higher $q_1 / p_1$ ratio corresponds to a lower $\beta$-value.
         \textbf{(b)} The ARI value of each algorithm when bounding the runtime of the algorithm.
    \label{fig:ari}
    }
\end{figure}
 
\paragraph{Experiments on the real-world dataset.} 
We further demonstrate the significance of our algorithm on the Interstate Disputes Dataset~(v3.1)~\cite{maozDyadicMilitarizedInterstate2019}, which
records every 
interstate dispute during 1816--2010, including the level of hostility resulting from the dispute and the number of casualties, and has been widely studied in the social and political sciences~\cite{mansfieldWhyDemocraciesCooperate2002, martinMakeTradeNot2008} as well as the machine learning    community~\cite{huDeepGenerativeModels2017, menonLinkPredictionMatrix2011, traagCommunityDetectionNetworks2009}.
For a given time period, we construct a graph from the data by representing each country with a vertex and adding an edge between each pair of countries weighted according to the severity of any military disputes between those countries.
Specifically, if there's a war\footnote{A war is defined by the maintainers of the dataset as a series of battles resulting  in at least 1,000 deaths.
}
between the two countries, the corresponding two vertices are 
 connected by an edge with weight $30$; for any other dispute that is not part of an interstate war, the two corresponding vertices are connected by an edge with weight $1$.
 We always use the USA as the starting vertex of the algorithm, and the algorithm's output, as visualised in Figure~\ref{fig:intro_undirected}, can be well explained by geopolitics. The 
 $\beta$-values of the pairs of clusters in Figures~\ref{fig:intro_undirected} are $0.361$, $0.356$, $0.170$ and $0.191$ respectively.
 
\subsection{Results for Directed Graphs}
Next we evaluate the performance of our algorithm for directed graphs on synthetic and real-world datasets.
Since there are no previous local directed graph clustering algorithms that achieve similar objectives, we compare the output of Algorithm~\ref{alg:directed_evo_cut} (\algecd) with the state-of-the-art
non-local algorithm proposed by Cucuringu et al.~\cite{cucuringuHermitianMatricesClustering2020}, and we refer this to as \algclsz\ in the following. 
 
\paragraph{Synthetic dataset.}
Let us first look at the \firstdef{cyclic block model}~(CBM) described in  Cucuringu et al.~\cite{cucuringuHermitianMatricesClustering2020} with parameters $k$, $n$, $p$, $q$, and $\eta$.
In this model, we generate a directed graph with $k$ clusters $\setc_1, \ldots, \setc_k$ of size $n$, and for $u, v \in \setc_i$, there is an edge between $u$ and $v$ with probability $p$ and the edge direction is chosen uniformly at random.
For $u \in \setc_i$ and $v \in \setc_{i + 1 \mod k}$, there is an edge between $u$ and $v$ with probability $q$, and the edge is directed from $u$ to $v$ with probability $\eta$ and from $v$ to $u$ with probability $1 - \eta$.
Throughout our experiments, we fix $p = 0.001$, $q = 0.01$, and $\eta = 0.9$.

Secondly, since the goal of our algorithm is to find local structure in a graph, we extend the cyclic block model with additional local clusters and refer to this model as CBM+.
In addition to the parameters of the CBM, we introduce the parameters $n', q'_1, q'_2$, and $\eta'$.
In this model, the clusters $\setc_1$ to $\setc_k$ are generated as in the CBM, and there are two additional clusters $\setc_{k+1}$ and $\setc_{k+2}$ of size $n'$.
For $u, v \in \setc_{k+i}$ for $i \in \{1, 2\}$, there is an edge between $u$ and $v$ with probability $p$
and for $u \in \setc_{k+1}$ and $v \in \setc_{k+2}$, there is an edge with probability $q'_1$;
the edge directions are chosen uniformly at random.
For $u \in \setc_{k+1} \union \setc_{k+2}$ and $v \in \setc_1$, there is an edge with probability $q'_2$.
If $u \in \setc_{k+1}$, the orientation is from $v$ to $u$ with probability $\eta'$ and from $u$ to $v$ with probability $1 - \eta'$ and if $u \in \setc_{k+2}$, the orientation is from $u$ to $v$ with probability $\eta'$ and from $v$ to $u$ with probability $1 - \eta'$.
We always fix $q'_1 = 0.5$, $q'_2 = 0.005$, and $\eta' = 1$.
Notice that the clusters $\setc_{k+1}$ and $\setc_{k+2}$ form a `local' cycle with the cluster $\setc_1$.

Table~\ref{tab:gdsbm} reports the average performance over 10 runs with a variety of parameters.
We can see that \algclsz\ uncovers the global structure in the CBM more accurately than \algecd.
On the other hand, \algclsz\ fails to identify the local cycle in the CBM+ model, while \algecd\ succeeds. 

\begin{table}[t]
\caption[Comparison of \algecd\ with \algclsz\ on synthetic data.]{\small Comparison of \algecd\ with \algclsz\ on synthetic data. \label{tab:gdsbm}}
\begin{center}
\begin{small}
\begin{sc}
\small\addtolength{\tabcolsep}{-1pt}
\begin{tabular}{cccccccc}
\toprule
& & & & \multicolumn{2}{c}{Time} & \multicolumn{2}{c}{ARI} \\
\cmidrule(lr){5-6}\cmidrule(lr){7-8}
Model & $n'$ & $n$ & $k$ & \algecd & \algclsz & \algecd & \algclsz \\
\midrule
CBM & - & $10^3$ & $5$ & $\mathbf{1.59}$ & $3.99$ & $0.92$ & $\mathbf{1.00}$ \\
CBM & - &$10^3$ & $50$ & $\mathbf{3.81}$ & $156.24$ & $0.99$ & $0.99$ \\
CBM+ & $10^2$ & $10^3$ & $3$ & $\mathbf{0.24}$ & $6.12$ & $\mathbf{0.98}$ & $0.35$ \\
CBM+ & $10^2$ & $10^4$ & $3$ & $\mathbf{0.32}$ & $45.17$ & $\mathbf{0.99}$ & $0.01$ \\
\bottomrule
\end{tabular}
\end{sc}
\end{small}
\end{center}
\vskip -0.1in
\end{table}

\paragraph{Real-world dataset.}
Now  we evaluate the algorithms' performance
on the US Migration Dataset~\cite{UnitedStatesCensus2000}.
For fair comparison, we   follow Cucuringu et al.~\cite{cucuringuHermitianMatricesClustering2020} and construct the directed graph as follows: every county in the mainland USA is represented by a vertex; for any vertices $j,\ell$, the edge weight of $(j,\ell)$ is given by
\[
\left|\frac{\matm_{j,\ell} - \matm_{\ell,j}}{\matm_{j,\ell} + \matm_{\ell,j}}\right|,
\]
where $\matm_{j,\ell}$ is the number of people who migrated from county $j$ to county $\ell$ between 1995 and 2000; in addition, the direction of $(j,\ell)$ is set to be from $j$ to $\ell$ if $\matm_{j,\ell}>\matm_{\ell,j}$, otherwise the direction is set to be the opposite.

For \algclsz, we follow their suggestion on the same dataset and set $k=10$.
Both algorithms' performance is evaluated with respect to the flow ratio, as well as the cut imbalance ratio used in their work.
For any vertex sets $\setl$ and $\setr$, the \firstdef{cut imbalance ratio} is defined by
\[
\mathrm{CI}(\setl,\setr) = \frac{1}{2}\cdot \left|\frac{\weight(\setl,\setr) - \weight(\setr,\setl)}{\weight(\setl,\setr) + \weight(\setr,\setl)}\right|,
\]
and a higher $\mathrm{CI}(\setl,\setr)$ value indicates the connection between $\setl$ and $\setr$ is more significant.  Using counties in Ohio, New York, California, and Florida as the starting vertices, the outputs of \algevocutdirected\ are  visualised in Figures~\ref{fig:intro_directed},
and in Table~\ref{tab:compalgo2} we compare them
to  the top $4$ pairs returned by \algclsz, which are shown in Figure~\ref{fig:migration}.
The new algorithm produces better outputs with respect to the two metrics.

\begin{figure*}[t]
\centering
    \begin{subfigure}{0.48\columnwidth}
    \includegraphics[width=\columnwidth]{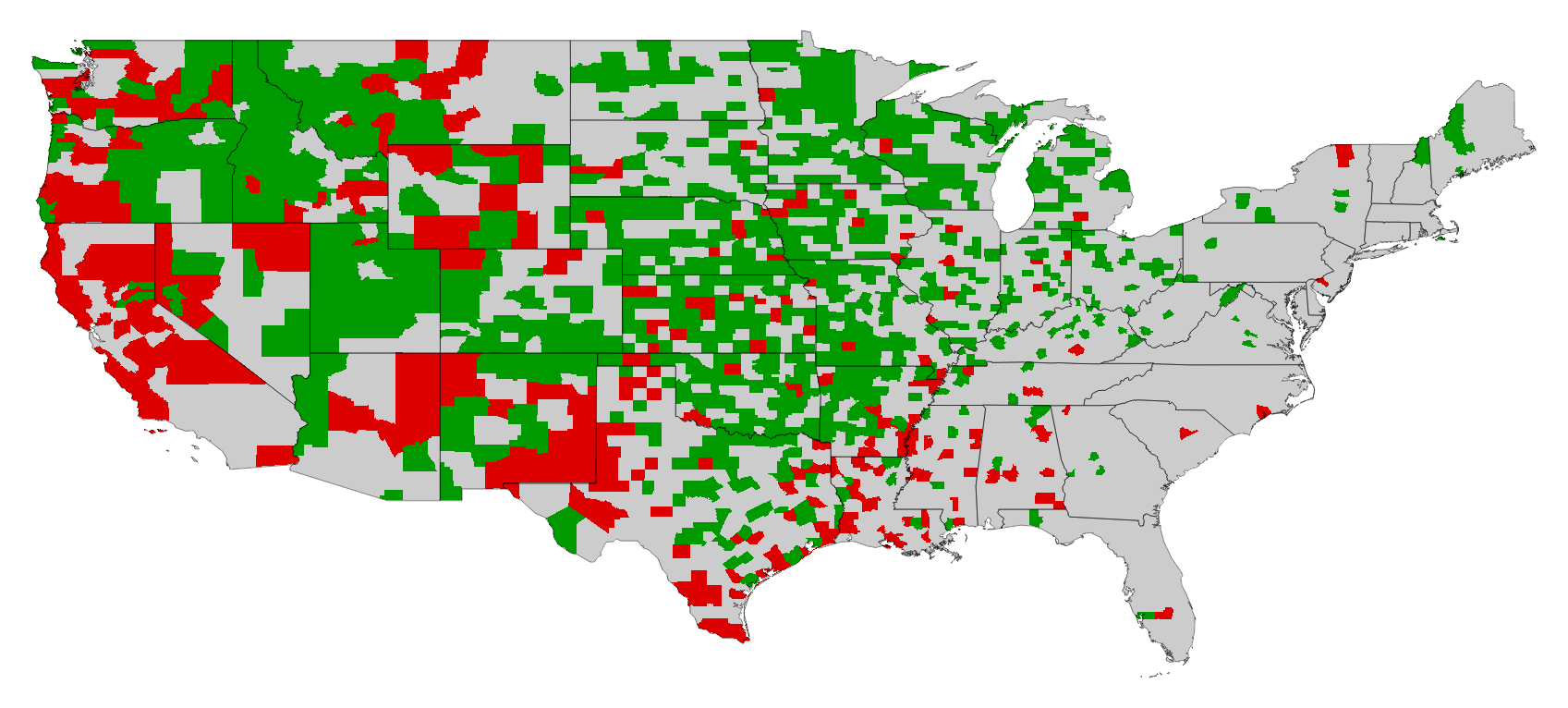}
    \caption{\algclsz, Pair 1}
    \end{subfigure}
    \hfill
    \begin{subfigure}{0.48\columnwidth}
    \includegraphics[width=\columnwidth]{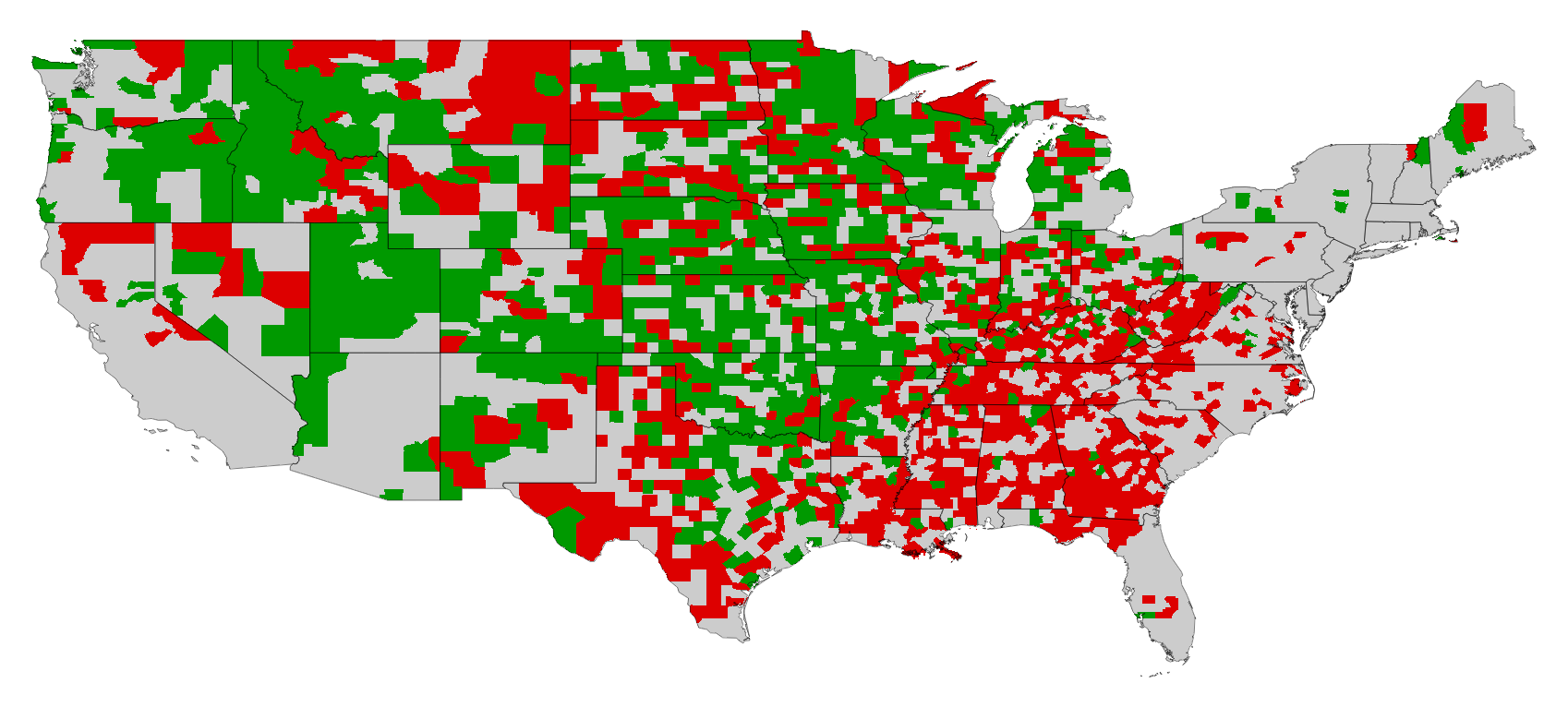}
    \caption{\algclsz, Pair 2}
    \end{subfigure}
    \begin{subfigure}{0.48\columnwidth} 
    \includegraphics[width=\columnwidth]{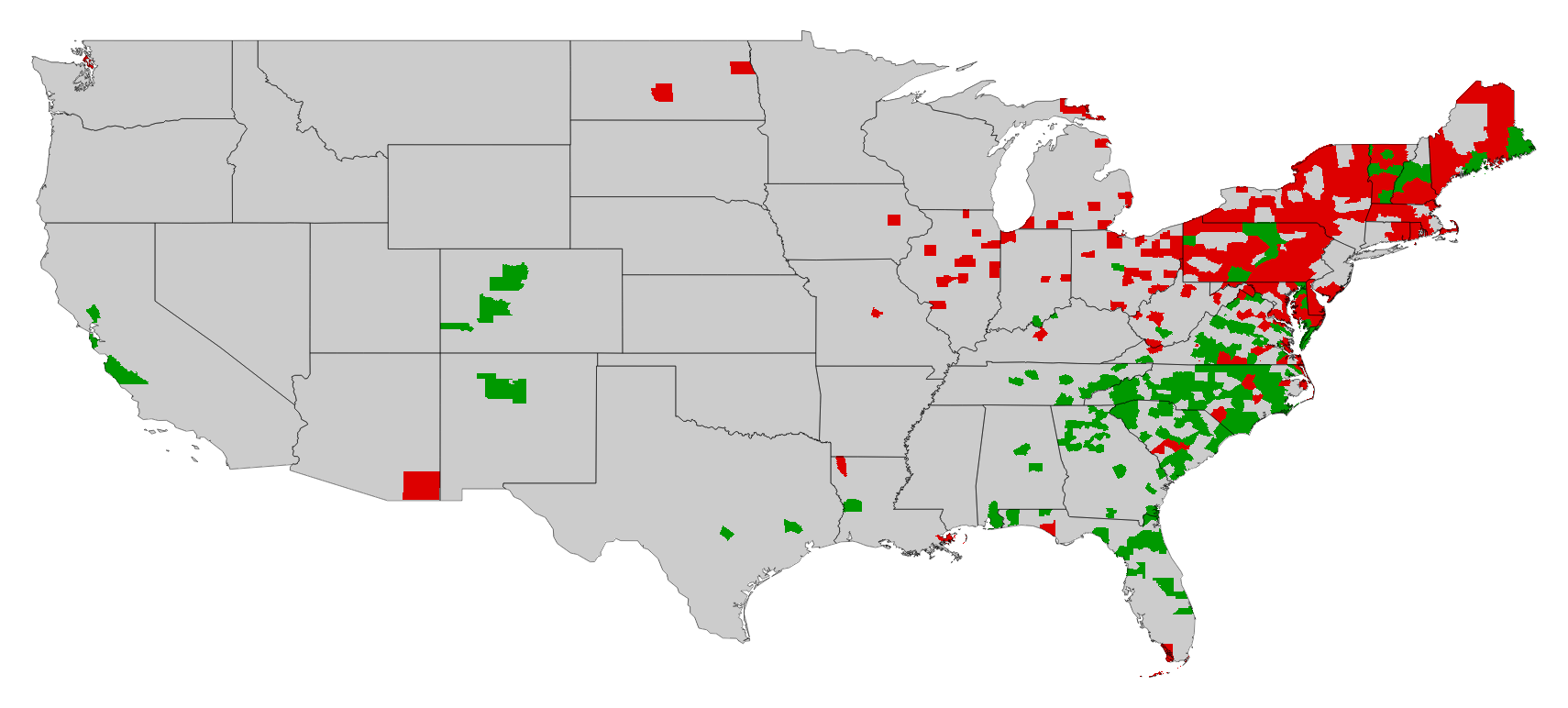}
    \caption{\algclsz, Pair 3}
    \end{subfigure}
    \hfill
    \begin{subfigure}{0.48\columnwidth}
    \includegraphics[width=\columnwidth]{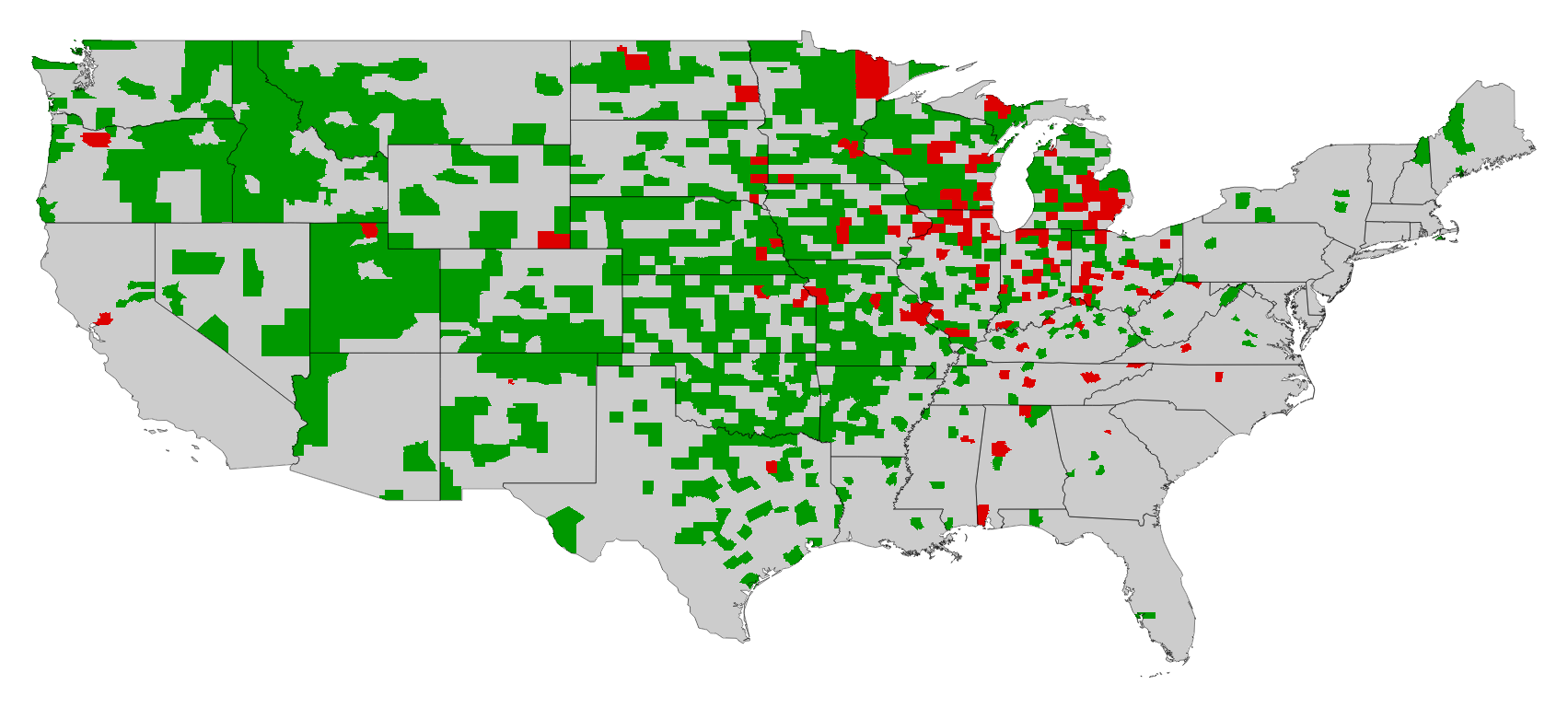}
    \caption{\algclsz, Pair 4}
    \end{subfigure}
    \caption[The top $4$ pairs of clusters found by \algclsz\ on the US Migration dataset]{
    The top $4$ pairs of clusters found by \algclsz\ on the US Migration dataset.
    \label{fig:migration}
    }
\end{figure*}

\begin{table*}[b]
\caption[Comparison of \algevocutdirected\ with \algclsz\ on the US migration dataset.]{\small Comparison of \algevocutdirected\ with \algclsz\ on the US migration dataset.}
\label{tab:compalgo2}
\vskip 0.15in
\begin{center}
\begin{small}
\begin{sc}
\begin{tabular}{ccccc}
\toprule
Figure & Algorithm      & Cluster         & Cut Imbalance & Flow Ratio \\
\midrule
\ref{fig:migration}(a)       & CLSZ           & Pair 1          & $0.41$        & $0.80$     \\
\ref{fig:migration}(b)       & CLSZ           & Pair 2          & $0.35$        & $0.83$     \\
\ref{fig:migration}(c)       & CLSZ           & Pair 3          & $0.32$        & $0.84$     \\
\ref{fig:migration}(d)       & CLSZ           & Pair 4          & $0.29$        & $0.84$     \\
\ref{fig:intro_directed}(a)       & EvoCutDirected & Ohio Seed       & $0.50$        & $0.56$     \\
\ref{fig:intro_directed}(b)       & EvoCutDirected & New York Seed   & $0.49$        & $0.58$     \\
\ref{fig:intro_directed}(c)       & EvoCutDirected & California Seed & $0.49$        & $0.67$     \\
\ref{fig:intro_directed}(d)       & EvoCutDirected & Florida Seed    & $0.42$        & $0.79$     \\
\bottomrule
\end{tabular}
\end{sc}
\end{small}
\end{center}
\vskip -0.1in
\end{table*}

These experiments suggest that local algorithms are not only more efficient, but can also be much more effective than global algorithms when learning certain  structures  in graphs. In particular, some localised structure
might be hidden when applying the objective function over the entire graph.

\section{Future Work}
 This chapter presents a novel graph reduction technique, and applies this to design  two   local algorithms for learning a certain structure of clusters in both undirected and directed graphs. Our experimental results further demonstrate their potential wide applications in learning meaningful structures of real-world datasets. Our work leaves several open questions for future work.
 
Firstly, there is a gap between the approximation guarantee of the algorithm for directed graphs when compared with undirected graphs, and it looks very challenging to close this gap with the techniques  presented in this chapter.
To understand this, it is insightful to discuss why Algorithm~\ref{alg:local_max_cut} cannot be applied for directed graphs, although the input graph is translated into an \emph{undirected} graph by the semi-double cover reduction.
This is because, when translating a directed graph into a bipartite undirected graph, the third property of the $\sigma$-operator in Lemma~\ref{lem:sigmaproperty} no longer holds, since the existence of the edge $\{u_1, v_2\}\in \edgeset_\graphh$ does not necessarily imply that $\{u_2, v_1\}\in \edgeset_\graphh$.
Indeed, Figure~\ref{fig:counterexample} gives a counterexample in which $\left(\sigma \circ (\lazywalkm\vecp)\right)(u) \not \leq \left(\lazywalkm(\sigma \circ \vecp)\right)(u)$.
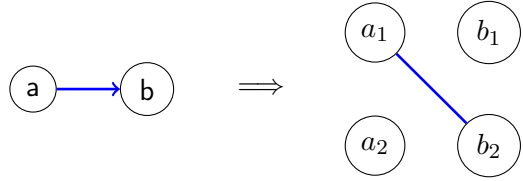
\begin{figure}[t] 
    \centering
    \begin{tikzpicture}
	\begin{pgfonlayer}{nodelayer}
		\node [style=basic] (0) at (0, 1.5) {$a_1$};
		\node [style=basic] (1) at (3, 1.5) {$b_1$};
		\node [style=basic] (2) at (0, -1.5) {$a_2$};
		\node [style=basic] (3) at (3, -1.5) {$b_2$};
		\node [style=basic] (8) at (-9, 0) {a};
		\node [style=basic] (9) at (-6, 0) {b};
		\node [style=none] (12) at (-3, 0) {$\Longrightarrow$};
	\end{pgfonlayer}
	\begin{pgfonlayer}{edgelayer}
		\draw [style=undirected blue] (0) to (3);
		\draw [style=directed blue] (8) to (9);
	\end{pgfonlayer}
\end{tikzpicture}
    \caption[Example showing that \alglocbipartdc\ cannot be used for directed graphs]{\small Consider the
    directed graph and its semi-double cover above. Suppose $\vecp(a_1) = \vecp(a_2) = 0.5$ and $\vecp(b_1) = \vecp(b_2) = 0$. It is straightforward to check that $(\sigma \circ (\lazywalkm\vecp))(b_2) = 0.25$ and $(\lazywalkm\sigmap)(b_2) = 0$.}
    \label{fig:counterexample}
\end{figure}
This means that the typical analysis of a \pagerank\ vector with the Lov\'asz-Simonovitz curve cannot be applied anymore.
In addition, it is interesting to note that one cannot apply the tighter analysis of the ESP process given by Andersen et al.\ \cite{andersenAlmostOptimalLocal2016} to the new algorithm in this chapter.
The key to their analysis is an
improved bound on the probability that a random walk escapes from the target cluster.
This bound given a better guarantee for the conductance of the output set at the cost of a relaxed guarantee on the overlap between the output set and the target set.
Since the analysis for directed graphs in this chapter relies on a very tight guarantee on the overlap of the output set with the target set, we cannot use their improvement in this setting.
These difficulties leave the following important open question.
\begin{openquestion}
Suppose we are given a directed graph $\geqve$ containing clusters $\setl, \setr \subset \vertexset$ with $F(\setl, \setr) = \phi$.
Is there a local algorithm which is guaranteed to return $\setl', \setr' \subset \vertexset$ with $F(\setl', \setr') = \widetilde{O}(\sqrt{\phi})$? 
\end{openquestion}

Secondly, although this chapter shows that two densely connected clusters can be approximately  recovered with local algorithms, it would be very interesting to know whether our presented technique can be generalised to find  clusters whose structure is defined by some meta-graph as described in Chapter~\ref{chap:meta}.
\begin{openquestion}
Suppose we are given a graph $\geqve$, and  $\sets_1, \ldots, \sets_k \subset \vertexset$
are clusters defined  with respect to some   meta-graph.
Is there a local algorithm that approximately recovers these  $\sets_1, \ldots, \sets_k$?
\end{openquestion}

\newcommand{\cals}{\mathcal{S}}         
\newcommand{\cali}{\mathcal{I}}
\newcommand{\cale}{\mathcal{E}}
\newcommand{\cut}{\mathrm{cut}}
\newcommand{\dxdt}{\frac{\mathrm{d} x}{\mathrm{d} t}}
\newcommand{\signlaph}{\signlap_\graphh}
\newcommand{\signlapg}{\signlap_\graphg}
\newcommand{\sfe}{\sets_{\vecf}(e)}
\newcommand{\sfte}{\sets_{\vecf_t}(e)}
\newcommand{\ife}{\set{I}_{\vecf}(e)}
\newcommand{\ifte}{\set{I}_{\vecf_t}(e)}

\newcommand{\deltaesquare}[1]{\left(\max_{u \in e} #1(u) + \min_{v \in e}#1(v)\right)^2}
\newcommand{\deltaeabs}[1]{\abs{\max_{u \in e} #1(u) + \min_{v \in e} #1(v)}}

\chapter{Finding Bipartite Components in Hypergraphs} \label{chap:hyper}
Hypergraphs are important objects to model ternary or higher-order relations of objects, and have a number of  applications in analysing many complex datasets occurring in practice.
In this chapter we   study a new heat diffusion process in hypergraphs, and employ this process to design a polynomial-time algorithm that approximately finds bipartite components in a hypergraph.
We theoretically study the performance of the proposed algorithm, and compare it against the previous state-of-the-art through extensive experimental analysis on both synthetic and real-world datasets.
We   see that the new algorithm consistently and significantly outperforms the previous state-of-the-art across a wide range of hypergraphs.
The problem studied in this chapter can be seen as a generalisation of the problem studied in the previous chapter, although we   study algorithms in the global setting rather than the local one.

 A unifying theme of this thesis is the design of efficient algorithms based on spectral methods, which study the graph Laplacian matrix in particular.
The results of the present chapter form part of a sequence of work to generalise many of the results of spectral graph theory to hypergraphs.
We introduce and study the non-linear Laplacian-type operator $\signlap_\graphh$ for any hypergraph $\graphh$.
While  we formally define the operator $\signlap_\graphh$ in Section~\ref{sec:diff_and_alg}, one can informally think about $\signlap_\graphh$ as a variant of the non-linear hypergraph Laplacian $\lap_\graphh$ studied in~\cite{chanSpectralPropertiesHypergraph2018, liSubmodularHypergraphsPLaplacians2018, takaiHypergraphClusteringBased2020} which generalises the \emph{signless Laplacian operator} of graphs.
We develop a polynomial-time algorithm that finds an eigenvalue $\lambda$ and its associated eigenvector of $\signlap_\graphh$.
 
\begin{mainresult}[See Theorem~\ref{thm:convergence} for the formal statement]
Given a hypergraph $\graphh$, there is a non-linear operator $\signlap_\graphh$ which generalises the signless graph Laplacian operator for hypergraphs. Moreover, there
 is a polynomial-time algorithm that finds an eigenvalue of $\signlap_\graphh$ along with its associated eigenvector.
\end{mainresult}

The algorithm is based on the following  heat diffusion process:
starting from an arbitrary vector $\vecf_0 \in \R^n$ that describes the initial heat distribution among the vertices, we use $\vecf_0$ to  construct some graph\footnote{Note that the word `graph' always refers to a non-hyper graph. Similarly, we always use $\lap_\graphh$ to refer to the \emph{non-linear} hypergraph Laplacian operator, and use $\lap_\graphg$ as the standard graph Laplacian.}
$\graphg_0$, and use the diffusion process in $\graphg_0$ to represent the  one in the original hypergraph $\graphh$ and update $\vecf_t$; this process continues until the time at which $\graphg_0$ cannot be used to appropriately simulate the diffusion process in $\graphh$ any more.
At this point, we use the currently maintained $\vecf_t$ to construct another graph $\graphg_t$ to simulate the diffusion process in $\graphh$, and update $\vecf_t$.
This process continues until the vector $\vecf_t$ converges; see Figure~\ref{fig:illustration} for illustration.
We prove that this heat diffusion process is unique, well-defined, and the maintained vector $\vecf_t$ converges to an eigenvector of $\signlap_\graphh$.

\begin{figure}[t]
    \centering
    \scalebox{0.8}{\singlespace
    \input{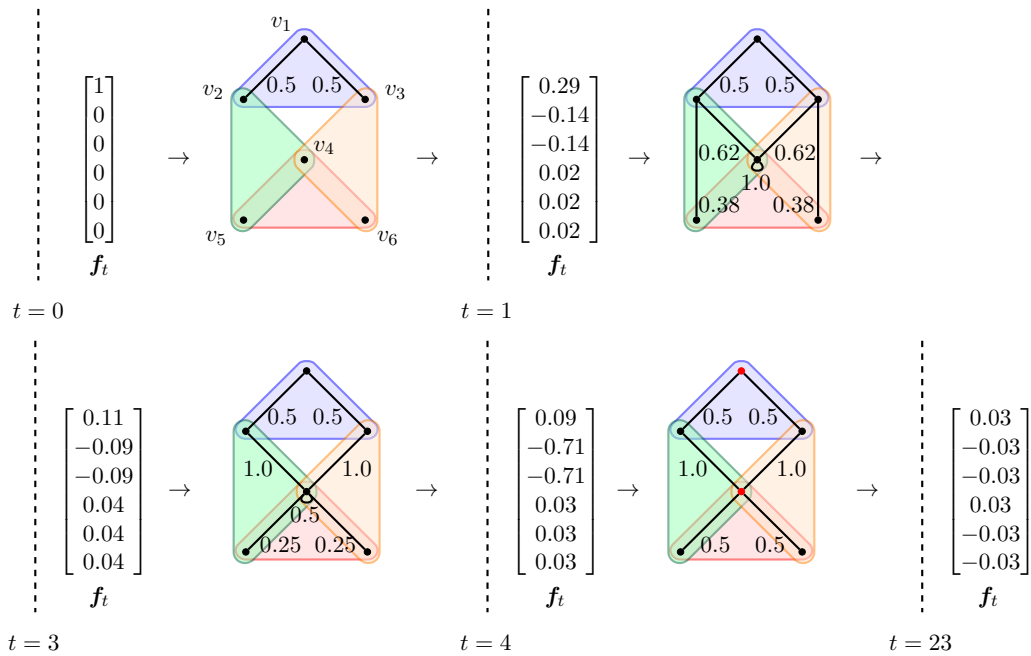}
    }
\caption[Hypergraph signless Laplacian diffusion process]{\small{Illustration of the proposed diffusion process. In each time step, we construct a graph $\graphg$ based on the current vector $\vecf_t$, and update $\vecf_t$ with the $\signlap_\graphg$ operator.
Notice that the graph $\graphg$ changes throughout the execution of the algorithm, and that the final $\vecf_t$ vector can be used to partition the vertices of $\graphh$ into two well-connected sets (all the edges are adjacent to both sets), by splitting according to positive and negative entries.}
\label{fig:illustration}}
\end{figure}

The second result of this chapter is a polynomial-time algorithm that, given a hypergraph $\graphh=(\vertexset_\graphh, \edgeset_\graphh, \weight)$ as input, 
finds disjoint subsets $\setl, \setr \subset \vertexset_\graphh$ that are highly connected with each other.
 \begin{mainresult}[See Theorem~\ref{thm:mainalg} for the formal statement] 
Suppose we are given a hypergraph $\graphh = (\vertexset_\graphh, \edgeset_\graphh, \weight)$, and an eigenvalue $\lambda$ of the operator $\signlap_\graphh$ with its corresponding eigenvector.
There is an algorithm that finds two clusters $\setl, \setr \subset \vertexset_\graphh$ such that $\bipart_\graphh(\setl, \setr) \leq \sqrt{2 \lambda}$, where $\bipart_\graphh$ is the \emph{hypergraph bipartiteness}.
\end{mainresult}

The key to this algorithm is a Cheeger-type inequality for hypergraphs that relates the spectrum of $\signlap_\graphh$ and the \emph{hypergraph bipartiteness} of $\graphh$, an analogue of the graph bipartiteness studied in Chapter~\ref{chap:local}.
Both the design and analysis of the algorithm is inspired  by Trevisan~\cite{trevisanMaxCutSmallest2012}, however the analysis is much more involved because of the non-linear operator $\signlap_\graphh$ and hyperedges of different ranks.
This result alone answers an open question posed by Yoshida~\cite{yoshidaCheegerInequalitiesSubmodular2019},  who asks whether there is a hypergraph operator which satisfies a Cheeger-type inequality for bipartiteness.

The significance of our algorithmic result is demonstrated by extensive experimental studies of the algorithm on both synthetic and real-world datasets.
In particular, on the well-known Penn Treebank corpus that contains 49,208 sentences and over 1 million words, 
our \emph{purely unsupervised} algorithm is able to identify a significant fraction of verbs from non-verbs in its two output clusters.
Hence, this work could potentially have many applications in unsupervised learning for hypergraphs.

Finally, we study the spectrum of the non-linear Laplacian-type operators for hypergraphs, and show that, for any $n \in \Z_{\geq 0}$, there is a hypergraph on $n$ vertices such that $\lap_\graphh$ and $\signlap_\graphh$ have more than $n$ eigenvectors.
This answers an open question posed by Chan~\etal~\cite{chanSpectralPropertiesHypergraph2018}, who ask whether $\lap_\graphh$ could have more than $2$ eigenvectors\footnote{Notice that, while the operator $\lap_\graphg$ of a graph $\graphg$ has $n$ eigenvalues, the number of eigenvalues of $\lap_\graphh$ is unknown because of its non-linearity.}.

\section{Additional Notation} \label{sec:hyper:prelim}
In order to study densely connected clusters in hypergraphs, we   need some new notation and definitions.
Suppose, we have a hypergraph $\graphh = (\vertexset, \edgeset, \weight)$, then for any $\seta, \setb \subset \vertexset$ we can overload the weight function to define the cut value between $\seta$ and $\setb$ by
\[
\weight(\seta, \setb) \triangleq
\sum_{e\in \edgeset} \weight(e) \cdot \iverson{e \intersect \seta \neq \emptyset \land e \intersect \setb \neq \emptyset}.
\]
Sometimes, we would like to analyse the weights of edges that intersect some vertex sets and not others. To this end, for any $\seta, \setb, \setc \subseteq \vertexset$, let 
\[
\weight(\seta, \setb~|~\setc) \triangleq \sum_{e \in \edgeset} \weight(e) \cdot \iverson{e \intersect \seta \neq \emptyset \land e \intersect \setb \neq \emptyset \land e \intersect C=\emptyset},
\]
and we can also define $\weight(\seta~|~\setc) \triangleq \weight(\seta, \seta~|~\setc)$ for simplicity.
Generalising the notion of the bipartiteness  of a graph, the \firstdef{hypergraph bipartiteness} of sets $\setl,\setr \subset \vertexset$ is defined by 
    \[
        \bipart_\graphh(\setl, \setr) \triangleq \frac{2\weight(\setl | \setcomplement{\setl}) + 2\weight(\setr | \setcomplement{\setr}) + \weight(\setl, \setcomplement{\setl} | \setr) + w(\setr, \setcomplement{\setr}| \setl)}{\vol(\lur)},
    \]
    and we define $$\bipart_H\triangleq \min_{\sets \subset \vertexset_\graphh}\bipart_\graphh(\sets, \vertexset_\graphh \setminus \sets).$$
This definition generalises the graph bipartiteness (Definition~\ref{def:bipartiteness}) in that, if every edge in $\graphh$ contains $2$ vertices, then $\graphh$ is a (non-hyper) graph and $\bipart_\graphh(\setl, \setr)$ is equivalent to the graph bipartiteness.
Informally, a low value of $\bipart_\graphh(\setl, \setr)$ indicates that most of the edges which intersect either $\setl$ or $\setr$ connect $\setl$ and $\setr$ together.
The goal of this chapter is to develop an algorithm to find such sets $\setl$ and $\setr$ in a hypergraph.

For any hypergraph $H$ and $\vecf \in \R^n$, we define the \firstdef{discrepancy} of an edge $e \in \edgeset_\graphh$ with respect to $\vecf$ as
\[
\discrep_{\vecf} (e) \triangleq \max_{u\in e} \vecf(u) + \min_{v\in e} \vecf(v).
\] 
Moreover, the weighted discrepancy of $e \in \edgeset_\graphh$ is
\[
    c_{\vecf}(e) \triangleq \weight(e) \abs{\discrep_{\vecf}(e)},
\]
and for any set $\sets \subseteq \edgeset_\graphh$  we have
\[
    c_{\vecf}(\sets) \triangleq \sum_{e \in \sets} c_{\vecf}(e).
\]
The discrepancy of $\graphh$ with respect to $\vecf$ is defined by
\[
D(\vecf) \triangleq\frac{\sum_{e\in \edgeset_\graphh} \weight(e) \cdot \discrep_{\vecf}(e)^2 }{\sum_{v\in \vertexset_\graphh} \deg_\graphh(v)\cdot \vecf(v)^2}.
\]
For any
$\vecf, \vecg \in \R^n$, the weighted inner product between $\vecf$ and $\vecg$ is defined by 
$$\inner{\vecf}{\vecg}_w \triangleq \vecf^\transpose \degm_\graphh \vecg,$$
where $\degm_\graphh \in \R^{n \times n}$ is the diagonal matrix consisting of the degrees of all the vertices of $\graphh$. The weighted norm of $\vecf$ is defined by $\norm{\vecf}_w^2 \triangleq  \inner{\vecf}{\vecf}_w$.

For any non-linear operator $\mat{J}: \R^n \mapsto \R^n$, we say that $(\lambda, \vecf)$ is an eigen-pair if and only if $\mat{J} \vecf = \lambda \vecf$. Note that
in contrast with linear operators, a non-linear operator does not necessarily have $n$ eigenvalues.
Additionally, we define the weighted Rayleigh quotient of the operator $\signlap$ and vector $\vecf$ to be
\[
    R_\signlap(\vecf) \triangleq \frac{\vecf^\transpose \signlap \vecf}{\norm{\vecf}_\weight^2}.
\]
It is  important to remember that throughout this chapter the letter $\graphh$ is always used to represent a hypergraph, and $\graphg$ to represent a graph. 

\section{Baseline Algorithms} \label{sec:hyper:baselines}
Before we present the new algorithm, let us first consider the baseline algorithms for finding bipartite components in hypergraphs.
One natural idea is to
construct a graph $\graphg$ from the original hypergraph $\graphh$ and apply a graph algorithm on $\graphg$ to find a good approximation of the cut structure of $\graphh$~\cite{chitraRandomWalksHypergraphs2019, liInhomogeneousHypergraphClustering2017, zhouLearningHypergraphsClustering2006}.
However, most graph reductions would introduce a factor of $r$, which relates to the rank of hyperedges in $\graphh$, into the approximation guarantee.
To see this, consider the following two natural graph reductions:
\begin{enumerate}
\item \firstdef{Random Reduction}: From $\graphh=(\vertexset,\edgeset_\graphh)$, we construct $\graphg=(\vertexset, \edgeset_\graphg)$ in which every hyperedge $e\in  \edgeset_\graphh$ is replaced by an edge connecting two randomly chosen vertices in $e$;
    \item \firstdef{Clique Reduction}: From $\graphh=(\vertexset,\edgeset_\graphh)$, we construct $\graphg=(\vertexset, \edgeset_\graphg)$ in which 
every hyperedge $e\in \edgeset_\graphh$ is replaced by a clique of $\rank(e)$ vertices in $\graphg$.
\end{enumerate}

To discuss the drawback of both reductions, we study the following two $r$-uniform  hypergraphs $\graphh_1$ and $\graphh_2$, in each of which the vertex set is $\setl\cup \setr$ and the edge set is defined as follows:
\begin{enumerate}
    \item in $\graphh_1$, every edge contains exactly one vertex from $\setl$, and $r-1$ vertices from $\setr$;
    \item in $\graphh_2$, every edge contains exactly $r/2$ vertices from $\setl$ and $r/2$ vertices from $\setr$.
\end{enumerate}
As such, we have that $\weight_{\graphh_1}(\setl,\setr) = \cardinality{\edgeset_{\graphh_1}}$, and $\weight_{\graphh_2}(\setl,\setr) = \cardinality{\edgeset_{\graphh_2}}$.
See  Figure~\ref{fig:example_hypergraphs} for illustration.

\begin{figure}[t]
    \centering
    \begin{tikzpicture}
	\begin{pgfonlayer}{nodelayer}
		\node [style=smallblack] (4) at (-5, 3) {};
		\node [style=smallblack] (6) at (-5, 1.75) {};
		\node [style=none] (15) at (-5, 1.5) {};
		\node [style=none] (17) at (-5, 3.25) {};
		\node [style=none] (18) at (-9, 2.5) {};
		\node [style=none] (19) at (-9, 2) {};
		\node [style=smallblack] (20) at (-9, 2.25) {};
		\node [style=none] (29) at (-10.5, 2.5) {$\setl$};
		\node [style=none] (30) at (-3.5, 2.5) {$\setr$};
		\node [style=none] (31) at (-7, -4.25) {Hypergraph $\graphh_1$};
		\node [style=smallblack] (34) at (9, 3.25) {};
		\node [style=smallblack] (35) at (5, 2) {};
		\node [style=smallblack] (36) at (9, 2) {};
		\node [style=smallblack] (49) at (5, 3.25) {};
		\node [style=none] (58) at (3.5, 2.5) {$\setl$};
		\node [style=none] (59) at (10.5, 2.5) {$\setr$};
		\node [style=none] (60) at (7, -4.25) {Hypergraph $\graphh_2$};
		\node [style=none] (61) at (9, 1.75) {};
		\node [style=none] (62) at (9, 3.5) {};
		\node [style=none] (63) at (5, 3.5) {};
		\node [style=none] (64) at (5, 1.75) {};
		\node [style=none] (81) at (-5, 2.575) {$\vdots$};
		\node [style=smallblack] (82) at (-5, 0.75) {};
		\node [style=smallblack] (83) at (-5, -0.5) {};
		\node [style=none] (84) at (-5, -0.75) {};
		\node [style=none] (85) at (-5, 1) {};
		\node [style=none] (86) at (-9, 0.25) {};
		\node [style=none] (87) at (-9, -0.25) {};
		\node [style=smallblack] (88) at (-9, 0) {};
		\node [style=none] (89) at (-5, 0.325) {$\vdots$};
		\node [style=smallblack] (90) at (-5, -1.5) {};
		\node [style=smallblack] (91) at (-5, -2.75) {};
		\node [style=none] (92) at (-5, -3) {};
		\node [style=none] (93) at (-5, -1.25) {};
		\node [style=none] (94) at (-9, -2) {};
		\node [style=none] (95) at (-9, -2.5) {};
		\node [style=smallblack] (96) at (-9, -2.25) {};
		\node [style=none] (97) at (-5, -1.925) {$\vdots$};
		\node [style=big vertical ellipse dashed] (100) at (-9, 0) {};
		\node [style=big vertical ellipse dashed] (101) at (-5, 0) {};
		\node [style=none] (102) at (5, 2.825) {$\vdots$};
		\node [style=none] (103) at (9, 2.825) {$\vdots$};
		\node [style=smallblack] (104) at (9, 0.75) {};
		\node [style=smallblack] (105) at (5, -0.5) {};
		\node [style=smallblack] (106) at (9, -0.5) {};
		\node [style=smallblack] (107) at (5, 0.75) {};
		\node [style=none] (108) at (9, -0.75) {};
		\node [style=none] (109) at (9, 1) {};
		\node [style=none] (110) at (5, 1) {};
		\node [style=none] (111) at (5, -0.75) {};
		\node [style=none] (112) at (5, 0.325) {$\vdots$};
		\node [style=none] (113) at (9, 0.325) {$\vdots$};
		\node [style=smallblack] (114) at (9, -1.75) {};
		\node [style=smallblack] (115) at (5, -3) {};
		\node [style=smallblack] (116) at (9, -3) {};
		\node [style=smallblack] (117) at (5, -1.75) {};
		\node [style=none] (118) at (9, -3.25) {};
		\node [style=none] (119) at (9, -1.5) {};
		\node [style=none] (120) at (5, -1.5) {};
		\node [style=none] (121) at (5, -3.25) {};
		\node [style=none] (122) at (5, -2.175) {$\vdots$};
		\node [style=none] (123) at (9, -2.175) {$\vdots$};
		\node [style=big vertical ellipse dashed] (124) at (9, 0) {};
		\node [style=big vertical ellipse dashed] (125) at (5, 0) {};
	\end{pgfonlayer}
	\begin{pgfonlayer}{edgelayer}
		\draw [style=undirected green] (19.center)
			 to [bend left=75, looseness=2.50] (18.center)
			 to [in=150, out=15, looseness=0.25] (17.center)
			 to [bend left=75] (15.center)
			 to [in=-15, out=-150, looseness=0.25] cycle;
		\draw [style=undirected green] (62.center)
			 to [bend left=60, looseness=1.25] (61.center)
			 to [bend left, looseness=0.50] (64.center)
			 to [bend left=60, looseness=1.25] (63.center)
			 to [bend left, looseness=0.50] cycle;
		\draw [style=undirected green] (87.center)
			 to [bend left=75, looseness=2.50] (86.center)
			 to [in=150, out=15, looseness=0.25] (85.center)
			 to [bend left=75] (84.center)
			 to [in=-15, out=-150, looseness=0.25] cycle;
		\draw [style=undirected green] (95.center)
			 to [bend left=75, looseness=2.50] (94.center)
			 to [in=150, out=15, looseness=0.25] (93.center)
			 to [bend left=75] (92.center)
			 to [in=-15, out=-150, looseness=0.25] cycle;
		\draw [style=undirected green, bend left=60, looseness=1.25] (109.center) to (108.center);
		\draw [style=undirected green, bend left, looseness=0.50] (108.center) to (111.center);
		\draw [style=undirected green, bend left=60, looseness=1.25] (111.center) to (110.center);
		\draw [style=undirected green, bend left, looseness=0.50] (110.center) to (109.center);
		\draw [style=undirected green, bend left=60, looseness=1.25] (119.center) to (118.center);
		\draw [style=undirected green, bend left, looseness=0.50] (118.center) to (121.center);
		\draw [style=undirected green, bend left=60, looseness=1.25] (121.center) to (120.center);
		\draw [style=undirected green, bend left, looseness=0.50] (120.center) to (119.center);
	\end{pgfonlayer}
\end{tikzpicture}
    \caption{Hypergraphs illustrating the downside of simple reductions
    \label{fig:example_hypergraphs}}
\end{figure}
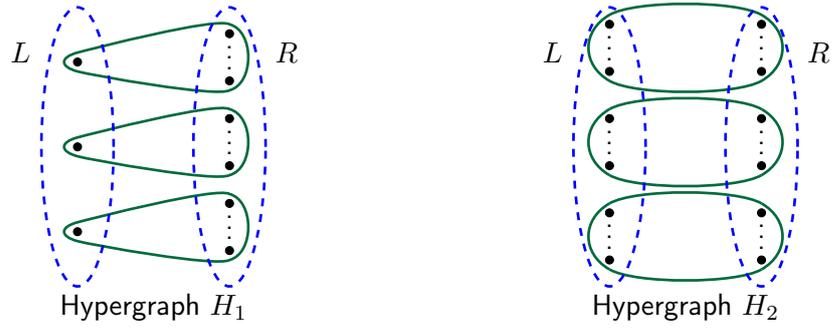

Now we analyse the size of the cut $(\setl, \setr)$ in the reduced graphs.
\begin{itemize}
    \item Let $\weight_{\graphg}(\setl,\setr)$ be the cut value of $(\setl,\setr)$ in the \emph{random} graph $\graphg$ constructed from $\graphh$ by the random reduction. We have for $\graphh_1$ that $\E[\weight_{\graphg}(\setl,\setr) ] = \Theta(1/r)\cdot \weight_{\graphh_1}(\setl,\setr)$  and for $\graphh_2$ that $\E[\weight_{\graphg}(\setl,\setr) ] = \Theta(1)\cdot \weight_{\graphh_2}(\setl,\setr)$. 
    \item 
Similarly, by setting $\weight_{\graphg}(\setl,\setr)$ to be the cut value of $(\setl,\setr)$ in the reduced graph $\graphg$ constructed from $\graphh$ by the clique reduction, we have for $\graphh_1$ that
$\weight_{\graphg}(\setl,\setr)= \Theta(r)\cdot \weight_{\graphh_1}(\setl,\setr)$
and for $\graphh_2$ that
$\weight_{\graphg}(\setl,\setr) = \Theta(r^2)\cdot \weight_{\graphh_2}(\setl,\setr)$.
\end{itemize}
Notice that these two approximation ratios differ by a factor of $r$ in the two reductions.
This suggests that reducing an $r$-uniform hypergraph $\graphh$ to a graph with some simple reduction would always introduce a factor of $r$ into the approximation guarantee.
This is one of the main reasons to develop spectral theory for hypergraphs through heat diffusion processes~\cite{chanSpectralPropertiesHypergraph2018, takaiHypergraphClusteringBased2020, yoshidaCheegerInequalitiesSubmodular2019}.

\section{Diffusion Process and Algorithm} \label{sec:diff_and_alg}
In this section, we introduce a new diffusion process in hypergraphs, and use it to design a polynomial-time algorithm  for finding bipartite components in hypergraphs.
We first study the heat diffusion process in graphs to give some intuition. 
Then, we generalise it to hypergraphs and describe the new algorithm.
Finally, we sketch some of the analysis which proves that the diffusion process is well defined.
The full analysis is given in Sections~\ref{sec:thm2proof}~and~\ref{sec:thm1proof}.

\subsection{Diffusion Process in Graphs} \label{sec:2graph-diffusion}
To discuss the intuition behind our designed diffusion process, let us look at the case of graphs. Let $\graphg=(\vertexset, \edgeset , \weight)$ be a graph, and we have for  any $\vecx \in \R^n$ that
\[
\frac{\vecx^\transpose \signlapn_\graphg \vecx}{\vecx^\transpose \vecx} = 
\frac{\vecx^\transpose ( \identity + \degmhalfneg_\graphg \adj_\graphg \degmhalfneg_\graphg) \vecx }{\vecx^\transpose \vecx }.
\]
By setting $\vecx = \degmhalf_\graphg \vecy$, we have that 
\begin{align}\label{eq:2graphcalculation}
\frac{\vecx^\transpose \signlapn_\graphg \vecx}{\vecx^\transpose \vecx}
& = \frac{\vecy^\transpose \degmhalf_\graphg \signlapn_\graphg \degmhalf_\graphg \vecy }{ \vecy^\transpose \degm_\graphg \vecy} \nonumber \\
& = \frac{\vecy^\transpose (\degm_\graphg + \adj_\graphg) \vecy }{\vecy^\transpose \degm_\graphg  \vecy } \nonumber \\
& = \frac{\sum_{\{u,v\}\in \edgeset_\graphg} \weight(u,v) \cdot (\vecy(u) + \vecy(v))^2}{\sum_{u\in \vertexset_\graphg} \deg_\graphg(u) \cdot \vecy(u)^2}.
\end{align}
It is easy to see  that   $\lambda_{1}(\signlapn_\graphg)=0$ if $\graphg$ is bipartite,  and it is known that  $\lambda_{1}(\signlapn_\graphg)$ and its corresponding eigenvector $\vecf_1(\signlapn_\graphg)$ are closely related to two densely connected components of $\graphg$~\cite{trevisanMaxCutSmallest2012}.
Moreover, similar to the heat equation for graph Laplacian $\lap_\graphg$, suppose $\degm_\graphg \vecf_t \in \R^n$ is some measure on the vertices of $\graphg$, then a diffusion process defined by the differential equation
\begin{equation} \label{eq:2graph_pde}
\frac{\mathrm{d} \vecf_t}{\mathrm{d} t} = -\degm_\graphg^{-1} \signlap_\graphg \vecf_t\
\end{equation}
converges to the eigenvector of $\degm_\graphg^{-1} \signlap_\graphg$ with minimum eigenvalue and
can be employed to find two densely connected components of the underlying graph.\footnote{For the reader familiar with the heat diffusion process of graphs~(e.g.,  \cite{chungHeatKernelPagerank2007, klosterHeatKernelBased2014}), notice that the above-defined process essentially employs the operator $\signlap_\graphg$ to
replace the Laplacian $\lap_\graphg$ when defining the heat diffusion: through $\signlap_\graphg$, the heat diffusion can be used to find two densely connected components of $\graphg$.}

\subsection{Hypergraph Diffusion and our Algorithm} \label{sec:algorithm}
Now we study whether one can construct a new hypergraph operator $\signlap_\graphh$ that generalises the diffusion in graphs to hypergraphs.
First of all, we focus on a fixed time $t$ with measure vector $\degm_\graphh \vecf_t \in \R^n$. 
We follow \eqref{eq:2graph_pde} and define the rate of change
\[
    \frac{\mathrm{d} \vecf_t}{\mathrm{d} t} = - \degm_\graphh^{-1} \signlap_\graphh \vecf_t
\]
so that the diffusion can proceed for an infinitesimal time step.
Our intuition is that the rate of change due to some edge $e \in \edgeset_\graphh$ should involve only the vertices in $e$ with the maximum or minimum value in the normalised measure $\vecf_t$.
To formalise this, for any edge $e \in \edgeset_\graphh$ we define
\[
    \sets_{\vecf}(e) \triangleq \{ v \in e : \vecf_t(v) = \max_{u \in e} \vecf_t(u) \} \mbox{\ \ \ and\ \ \ }
    \set{I}_{\vecf}(e) \triangleq \{ v \in e : \vecf_t(v) = \min_{u \in e} \vecf_t(u) \}.
\]
That is, for any edge $e$ and normalised measure $\vecf_t$, $\sets_{\vecf}(e) \subseteq e$ consists of the vertices $v$ adjacent to $e$ whose $\vecf_t(v)$ values are maximum, and $\set{I}_{\vecf}(e) \subseteq e$ consists of the vertices $v$ adjacent to $e$ whose $\vecf_t(v)$ values are minimum.
See Figure~\ref{fig:example_hyperedge} for an example.
Then, applying the $\signlap_\graphh$ operator to a vector $\vecf_t$ should be equivalent to applying the operator $\signlap_\graphg$ for some graph $\graphg$ which we construct by splitting the weight of each hyperedge $e \in \edgeset_\graphh$ between the edges in $\sets_{\vecf}(e) \times \set{I}_{\vecf}(e)$.
Similar to the case for graphs and \eqref{eq:2graphcalculation}, for any $\vecx = \degmhalf_\graphh \vecf_t$ this gives us the quadratic form
\begin{align*}
\frac{\vecx^\transpose \degmhalfneg_\graphh \signlap_\graphh \degmhalfneg_\graphh \vecx}{\vecx^\transpose \vecx}
= \frac{\vecf_t^\transpose \signlap_\graphg \vecf_t }{\vecf_t^\transpose \degm_\graphh  \vecf_t } & 
= \frac{\sum_{\{u,v\}\in \edgeset_\graphg} \weight_\graphg(u, v) \cdot (\vecf_t(u) + \vecf_t(v))^2}{\sum_{u\in \vertexset_\graphg} \deg_\graphh(u) \cdot \vecf_t(u)^2} \\
& = \frac{\sum_{e \in \edgeset_\graphh} \weight_\graphh(e) (\max_{u \in e} \vecf_t(u) + \min_{v \in e} \vecf_t(v))^2}{\sum_{u \in \vertexset_\graphh} \deg_\graphh(u)\cdot \vecf_t(u)^2},
\end{align*}
where   $\weight_\graphg(u, v)$ is the weight of the edge $\{u, v\}$ in $\graphg$, and $\weight_\graphh(e)$ is the weight of the edge $e$ in $\graphh$.
We   see later in this chapter that $\signlap_\graphh$ has an eigenvalue of $0$ if the hypergraph is $2$-colourable\footnote{Hypergraph $\graphh$ is $2$-colourable if there are disjoint sets $\setl, \setr \subset \vertexset_\graphh$ such that every edge intersects $\setl$ and $\setr$.}, and the spectrum of $\signlap_\graphh$ is closely related to the hypergraph bipartiteness.

\begin{figure}[t]
    \centering
    \begin{tikzpicture}
	\begin{pgfonlayer}{nodelayer}
		\node [style=smallblack] (0) at (-1, 2) {};
		\node [style=smallblack] (1) at (1, 2) {};
		\node [style=smallblack] (2) at (-1, 0) {};
		\node [style=smallblack] (3) at (1, 0) {};
		\node [style=smallblack] (4) at (-1, -2) {};
		\node [style=smallblack] (5) at (1, -2) {};
		\node [style=none] (8) at (-2.25, -2.2) {};
		\node [style=none] (9) at (-1.25, -3.2) {};
		\node [style=none] (10) at (0.75, -3.2) {};
		\node [style=none] (11) at (1.75, -2.2) {};
		\node [style=none] (12) at (1.75, 1.8) {};
		\node [style=none] (13) at (0.75, 2.8) {};
		\node [style=none] (14) at (-1.25, 2.8) {};
		\node [style=none] (15) at (-2.25, 1.8) {};
		\node [style=none] (16) at (-1.5, 1.5) {$-2$};
		\node [style=none] (17) at (0.5, 1.5) {$-2$};
		\node [style=none] (18) at (-1.5, -0.5) {$0$};
		\node [style=none] (19) at (0.5, -0.5) {$1$};
		\node [style=none] (20) at (0.5, -2.5) {$3$};
		\node [style=none] (21) at (-1.5, -2.5) {$1$};
		\node [style=none] (22) at (-2, 2) {};
		\node [style=none] (23) at (-2, 1.5) {};
		\node [style=none] (24) at (-1.5, 1) {};
		\node [style=none] (25) at (1, 1) {};
		\node [style=none] (26) at (1.5, 1.5) {};
		\node [style=none] (27) at (1.5, 2) {};
		\node [style=none] (28) at (1, 2.5) {};
		\node [style=none] (29) at (-1.5, 2.5) {};
		\node [style=none] (30) at (1, -1.5) {};
		\node [style=none] (31) at (1.5, -2) {};
		\node [style=none] (32) at (1.5, -2.5) {};
		\node [style=none] (33) at (1, -3) {};
		\node [style=none] (34) at (0.5, -3) {};
		\node [style=none] (35) at (0, -2.5) {};
		\node [style=none] (36) at (0, -2) {};
		\node [style=none] (37) at (0.5, -1.5) {};
		\node [style=none] (38) at (-0.25, -3.75) {Hyperedge $e$};
		\node [style=none] (39) at (2.75, 1.75) {$\ife$};
		\node [style=none] (40) at (2.75, -2.25) {$\sfe$};
	\end{pgfonlayer}
	\begin{pgfonlayer}{edgelayer}
		\draw [style=undirected] (9.center)
			 to (10.center)
			 to [bend left=315, looseness=1.25] (11.center)
			 to (12.center)
			 to [bend right=45, looseness=1.25] (13.center)
			 to (14.center)
			 to [bend right=45, looseness=1.25] (15.center)
			 to (8.center)
			 to [bend right=45, looseness=1.25] cycle;
		\draw [style=blue dashed, fill={blue!20}, opacity=0.5] (29.center)
			 to (28.center)
			 to [bend left=45] (27.center)
			 to (26.center)
			 to [bend left=45] (25.center)
			 to (24.center)
			 to [bend right=315] (23.center)
			 to (22.center)
			 to [bend left=45] cycle;
		\draw [style=green dashed, fill={rgb,255: red,149; green,255; blue,179}, opacity=0.5] (33.center)
			 to (34.center)
			 to [bend right=315, looseness=1.25] (35.center)
			 to (36.center)
			 to [bend left=45] (37.center)
			 to (30.center)
			 to [bend left=45, looseness=1.25] (31.center)
			 to (32.center)
			 to [bend left=45] cycle;
	\end{pgfonlayer}
\end{tikzpicture}
    \caption[Illustration of $\sfe$ and $\ife$]{\small{Illustration of 
      $\sfe$ and $\ife$. Vertices are labelled with their value in $\vecf_t$.
    \label{fig:example_hyperedge}}}
\end{figure}

 For this reason, we would expect that the diffusion process based on the operator $\signlap_\graphh$ can be used to find sets with small hypergraph bipartiteness.
However, one needs to be very  cautious here as,   by the nature of the diffusion process,  the values $\vecf_t(v)$ of all the vertices $v$  change over time and, as a result, the sets $\sets_{\vecf}(e)$ and $\set{I}_{\vecf}(e)$ that consist of the vertices with the maximum and minimum $\vecf_t$-value
might change after an \emph{infinitesimal} time step; this   prevents the process from continuing. We     discuss this  issue in detail through the so-called \firstdef{diffusion continuity condition} in Section~\ref{sec:natural_assumption}.
In essence,  the diffusion continuity condition ensures that one  can always    construct a graph $\graphg$ by allocating the weight of each hyperedge $e$ to the edges in $\sets_{\vecf}(e) \times \set{I}_{\vecf}(e)$ such that the sets $\sets_{\vecf}(e)$ and $\set{I}_{\vecf}(e)$ don't change in infinitesimal time although $\vecf_t$ changes according to $(\mathrm{d} \vecf_t)/(\mathrm{d} t) = - \degm_\graphh^{-1} \signlap_\graphg \vecf_t$.
We   also develop an efficient procedure in Section~\ref{sec:natural_assumption}  to compute the weights of edges in $\sfe \times \ife$. All of these guarantee that (i) every graph that corresponds to the hypergraph diffusion process at any time step   can be efficiently constructed;
(ii) with this sequence of constructed graphs, the diffusion process defined by $\signlaph$ is able to continue until the heat distribution converges.
With this, we can summarise the main idea of the new algorithm as follows: 
\begin{itemize}\itemsep -2pt
    \item First of all, we introduce some arbitrary $\vecf_0 \in \R^n$ as the initial diffusion vector, and a step size parameter $\epsilon>0$ to discretise
    the diffusion process. At each step, the algorithm constructs the graph $\graphg$ guaranteed by the diffusion continuity condition,  and updates  $\vecf_t \in \R^n$ according to the rate of change $(\mathrm{d} \vecf_t)/(\mathrm{d}t) = - \degm_\graphh^{-1} \signlapg \vecf_t$. The algorithm terminates when  
     $\vecf_t$ has converged, i.e., the ratio between the current Rayleigh quotient $(\vecf_t^\transpose \signlapg \vecf_t)/(\vecf_t^\transpose \degm_\graphh \vecf_t)$
     and the one in the previous time step is bounded by some predefined constant.
    \item Secondly, similar to many previous spectral graph clustering algorithms~(e.g. \cite{andersenLocalGraphPartitioning2006, takaiHypergraphClusteringBased2020, trevisanMaxCutSmallest2012}),
    the algorithm constructs  the sweep sets defined by $\vecf_t$ and returns the two sets with minimum
    $\bipart_\graphh$-value among all the constructed sweep sets. 
    Specifically, for every $i \in [n]$, the algorithm constructs    $\setl_j = \{v_i : \abs{\vecf_t(v_i)} \geq \abs{\vecf_t(v_j)} \land \vecf_t(v_i) < 0\}$ and $\setr_j = \{v_i : \abs{\vecf_t(v_i)} \geq \abs{\vecf_t(v_j)} \land \vecf_t(v_i) \geq 0\}$.
    Then, among the $n$ pairs $(\setl_j, \setr_j)$, the algorithm returns the one with the minimum $\bipart_\graphh$-value. 
\end{itemize}
  See Algorithm~\ref{algo:main} for the formal description,
and its performance is summarised in  Theorem~\ref{thm:mainalg}.
We prove the theorem in Section~\ref{sec:thm1proof}.

\begin{theorem} \label{thm:mainalg}
Given a hypergraph $\graphh = (\vertexset_\graphh, \edgeset_\graphh, \weight)$ and parameter $\epsilon>0$, the following holds:
\begin{enumerate}\itemsep -2pt
    \item There is an algorithm that finds an eigen-pair ($\lambda, \vecf$) of the operator $\signlaph$ such that $\lambda \leq \lambda_1(\signlapg)$, where $\graphg$ is the clique reduction of $\graphh$ and the inequality is strict if $\min_{e \in \edgeset_\graphh} \rank(e) > 2$. The algorithm has running time $\poly(\cardinality{\vertexset_\graphh}, \cardinality{\edgeset_\graphh} ,1/\epsilon)$.
    \item Given an eigen-pair $(\lambda, \vecf)$ of the operator $\signlaph$, there is an algorithm that constructs the two-sided sweep sets  defined  on $\vecf$, and  finds sets $\setl$ and $\setr$ such that $\bipart_\graphh(\setl, \setr) \leq \sqrt{2 \lambda}$. The algorithm has running time $\poly(\cardinality{\vertexset_\graphh}, \cardinality{\edgeset_\graphh})$.\\
\end{enumerate}
\end{theorem}

\begin{algorithm} \SetAlgoLined
\SetKwInOut{Input}{Input}
\SetKwInOut{Output}{Output}
\Input{Hypergraph $\graphh$, starting vector $\vecf_0\in \R^n$, step size $\epsilon>0$}
\Output{Sets $\setl$ and $\setr$} 
$t := 0$ \\
\While{$\vecf_t$ has not converged}{
    Use $\vecf_t$ to construct graph $\graphg$ satisfying the diffusion continuity condition
     \\
    $\vecf_{t + \epsilon} := \vecf_t - \epsilon \degm_\graphh^{-1} \signlapg \vecf_t $\\
    $t := t + \epsilon$ 
} 
Set $j := \argmin_{1 \leq i \leq n}\bipart_\graphh(\setl_i, \setr_i)$ \\
\Return $(\setl_j, \setr_j)$
\caption[Find densely connected clusters in a hypergraph: \diffalgshortname$(H, \vecf_0, \epsilon)$]{\diffalgname$(H, \vecf_0, \epsilon)$}\label{algo:main}
\end{algorithm}

\subsection{Dealing with the Diffusion Continuity Condition} \label{sec:natural_assumption}
It remains for us to discuss the diffusion continuity condition, which guarantees that $\sfe$ and $\ife$ 
don't change in infinitesimal time and the diffusion process
  eventually converges to some stable distribution.
Formally, let $\vecf_t$ be the normalised measure on the vertices of $\graphh$,
and
\begin{equation*}
\vecr \triangleq \dfdt = -\degm_\graphh^{-1} \signlaph \vecf_t
\end{equation*}
be the derivative of $\vecf_t$, which describes the rate of change   for every vertex at the current time $t$.
We write $\vecr(v)$ for any $v\in \vertexset_\graphh$ as 
$
\vecr(v) = \sum_{e\in \edgeset_\graphh} \vecr_e(v)$,  
where $\vecr_e(v)$ is the contribution of edge $e$ towards the rate of change of $\vecf_t(v)$.
Now we discuss three rules that we expect the diffusion process to satisfy, and later prove that these three rules uniquely define the rate of change $\vecr$.

First of all, as we mentioned in Section~\ref{sec:algorithm},
we expect that only the vertices in $\sfe \cup \ife$   participate in the diffusion process,
i.e., $\vecr_e(u)=0$ unless $u\in \sfe \cup \ife$. Moreover, any vertex $u$  participating in the diffusion process must satisfy the following: 
\begin{center}
\begin{minipage}{0.99\textwidth}
\centering
\begin{itemize}\itemsep -2pt
\item Rule~(0a): if $\abs{\vecr_e(u)}>0$ and $u\in \sfe$, then  $\vecr(u) = \max_{v \in \sfe} \{\vecr(v)\}$. 
\item Rule~(0b): if $\abs{\vecr_e(u)}>0$ and $u\in \ife$, then $\vecr(u) = \min_{v \in \ife} \{\vecr(v)\}$. 
\end{itemize}
\end{minipage}
\end{center}
To explain Rule~(0), notice that for an infinitesimal time, $\vecf_t(u)$ is increased  according to $\left(\mathrm{d}\vecf_t/\mathrm{d}t\right)(u) = \vecr(u)$. Hence, by Rule~(0) we know that, if $u\in \sfe$ (resp.\ $u \in \ife$) participates in the diffusion process in edge $e$, then in an infinitesimal time $\vecf(u)$ remains the maximum (resp.\ minimum) among the vertices in $e$.
Such a rule is necessary to ensure that the vertices involved in the diffusion in edge $e$ do not change in infinitesimal time, and the diffusion process is able to continue.

Our next rule states that the total rate of change of the measure due to edge $e$ is equal to $-\weight(e) \cdot \discrep_{\vecf}(e)$:
\begin{center}
\begin{minipage}{0.99\textwidth}
\centering
\begin{itemize}
\item Rule~(1): $\sum_{v\in \sfe} \deg(v) \vecr_e(v)  = \sum_{v\in \ife} \deg(v) \vecr_e(v) = -\weight(e)\cdot \discrep_{\vecf}(e)$ for all $e \in \edgeset_\graphh$.
\end{itemize}
\end{minipage}
\end{center}
 
This rule is a generalisation from the operator $\signlapg$ in graphs.
In particular, since
\[
\degm_\graphg^{-1} \signlapg \vecf_t(u) =   \sum_{\{u, v\}\in \edgeset_\graphg} \weight_\graphg(u, v) (\vecf_t(u) + \vecf_t(v))/\deg_\graphg(u),
\]
the rate of change of $\vecf_t(u)$ due to the edge $\{u, v\} \in \edgeset_\graphg$ is $- \weight_\graphg(u, v) (\vecf_t(u) + \vecf_t(v))/\deg_\graphg(u)$.
Rule (1) states that in the hypergraph case the rate of change of the vertices in $\sfe$ and $\ife$ together behave like the rate of change of $u$ and $v$ in the graph case.
 
One might have expected that Rules~(0) and (1) together would define a unique process. Unfortunately, this isn't the case.
For example, let us define an unweighted hypergraph $\graphh= (\vertexset_\graphh, \edgeset_\graphh)$, where $\vertexset_\graphh=\{u,v,w\}$ and $\edgeset_\graphh = \{ \{ u,v,w\} \}$. By setting the measure to be $\vecf_t=(1,1,-2)^\transpose$ and $e=\{u,v,w\}$, we have $\discrep_{\vecf_t}(e)=-1$, and $\vecr(u) + \vecr(v) = \vecr(w)= 1$ by Rule (1). In such a scenario, either $\{u,w\}$ or $\{v,w\}$ can participate in the diffusion and satisfy Rule (0), which makes the process not uniquely defined.
To overcome this, we introduce the following stronger rule to replace Rule~(0): 
\begin{center}
\begin{minipage}{0.99\textwidth}
\centering
\begin{itemize}\itemsep -2pt
\item Rule~(2a): Assume that $\abs{\vecr_e(u)}>0$ and $u\in \sfe$:
\begin{itemize}
    \item if $\discrep_{\vecf}(e)>0$, then $\vecr(u) = \max_{v\in \sfe) }\{\vecr(v) \}$;
    \item if $\discrep_{\vecf}(e)<0$, then $\vecr(u)= \vecr(v)$ for all $v\in \sfe$.
\end{itemize}
\item Rule~(2b): Assume that $\abs{\vecr_e(u)}>0$ and $u\in \ife$:
\begin{itemize}
    \item if $\discrep_{\vecf}(e)<0$, then $\vecr(u) = \min_{v\in \ife)} \{\vecr(v)\}$;
    \item if $\discrep_{\vecf}(e)>0$, then $\vecr(u)= \vecr(v)$ for all $v\in \ife$.
\end{itemize}
\end{itemize}
\end{minipage}
\end{center}
Notice that the first conditions of Rules~(2a) and (2b) correspond to Rules (0a) and (0b) respectively; the second conditions are introduced for purely technical reasons: they state that, if the discrepancy of $e$ is negative  (resp.\ positive), then all the vertices $u\in \sfe$ (resp.\ $u \in \ife$)   have the same value of $\vecr(u)$.
Lemma~\ref{lem:rulesimplydiffusion} shows that there is a unique $\vecr\in\R^n$ that satisfies Rules~(1) and (2), and $\vecr$ can be computed in polynomial time. 
Therefore, our two rules uniquely define a diffusion process, and we can use the computed $\vecr$ to simulate the continuous diffusion process with a discretised version.\footnote{Note that the graph $\graphg$ used for the diffusion at time $t$ can be easily computed from the $\{\vecr_e(v)\}$ values, although in practice this is not actually needed since the $\vecr(u)$ values can be used to update the diffusion directly.}

\begin{lemma}[Existence of diffusion process]
     \label{lem:rulesimplydiffusion}
    For any given $\vecf_t\in\R^n$, there is a unique $\vecr=\mathrm{d}\vecf_t/\mathrm{d}t$ and associated $\{\vecr_e(v)\}_{e\in \edgeset, v\in \vertexset}$ that satisfy  Rule~(1) and (2), and  $\vecr$ can be computed in polynomial time by linear programming. 
\end{lemma}

It is worth emphasising that  
our defined rules and the proof of Lemma~\ref{lem:rulesimplydiffusion} are more involved than those used in~\cite{chanSpectralPropertiesHypergraph2018} to define the hypergraph Laplacian operator.
In particular, in contrast to~\cite{chanSpectralPropertiesHypergraph2018}, in our case the discrepancy $\discrep_{\vecf}(e)$ within a hyperedge $e$ can be either positive or negative.
This results in the four different cases in Rule~(2) which have to be carefully considered throughout the proof of Lemma~\ref{lem:rulesimplydiffusion}.

\section{Existence of Diffusion Process} \label{sec:thm2proof}

In this section we prove Lemma~\ref{lem:rulesimplydiffusion} which shows that the diffusion process introduced in Section~\ref{sec:diff_and_alg} is well-defined.
The proof consists of two parts.
First, we construct a linear program which can compute the rate of change $\vecr$ satisfying the rules of the diffusion process.
Then we analyse the new linear program which establishes Lemma~\ref{lem:rulesimplydiffusion}.

\subsection{Computing \texorpdfstring{$\vecr$}{r} by a Linear Program}
Now we present an algorithm that computes the vector $\vecr=\mathrm{d}\vecf_t/\mathrm{d}t$ for any $\vecf_t\in\R^n$. 
Without loss of generality, let us fix $\vecf_t\in\R^n$, and define a set of equivalence classes $\mathcal{U}$ on $\vertexset$ such that vertices $u,v\in \vertexset_\graphh$ are in the same equivalence class if $\vecf_t(u)=\vecf_t(v)$. 
Next we study every equivalence class $\setu \in\mathcal{U}$ in turn, and set the $\vecr$-value of the vertices in $\setu$ recursively.
 In each iteration, we fix the $\vecr$-value of some subset $\setp \subseteq \setu$ and recurse on $\setu \setminus \setp$.
As we prove later, it's important to highlight that the recursive procedure ensures that the $\vecr$-values assigned to the vertices always \emph{decrease} after each recursion.
Notice that it suffices to consider the edges $e$ in which $(\sfte \cup \ifte)\cap \setu \neq\emptyset$, since the diffusion process induced by other edges $e$ has no impact on   $\vecr(u)$ for any $u\in \setu$. Hence, we introduce the sets
\begin{align*}
\mathcal{S}_{\setu} & \triangleq \{e\in \edgeset_\setu: \sfte \intersect \setu \neq\emptyset \} \\
\mathcal{I}_{\setu} & \triangleq \{e\in \edgeset_\setu: \ifte \intersect \setu \neq\emptyset\},
\end{align*}
where $\edgeset_\setu$ consists of the edges adjacent to some vertex in $\setu$. To work with the four cases listed in Rule~(2a) and (2b), we define
\begin{align}
\mathcal{S}_{\setu}^+& \triangleq \{ e\in \cals_\setu: \discrep_{\vecf_t}(e)<0\},\nonumber\\
\cals_{\setu}^- & \triangleq \{ e\in\cals_\setu: \discrep_{\vecf_t}(e)>0\},\nonumber\\
\cali_\setu^+ & \triangleq \{ e\in \cali_\setu: \discrep_{\vecf_t}(e)<0\},\nonumber\\
\cali_\setu^- & \triangleq \{ e\in \cali_\setu: \discrep_{\vecf_t}(e)>0\}.\nonumber
\end{align}
Our objective is to find some  $\setp \subseteq \setu$ and assign the same $\vecr$-value to every vertex in $\setp$.
To this end, for any $\setp \subseteq \setu$ we define
\newcommand{\spp}[2]{\cals_{\set{#1}, \set{#2}}^+}
\newcommand{\spm}[2]{\cals_{\set{#1}, \set{#2}}^-}
\newcommand{\ipp}[2]{\cali_{\set{#1}, \set{#2}}^+}
\newcommand{\ipm}[2]{\cali_{\set{#1}, \set{#2}}^-}
\begin{align*}
\spp{U}{P} & \triangleq \left\{e\in\cals_{\setu}^+: \sfte \subseteq \setp \right\},\\
\ipp{U}{P} & \triangleq 
\left\{e\in\cali_{\setu}^+: \ifte \subseteq \setp \right\}, \\
\spm{U}{P} & \triangleq \left\{ e\in \cals_{\setu}^-: \sfte \cap \setp \neq\emptyset \right\},\\
\ipm{U}{P} &\triangleq \{e \in \cali_{\setu}^-: \ifte \subseteq \setp \}.
\end{align*}
These are the edges contributing to the rate of change of the vertices in $\setp$.
Before continuing the analysis, we briefly explain the intuition behind these four definitions:
\begin{enumerate}
    \item For $\spp{U}{P}$, since every $e\in \cals_{\setu}^+$ satisfies $\discrep_{\vecf_t}(e)<0$ and all the vertices in $\sfte$ must have the same value by Rule~(2a), all such $e\in \spp{U}{P}$ must satisfy that $\sfte \subseteq \setp$, since the unassigned vertices would receive lower values of $\vecr$ in the remaining part of the recursion process.
    \item For $\ipp{U}{P}$, since every $e\in\cali_{\setu}^+$ satisfies $\discrep_{\vecf_t}(e)<0$, Rule (2b) implies that if $\vecr_e(u) \neq 0$ then $\vecr(u)\leq \vecr(v)$ for all $v\in \ifte$. Since unassigned vertices would receive lower values of $\vecr$ later, such $e\in\ipp{U}{P}$ must satisfy $\ifte \subseteq \setp$.
    \item For $\spm{U}{P}$,  since every $e\in\cals_{\setu}^-$ satisfies $\discrep_{\vecf_t}(e)>0$, by Rule~(2a) it suffices that some vertex in $\sfte$ receives the assignment in the current iteration, i.e., every such $e$ must satisfy $\sfte \cap \setp \neq\emptyset$.
    \item The case for $\ipm{U}{P}$ is the same as $\spp{U}{P}$. 
\end{enumerate}

As we expect all the vertices $u\in \sfte$ to have the same $\vecr$-value for every $e$ as long as $\discrep_{\vecf_t}(e)<0$ by Rule~(2a) and at the moment we are only considering the assignment of the vertices in $\setp$,  we expect that 
\begin{equation}\label{eq:condition1}
\left\{ e\in \cals_{\setu}^+\setminus \spp{U}{P}: \sfte \cap \setp \neq \emptyset \right\} =\emptyset,
\end{equation}
and this   ensures that, as long as $\discrep_{\vecf_t}(e)<0$ and some $u\in \sfte$ gets its $\vecr$-value, all the other vertices in $\sfte$ would be assigned the same value as $u$.
Similarly, by Rule~(2b), we expect all the vertices $u\in \ifte$ to have the same $\vecr$-value for every $e$ as long as $\discrep_{\vecf_t}(e)>0$, and so we expect that
\begin{equation}\label{eq:condition2}
\left\{ 
e\in \cali_{\setu}^- \setminus \ipm{U}{P}: \ifte \cap \setp \neq\emptyset
\right\} =\emptyset.
\end{equation}
We   set the $\vecr$-value by dividing the total discrepancy of the edges in $\ipp{U}{P} \union \spp{U}{P} \union \ipm{U}{P} \union \spm{U}{P}$ between the vertices in $\setp$. As such, 
we would like to find some $\setp \subseteq \setu$ that maximises the value of 
\[
\frac{1}{\vol(\setp)}\cdot \left(\sum_{e\in \spp{U}{P}\cup\ipp{U}{P}} c_{\vecf_t}(e)  - \sum_{e\in \spm{U}{P} \cup \ipm{U}{P}} c_{\vecf_t}(e)\right).
\]

Taking all of these requirements into account, we   show that, for any equivalence class $\setu$, we can find the desired set $\setp$ by solving the following linear program:
\begin{alignat}{2}
   & \text{maximise } & & c(\vecx) = \sum_{e \in \cals_{\setu}^+\cup \cali_{\setu}^+ } c_{\vecf_t}(e)\cdot  x_e  - \sum_{e \in \cals_{\setu}^-\cup\cali_{\setu}^-} c_{\vecf_t}(e)\cdot x_e \label{eq:lp} \\
   & \text{subject to }& \quad & \sum_{v \in \setu}
   \begin{aligned}[t]
                \deg(v) y_v & = 1 \\[3ex]
                x_e & = y_u & \quad e & \in \cals_{\setu}^+, u \in \sfte, \\
                x_e & \leq y_u & \quad e & \in \cali_{\setu}^+, u \in \ifte, \\
                x_e & \geq y_u & \quad e & \in \cals_{\setu}^-, u \in \sfte, \\
                x_e & = y_u & \quad e & \in \cali_{\setu}^-, u \in \ifte, \\
                x_e, y_v & \geq 0 & & \forall  v\in \setu,  e\in \edgeset_{\setu}.
   \end{aligned}\nonumber
\end{alignat}
Since the linear program only gives partial assignment to the vertices' $\vecr$-values, we solve the same linear program on the reduced instance given by the set $\setu \setminus \setp$. The formal description of our algorithm is given in Algorithm~\ref{algo:computechangerate}.
\begin{algorithm} \SetAlgoLined
\SetKwInOut{Input}{Input}
\SetKwInOut{Output}{Output}
\Input{vertex set $\setu \subseteq \vertexset$, and edge set $\edgeset_{\setu}$ }
\Output{Values of $\{\vecr(v)\}_{v\in \setu}$}
Construct sets $\cals_{\setu}^+$, $\cals_{\setu}^-$, $\cali_{\setu}^+$, and $\cali_{\setu}^-$\\
Solve the linear program defined by \eqref{eq:lp}, and define $\setp :=\{v\in \setu: y(v) >0\}$\\
Construct sets $\spp{U}{P}$, $\spm{U}{P}$, $\ipp{U}{P}$, and $\ipm{U}{P}$\\
Set $C(\setp) := c_{\vecf_t}\left( \spp{U}{P} \right) + c_{\vecf_t} \left(\ipp{U}{P} \right) - c_{\vecf_t} \left(\spm{U}{P}  \right) - c_{\vecf_t} \left( \ipm{U}{P} \right)$\\
Set $\delta(\setp): = C(\setp) / \vol(\setp)$\\
Set $\vecr(u) :=\delta(\setp)$ for every $u\in \setp$\\
 \algcomputechangerate$\left(\setu \setminus \setp, \edgeset_\setu \setminus \left(\spp{U}{P} \cup \ipp{U}{P} \cup \spm{U}{P} \cup \ipm{U}{P} \right)\right)$
 \caption[Diffusion helper method: \algcomputechangerate$(\setu, \edgeset_\setu)$]{\algcomputechangerate$(\setu, \edgeset_\setu)$}\label{algo:computechangerate}
\end{algorithm}

\subsection{Analysis of the Linear Program}
Now we analyse Algorithm~\ref{algo:computechangerate}, and the properties of the $\vecr$-values it computes.
 Specifically, we show the following facts which   together allow us to establish Lemma~\ref{lem:rulesimplydiffusion}.
\begin{enumerate}
    \item Algorithm~\ref{algo:computechangerate} always produces a unique vector $\vecr$, no matter which optimal result is returned when computing the linear program~\eqref{eq:lp}.
    \item If there is any vector $\vecr$ which is consistent with Rules (1) and (2), then it must be equal to the output of Algorithm~\ref{algo:computechangerate}.
    \item The vector $\vecr$ produced by Algorithm~\ref{algo:computechangerate} is consistent with Rules (1) and (2).
\end{enumerate}
 
\subsubsection{Output of Algorithm~\ref{algo:computechangerate} is Unique}
First of all, for any $\setp \subseteq \setu$ that satisfies \eqref{eq:condition1} and \eqref{eq:condition2}, we define vectors $\vecx_\setp$ and $\vecy_\setp$ by
\[
    \vecx_\setp(e) = \twopartdefow{\frac{1}{\vol(\setp)}}{e \in \spp{U}{P}  \union \ipp{U}{P} \union \spm{U}{P} \union \ipm{U}{P}}{0},
\]
\[
    \vecy_\setp(v) = \twopartdefow{\frac{1}{\vol(\setp)}}{v \in \setp}{0},
\]
and $z_\setp=\left(\vecx_\setp, \vecy_\setp \right)$.
It is easy to verify that $\left(\vecx_\setp, \vecy_\setp \right)$ is a feasible solution to \eqref{eq:lp} with the objective value $c\left(\vecx_\setp \right)=\delta(\setp)$. We   prove that \algcomputechangerate~(Algorithm~\ref{algo:computechangerate}) computes a unique vector $\vecr$ regardless of how ties are broken when computing the subsets $\setp$.

For any feasible solution $z=(\vecx,\vecy)$, we say that a non-empty set $\set{Q}$ is a \emph{level set} of $z$ if there is some $t>0$ such that $\set{Q}=\left\{u\in \setu: y_u\geq t\right\}$. We first show that any non-empty level set of an optimal solution $z$ also corresponds to an optimal solution.

\begin{lemma}
    Suppose that $z^{\star}= \left(\vecx^{\star}, \vecy^{\star} \right)$ is an optimal solution of the linear program \eqref{eq:lp}. Then, any non-empty level set $\set{Q}$ of $z^{\star}$ corresponds to an optimal solution of \eqref{eq:lp} as well.
\end{lemma}
\begin{proof}
Let $\setp = \{v \in \vertexset_\graphh : \vecy^\star(v) > 0 \}$.
The proof is by case distinction.
We first look at the case in which all the vertices in $\setp$ have the same value of  $\vecy^{\star}(v)$ for any $v\in \setp$.
Then, it must be that $z^\star = z_\setp$ and
every non-empty level set $\set{Q}$ of $z^{\star}$ is equal to $\setp$; as such, the statement holds trivially.

Secondly, we assume that the vertices in $\setp$ have at least two different $\vecy^{\star}$-values. We define $\alpha=\min\left\{ y_v^{\star}: v\in \setp \right\}$, and have
\[
\alpha\cdot\vol(\setp) < \sum_{u\in \setu} \deg(u)\cdot y_u^{\star}=1.
\]
We introduce $\widehat{z}=\left( \widehat{\vecx}, \widehat{\vecy} \right)$ defined by
 \begin{align*}
        \widehat{\vecx}(e) & = \twopartdefow{\frac{\vecx^{\star}(e) - \alpha}{1 - \alpha \vol(\setp)}}{\vecx^{\star}(e) \geq 0}{0}, 
\end{align*}
and
\begin{align*}
        \widehat{\vecy}(v) & = \twopartdefow{\frac{\vecy^{\star}(v) - \alpha}{1 - \alpha \vol(\setp)}}{v \in \setp}{0}.
    \end{align*}
This implies that 
\begin{equation}\label{eq:linearcombination1}
\vecx^{\star} = \left( 1-\alpha\vol(\setp) \right)\widehat{\vecx} + \alpha\cdot \constvec_\setp = \left( 1-\alpha\vol(\setp) \right)\widehat{\vecx} + \alpha\cdot \vol(\setp)\cdot \vecx_\setp,
\end{equation}
and
\begin{equation} \label{eq:linearcombination2}
\vecy^{\star} = \left( 1-\alpha\vol(\setp) \right)\widehat{\vecy} + \alpha\cdot \constvec_\setp = \left( 1-\alpha\vol(\setp) \right)\widehat{\vecy} + \alpha\cdot \vol(\setp)\cdot \vecy_\setp,
\end{equation}
where $\constvec_\setp$ is the indicator vector of the set $\setp$.
Notice that 
 $\widehat{z}$ preserves the relative ordering of the vertices and edges with respect to $\vecx^{\star}$ and $\vecy^{\star}$, and all the constraints in \eqref{eq:lp} hold for $\widehat{z}$. These imply that $\widehat{z}$ is a feasible solution to \eqref{eq:lp} as well.
 Moreover, it's not difficult to see that $\widehat{z}$ is an optimal solution of \eqref{eq:lp}, since otherwise by the linearity of \eqref{eq:linearcombination1} and \eqref{eq:linearcombination2} $z_\setp$ would have a higher objective value than $z^{\star}$, contradicting the fact that $z^{\star}$ is an optimal solution. Hence, the non-empty level set  defined by $\widehat{z}$ corresponds to an optimal solution.
Finally, by applying the second case inductively, we prove the claimed statement of the lemma.
\end{proof}

By applying the lemma above and the linearity of the objective function of \eqref{eq:lp}, we obtain the following corollary.
\begin{corollary} \label{cor:lpmaximal}
    The following statements hold:
    \begin{itemize}
        \item Suppose that $\setp_1$ and $\setp_2$ are optimal subsets of $\setu$. Then, $\setp_1\cup \setp_2$, as well as $\setp_1\cap \setp_2$ satisfying $\setp_1\cap \setp_2\neq\emptyset$,  is  an optimal subset of $\setu$. 
        \item The optimal set of maximum size is unique, and contains all optimal subsets.
    \end{itemize}
\end{corollary}

Now we are ready to show that the procedure \algcomputechangerate~(Algorithm~\ref{algo:computechangerate}) and the linear program \eqref{eq:lp} together always give us the same set of $\vecr$-values regardless of which optimal solution of \eqref{eq:lp} is used for the recursive construction of the entire vector $\vecr$.

\begin{lemma}
    Let $(\setu, \edgeset_\setu)$ be the input to \algcomputechangerate, and $\setp \subset \setu$ be the set returned by \eqref{eq:lp}. Moreover, let $\left(\setu'=\setu \setminus \setp, \edgeset_{\setu'}\right)$ be the input to the recursive call \algcomputechangerate$(\setu', \edgeset_{\setu'})$. Then, it holds for any $\setp' \subseteq \setu'$ that $\delta(\setp')\leq \delta(\setp)$, where the equality holds if and only if $\delta(\setp \cup \setp') = \delta(\setp)$.
\end{lemma}
\begin{proof}
By the definition of the function $c$ and sets $\cals^+, \cals^-, \cali^+, \cali^-$, it holds that
\[
c\left( \spp{U'}{P'} \right) = c\left( \spp{U}{ P\cup P'} \right) - c\left( \spp{U}{P} \right), 
\]
and the same equality holds for sets $\cals^-, \cali^+$ and $\cali^-$.
We have that 
    \begin{align*}
        \delta(\setp') & = \frac{c\left(\spp{U'}{P'}\right) + c\left(\ipp{U'}{P'}\right) - c\left(\spm{U'}{P'}\right) - c\left(\ipm{U'}{P'}\right)}{\vol(\setp')} \\
        & = \frac{\delta\left(\setp \union \setp'\right)\cdot \vol(\setp \union \setp') - \delta(\setp) \cdot \vol(\setp)}{\vol(\setp \union \setp') - \vol(\setp)}.
    \end{align*}
    Therefore, it holds  for any operator $\bowtie\ \in \{<, =, >\}$ that
    \begin{alignat*}{2}
        & & \delta(\setp') & \bowtie \delta(\setp) \\
        \iff & \quad & \frac{\delta(\setp \union \setp')\cdot \vol(\setp \union \setp') - \delta(\setp)\cdot \vol(\setp)}{\vol(\setp \union \setp') - \vol(\setp)} & \bowtie \delta(\setp) \\
        \iff & & \delta(\setp \union \setp') & \bowtie \delta(\setp),
    \end{alignat*}
    which implies that  $\delta(\setp') \leq \delta(\setp)$ if and only if $\delta(\setp \union \setp') \leq \delta(\setp)$ with equality if and only if $\delta(\setp \union \setp') = \delta(\setp)$.
    Since $\setp$ is optimal, it cannot be the case that $\delta\left( \setp \cup \setp' \right)>\delta(\setp)$, and therefore the lemma follows.
 \end{proof}
 Combining everything together, we have that
for any input instance $(\setu, \edgeset_\setu)$,  Algorithm~\ref{algo:computechangerate} always returns the same output $\vecr\in\mathbb{R}^{\cardinality{\setu}}$ no matter which optimal sets  are returned by solving the linear program~\eqref{eq:lp}. In particular, Algorithm~\ref{algo:computechangerate} always finds the unique optimal set $\setp \subseteq \setu$ of maximum size, and assigns $\vecr(u) = \delta(\setp)$ to every $u\in \setp$. After removing the computed $\setp \subset \setu$, the computed $\vecr(v)=\delta(\setp')$ for some $\setp'\subseteq \setu \setminus \setp$ and  any $v\in \setp'$ is always strictly less than $\vecr(u)=\delta(\setp)$ for any $u\in \setp$.

\subsubsection{Any $\vecr$ Satisfying Rules (1) and (2) is Computed by Algorithm~\ref{algo:computechangerate}}
 Next we show that if there is any vector $\vecr$ satisfying Rules (1) and (2), it must be equal to the output of Algorithm~\ref{algo:computechangerate}.
\begin{lemma} \label{lem:alg_computes_r}
For any hypergraph $\graphh = (\vertexset_\graphh, \edgeset_\graphh, \weight)$ and $\vecf_t \in \R^n$, if there is a vector $\vecr = \mathrm{d}\vecf_t/\mathrm{d}t$ with an associated $\{\vecr_e(v)\}_{e \in \edgeset_\graphh, v \in \vertexset_\graphh}$ satisfying Rules~(1) and (2), then $\vecr$ is equal to the output of Algorithm~\ref{algo:computechangerate}.
\end{lemma} 
\begin{proof}
We   focus our attention on a single equivalence class $\setu \subset \vertexset$ where for any $u, v \in \setu$, $\vecf(u) = \vecf(v)$.
    Recall that for each $e \in \edgeset_\setu$, $c_{\vecf_t}(e) = \weight(e) \abs{\discrep_{\vecf_t}(e)}$, which is the rate of flow due to $e$ into $\setu$ (if $e \in \cals_{\setu}^+ \union \cali_{\setu}^+$) or out of $\setu$ (if $e \in \cals_{\setu}^- \union \cali_{\setu}^-$).
    Let $\vecr\in\R^n$ be the vector supposed to satisfy Rules~(1) and (2). We assume that $\setu \subseteq \vertexset$ is an arbitrary equivalence class, and define $$\set{T} \triangleq \left\{ u\in \setu: \vecr(u) = \max_{v\in \setu} \vecr(v)\right\}.$$
    Let us study which properties $\vecr$ must satisfy according to Rules~(1) and (2).
\begin{itemize}
    \item Assume that $e\in\cals_\setu^-$, i.e., it holds that $\sfte \cap U \neq\emptyset$ and $\discrep_{\vecf_t}(e)>0$. To satisfy Rule~(2a), it suffices to have that $$c_{\vecf_t}(e)=\weight(e) \cdot  \discrep_{\vecf_t}(e) = - \sum_{v \in \sfe} \deg(v) \vecr_e(v) = - \sum_{v\in 
    \set{T}} \deg(v) \vecr_e(v)$$ if $\sfte \cap \set{T} \neq\emptyset$, and $\vecr_e(v) = 0$ for all $v \in \set{T}$ otherwise. 
    \item Assume that $e\in\cals_\setu^+$, i.e., it holds that $\sfte \cap \setu \neq\emptyset$ and $\discrep_{\vecf_t}(e)<0$. To satisfy Rule~(2a), it suffices to have $\sfte \subseteq \set{T}$, or $\sfte \cap \set{T}=\emptyset$.
    \item Assume that $e\in \cali_\setu^+$, i.e., it holds that $\ife \cap \setu \neq\emptyset$ and $\discrep_{\vecf_t}(e)<0$. To satisfy Rule~(2b), it suffices to have that $$c_{\vecf_t}(e) = \sum_{v\in \ife} \deg(v) \vecr_e(v) = \sum_{v\in \set{T}} \deg(v) \vecr_e(v)$$ if $\ife \subseteq \set{T}$, and $\vecr_e(v) = 0$ for all $v \in \set{T}$ otherwise. 
    \item Assume that $e\in\cali_\setu^-$, i.e., it holds that $\ife \cap \setu \neq\emptyset$ and $\discrep_{\vecf_t}(e)>0$. To satisfy Rule~(2b), it suffices to have $\ife \subseteq \set{T}$, or $\ife \cap \set{T} =\emptyset$.
\end{itemize}
    Notice that the four conditions above needed to satisfy Rule~(2) naturally reflect our definitions of the sets $\spp{U}{P},\ipp{U}{P}, \spm{U}{P}$, and $\ipm{U}{P}$ and for all $u \in \set{T}$, it must be that $\vecr(u) = \delta(\set{T})$.
    
Now we show that the output set $\setp$ returned by solving the linear program~\eqref{eq:lp}
 is the set $\set{T}$. To prove this, notice that on one hand, by  Corollary~\ref{cor:lpmaximal}, the linear program gives us the unique maximal optimal subset $\setp \subseteq \setu$, and every $v\in \setp$ satisfies that $\vecr(v)=\delta(\setp)\leq \vecr(u) = \delta(\set{T})$ for any $u\in \set{T}$ as every vertex in $\set{T}$ has the maximum $\vecr$-value. On the other side, we have that $\delta(\set{T}) \leq \delta(\setp)$ since  $\setp$ is the set returned by the linear program, and therefore  $\set{T}=\setp$. We can apply this argument recursively, and this proves that Algorithm~\ref{algo:computechangerate} must return the vector $\vecr$.
\end{proof}

\subsubsection{Output of Algorithm~\ref{algo:computechangerate} Satisfies Rules (1) and (2)}
 Now we show that the output of Algorithm~\ref{algo:computechangerate} does indeed satisfy Rules (1) and (2) which, together with Lemma~\ref{lem:alg_computes_r}, implies that there is exactly one such vector satisfying the rules.
\begin{lemma} \label{lem:algo_satisfies_rules}
    For any hypergraph $\graphh = (\vertexset_\graphh, \edgeset_\graphh, \weight)$ and vector $\vecf_t \in \R^n$, the vector $\vecr$ constructed by Algorithm~\ref{algo:computechangerate} has corresponding $\{\vecr_e(v)\}_{e \in \edgeset_\graphh, v \in \vertexset_\graphh}$ which satisfies Rules~(1) and (2).
    Moreover, the $\{\vecr_e(v)\}_{e \in \edgeset_\graphh, v \in \vertexset_\graphh}$ values can be computed in polynomial time using the vector $\vecr$.
\end{lemma} 
\begin{proof} 
\newcommand{\setetp}{\mathcal{E}_{\set{T}}^+}
\newcommand{\setetpp}{\mathcal{E}_{\set{T}'}^+}
\newcommand{\setetm}{\mathcal{E}_{\set{T}}^-}
\newcommand{\setetpm}{\mathcal{E}_{\set{T}'}^-}
We focus on a single iteration of the algorithm, in which $r(v)$ is assigned for the vertices in some set $\set{T} \subset \vertexset_\graphh$.
We use the notation \[
\setetp = \ipp{U}{T} \union \spp{U}{T}, \qquad \setetm = \ipm{U}{T} \union \spm{U}{T}
\]
and show that the values of $\vecr_e(v)$ for $e \in \setetp \union \setetm$ can be computed and satisfy Rules (1) and (2). Therefore, by applying this argument to each recursive call of the algorithm, we establish the lemma.
Given the set $\set{T}$, construct the following undirected flow graph, which is illustrated in  Figure~\ref{fig:maxflow}.
\begin{itemize}
    \item The vertex set is $\setetp \union \setetm \union \set{T} \union \{s, t\}$.
    \item For all $e \in \setetp$, there is an edge $(s, e)$ with capacity $c_{\vecf_t}(e)$.
    \item For all $e \in \setetm$, there is an edge $(e, t)$ with capacity $c_{\vecf_t}(e)$.
    \item For all $v \in \set{T}$, if $\delta(\set{T}) \geq 0$, there is an edge $(v, t)$ with capacity $\deg(v) \delta(\set{T})$.
    Otherwise, there is an edge $(s, v)$ with capacity $\deg(v)\abs{\delta(\set{T})}$.
    \item For each $e \in \setetp \union \setetm$, and each $v \in \set{T} \intersect \left(\sfte \union \ifte \right)$, there is an edge $(e, v)$ with  capacity $\infty$.
\end{itemize}

\begin{figure}[t]
    \centering
    \begin{subfigure}{0.49\textwidth}
    \centering
        \begin{tikzpicture}
	\begin{pgfonlayer}{nodelayer}
		\node [style=smallblack] (0) at (-5, 0) {};
		\node [style=square] (1) at (-2, 1.5) {};
		\node [style=square] (2) at (-2, 2.5) {};
		\node [style=square] (4) at (-2, -0.5) {};
		\node [style=none] (5) at (-5.35, -0.25) {s};
		\node [style=basic] (7) at (2, -0.5) {};
		\node [style=basic] (8) at (2, -1.5) {};
		\node [style=basic] (9) at (2, -2.5) {};
		\node [style=square] (12) at (2, 1.5) {};
		\node [style=square] (13) at (2, 2.5) {};
		\node [style=square] (15) at (2, 0.5) {};
		\node [style=none] (17) at (2, -3.75) {$T$};
		\node [style=smallblack] (18) at (5, 0) {};
		\node [style=none] (19) at (5.3, -0.25) {t};
		\node [style=square] (29) at (-2, 0.5) {};
		\node [style=square] (30) at (-2, -1.5) {};
		\node [style=square] (31) at (-2, -2.5) {};
		\node [style=vertical ellipse dashed] (32) at (-2, 0) {};
		\node [style=none] (33) at (-2, 3.75) {};
		\node [style=none] (34) at (-2, 3.75) {$\mathcal{E}_T^+$};
		\node [style=small vertical ellipse dashed] (36) at (2, 1.5) {};
		\node [style=none] (37) at (2, 3.75) {$\mathcal{E}_T^-$};
		\node [style=small vertical ellipse dashed] (38) at (2, -1.5) {};
		\node [style=none] (40) at (0, 0) {$\vdots$};
		\node [style=none] (42) at (-4.15, 1.75) {\small $c(e)$};
		\node [style=none] (44) at (4.375, 1.75) {\small $c(e)$};
		\node [style=none] (46) at (4.375, -2) {\small $d(v)\delta(T)$};
		\node [style=none] (47) at (0.25, 1.75) {\small $\infty$};
	\end{pgfonlayer}
	\begin{pgfonlayer}{edgelayer}
		\draw [style=undirected] (13) to (18);
		\draw [style=undirected] (12) to (18);
		\draw [style=undirected] (15) to (18);
		\draw [style=undirected] (7) to (18);
		\draw [style=undirected] (8) to (18);
		\draw [style=undirected] (9) to (18);
		\draw [style=undirected] (0) to (2);
		\draw [style=undirected] (0) to (1);
		\draw [style=undirected] (0) to (4);
		\draw [style=undirected] (4) to (9);
		\draw [style=undirected] (0) to (29);
		\draw [style=undirected] (0) to (30);
		\draw [style=undirected] (0) to (31);
		\draw [style=undirected] (31) to (8);
		\draw [bend left=60, looseness=1.75] (8) to (15);
		\draw [style=undirected] (2) to (7);
	\end{pgfonlayer}
\end{tikzpicture}
        \caption{}
    \end{subfigure}
    \begin{subfigure}{0.49\textwidth}
    \centering
        \begin{tikzpicture}
	\begin{pgfonlayer}{nodelayer}
		\node [style=smallblack] (0) at (-5, -0.5) {};
		\node [style=square] (1) at (-2, 2) {};
		\node [style=square] (2) at (-2, 3) {};
		\node [style=square] (4) at (-2, -1) {};
		\node [style=none] (5) at (-5.375, -0.75) {s};
		\node [style=basic] (7) at (2, 1) {};
		\node [style=basic] (8) at (2, -2) {};
		\node [style=basic] (9) at (2, -3) {};
		\node [style=square] (12) at (2, 2) {};
		\node [style=square] (13) at (2, 3) {};
		\node [style=square] (15) at (2, -1) {};
		\node [style=none] (17) at (2, -4) {$T'$};
		\node [style=smallblack] (18) at (5, -0.5) {};
		\node [style=none] (19) at (5.3, -0.75) {t};
		\node [style=square] (29) at (-2, 1) {};
		\node [style=square] (30) at (-2, -2) {};
		\node [style=square] (31) at (-2, -3) {};
		\node [style=none] (34) at (-2, -0.25) {$\mathcal{E}_{T'}^+$};
		\node [style=none] (37) at (2, -0.25) {$\mathcal{E}_{T'}^-$};
		\node [style=none] (40) at (0, 2) {$\vdots$};
		\node [style=none] (42) at (-4.675, 1.325) {\small $c(e)$};
		\node [style=none] (44) at (4.2, 1.55) {\small $c(e)$};
		\node [style=none] (46) at (4.575, -2.25) {\small $d(v)\delta(T)$};
		\node [style=tiny vertical ellipse dashed] (47) at (2, -2.5) {};
		\node [style=none] (48) at (-3, 3.25) {};
		\node [style=none] (49) at (-3, 1) {};
		\node [style=none] (50) at (-2.5, 0.5) {};
		\node [style=none] (51) at (2.5, 0.5) {};
		\node [style=none] (52) at (3, 0) {};
		\node [style=none] (53) at (3, -3.25) {};
		\node [style=none] (54) at (0, -1.325) {\small $\infty$};
	\end{pgfonlayer}
	\begin{pgfonlayer}{edgelayer}
		\draw [style=undirected] (13) to (18);
		\draw [style=undirected] (12) to (18);
		\draw [style=undirected] (15) to (18);
		\draw [style=undirected] (7) to (18);
		\draw [style=undirected] (8) to (18);
		\draw [style=undirected] (9) to (18);
		\draw [style=undirected] (0) to (2);
		\draw [style=undirected] (0) to (1);
		\draw [style=undirected] (0) to (4);
		\draw [style=undirected] (4) to (9);
		\draw [style=undirected] (0) to (29);
		\draw [style=undirected] (0) to (30);
		\draw [style=undirected] (0) to (31);
		\draw [style=undirected] (31) to (8);
		\draw [bend left=60, looseness=1.75] (8) to (15);
		\draw [style=blue dashed] (48.center)
			 to (49.center)
			 to [bend right=45] (50.center)
			 to (51.center)
			 to [bend left=45] (52.center)
			 to (53.center);
	\end{pgfonlayer}
\end{tikzpicture}
        \caption{}
    \end{subfigure}
    \caption[Max-flow graph used in proof of Lemma~\ref{lem:algo_satisfies_rules}]{ \textbf{(a)} An illustration of the constructed max-flow graph, when $\delta(\set{T}) \geq 0$.
    The minimum cut is given by $\{s\}$. \textbf{(b)} A cut induced by $\set{T}' \subset \set{T}$. We can assume that every $e \in \setetp \union \setetm$ connected to $\set{T}'$ is on the same side of the cut as $\set{T}'$. Otherwise, there would be an edge with infinite capacity crossing the cut.}
    \label{fig:maxflow}
\end{figure}
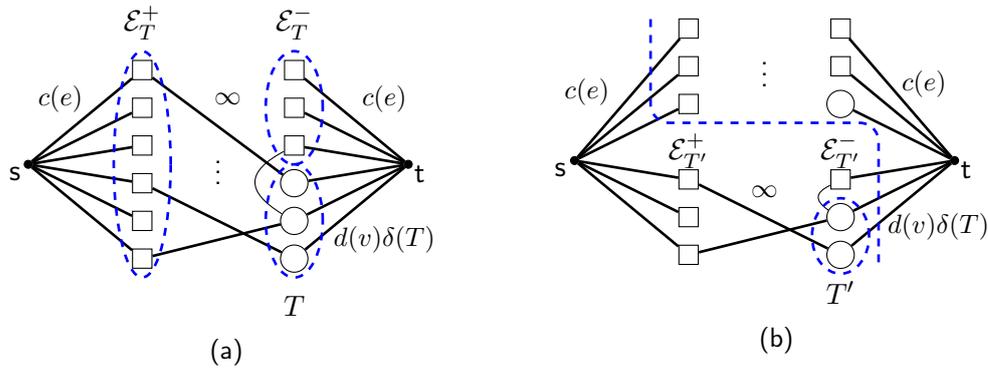

We use $\cut(\seta)$ to denote the weight of the cut defined by the set $\seta$ in this constructed graph, and note that $\cut(\{s\}) = \cut(\{t\})$ since
\begin{align*}
    \cut(\{s\}) - \cut(\{t\}) = \sum_{e \in \setetp} c_{\vecf}(e) - \sum_{e \in \setetm} c_{\vecf}(e) - \vol(\set{T}) \delta(\set{T}) = 0
\end{align*}
by the definition of $\delta(\set{T})$.

Now, suppose that the maximum flow value on this graph is $\cut(\{s\})$, and let the corresponding flow from $u$ to $v$ be given by $\Theta(u, v) = - \Theta(v, u)$.
Then, set $\deg(v) \vecr_e(v) = \Theta(e, v)$ for any $e \in \setetp \union \setetm$ and $v \in \set{T} \intersect e$.
This configuration of the values $\vecr_e(v)$ would be compatible with the vector $\vecr$ computed by Algorithm~\ref{algo:computechangerate}, and would satisfy the rules of the diffusion process for the following reasons:
\begin{itemize}
    \item For all $v \in \set{T}$, the edge $(v, t)$ or $(s, v)$ is saturated and so $$\sum_{e \in \edgeset} \deg(v) \vecr_e(v) = \deg(v) \delta(\set{T}) = \deg(v) \vecr(v)$$.
    \item For any $e \in \setetp \union \setetm$, the edge $(s, e)$ or $(e, t)$ is saturated and so we have $$\sum_{v \in \set{T} \intersect e} \deg(v) \vecr_v(e) = - \weight(e) \discrep(e).$$ Since $\setetp \union \setetm$ is removed in the recursive step of the algorithm, $\vecr_e(v) = 0$ for all $v \in \setu \setminus \set{T}$ and so $\sum_{v \in \setu \intersect e} \deg(v) \vecr_e(v) = - \weight(e) \discrep(e)$. This establishes Rule~(1) since $\setu \intersect e$ is equal to either $\sfe$ or $\ife$.
    \item For edges in $\spp{U}{T}$ (resp.\ $\ipp{U}{T}$, $\ipm{U}{T}$), since $\sfte$ (resp.\ $\ifte$, $\ifte$) is a subset of $\set{T}$ and every $v \in \set{T}$ has the same value of $\vecr(v)$, Rule~(2) is satisfied.
    For edges in $\spm{U}{T}$, for any $v \not \in \set{T}$, we have $\vecr_e(v) = 0$ and $\vecr(v) < \delta(\set{T})$ which satisfies Rule~(2).
\end{itemize}

We now show that every cut separating $s$ and $t$ has weight at least $\cut(\{s\})$ which  would establish that the maximum flow on this graph is $\cut(\{s\})$ by the max-flow min-cut theorem.

Consider some arbitrary cut given by $\set{X} = \{s\} \union \set{T}' \union \setetpp \union \setetpm$ where $\set{T}'$ (resp.\ $\setetpp$, $\setetpm$) is a subset of $\set{T}$ (resp.\ $\setetp$, $\setetm$).
Figure~\ref{fig:maxflow} illustrates this cut.
Since all of the edges not connected to $s$ or $t$ have infinite capacity, we can assume that no such edge crosses the cut which implies that
\begin{itemize}
    \item for all $e \in \setetpp$, $e \intersect (\set{T} \setminus \set{T}') = \emptyset$;
    \item for all $e \in \setetpm$, $e \intersect (\set{T} \setminus \set{T}') = \emptyset$;
    \item for all $e \in (\setetp \setminus \setetpp)$, $e \intersect \set{T}' = \emptyset$;
    \item for all $e \in (\cale_T^- \setminus \setetpm)$, $e \intersect \set{T}' = \emptyset$.
\end{itemize}
These conditions, along with the definition of $\setetp$ and $\setetm$, allow us to assume that $\setetpp = \ipp{U}{T'} \union \spp{U}{T'}$ and $\setetpm = \ipm{U}{T'} \union \spm{U}{T'}$.
The size of this arbitrary cut is
\[
    \cut(\set{X}) = \cut(\{s\}) - \sum_{e \in \setetpp} c(e) + \sum_{e \in \setetpm} c(e) + \sum_{v \in \set{T}'} \deg(v) \delta(\set{T}).
\]
Since $\set{T}$ maximises the objective function $\delta$, we have
\[
    \sum_{e \in \setetpp} c(e) - \sum_{e \in \setetpm} c(e) = \vol(\set{T}') \delta(\set{T}') \leq \vol(\set{T}') \delta(\set{T}) = \sum_{v \in \set{T}'} \deg(v) \delta(\set{T})
\]
and can conclude that $\cut(\set{X}) \geq \cut(\{s\})$. This completes the proof.
\end{proof}

We can now combine the results in Lemmas~\ref{lem:alg_computes_r} and \ref{lem:algo_satisfies_rules} to prove Lemma~\ref{lem:rulesimplydiffusion}.
\begin{proof}[Proof of Lemma~\ref{lem:rulesimplydiffusion}.]
    Lemma~\ref{lem:alg_computes_r} and Lemma~\ref{lem:algo_satisfies_rules} together imply that there is a unique vector $\vecr$ and corresponding $\{\vecr_e(v)\}_{e \in \edgeset_\graphh, v \in \vertexset_\graphh}$ which satisfies Rules~(1) and (2).
    Lemma~\ref{lem:alg_computes_r} further shows that Algorithm~\ref{algo:computechangerate} computes this vector $\vecr$, and the proof of Lemma~\ref{lem:algo_satisfies_rules} gives a polynomial-time algorithm for computing the $\{\vecr_e(v)\}$ values by solving a sequence of max-flow problems.
\end{proof}

\section{Analysis of Algorithm~\ref{algo:main}} \label{sec:thm1proof}
In this section we 
analyse the \diffalgname\ algorithm and
prove Theorem~\ref{thm:mainalg}.
This section consists of two parts corresponding to the two statements in Theorem~\ref{thm:mainalg}.
First, we show that the diffusion process converges to an eigenvector of $\signlaph$.
We then show that this allows us to find sets $\setl, \setr \subset \vertexset_\graphh$ with low hypergraph bipartiteness.

\subsection{Convergence of the Diffusion Process}
We now show that the diffusion process determined by the operator $\signlaph$ converges in polynomial time to an eigenvector of $\signlaph$.

\begin{theorem} \label{thm:convergence}
For any $\epsilon > 0$, there is some $t = \bigo{1 / \epsilon^3}$ such that for any starting vector $\vecf_0$, there is an interval $[c, c + 2\epsilon]$ such that
\[
    \frac{1}{\norm{\vecf_t}_\weight}\sum_{\vec{u}_i : \lambda_i \in [c, c + 2\epsilon]} \inner{\vecf_t}{\vec{u}_i}_\weight \geq 1 - \epsilon,
\]
where $(\vec{u}_i, \lambda_i)$ are the eigen-pairs of $\degm_\graphh^{-1} \signlap_t = \degm_\graphh^{-1}(\degm_{\graphg_t} + \adj_{\graphg_t})$ and $\graphg_t$ is the graph constructed to represent the diffusion operator $\signlaph$ at time $t$.
\end{theorem}
By taking $\epsilon$ to be some small constant, this shows that the vector $\vecf_t$ converges to an eigenvector of the hypergraph operator in polynomial time.  

\begin{proof}
 We show that the Rayleigh quotient $R_{\signlaph}(\vecf_t)$ is always decreasing at a rate of at least $\epsilon^3$ whenever the conclusion of Theorem \ref{thm:convergence} does not hold. Since $R_{\signlaph}(\vecf_t)$ can only decrease by a constant amount, the theorem follows.

    \newcommand{\sumj}{\sum_{j = 1}^n}
    \newcommand{\ajs}{\alpha_j^2}
    \newcommand{\halfdelta}{\frac{\delta}{2}}
    First, we derive an expression for the rate of change of the Rayleigh quotient $R_{\signlaph}(\vecf_t)$.
    Let $\vecx_t \triangleq \degmhalf_\graphh \vecf_{t}$, and   \[\signlapn_t = \degmhalfneg_\graphh (\degm_{\graphg_t} + \adj_{\graphg_t})\degmhalfneg_\graphh.\]  Then, we have
     \[
        R_{\signlap_\graphh}(\vecf_t) = \frac{\vecf_t^\transpose (\degm_{\graphg_t} + \adj_{\graphg_t}) \vecf_t}{\vecf_t^\transpose \degm_\graphh \vecf_t} = \frac{\vecx_t^\transpose \signlapn_t \vecx_t}{\vecx_t^\transpose \vecx_t}.
    \]
    For every eigenvector $\vecu_j$ of $\degm_\graphh^{-1} \signlaph$, the vector $\vecv_j = \degmhalf_\graphh \vecu_j$ is an eigenvector of $\signlapn$ with the same eigenvalue $\lambda_i$.
 Since $\signlapn$ is symmetric, the vectors $\vecv_1, \ldots, \vecv_n$ are orthogonal. 
    Additionally, notice that
    \[
    \inner{\vecx_t}{\vecv_j} = \vecf^\transpose_t \degm_\graphh \vecu_j = \inner{\vecf_{t}}{\vecu_j}_\weight
    \]
    and we define $\alpha_j = \inner{\vecf_t}{\vecu_j}_w$ so we can write $\vecx_t = \sum_{j = 1}^n \alpha_j \vecv_j$ and \[\vecf_{t} = \degmhalfneg_\graphh \vecx_t = \sum_{j = 1}^n \alpha_j \vecu_j.\]
    Now, we have that
    \begin{align*}
        \frac{\mathrm{d}}{\mathrm{d}t} \inner{\vecx_t}{\signlapn_{t} \vecx_t} & = \inner{\frac{\mathrm{d}}{\mathrm{d}t} \vecx_t}{\signlapn_t \vecx_t} + \inner{\vecx_t}{\frac{\mathrm{d}}{\mathrm{d}t} \signlapn_t \vecx_t} \\
        & = - \vecx_t^\transpose \signlapn_t^2 \vecx_t - \vecx_t^\transpose \signlapn_t^2 \vecx_t \\
        & = - 2 \sum_{j = 1}^n \alpha_j^2 \lambda_j^2.
    \end{align*}
    Additionally, we have
    \[
        \frac{\mathrm{d}}{\mathrm{d}t} \vecx_t^\transpose \vecx_t = - \vecx_t^\transpose \signlapn_t \vecx_t - \vecx_t^\transpose \signlapn_t \vecx_t 
        = - 2 \vecx_t^\transpose \signlapn_t \vecx_t.
    \]
    Recalling that $\vecx_t^\transpose \vecx_t = \sum_{j = 1}^n \alpha_j^2$, this gives 
    \newcommand{\dddelta}{\frac{\mathrm{d}}{\mathrm{d} \delta}}
    \begin{align}
        \frac{\mathrm{d}}{\mathrm{d} t} R(\vecf_t) & = \frac{1}{(\vecx_t^\transpose \vecx_t)^2} \left(\left(\frac{\mathrm{d}}{\mathrm{d} t} (\vecx_t^\transpose \signlapn_t \vecx_t) \right) (\vecx_t^\transpose \vecx_t) - \left(\frac{\mathrm{d}}{\mathrm{d}t} (\vecx_t^\transpose \vecx_t) \right) (\vecx_t^\transpose \signlapn_t \vecx_t) \right) \nonumber \\
        & = \frac{1}{\vecx_t^\transpose \vecx_t} \left(\frac{\mathrm{d}}{\mathrm{d}t} (\vecx_t^\transpose \signlapn_t \vecx_t)\right) + 2 R(\vecf_t)^2 \nonumber \\
        & = 2 \left( R(\vecf_t)^2 - \frac{1}{\sum_{j = 1}^n \alpha_j^2} \sumj \ajs \lambda_j^2\right) \label{eq:ddt_gt}
    \end{align}
    We now show that at any time $t$, if the conclusion of the theorem does not hold, then
    \begin{equation} \label{eq:gt_decreases}
        \frac{\mathrm{d}}{\mathrm{d} t} R(\vecf_t) \leq - \epsilon^3.
    \end{equation}
    Assuming that this holds and using the fact that $R(\vecf_0) \leq 2$, when $t = 2 / \epsilon^3$, either $R(\vecf_t) = 0$ or there is some $t' < t$ when $(\mathrm{d}/\mathrm{d}t') R(\vecf_{t'}) > - \epsilon^3$ and the conclusion of the theorem holds.
    
    Now, to show \eqref{eq:gt_decreases}, consider the partial derivative
    \begin{align}
        \frac{\partial}{\partial \lambda_i} 2 \left( R(\vecf_t)^2 - \frac{1}{\sum_{j = 1}^n \alpha_j^2} \sumj \ajs \lambda_j^2 \right) & = 2 \left( \frac{2 \alpha_i^2}{\sum_{j = 1}^n \alpha_j^2} R(\vecf_t) - \frac{2 \alpha_i^2}{\sum_{j = 1}^n \alpha_j^2} \lambda_i \right) \nonumber \\
        & = \frac{4 \alpha_i^2}{\sum_{j = 1}^n \alpha_j^2} (R(\vecf_t) - \lambda_i), \label{eq:partial}
    \end{align}
    where we use the fact that $R(\vecf_t) = (\sum_{j = 1}^n \alpha_j^2 \lambda_j) / (\sum_{j = 1}^n \alpha_j^2)$.
    Notice that
    the derivative in \eqref{eq:partial} is greater than $0$
    if $\lambda_i < R(\vecf_t)$,
    and
    the derivative is less than $0$
    if $\lambda_i > R(\vecf_t)$.
    This means that, in order to establish an upper-bound of $(\mathrm{d}/\mathrm{d} t) R(\vecf_t)$, we can assume that the eigenvalues $\lambda_j$ are as close to the value of $R(\vecf_t)$ as possible.
    
    Now, we assume that at time $t$ the conclusion of the theorem does not hold. Then, one of the following cases must hold: 
    \begin{enumerate}
        \item $(\sum_{j : \lambda_j > R(\vecf_t) + \epsilon} \alpha_j^2) / (\sum_{j = 1}^n \alpha_j^2) > \epsilon / 2$;
        \item $(\sum_{j : \lambda_j < R(\vecf_t) - \epsilon} \alpha_j^2) / (\sum_{j = 1}^n \alpha_j^2) > \epsilon / 2$.
    \end{enumerate}
    Suppose the first case holds.
    By the conclusions we draw from \eqref{eq:partial}, we can assume that there is an eigenvalue $\lambda_i = R(\vecf_t) + \epsilon$ such that $\alpha_i^2 / (\sum_{j = 1}^n \alpha_j^2) = \epsilon / 2$ and that \[\left(\sum_{j: \lambda_j < R(\vecf_t)} \alpha_j^2\right) \Bigg/ \left(\sum_{j = 1}^n \alpha_j^2\right) = 1 - \epsilon / 2.\]
    Then, since \[R(\vecf_t) = \left(\sum_{j = 1}^n \alpha_j^2 \lambda_j\right) \Bigg/ \left(\sum_{j = 1}^n \alpha_j^2\right),\] we have
    \begin{align*}
        \frac{1}{\sum_{j = 1}^n \alpha_j^2} \sum_{j: \lambda_j < R(\vecf_t)} \alpha_j^2 \lambda_j = R(\vecf_t) - \frac{\epsilon}{2}(R(\vecf_t) + \epsilon),
    \end{align*}
    which is equivalent to
    \begin{equation} \label{eq:ljminusgt}
        \frac{1}{\sumj \ajs} \sum_{j: \lambda_j < R(\vecf_t)} \ajs (\lambda_j - R(\vecf_t)) = \frac{\epsilon}{2}\cdot R(\vecf_t) - \frac{\epsilon}{2}\cdot (R(\vecf_t) + \epsilon) = - \frac{\epsilon^2}{2}.
    \end{equation}
    Now, notice that for any $\lambda_j < R(\vecf_t)$ we have
    \begin{align*}
        (\lambda_j^2 - R(\vecf_t)^2) & = (\lambda_j + R(\vecf_t)) (\lambda_j - R(\vecf_t)) \\
        & \geq 2 R(\vecf_t) (\lambda_j - R(\vecf_t)),
    \end{align*}
    since $\lambda_j - R(\vecf_t) < 0$.
    As such, we have
    \begingroup
    \allowdisplaybreaks
    \begin{align*}
        \frac{\mathrm{d}}{\mathrm{d}t} R(\vecf_t) & = 2 \left( R(\vecf_t)^2 - \frac{1}{\sum_{j = 1}^n \alpha_j^2} \sumj \ajs \lambda_j^2\right) \\
        & = - \frac{2}{\sum_{j = 1}^n \alpha_j^2} \sumj \ajs (\lambda_j^2 - R(\vecf_t)^2) \\
        & = - \epsilon \left((R(\vecf_t) + \epsilon)^2 - R(\vecf_t)^2\right) - \frac{2}{\sumj \ajs} \sum_{j: \lambda_j < R(\vecf_t)} \alpha_j^2 (\lambda_j^2 - R(\vecf_t)^2) \\
        & \leq - \epsilon \left(2 \epsilon R(\vecf_t) + \epsilon^2\right) - \frac{2}{\sumj \ajs} \sum_{j: \lambda_j < R(\vecf_t)} 2 \ajs R(\vecf_t) (\lambda_j - R(\vecf_t)) \\
        & = - 2 \epsilon^2 R(\vecf_t) - \epsilon^3 + 2 \epsilon^2 R(\vecf_t) \\
        & = - \epsilon^3,
    \end{align*}
    \endgroup
    where the second last equality follows by \eqref{eq:ljminusgt}.
    We now consider the second case.
    We assume that there is an eigenvalue $\lambda_i = R(\vecf_t) - \epsilon$ such that $\alpha_i^2 / (\sumj \ajs) = \epsilon / 2$, and that $(\sum_{j: \lambda_j > R(\vecf_t)} \ajs) / (\sumj \ajs) = 1 - \epsilon / 2$.
    Then, we have
    \begin{align*}
        \frac{1}{\sum_{j = 1}^n \alpha_j^2} \sum_{j: \lambda_j > R(\vecf_t)} \alpha_j^2 \lambda_j = R(\vecf_t) - \frac{\epsilon}{2}\left(R(\vecf_t) - \epsilon\right),
    \end{align*}
    which is equivalent to
    \begin{equation} \label{eq:ljminusgt2}
        \frac{1}{\sumj \ajs} \sum_{j: \lambda_j > R(\vecf_t)} \ajs (\lambda_j - R(\vecf_t)) = \frac{\epsilon}{2} R(\vecf_t) - \frac{\epsilon}{2}(R(\vecf_t) - \epsilon) = \frac{\epsilon^2}{2}.
    \end{equation}
    Now, notice that for any $\lambda_j > R(\vecf_t)$ we have
    \begin{align*}
        \lambda_j^2 - R(\vecf_t)^2 & = (\lambda_j + R(\vecf_t)) (\lambda_j - R(\vecf_t)) \\
        & \geq 2 R(\vecf_t) \cdot (\lambda_j - R(\vecf_t)).
    \end{align*}
    As such, we have
    \begin{align*}
        \frac{\mathrm{d}}{\mathrm{d} t} R(\vecf_t) & = - \frac{2}{\sum_{j = 1}^n \alpha_j^2} \sumj \ajs (\lambda_j^2 - R(\vecf_t)^2) \\
        & = - \epsilon \left((R(\vecf_t) - \epsilon)^2 - R(\vecf_t)^2\right) - \frac{2}{\sumj \ajs} \sum_{j: \lambda_j > R(\vecf_t)} \alpha_j^2 (\lambda_j^2 - R(\vecf_t)^2) \\
        & \leq - \epsilon \left(\epsilon^2 - 2 \epsilon R(\vecf_t)\right) - \frac{2}{\sumj \ajs} \sum_{j: \lambda_j < R(\vecf_t)} 2 \ajs R(\vecf_t) (\lambda_j - R(\vecf_t)) \\
        & = 2 \epsilon^2 g(t) - \epsilon^3 - 2 \epsilon^2 R(\vecf_t) \\
        & = - \epsilon^3
    \end{align*}
    where the second last equality follows by \eqref{eq:ljminusgt2}.
    These two cases establish \eqref{eq:gt_decreases} and complete the proof of the theorem.
\end{proof}

\subsubsection{Relationship of the Eigenvalue to the Clique Graph}
 We now show that the eigenvalue corresponding to the eigenvector to which the algorithm converges is at most the minimum eigenvalue of $\signlapg$ where $\graphg$ is the clique reduction of $\graphh$.
 We start by showing the following technical lemma.
\begin{lemma} \label{lem:cliqueedge}
    For any hypergraph $\graphh = (\vertexset_\graphh, \edgeset_\graphh, \weight)$, vector $\vecf \in \R^n$, and edge $e \in \edgeset_\graphh$, it holds that
    \[
        \left(\max_{u \in e} \vecf(u) + \min_{v \in e} \vecf(v)\right)^2 \leq \sum_{u, v \in e} \frac{1}{\rank(e) - 1} (\vecf(u) + \vecf(v))^2.
    \]
    The equality holds if and only if there is exactly one vertex $v \in e$ with $\vecf(v) \neq 0$ or $\rank(e) = 2$.
 
\end{lemma}
\begin{proof}
    We consider some ordering of the vertices in $e$,
    \[
    u_1, u_2, \ldots, u_{\rank(e)},
    \]
    such that $u_1 = \argmax_{u \in e} \vecf(u)$ and $u_2 = \argmin_{u \in e} \vecf(u)$ and the remaining vertices are ordered arbitrarily. 
    Then, for any $2 \leq k \leq \rank(e)$, we define
    \[
        C_k = \sum_{u, v \in \{u_1, \ldots, u_k\}} \frac{1}{k - 1} (\vecf(u) + \vecf(v))^2.
    \]
    We show by induction on $k$ that 
    \begin{equation} \label{eq:induction}
    C_k \geq \left(\max_{u \in e} \vecf(u) + \min_{v \in e} \vecf(v)\right)^2
    \end{equation}
    for all $2 \leq k \leq \rank(e)$, with equality if and only if $k = 2$ or there is exactly one vertex $u_i \in e$ with $\vecf(v) \neq 0$.
    The lemma follows by setting $k = \rank(e)$.
 
    The base case of $k = 2$ follows trivially by the definitions and the choice of $u_1$ and $u_2$.

    For the inductive step, we assume that~\eqref{eq:induction} holds for some $k$, and show that it holds for $k + 1$. 
    We have that
    \[
        C_{k + 1} = \sum_{u, v \in \{u_1, \ldots, u_{k+1}\}} \frac{1}{k}\cdot (\vecf(u) + \vecf(v))^2,
    \]
    which is equivalent to
    \begin{align*}
        C_{k + 1} & = \frac{1}{k} \sum_{i = 1}^k (\vecf(u_i) + \vecf(u_{k+1}))^2 + \frac{1}{k} \sum_{u, v \in \{u_1, \ldots, u_k\}} (\vecf(u) + \vecf(v))^2 \\
        & = \frac{1}{k} \sum_{i = 1}^k (\vecf(u_i) + \vecf(u_{k + 1}))^2 + \frac{k - 1}{k} C_k \\
        & \geq \left(1 - \frac{1}{k}\right)\left(\max_{u \in e}\vecf(u) + \min_{v \in e} \vecf(v)\right)^2 + \frac{1}{k} \sum_{i = 1}^k (\vecf(u_1) + \vecf(u_{k + 1}))^2
    \end{align*}
    where the final inequality holds by the induction hypothesis.
    Therefore, it is sufficient to show that
    \[
        \sum_{i = 1}^k \left(\vecf(u_i) + \vecf(u_{k + 1})\right)^2 \geq \left(\max_{u \in e} \vecf(u) + \min_{v \in e} \vecf(v)\right)^2.
    \]
    We instead show the stronger fact that
    \begin{equation} \label{eq:stronginequality}
        \left(\max_{v \in e}\vecf(v) + \vecf(u_{k + 1})\right)^2 +
        \left(\min_{v \in e}\vecf(v) + \vecf(u_{k + 1})\right)^2 
        \geq
        \left(\max_{u \in e}\vecf(u) + \min_{v \in e}\vecf(v)\right)^2.
    \end{equation}
    The proof is by case distinction.
    The first case is when 
    \[
    \sign(\max_{u \in e}\vecf(v)) = \sign(\min_{u \in e}\vecf(u)).
    \] 
    Assume without loss of generality that the sign is positive.
    Then, since $\vecf(u_{k+1}) \geq \min_{v \in e} \vecf(v)$, we have
    \[
        \left(\max_{v \in e} \vecf(v) + \vecf(u_{k + 1})\right)^2 \geq \left(\max_{v \in e} \vecf(v) + \min_{u \in e} \vecf(v)\right)^2
    \]
    and~\eqref{eq:stronginequality} holds.
    Moreover, the inequality is strict if $\abs{\min_{u \in e}\vecf(u)} > 0$ or $\abs{\vecf(u_{k+1})} > 0$.
 
    The second case is when
    \[
    \sign(\min_{u \in e}\vecf(u)) \neq \sign(\max_{v \in e}\vecf(v)).
    \]
    Expanding~\eqref{eq:stronginequality}, we would like to show
    \begin{multline*}
        \left(\max_{u \in e} \vecf(u)\right)^2 + \left(\min_{v \in e}\vecf(v)\right)^2 + 2 \vecf(u_{k + 1})\left(\max_{u \in e}\vecf(u)\right) + 2 \vecf(u_{k + 1})\left(\min_{v \in e}\vecf(v)\right) \\
        + 2 \vecf(u_{k + 1})^2
        \geq \left(\max_{u \in e}\vecf(u)\right)^2 + \left(\min_{u \in e} \vecf(u)\right)^2 - 2\left(\max_{u \in e}\vecf(u)\right)\abs{\min_{u \in e}\vecf(u)},
    \end{multline*}
    which is equivalent to
    \[
        2 \vecf(u_{k + 1})^2 + 2 \vecf(u_{k + 1})\left(\max_{u \in e}\vecf(u)\right) + 2 \vecf(u_{k + 1})\left(\min_{v \in e}\vecf(v)\right) \geq -2 \left(\max_{u \in e}\vecf(u)\right) \abs{\min_{v \in e}\vecf(v)}.
    \]
    Notice that exactly one of the terms on the left hand side is negative.
    Recalling that $\min_{u \in e}\vecf(u) \leq \vecf(u_{k + 1}) \leq \max_{v \in e} \vecf(v)$, it is clear that
    \begin{itemize}
        \item if $\vecf(u_{k+1}) < 0$, then \[-2\left(\max_{u \in e}\vecf(u)\right)\abs{\min_{v \in e}\vecf(v)} \leq 2 \vecf(u_{k + 1})\left(\max_{v \in e}\vecf(v)\right) \leq 0\] and the inequality holds.
        \item if $\vecf(u_{k+1}) \geq 0$, then \[-2\left(\max_{u \in e}\vecf(u)\right)\abs{\min_{v \in e}\vecf(v)} \leq 2 \vecf(u_{k + 1})\left(\min_{u \in e}\vecf(u)\right) \leq 0,\] and the inequality holds.
    \end{itemize}
    Moreover, in both cases the inequality is strict if $-2 \left(\max_{v \in e}\vecf(v)\right) \abs{\min_{u \in e}\vecf(u)}< 0$ or $\abs{\vecf(u_{k_1})} > 0$.
\end{proof}

 Now, we show that one always finds an eigenvector whose corresponding eigenvalue is at most the minimum eigenvalue of the clique reduction.

\begin{lemma} \label{lem:bettereigenvalue}
        For any hypergraph $\graphh = (\vertexset_\graphh, \edgeset_\graphh, \weight)$ with clique reduction $\graphg$, if $\vecf$ is the eigenvector corresponding to $\lambda_1(\degm_\graphg^{-1} \signlapg)$, then
        \[
            \frac{\vecf^\transpose \signlaph \vecf}{\vecf^\transpose \degm_\graphh \vecf} \leq \lambda_1(\degm_\graphg^{-1} \signlapg)
        \]
        and the inequality is strict if $\min_{e \in \edgeset_\graphh} \rank(e) > 2$.
\end{lemma}
\begin{proof}
    Since $\lambda_1(\degm_\graphg^{-1} \signlapg) = (\vecf^\transpose \signlapg \vecf) / (\vecf^\transpose \degm_\graphg \vecf)$ and $(\vecf^\transpose \degm_\graphg \vecf) = (\vecf^\transpose \degm_\graphh \vecf)$ by the construction of the clique graph, it suffices to show that
    \[
        \vecf^\transpose \signlaph \vecf \leq \vecf^\transpose \signlapg \vecf.
    \]
    This is equivalent to
    \begin{align*}
        \sum_{e \in \edgeset_\graphh} \weight(e) \left(\max_{v \in e} \vecf(v) + \min_{u \in e}\vecf(u)\right)^2 & \leq \sum_{(u, v) \in \edgeset_\graphg} \weight_\graphg(u, v) (\vecf(u) + \vecf(v))^2 \\
        & = \sum_{e \in \edgeset_\graphh} \weight(e) \sum_{u, v \in e} \frac{1}{\rank(e) - 1} (\vecf(u) + \vecf(v))^2
    \end{align*}
    which holds by Lemma~\ref{lem:cliqueedge}.
    
    Furthermore, if $\min_{e \in \edgeset_\graphh} \rank(e) > 2$, then by Lemma~\ref{lem:cliqueedge} the inequality is strict unless every edge $e \in \edgeset_\graphh$ contains at most one $v \in e$ with $\vecf(v) \neq 0$.
    Suppose the inequality is not strict, then it must be that
    \begin{align*}
        \lambda_1(\degm_\graphg^{-1} \signlapg) & = \frac{\sum_{(u, v) \in \edgeset_\graphg} \weight_\graphg(u, v)(\vecf(v) + \vecf(u))^2}{\sum_{v \in \vertexset_\graphg} \deg_\graphg(v) \vecf(v)^2} \\
        & =  \frac{\sum_{v \in \vertexset_\graphg} \deg_\graphg(v) \vecf(v)^2}{\sum_{v \in \vertexset_\graphg} \deg_\graphg(v) \vecf(v)^2} \\
        & = 1,
    \end{align*}
    since for every edge $(u, v) \in \edgeset_\graphg$, at most one of $\vecf(u)$ or $\vecf(v)$ is not equal to $0$.
    This cannot be the case, since it is well-known that the maximum eigenvalue $\lambda_n(\degm_\graphg^{-1} \signlapg) = 2$ and so $\sum_{i = 1}^n \lambda_i(\degm_\graphg^{-1} \signlapg) \geq (n - 1) + 2 = n + 1$ which contradicts the fact that the trace $\tr(\degm_\graphg^{-1} \signlapg)$ is equal to $n$. This proves the final statement of the lemma.
\end{proof}

\subsection{Cheeger-type Inequality for Hypergraph Bipartiteness}
 
We now prove some intermediate facts about the new hypergraph operator $\signlaph$.
These allow us to show that the operator $\signlaph$ has a well-defined minimum eigenvector, and the corresponding eigenvalue satisfies a Cheeger-type inequality for hypergraph bipartiteness.
Given a hypergraph $\graphh = (\vertexset_\graphh, \edgeset_\graphh, \weight)$, any edge $e \in \edgeset_\graphh$ and weighted measure vector $\vecf_t$, let $r_e^{\sets} \triangleq \max_{v \in \sfte} \{\vecr(v)\}$ and $r_e^{\set{I}} \triangleq \min_{v \in \ifte} \{\vecr(v)\}$; recall that $c_{\vecf_t}(e) = \weight(e) \abs{\discrep_{\vecf_t}(e)}$.
\begin{lemma} \label{lem:rnorm}
Given a hypergraph $\graphh = (\vertexset_\graphh, \edgeset_\graphh, \weight)$ and normalised measure vector $\vecf_t$, let
\[
    \vecr = \dfdt = - \degm_\graphh^{-1} \signlaph \vecf_t.
\]
Then, it holds that 
    \[
        \norm{\vecr}_\weight^2 = - \sum_{e \in \edgeset} \weight(e) \discrep_{\vecf_t}(e) \left(r_e^{\sets} + r_e^{\set{I}}\right).
    \]
\end{lemma}
\begin{proof}
Let $\setp \subset \vertexset$ be one of the densest vertex sets defined with respect to the solution of the linear program~\eqref{eq:lp}. By the description of Algorithm~\ref{algo:computechangerate}, we have $\vecr(u)=\delta(\setp)$ for every $u\in \setp$, and therefore
\begin{align*}
    \lefteqn{\sum_{u \in \setp} \deg(u) \vecr(u)^2}\\
    & = \vol(\setp)\cdot \delta(\setp)^2 \\
    & = \left(c_{\vecf_t}\left(\spp{U}{P}\right) + c_{\vecf_t}\left(\ipp{U}{P}\right) - c_{\vecf_t}\left(\spm{U}{P}\right) - c_{\vecf_t}\left(\ipm{U}{P}\right)\right)\cdot \delta(\setp) \\
     & = \left( \sum_{e\in\spp{U}{P}} c_{\vecf_t}(e) + \sum_{e\in\ipp{U}{P}} c_{\vecf_t}(e) - \sum_{e\in\spm{U}{P}} c_{\vecf_t}(e) - \sum_{e\in\ipm{U}{P}} c_{\vecf_t}(e) \right)\cdot\delta(\setp)\\
     & =   \sum_{e\in\spp{U}{P}} c_{\vecf_t}(e) \cdot r_e^{\sets} + \sum_{e\in\ipp{U}{P}} c_{\vecf_t}(e)\cdot r_e^{\set{I}} - \sum_{e\in\spm{U}{P}} c_{\vecf_t}(e)\cdot r_e^{\sets} - \sum_{e\in\ipm{U}{P}} c_{\vecf_t}(e)\cdot r_e^{\set{I}}.
\end{align*}
Since each vertex is included in exactly one set $\setp$ and each edge appears either in one each of $\spp{U}{P}$ and $\ipp{U}{P}$ or in one each of $\spm{U}{P}$ and $\ipm{U}{P}$, it holds that 
    \begin{align*}
    \|\vecr\|^2_\weight & = 
        \sum_{v \in \vertexset} \deg(v) \vecr(v)^2\\
        & = \sum_{\setp} \sum_{v\in \setp} \deg(v) \vecr(v)^2\\
        & = \sum_{\setp} \left( \sum_{e\in\spp{U}{P}} c_{\vecf_t}(e) \cdot r_e^\sets + \sum_{e\in\ipp{U}{P}} c_{\vecf_t}(e)\cdot r_e^{\set{I}} - \sum_{e\in\spm{U}{P}} c_{\vecf_t}(e)\cdot r_e^\sets - \sum_{e\in\ipm{U}{P}} c_{\vecf_t}(e)\cdot r_e^{\set{I}}\right) \\
         & = - \sum_{e \in \edgeset} \weight(e) \discrep_{\vecf_t}(e) \left(r_e^\sets + r_e^{\set{I}}\right),
    \end{align*}
    which proves the lemma.
\end{proof}

Next, we define $\gamma_1 = \min_\vecf D(\vecf)$ and show that any vector $\vecf$ that satisfies $\gamma_1 = D(\vecf)$  is an eigenvector of $\signlaph$ with eigenvalue $\gamma_1$. 
We start by showing that the Raleigh quotient of the new operator is equivalent to the discrepancy ratio of the hypergraph.
\begin{lemma} \label{lem:rqdisc}
    For any hypergraph $\graphh$ and vector $\vecf_t \in \R^n$, it holds that $D(\vecf_t) = R_{\signlaph}(\vecf_t)$.
\end{lemma}
\begin{proof}
    Since $\vecf_t^\transpose \degm_\graphh \vecf_t = \sum_{v \in \vertexset} \deg(v) \vecf_t(v)^2$, it is sufficient to show that
    \[
        \vecf_t^\transpose \signlaph \vecf_t = \sum_{e \in \edgeset_\graphh} \weight(e) \left(\max_{u \in e} \vecf_t(u) + \min_{v \in e} \vecf_t(v)\right)^2.
    \]
    Recall that for some graph $\graphg_t$, $\signlaph = \degm_{\graphg_t} + \adj_{\graphg_t}$.
    Then, we have that
    \begin{align*}
        \vecf_t^\transpose \signlaph \vecf_t & = \vecf_t^\transpose (\degm_{\graphg_t} + \adj_{\graphg_t}) \vecf_t \\
        & = \sum_{(u, v) \in \edgeset_\graphg} \weight_\graphg(u, v) (\vecf_t(u) + \vecf_t(v))^2 \\
        & = \sum_{e \in \edgeset_\graphh} \left(\sum_{(u, v) \in \sfte \times \ifte} \weight_{\graphg_t}(u, v) (\vecf_t(u) + \vecf_t(v))^2\right) \\
        & = \sum_{e \in \edgeset_\graphh} \weight(e) \left(\max_{u \in e} \vecf_t(u) + \min_{v \in e} \vecf_t(v)\right)^2, 
    \end{align*}
    which follows since the graph $\graphg_{t}$ is constructed by splitting the weight of each hyperedge $e \in \edgeset_\graphh$ between the edges $\sfte \times \ifte$.
\end{proof}

\begin{lemma} \label{lem:derivatives}
    For a hypergraph $\graphh$, operator $\signlaph$, and vector $\vecf_t$, the following statements hold:
    \begin{enumerate}
        \item $\frac{\mathrm{d}}{\mathrm{d} t} \norm{\vecf_t}_\weight^2 = -2 \vecf_t^\transpose \signlaph \vecf_t$;
        \item $\frac{\mathrm{d}}{\mathrm{d} t} (\vecf_t^\transpose \signlaph \vecf_t) = -2 \norm{\degm_\graphh^{-1} \signlaph \vecf_t}_\weight^2$;
        \item $\frac{\mathrm{d}}{\mathrm{d} t} R(\vecf_t) \leq 0$ with equality if and only if $\degm_\graphh^{-1} \signlaph \vecf_t \in \mathrm{span}(\vecf_t)$.
    \end{enumerate}
\end{lemma}
\begin{proof}
 
By definition, we have that 
\begin{align*}
    \frac{\mathrm{d} \norm{\vecf_t}_\weight^2}{\mathrm{d}t} & = \frac{\mathrm{d}}{\mathrm{d}t} \sum_{v \in \vertexset} \deg(v) \vecf_t(v)^2 \\
    & = \sum_{v \in \vertexset} \deg(v)\cdot \frac{\mathrm{d} \vecf_t(v)^2}{\mathrm{d} t} \\
    & = \sum_{v \in \vertexset} \deg(v)\cdot \frac{\mathrm{d} \vecf_t(v)^2}{\mathrm{d} \vecf_t(v)} \frac{\mathrm{d} \vecf_t(v)}{\mathrm{d}t} \\
    & = 2 \sum_{v \in \vertexset} \deg(v) \vecf_t(v)\cdot \frac{\mathrm{d} \vecf_t(v)}{\mathrm{d}t} \\
    & = 2 \inner{\vecf_t}{\dfdt}_\weight = -2 \inner{\vecf_t}{\degm_\graphh^{-1} \signlaph \vecf_t}_\weight,
\end{align*}
which proves the first statement.

For the second statement, by Lemma~\ref{lem:rqdisc} we have
\[
\vecf_t^\transpose \signlaph \vecf_t = \sum_{e \in \edgeset} \weight(e) \left(\max_{u \in e} \vecf_t(u) + \min_{v \in e} \vecf_t(v)\right)^2,
\]
and therefore 
\begin{align}
    \frac{\mathrm{d}}{\mathrm{d}t} \vecf_t^\transpose \signlaph \vecf_t & = \frac{\mathrm{d}}{\mathrm{d}t} \sum_{e \in \edgeset} \weight(e) \left(\max_{u \in e} \vecf_t(u) + \min_{v \in e} \vecf_t(v) \right)^2 \nonumber\\
    & = 2 \sum_{e \in \edgeset} \weight(e) \discrep_{\vecf_t}(e)\cdot \frac{\mathrm{d}}{\mathrm{d}t} \left(\max_{u \in e} \vecf_t(u) + \min_{v \in e} \vecf_t(v)\right)\nonumber\\
    &= 2 \sum_{e \in \edgeset} \weight(e) \discrep_{\vecf_t}(e)\cdot  \left(r_e^{\sets} + r_e^{\set{I}} \right), \label{eq:calculation1}
\end{align}
where the last equality holds by the way that all the vertices receive their $\vecr$-values by the algorithm and the definitions of $r_e^{\sets}$ and $r_e^{\set{I}}$.  On the other side, 
by definition $\vecr=- \degm_\graphh^{-1} \signlaph \vecf_t$ and so by Lemma~\ref{lem:rnorm},
\begin{equation}\label{eq:calculation2}
\norm{\degm_\graphh^{-1} \signlaph \vecf_t}_\weight^2 = 
\norm{\vecr}_\weight^2 = - \sum_{e\in \edgeset} \weight(e)\discrep_{\vecf_t}(e) \left( r_e^{\sets}+r_e^{\set{I}}
\right).
\end{equation}
By combining \eqref{eq:calculation1} with \eqref{eq:calculation2}, we have the second statement.  

 For the third statement, notice that we can write $\vecf_t^\transpose \signlaph \vecf_t$ as $\inner{\vecf_t}{\degm_\graphh^{-1} \signlaph \vecf_t}_\weight$. Then, we have that
\begin{align*}
\lefteqn{\frac{\mathrm{d}}{\mathrm{d}t} \frac{ \langle \vecf_t, \degm_\graphh^{-1} \signlaph \vecf_t\rangle_\weight}{ \norm{\vecf_t}_\weight^2}}\\ & = \frac{1}{ \norm{\vecf_t}_\weight^2} \cdot \frac{\mathrm{d}~\langle \vecf_t, \degm_\graphh^{-1} \signlaph \vecf_t\rangle_\weight}{\mathrm{d}t}    -  \langle \vecf_t, \degm_\graphh^{-1} \signlaph \vecf_t\rangle_\weight \cdot   \frac{1}{ \norm{\vecf_t}_\weight^4} \frac{\mathrm{d}~\norm{\vecf_t}^2_\weight}{\mathrm{d}t} \\
& = -\frac{1}{\norm{\vecf_t}_\weight^4}\cdot \left(2\cdot \norm{\vecf_t}_\weight^2\cdot  \norm{\degm_\graphh^{-1} \signlaph \vecf_t}_\weight^2 + \left\langle \vecf_t, \degm_\graphh^{-1} \signlaph \vecf_t\right\rangle_\weight  \cdot \frac{\mathrm{d}~\norm{\vecf_t}^2_\weight}{\mathrm{d}t} \right) \\
& = - \frac{2}{\norm{\vecf_t}_\weight^4}\cdot \left( \norm{\vecf_t}_\weight^2\cdot  \norm{\degm_\graphh^{-1} \signlaph \vecf_t}_\weight^2 -\left\langle \vecf_t, \degm_\graphh^{-1} \signlaph \vecf_t\right\rangle_\weight^2  \right)\\
& \leq 0,
\end{align*}
where the last inequality holds by the Cauchy-Schwarz inequality on the inner product $\inner{\cdot}{\cdot}_\weight$ with the equality if and only if $\degm_\graphh^{-1} \signlaph \vecf_t\in\mathrm{span}(\vecf_t)$.
\end{proof}

This allows us to establish the following key lemma.
 
\begin{lemma} \label{lem:eigenvalueexists}
    For any hypergraph $\graphh$, $\gamma_1 = \min_{\vecf} D(\vecf)$ is an eigenvalue of $\signlaph$ and any minimiser $\vecf$ is its corresponding eigenvector.
\end{lemma}
\begin{proof}
By Lemma~\ref{lem:rqdisc}, it holds that $R(\vecf)= D(\vecf)$ for any $\vecf\in\R^n$. When $\vecf$ is a minimiser of $D(\vecf)$, it holds that 
\[
\frac{\mathrm{d} R(\vecf)}{\mathrm{d}t}=0,
\]
which implies by Lemma~\ref{lem:derivatives} that $\degm_\graphh^{-1} \signlaph \vecf\in\mathrm{span}(\vecf)$ and proves that $\vecf$ is an eigenvector.
\end{proof}

We are now ready to prove a Cheeger-type inequality for our operator and the hypergraph bipartiteness.
We first show that the minimum eigenvalue of $\signlaph$ is at most twice the hypergraph bipartiteness.

\begin{lemma} \label{lem:trev_cheeg_easy}
    Given a hypergraph $\graphh = (\vertexset_\graphh, \edgeset_\graphh)$ with sets $\setl, \setr \subset \vertexset_\graphh$ such that $\bipart(\setl, \setr) = \bipart$, it holds that
    \[
        \gamma_1 \leq 2 \bipart
    \]
    where $\gamma_1$ is the smallest eigenvalue of $\signlaph$.
\end{lemma}
\begin{proof}
    Let $\vec{\chi}_{\setl, \setr}\in\{-1,0,1\}^n$ be the indicator vector of the cut $\setl, \setr$ such that
    \[
        \vec{\chi}_{\setl, \setr}(u) = \threepartdefow{1}{u \in \setl}{-1}{u \in \setr}{0}.
    \]
    Then, by Lemma~\ref{lem:rqdisc}, the Rayleigh quotient is given by
    \begin{align*}
        R_{\signlaph}\left(\vec{\chi}_{\setl,\setr}\right) & = \frac{\sum_{e \in \edgeset} \weight(e) \left(\max_{u \in e} \vec{\chi}_{\setl,\setr}(u) + \min_{v \in e} \vec{\chi}_{\setl,\setr}(v)\right)^2}{\sum_{v \in \vertexset} \deg(v) \vec{\chi}_{\setl,\setr}(v)^2} \\
        & = \frac{4 \weight(\setl | \setcomplement{\setl}) + 4 \weight(\setr | \setcomplement{\setr}) + \weight(\setl, \setcomplement{\lur} | \setr) + \weight(\setr, \setcomplement{\lur} | \setl)}{\vol(\lur)} \\
        & \leq 2 \bipart(\setl, \setr),
    \end{align*}
    which completes the proof since $\gamma_1 = \min_{\vecf} R_{\signlaph}(\vecf)$ by Lemma~\ref{lem:eigenvalueexists}.
\end{proof}

The proof of the other direction of the Cheeger-type inequality is more involved, and forms the basis of the second claim in Theorem~\ref{thm:mainalg}.

\begin{lemma} \label{lem:trev-cheeg}
For any hypergraph $\graphh =(\vertexset,\edgeset,\weight)$, let $\signlaph$ be the operator defined with respect to $\graphh$, and $\gamma_1 = \min_{\vecf} D(\vecf) = \min_{\vecf} R(\vecf)$ be the minimum eigenvalue of $\degm_\graphh^{-1} \signlaph$. Then, there are disjoint sets $\setl,\setr \subset \vertexset$ such that
\[
\bipart(\setl,\setr)\leq\sqrt{2\gamma_1}.
\]
\end{lemma}
\begin{proof}
    \newcommand{\maxfu}{a_e}
    \newcommand{\minfu}{b_e}
    Let $\vecf\in\R^n$ be the vector such that $D(\vecf)=\gamma_1$.
    For any threshold $t \in [0, \max_u \vecf(u)^2]$, define $\vecx_t$ such that
    \[
        \vecx_t(u) = \threepartdefow{1}{\vecf(u) \geq \sqrt{t}}{-1}{\vecf(u) \leq -\sqrt{t}}{0}.
    \]
    We show that choosing $t \in \left[0, \max_u \vecf(u)^2\right]$ uniformly at random gives
    \begin{equation} \label{eq:expectation_goal}
        \E\left[\sum_{e \in \edgeset_\graphh} \weight(e) \deltaeabs{\vecx_t} \right] \leq \sqrt{2 \gamma_1}\cdot \E\left[\sum_{v \in \vertexset_\graphh} \deg(v) \abs{\vecx_t(v)} \right].
    \end{equation}
    Notice that every such $\vecx_t$ defines disjoints vertex sets $\setl$ and $\setr$ that satisfy
     \[
        \bipart(\setl, \setr) = \frac{\sum_{e \in \edgeset_\graphh}\weight(e) \deltaeabs{\vecx_t}}{\sum_{v \in \vertexset_\graphh} \deg(v) \abs{\vecx_t(v)}}.
    \]
    Hence, \eqref{eq:expectation_goal} would imply that there is some $\vecx_t$ such that the disjoint $\setl,\setr$ defined by $\vecx_t$ would satisfy
    \[
     \sum_{e \in \edgeset_\graphh}\weight(e) \deltaeabs{\vecx_t}   \leq \sqrt{2 \gamma_1}\cdot  \sum_{v \in \vertexset_\graphh} \deg(v) \abs{\vecx_t(v)},
    \]
    which implies that
    \[
    \bipart(\setl,\setr)\leq\sqrt{2\gamma_1}.
    \]
    Hence, it suffices to prove \eqref{eq:expectation_goal}.
        We assume without loss of generality that $\max_u\left\{\vecf(u)^2\right\}=1$, so $t$ is chosen uniformly from $[0,1]$.
    First of all, we have that
    \begin{align*}
        \E \left[\sum_{v \in \vertexset_\graphh} \deg(v) \abs{\vecx_t(v)} \right] = \sum_{v \in \vertexset_\graphh} \deg(v) \E\left[\abs{\vecx_t(v)}\right] = \sum_{v \in \vertexset_\graphh} \deg(v) \vecf(v)^2.
    \end{align*}
    To analyse the left-hand side of \eqref{eq:expectation_goal}, we focus on a particular edge $e \in \edgeset_\graphh$. Let $\maxfu = \max_{u \in e} \vecf(u)$ and $\minfu = \min_{v \in e} \vecf(v)$.
    We drop the subscript $e$ when it is clear from context.
    We show that
    \begin{equation} \label{eq:edge_expectation}
        \E\left[\deltaeabs{\vecx_t}\right] \leq \abs{\maxfu + \minfu}(\abs{\maxfu} + \abs{\minfu}).
    \end{equation}
    \begin{enumerate}
        \item Suppose $\sign(\maxfu) = \sign(\minfu)$. Our analysis is by case distinction: 
        \begin{itemize}
            \item $\deltaeabs{\vecx_t} = 2$ with probability $\min(\maxfu^2, \minfu^2)$;
            \item $\deltaeabs{\vecx_t} = 1$ with probability $\abs{\maxfu^2 - \minfu^2}$;
            \item $\deltaeabs{\vecx_t} = 0$ with probability $1 - \max(\maxfu^2, \minfu^2)$.
        \end{itemize}
        Assume without loss of generality that $\maxfu^2 = \min\left(\maxfu^2, \minfu^2\right)$. Then, it holds hat 
        \[\E \left[\deltaeabs{\vecx}\right] = 2 \maxfu^2 + \abs{\maxfu^2 - \minfu^2} = \maxfu^2 + \minfu^2 \leq \abs{\maxfu + \minfu}(\abs{\maxfu} + \abs{\minfu}).\]
        \item Suppose $\sign(\maxfu) \neq \sign(\minfu)$. Our analysis is by case distinction: 
        \begin{itemize}
            \item $\deltaeabs{\vecx_t} = 2$ with probability $0$;
            \item $\deltaeabs{\vecx_t} = 1$ with probability $\abs{\maxfu^2 - \minfu^2}$;
            \item $\deltaeabs{\vecx_t} = 0$ with probability $\min(\maxfu^2, \minfu^2)$.
        \end{itemize}
        Assume without loss of generality that $\maxfu^2 = \min\left(\maxfu^2, \minfu^2\right)$. Then, it holds that 
        \begin{align*}
        \E \left[\deltaeabs{\vecx_t}\right] & = \abs{\maxfu^2 - \minfu^2} \\
        & = (\abs{\maxfu} - \abs{\minfu})(\abs{\maxfu} + \abs{\minfu}) \\
        & = \abs{\maxfu + \minfu}(\abs{\maxfu} + \abs{\minfu}),
        \end{align*}
        where the final equality follows because $\maxfu$ and $\minfu$ have different signs.
    \end{enumerate}
    These two cases establish \eqref{eq:edge_expectation}.
    Now, we have that 
    \begin{align*}
        \lefteqn{\E  \left[\sum_{e \in \edgeset_\graphh} \weight(e) \deltaeabs{\vecx_t}\right]}\\
        & \leq \sum_{e \in \edgeset_\graphh} \weight(e) \abs{\maxfu + \minfu}(\abs{\maxfu} + \abs{\minfu}) \\
        & \leq \sqrt{\sum_{e \in \edgeset_\graphh}\weight(e)\abs{\maxfu + \minfu}^2} \sqrt{\sum_{e \in \edgeset_\graphh}\weight(e) (\abs{\maxfu} + \abs{\minfu})^2} \\
        & = \sqrt{\sum_{e \in \edgeset_\graphh} \weight(e) \deltaesquare{\vecf}} \sqrt{\sum_{e \in \edgeset_\graphh}\weight(e) (\abs{\maxfu} + \abs{\minfu})^2}.
    \end{align*}
    By our assumption that $\vecf$ is the eigenvector corresponding to the eigenvalue $\gamma_1$, it holds that
    \[
        \sum_{e \in \edgeset_\graphh}\weight(e) \deltaesquare{\vecf} \leq \gamma_1 \sum_{v \in \vertexset_\graphh}\deg(v) \vecf(v)^2.
    \]
    On the other side, we have that 
    \[
        \sum_{e \in \edgeset_\graphh}\weight(e)(\abs{\maxfu} + \abs{\minfu})^2 \leq 2 \sum_{e \in \edgeset_\graphh}\weight(e)(\abs{\maxfu}^2 + \abs{\minfu}^2) \leq 2 \sum_{v \in \vertexset_\graphh}\deg(v) \vecf(v)^2.
    \]
    This gives us that 
    \begin{align*}
        \E \left[\sum_{e \in \edgeset_\graphh}\weight(e) \deltaeabs{\vecx}\right] & \leq \sqrt{2 \gamma_1} \sum_{v \in \vertexset_\graphh}\deg(v) \vecf(v)^2\\
        & = \sqrt{2 \gamma_1}\cdot  \E \left[\sum_{v \in \vertexset_\graphh}\deg(v) \abs{\vecx(v)}\right],
    \end{align*}
    which proves the statement.
\end{proof}
 Now, we are able to combine these results to prove Theorem~\ref{thm:mainalg}.
\begin{proof}[Proof of Theorem~\ref{thm:mainalg}.]
    The first statement of the theorem follows by setting the starting vector $\vecf_0$ of the diffusion to be the minimum eigenvector of the clique graph $\graphg$. By Lemma~\ref{lem:bettereigenvalue}, we have that $R_{\signlaph}(\vecf_0) \leq \lambda_1(\degm_\graphg^{-1} \signlapg)$ and the inequality is strict if $\min_{e \in \edgeset_\graphh} \rank(e) > 2$.
    Then, Theorem~\ref{thm:convergence} shows that the diffusion process converges to an eigenvector and the Rayleigh quotient only decreases during convergence, and so the inequality holds.
    The algorithm runs in polynomial time, since by Theorem~\ref{thm:convergence} the diffusion process converges in polynomial time and each step of Algorithm~\ref{algo:main} can be computed in polynomial time using a standard algorithm for solving   linear programs.
 
    The second statement is a restatement of Lemma~\ref{lem:trev-cheeg}. The sweep set algorithm runs in polynomial time, since there are $n$ different sweep sets and computing the hypergraph bipartiteness for each one also takes   polynomial time.
\end{proof}

\section{Further Study of the Hypergraph Laplacian Spectrum}
 In this section, we further study the spectrum of the operators $\lap_\graphh$ and $\signlap_\graphh$.
We first demonstrate that there exist infinite families  of hypergraphs $\{H\}$ for which the number of eigenvectors of $\lap_\graphh$ and $\signlap_\graphh$ exceeds the number of vertices.
We then show that finding the eigenvector of $\signlap_\graphh$ corresponding to its smallest eigenvalue is $\NP$-hard.

\subsection{Discussion on the Number of Eigenvectors of \texorpdfstring{$\signlaph$}{Jh} and \texorpdfstring{$\lap_\graphh$}{Lh}} \label{sec:numeigs}
 We now  investigate the spectrum of the non-linear hypergraph Laplacian operator and our new operator $\signlaph$ by considering some example hypergraphs.
In particular, we show that the hypergraph Laplacian $\lap_\graphh$ can have more than $2$ eigenvalues, which answers the following open question posed by Chan~\etal~\cite{chanSpectralPropertiesHypergraph2018}.

\begin{displayquote}[Chan~\etal~\cite{chanSpectralPropertiesHypergraph2018}]
\ldots\ apart from $\vecf_1$, the Laplacian has another eigenvector $\vecf_2$ \ldots\ it is not clear if $\lap$ has any other eigenvalues. We leave as an open problem the task of investigating if other eigenvalues exist.
\end{displayquote}

Furthermore, we show that the number of eigenvectors of the new operator $\signlaph$ can be exponential in the number of vertices.

\subsubsection{\texorpdfstring{$\lap_\graphh$}{Lh} Can Have More Than $2$ Eigenvalues}
Similar to our new operator $\signlaph$, for some vector $\vecf$, the operator $\lap_\graphh$ behaves like the graph Laplacian $\lap_\graphg = \degm_\graphg - \adj_\graphg$ for some graph $\graphg$ constructed by splitting the weight of each hyperedge $e$ between the edges in $\sfe \times \ife$.
We refer the reader to~\cite{chanSpectralPropertiesHypergraph2018} for the full details of this construction.

\begin{figure}[t]
    \centering
    \scalebox{0.8}{
    \begin{tikzpicture}
	\begin{pgfonlayer}{nodelayer}
		\node [style=smallblack] (0) at (-9, 2) {};
		\node [style=smallblack] (1) at (-11.5, -2) {};
		\node [style=smallblack] (2) at (-6.5, -2) {};
		\node [style=none] (3) at (-9, 2.625) {$v_1$};
		\node [style=none] (4) at (-12, -2.5) {$v_2$};
		\node [style=none] (5) at (-6, -2.5) {$v_3$};
		\node [style=none] (6) at (-10, 2) {};
		\node [style=none] (7) at (-9, 3) {};
		\node [style=none] (8) at (-8, 2) {};
		\node [style=none] (9) at (-5.5, -2) {};
		\node [style=none] (11) at (-6.5, -3) {};
		\node [style=none] (12) at (-11.5, -3) {};
		\node [style=none] (13) at (-12.5, -2) {};
		\node [style=none] (14) at (-9, -4) {$H$};
		\node [style=smallblack] (15) at (0, 2) {};
		\node [style=smallblack] (16) at (-2.5, -2) {};
		\node [style=smallblack] (17) at (2.5, -2) {};
		\node [style=none] (18) at (0, 2.625) {$v_1$};
		\node [style=none] (19) at (-3, -2.5) {$v_2$};
		\node [style=none] (20) at (3, -2.5) {$v_3$};
		\node [style=none] (28) at (0, -4) {$G_1$};
		\node [style=smallblack] (29) at (9, 2) {};
		\node [style=smallblack] (30) at (6.5, -2) {};
		\node [style=smallblack] (31) at (11.5, -2) {};
		\node [style=none] (32) at (9, 2.625) {$v_1$};
		\node [style=none] (33) at (6, -2.5) {$v_2$};
		\node [style=none] (34) at (12, -2.5) {$v_3$};
		\node [style=none] (42) at (9, -4) {$G_2$};
		\node [style=smallblack] (43) at (18, 2) {};
		\node [style=smallblack] (44) at (15.5, -2) {};
		\node [style=smallblack] (45) at (20.5, -2) {};
		\node [style=none] (46) at (18, 2.625) {$v_1$};
		\node [style=none] (47) at (15, -2.5) {$v_2$};
		\node [style=none] (48) at (21, -2.5) {$v_3$};
		\node [style=none] (56) at (18, -4) {$G_3$};
		\node [style=none] (57) at (-1.75, 0.5) {$\frac{1}{3}$};
		\node [style=none] (58) at (0, -2.75) {$\frac{1}{3}$};
		\node [style=none] (59) at (1.75, 0.5) {$\frac{1}{3}$};
		\node [style=none] (60) at (7.25, 0.5) {$\frac{1}{2}$};
		\node [style=none] (61) at (9, -2.75) {$\frac{1}{2}$};
		\node [style=none] (62) at (16.25, 0.5) {$1$};
		\node [style=none] (63) at (-4.5, 3.75) {};
		\node [style=none] (64) at (-4.5, -4.5) {};
		\node [style=none] (65) at (4.5, -4.5) {};
		\node [style=none] (66) at (4.5, 3.75) {};
		\node [style=none] (67) at (13.5, 3.75) {};
		\node [style=none] (68) at (13.5, -4.5) {};
	\end{pgfonlayer}
	\begin{pgfonlayer}{edgelayer}
		\draw [style=blue filled] (6.center)
			 to (13.center)
			 to [in=180, out=-120, looseness=1.25] (12.center)
			 to (11.center)
			 to [in=-60, out=0, looseness=1.25] (9.center)
			 to (8.center)
			 to [in=0, out=120] (7.center)
			 to [in=60, out=180] cycle;
		\draw [style=undirected] (15) to (16);
		\draw [style=undirected] (16) to (17);
		\draw [style=undirected] (17) to (15);
		\draw [style=undirected] (29) to (30);
		\draw [style=undirected] (30) to (31);
		\draw [style=undirected] (43) to (44);
		\draw [style=undirected] (63.center) to (64.center);
		\draw [style=undirected] (66.center) to (65.center);
		\draw [style=undirected] (67.center) to (68.center);
	\end{pgfonlayer}
\end{tikzpicture}
    }
\caption[Graphs corresponding to hypergraph Laplacian eigenvalues]{\small{Given the hypergraph $\graphh$, there are three graphs $\graphg_1$, $\graphg_2$, and $\graphg_3$ which correspond to eigenvalues of the hypergraph Laplacian $\lap_\graphh$.}
\label{fig:laplacian_eigenvalues}}
\end{figure}

Now, with reference to Figure~\ref{fig:laplacian_eigenvalues}, we consider the simple hypergraph $\graphh$ where $\vertexset_\graphh = \{v_1, v_2, v_3\}$ and $\edgeset_\graphh = \{\{v_1, v_2, v_3\}\}$.
Letting $\vecf_1 = [1, 1, 1]^\transpose$, notice that the graph $\graphg_1$ in Figure~\ref{fig:laplacian_eigenvalues} is the graph for which $\lap_\graphh \vecf_1$ is equivalent to $\lap_{\graphg_1} \vecf_1$.
In this case,
\[
\renewcommand\arraystretch{1.6}
    \lap_{\graphg_1} = \begin{bmatrix}
        \frac{2}{3} & -\frac{1}{3} & -\frac{1}{3} \\
        -\frac{1}{3} & \frac{2}{3} & -\frac{1}{3} \\
        -\frac{1}{3} & -\frac{1}{3} & \frac{2}{3}
    \end{bmatrix}
\]
and we have
\[
\renewcommand\arraystretch{1.6}
    \lap_\graphh \vecf_1 = \lap_{\graphg_1} \vecf_1 = [0, 0, 0]^\transpose, 
\]
which shows that $\vecf_1$ is the trivial eigenvector of $\lap_\graphh$ with eigenvalue $0$.

Now, consider $\vecf_2 = [1, -2, 1]^\transpose$. In this case, $\graphg_2$ shown in Figure~\ref{fig:laplacian_eigenvalues} is the graph for which $\lap_\graphh \vecf_2$ is equivalent to $\lap_{\graphg_2} \vecf_2$.
Then, we have
\[
\renewcommand\arraystretch{1.6}
    \lap_{\graphg_2} = \begin{bmatrix}
        \frac{1}{2} & -\frac{1}{2} & 0 \\
        -\frac{1}{2} & 1 & -\frac{1}{2} \\
        0 & -\frac{1}{2} & \frac{1}{2}
    \end{bmatrix}
\]
and
\[
    \lap_\graphh \vecf_2 = \lap_{\graphg_2} \vecf_2 = \left[\frac{3}{2}, -3, \frac{3}{2} \right]^\transpose
\]
which shows that $\vecf_2$ is an eigenvector of $\lap_\graphh$ with eigenvalue $3/2$.

Finally, we consider $\vecf_3 = [1, -1, 0]$ and notice that $\graphg_3$ is the graph such that $\lap_\graphh \vecf_3 = \lap_\graphg \vecf_3$. 
Then,
\[
\renewcommand\arraystretch{1.6}
    \lap_{\graphg_3} = \begin{bmatrix}
        \frac{1}{2} & -\frac{1}{2} & 0 \\
        - \frac{1}{2} & \frac{1}{2} & 0 \\
        0 & 0 & 0
    \end{bmatrix}
\]
and
\[
    \lap_\graphh \vecf_3 = \lap_{\graphg_3} \vecf_3 = \left[2, -2, 0 \right]^\transpose
\]
which shows that $\vecf_3$ is an eigenvector of $\lap_\graphh$ with eigenvalue $2$.

Through an exhaustive search of the other possible constructed graphs on $\{v_1, v_2, v_3\}$, we find that these are the only eigenvalues of $\lap_\graphh$.
By the symmetries of $\vecf_2$ and $\vecf_3$, this means that the operator $\lap_\graphh$ has a total of $7$ different eigenvectors and $3$ distinct eigenvalues.
It is useful to point out that, since the $\lap_\graphh$ operator is non-linear, a linear combination of eigenvectors with the same eigenvalue is \emph{not}, in general, an eigenvector.
Additionally, notice that the constructions of $\vecf_1$ and $\vecf_2$ can be generalised to hypergraphs with more than $3$ vertices, demonstrating that the number of eigenvectors of $\lap_\graphh$ can grow faster than the number of vertices.

\subsubsection{\texorpdfstring{$\signlap_\graphh$}{Jh} Can Have an Exponential Number of Eigenvectors}
\begin{figure}[t]
    \centering
    \scalebox{0.9}{
    \scalebox{1.2}{\begin{tikzpicture}
	\begin{pgfonlayer}{nodelayer}
		\node [style=smallblack] (4) at (-5, 2.875) {};
		\node [style=smallblack] (6) at (-5, 1.875) {};
		\node [style=none] (15) at (-5, 1.5) {};
		\node [style=none] (17) at (-5, 3.25) {};
		\node [style=none] (18) at (-9, 2.5) {};
		\node [style=none] (19) at (-9, 2) {};
		\node [style=smallblack] (20) at (-9, 2.25) {};
		\node [style=none] (29) at (-10.75, 0.5) {$K_{\frac{n}{2}}$};
		\node [style=none] (30) at (-3.25, 0.5) {$K_{\frac{n}{2}}$};
		\node [style=none] (31) at (-7, -4.25) {$H$};
		\node [style=smallblack] (90) at (-9, -1.625) {};
		\node [style=smallblack] (91) at (-9, -2.625) {};
		\node [style=none] (92) at (-9, -3) {};
		\node [style=none] (93) at (-9, -1.25) {};
		\node [style=none] (94) at (-5, -2) {};
		\node [style=none] (95) at (-5, -2.5) {};
		\node [style=smallblack] (96) at (-5, -2.25) {};
		\node [style=big vertical ellipse dashed] (101) at (-5, 0) {};
		\node [style=big vertical ellipse dashed] (102) at (-9, 0) {};
		\node [style=none] (103) at (-7, 0.25) {$\vdots$};
		\node [style=smallblack] (104) at (7, 2.875) {};
		\node [style=smallblack] (105) at (7, 1.875) {};
		\node [style=smallblack] (110) at (3, 2.25) {};
		\node [style=none] (111) at (1.25, 0.5) {$K_{\frac{n}{2}}$};
		\node [style=none] (112) at (8.75, 0.5) {$K_{\frac{n}{2}}$};
		\node [style=none] (113) at (5, -4.25) {$G$};
		\node [style=smallblack] (114) at (3, -1.625) {};
		\node [style=smallblack] (115) at (3, -2.625) {};
		\node [style=smallblack] (120) at (7, -2.25) {};
		\node [style=big vertical ellipse dashed] (121) at (7, 0) {};
		\node [style=big vertical ellipse dashed] (122) at (3, 0) {};
		\node [style=none] (123) at (5, 0.25) {$\vdots$};
		\node [style=none] (124) at (-1, 0.5) {$\rightarrow$};
		\node [style=none] (125) at (-9, 0.25) {$\vdots$};
		\node [style=none] (126) at (-5, 0.25) {$\vdots$};
		\node [style=none] (127) at (3, 0.25) {$\vdots$};
		\node [style=none] (128) at (7, 0.25) {$\vdots$};
	\end{pgfonlayer}
	\begin{pgfonlayer}{edgelayer}
		\draw [style=undirected green] (19.center)
			 to [bend left=75, looseness=2.50] (18.center)
			 to [in=150, out=15, looseness=0.25] (17.center)
			 to [bend left=75] (15.center)
			 to [in=-15, out=-150, looseness=0.25] cycle;
		\draw [style=undirected green] (95.center)
			 to [bend right=75, looseness=2.50] (94.center)
			 to [in=30, out=165, looseness=0.25] (93.center)
			 to [bend right=75] (92.center)
			 to [in=-165, out=-30, looseness=0.25] cycle;
		\draw [style=undirected red] (110) to (104);
		\draw [style=undirected red] (115) to (120);
	\end{pgfonlayer}
\end{tikzpicture}}
    }
\caption[Graphs corresponding to hypergraph signless Laplacian eigenvalues]{\small{Given the hypergraph $\graphh$, there are $2^{n/3}$ possible graphs $\graphg$, each of which corresponds to a different eigenvector of the hypergraph operator $\signlaph$.}
\label{fig:new_operator_eigenvalues}}
\end{figure}
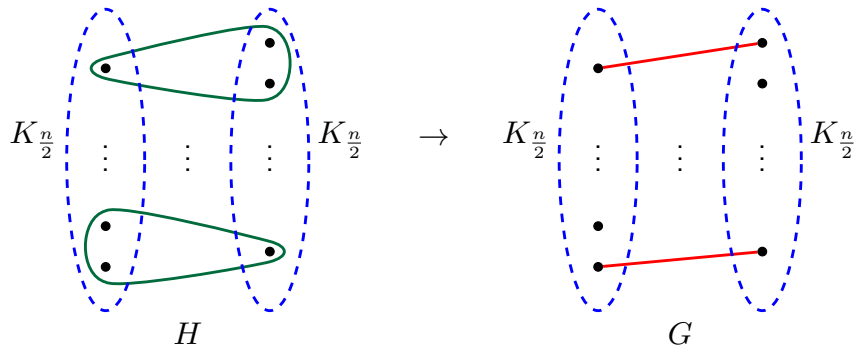
To study the spectrum of our new operator $\signlaph$, we construct a hypergraph $\graphh$ in the following way:
\begin{itemize}
    \item There are $n$ vertices split into two clusters $\setl$ and $\setr$ of size $n/2$. There is a clique on each cluster.
    \item The cliques are joined by $n/3$ edges of rank $3$, such that every vertex is a member of exactly one such edge.
\end{itemize}
Now, we can construct a vector $\vecf$ as follows:
for each edge $e$ of rank $3$ in the hypergraph, let $u \in e$ be the vertex alone in one of the cliques, and $v, w \in e$ be the two vertices in the other clique.
Then, we set $\vecf(u) = 1$ and one of $\vecf(v)$ or $\vecf(w)$ to be $-1$ and the other to be $0$.
Notice that there are $2^{n/3}$ such vectors.
Each one corresponds to a different graph $\graphg$, as illustrated in Figure~\ref{fig:new_operator_eigenvalues}, which is the graph constructed such that $\signlaph$ is equivalent to $\signlapg$ when applied to the vector $\vecf$.

Notice that, by the construction of the graph $\graphh$, half of the edges of rank $3$ must have one vertex in $\setl$ and two vertices in $\setr$, and half of the rank-$3$ edges must have two vertices in $\setl$ and one vertex in $\setr$.
This means that, within each cluster $\setl$ or $\setr$, one third of the vertices have $\vecf$-value $1$, one third have $\vecf$-value $-1$, and one third have $\vecf$-value $0$.

Now, we have that
\begin{align*}
    \left(\degm_\graphh^{-1} \signlaph \vecf\right)(u)  = \left(\degm_\graphh^{-1} \signlapg \vecf\right)(u)  = \frac{2}{n} \sum_{\{u, v\} \in \edgeset_\graphg} (\vecf(u) + \vecf(v)).
\end{align*}
Suppose that $\vecf(u) = 0$, meaning that it does not have an adjacent edge in $\graphg$ from its adjacent rank-$3$ edge in $\graphh$.
In this case,
\begin{align*}
    \left(\degm_\graphh^{-1} \signlaph \vecf\right)(u) & = \frac{2}{n} \sum_{u \sim_\graphg v} \vecf(v) \\
    & = \frac{2}{n} \left[ \frac{n}{3 \cdot 2} \cdot 1 + \frac{n}{3 \cdot 2} \cdot (-1) \right] \\
    & = 0.
\end{align*}
Now, suppose that $\vecf(u) = 1$, and so it has an adjacent edge in $\graphg$ from its adjacent rank-$3$ edge in $\graphh$.
Then,
\begin{align*}
    \left(\degm_\graphh^{-1} \signlaph \vecf\right)(u) & = \frac{2}{n} \sum_{u \sim_\graphg v} (1 + \vecf(v)) \\
    & = \frac{2}{n} \left[ \frac{n}{2} + \frac{n}{3 \cdot 2} \cdot (-1) + \left(\frac{n}{3 \cdot 2} - 1\right) \cdot 1 - 1 \right] \\
    & = 1 - \frac{4}{n}.
\end{align*}
Similarly, if $\vecf(u) = -1$, then we find that
\[
    \left(\degm_\graphh^{-1} \signlaph \vecf\right)(u) = \frac{4}{n} - 1,
\]
and so we can conclude that $\vecf$ is an eigenvector of the operator $\signlaph$ with eigenvalue $(n - 4)/n$.
Since there are $2^{n/3}$ possible such vectors $\vecf$, there are an exponential number of eigenvectors with this eigenvalue.
Once again, we highlight that due to the non-linearity of $\signlaph$, a linear combination of these eigenvectors is in general not an eigenvector of the operator.

\subsection{Hardness of Finding the Minimum Eigenvalue of \texorpdfstring{$\signlaph$}{J\_H}} 
The second statement of Theorem~\ref{thm:mainalg} answers an open question posed by Yoshida~\cite{yoshidaCheegerInequalitiesSubmodular2019} by showing that our constructed hypergraph operator satisfies a Cheeger-type inequality for hypergraph bipartiteness.
However, the following theorem shows that there is no polynomial-time algorithm to compute the eigenvector corresponding to the minimum eigenvalue of \emph{any such operator} unless $\P=\NP$.
This means that the problem we consider in this work is fundamentally more difficult than the equivalent problem for graphs, as well as the problem of finding a sparse cut in a hypergraph.

\begin{theorem} \label{thm:min_eigvec_np}
     Given any operator that satisfies a Cheeger-type inequality for hypergraph bipartiteness, there is no polynomial-time algorithm that computes a multiplicative-factor approximation of the minimum eigenvalue or eigenvector unless $\P=\NP$.
\end{theorem}
\begin{proof}
The proof is based on considering the following \textsc{Max Set Splitting} problem:  For any given hypergraph, the problem looks for a partition $\setl,\setr$ with $\setl \union \setr=\vertexset_\graphh$ and $\setl \intersect \setr=\emptyset$, such that it holds for any $e\in \edgeset_\graphh$ that  $e\cap \setl \neq\emptyset$ and $e\cap \setr \neq\emptyset$. This problem is known to be $\NP$-complete~\cite{gareyComputersIntractability1979}.
This is also referred to as \textsc{Hypergraph 2-Colourability} and we can consider the problem of colouring every vertex in the hypergraph either red or blue such that every edge contains at least one vertex of each colour.

We assume that there is some operator $\signlap$ which satisfies the Cheeger-type inequality
\begin{equation} \label{eq:cheeg_ineq_proof}
    \frac{\lambda_1(\signlap)}{2} \leq \bipart_H \leq \sqrt{2 \lambda_1(\signlap)}
\end{equation}
and that there is an algorithm which can compute the minimum eigen-pair of the operator in polynomial time.
We   show that this would allow us to solve the \textsc{Max Set Splitting} problem in polynomial time.

Given a 2-colourable hypergraph $\graphh$ with colouring $(\setl, \setr)$, we   use the eigenvector of the operator $\signlaph$ to find a valid colouring.
By definition, we have that $\bipart(\setl, \setr) = 0$ and $\gamma_1 = 0$ by \eqref{eq:cheeg_ineq_proof}.
Furthermore, by the second statement of Theorem~\ref{thm:mainalg}, we can compute disjoint sets $\setl', \setr'$ such that $\bipart(\setl', \setr') = 0$. Note that in general $\setl'$ and $\setr'$ are different from $\setl$ and $\setr$.

Now, let $\sete' = \{e \in \edgeset_\graphh : e \intersect (\setl' \union \setr') \neq \emptyset\}$.
Then, by the definition of bipartiteness, for all $e \in \sete'$ we have $e \intersect \setl' \neq \emptyset$ and $e \intersect \setr' \neq \emptyset$.
Therefore, if we colour every vertex in $\setl'$ blue and every vertex in $\setr'$ red, then every $e \in \sete'$ is satisfied. That is, they contain at least one vertex of each colour.

It remains to colour the vertices in $\vertexset \setminus (\setl' \union \setr')$ such that the edges in $\edgeset_\graphh \setminus \sete'$ are satisfied.
The hypergraph $\graphh' = (\vertexset \setminus (\setl' \union \setr'), \edgeset_\graphh \setminus \sete')$ must also be 2-colourable, and so we can recursively apply our algorithm until every vertex is coloured.
This algorithm   runs in polynomial time, since there are at most $\bigo{n}$ iterations and each iteration runs in polynomial time by our assumption.
\end{proof}

\section{Experimental Results} \label{sec:experiments}
In this section, we evaluate the performance of Algorithm~\ref{algo:main} on synthetic and real-world datasets. 
All algorithms are implemented in Python 3.6, using the \texttt{scipy} library for sparse matrix representations and linear programs. The experiments are performed using an Intel(R) Core(TM) i5-8500 CPU @ 3.00GHz processor, with 16 GB RAM. Code to reproduce these experiments can be downloaded from 
\begin{center}
\href{https://github.com/pmacg/hypergraph-bipartite-components}{https://github.com/pmacg/hypergraph-bipartite-components}.
\end{center}

Since this is the first proposed algorithm for approximating hypergraph bipartiteness, we compare it to a simple and natural baseline algorithm, which we call \algcliquecut~(\algcc).
In this algorithm, we construct the clique reduction of the hypergraph, and 
use the two-sided sweep-set algorithm described in \cite{trevisanMaxCutSmallest2012} to find a set with low bipartiteness in the clique reduction.\footnote{We choose to use the algorithm in \cite{trevisanMaxCutSmallest2012} here since, as far as we know, this is the only non-\textsf{SDP} based algorithm for solving the MAX-CUT problem for graphs. Notice that, although SDP-based algorithms achieve a better approximation ratio for the MAX-CUT problem, they are not practical even for hypergraphs of medium sizes.}

Additionally, we compare two versions of our proposed algorithm.  \diffalgname\ (\diffalgshortname) is the new algorithm described in Algorithm~\ref{algo:main} and \diffalgapproxname\ (\diffalgapproxshortname) is an approximate version in which we do not solve the linear programs in Lemma~\ref{lem:rulesimplydiffusion} to compute the graph $\graphg$.
Instead, at each step of the algorithm, we construct $\graphg$ by splitting the weight of each hyperedge $e$ evenly between the edges in $\sets(e) \times \set{I}(e)$.

We always set the parameter $\epsilon = 1$ for \diffalgshortname\ and \diffalgapproxshortname, and   set the starting vector $\vecf_0 \in\R^n$ for the diffusion to be the eigenvector corresponding to the minimum eigenvalue of $\signlapg$, where $\graphg$ is the clique reduction of the hypergraph $\graphh$.

\subsection{Synthetic Datasets}
We first evaluate the algorithms using a random hypergraph model.
Given the parameters $n$, $r$, $p$, and $q$, we generate an $n$-vertex $r$-uniform hypergraph in the following way:
the vertex set $\vertexset$ is divided into two clusters $\setl$ and $\setr$ of size $n/2$.
For every set $\sets \subset \vertexset$ with $\cardinality{\sets} = r$, we add the hyperedge $\sets$ with probability $p$ if $\sets \subset \setl$ or $\sets \subset \setr$, and   we add the hyperedge with probability $q$ otherwise.
 We remark that this is a special case of the hypergraph stochastic block model (e.g., \cite{chienCommunityDetectionHypergraphs2018}). We limit the number of free parameters for simplicity while maintaining enough flexibility to generate random hypergraphs with a wide range of optimal $\bipart_\graphh$-values.

We compare the algorithms' performance using four metrics: the hypergraph bipartiteness $\bipart_\graphh(\setl, \setr)$, the clique graph bipartiteness $\bipart_\graphg(\setl, \setr)$, the $F_1$-score~\cite{vanrijsbergenGeometryInformationRetrieval2004} of the returned clustering, and the runtime of the algorithm.
We always report the average result on $10$ hypergraphs randomly generated with each parameter configuration.

\paragraph{Comparison of \diffalgshortname\ and \diffalgapproxshortname.}
We first fix the values $n = 200$, $r = 3$,   $p = 10^{-4}$, and vary the ratio of $q/p$ from $2$ to $6$ which produces hypergraphs with $250$ to $650$ edges.
The performance of each algorithm on these hypergraphs is shown in Figure~\ref{fig:synthetic_rank3}, from which we can make the following observations.
\begin{itemize}
    \item From Figure~\ref{fig:synthetic_rank3}(a) we observe that \diffalgshortname\ and \diffalgapproxshortname\ find sets with very similar bipartiteness and they perform better than the \algcliquecut\ baseline.
    \item Figure~\ref{fig:synthetic_rank3}(b) shows that the \diffalgapproxshortname\ algorithm is much faster than \diffalgshortname.
\end{itemize}
From this, we conclude that in practice it is sufficient to use the much faster \diffalgapproxshortname\ algorithm in place of the \diffalgshortname\ algorithm.

\begin{figure}[t]
\captionsetup[subfigure]{justification=centering}
\centering
    \begin{subfigure}{0.45\textwidth}
        \includegraphics[width=\textwidth]{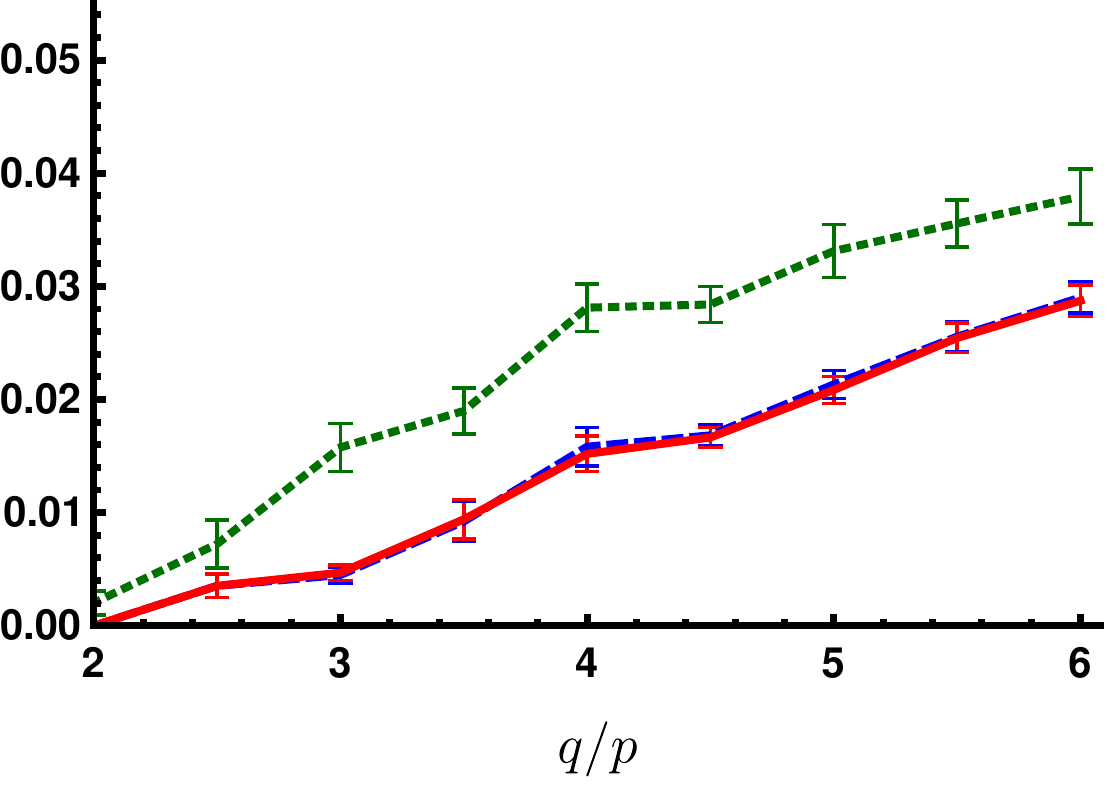}
        \caption{\centering $\bipart_\graphh$-value}
    \end{subfigure}
    \hspace{1em}
    \begin{subfigure}{0.45\textwidth}
        \includegraphics[width=\textwidth]{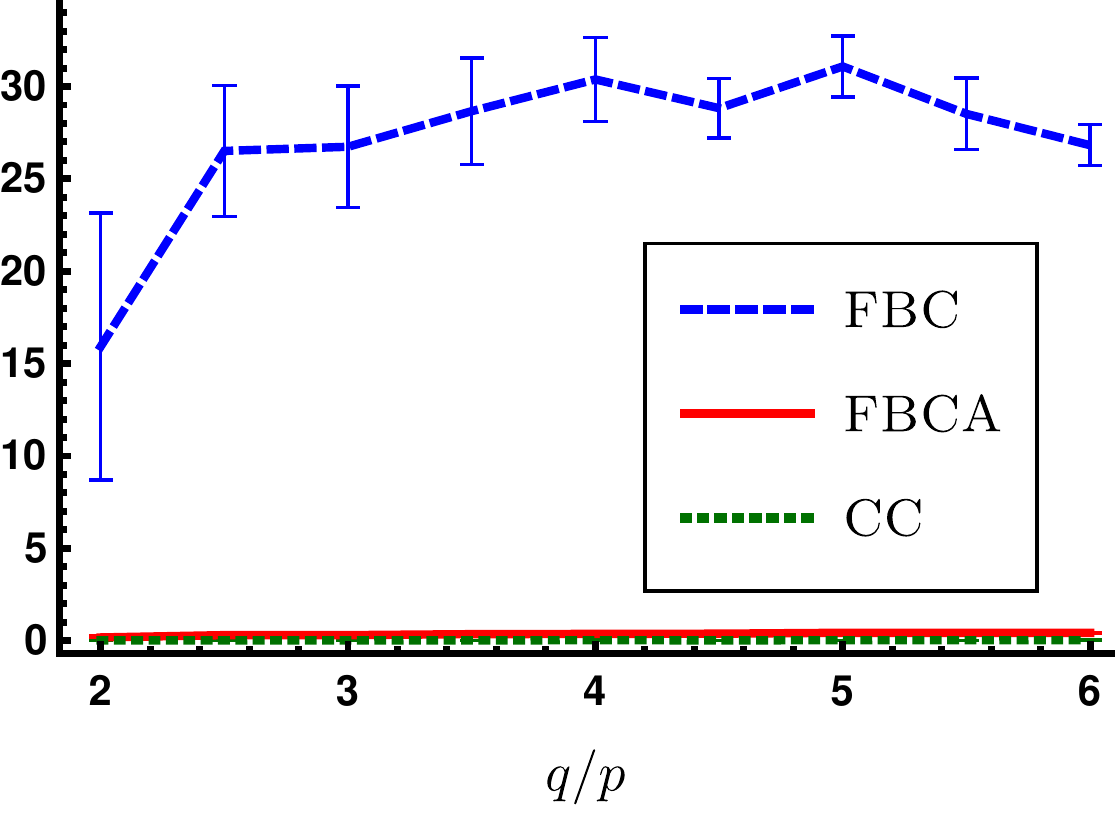}
        \caption{Runtime}
    \end{subfigure}
    \caption[Experimental results of hypergraph partitioning on synthetic data]{The average performance and standard error of each algorithm, when $n = 200$, $r = 3$ and $p = 10^{-4}$. 
    \label{fig:synthetic_rank3}
    }
\end{figure}

\paragraph{Experiments on larger graphs.}

We now compare only the \diffalgapproxshortname\ and \algcliquecut\ algorithms, which allows us to run on hypergraphs with higher rank and number of vertices.
We fix the parameters $n = 2,000$, $r = 5$, and $p = 10^{-11}$, producing hypergraphs with between $5,000$ and $75,000$ edges\footnote{In the random hypergraph model, small values of $p$ and $q$ are needed, since in an $n$-vertex, $r$-uniform hypergraph there are $\binom{n}{r}$ possible edges which can be a very large number. In this case, $\binom{2000}{5} \approx 2.6 \times 10^{14}$.} and show the results in Figure~\ref{fig:synthetic_rank5}.
The new algorithm consistently and significantly outperforms the baseline on every metric and across a wide variety of input hypergraphs.

\begin{figure}[t]
\centering
    \begin{subfigure}{0.45\textwidth}
        \includegraphics[width=\textwidth]{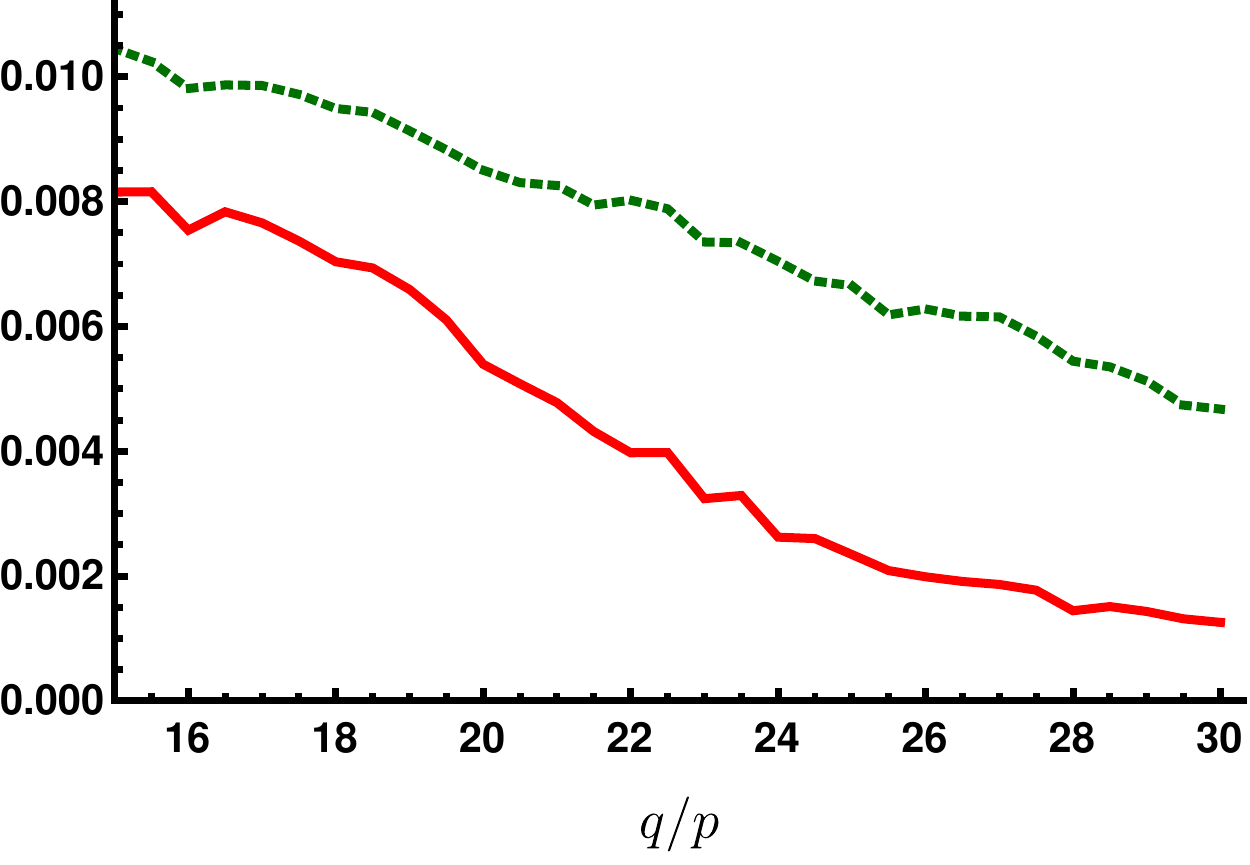}
        \caption{$\beta_H$-value}
    \end{subfigure}
    \hspace{1em}
    \begin{subfigure}{0.45\textwidth}
        \includegraphics[width=\textwidth]{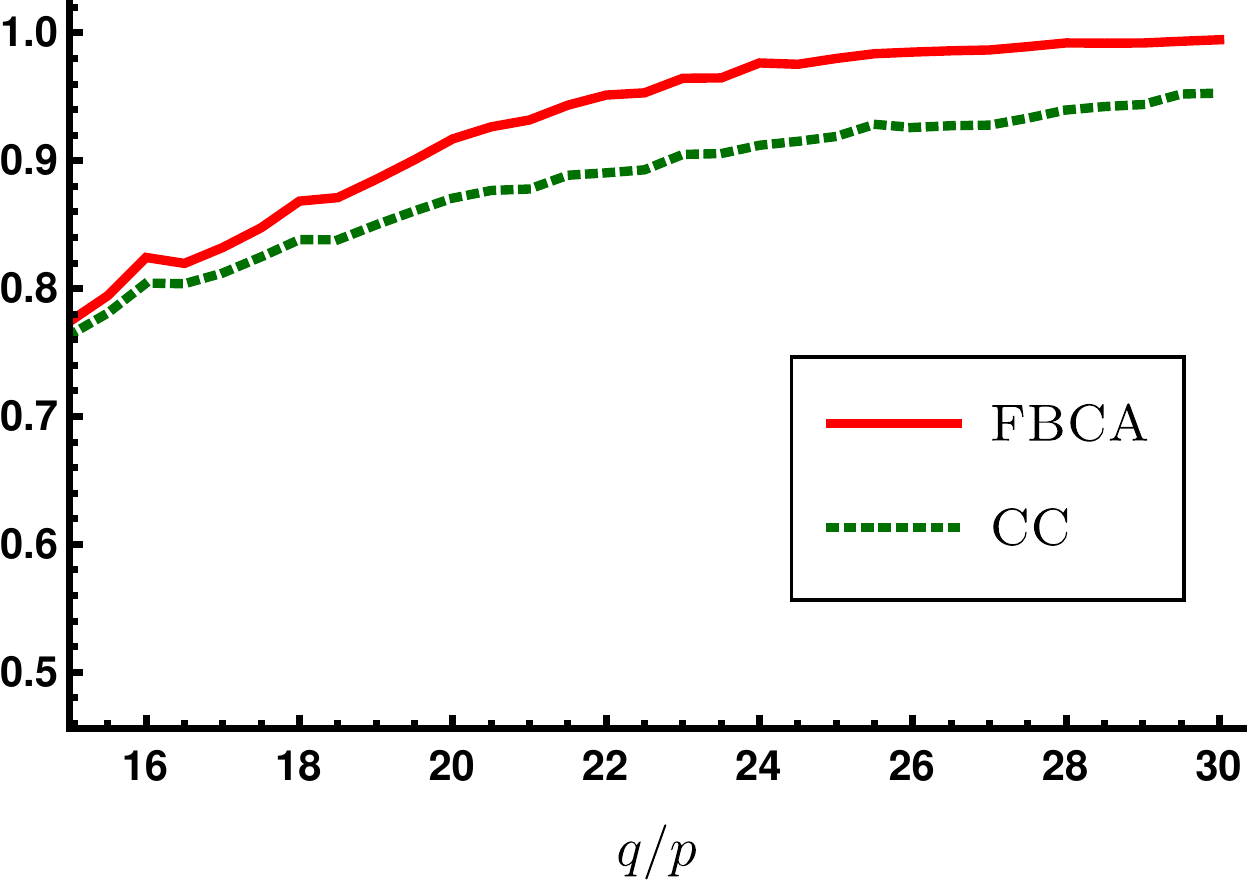}
        \caption{$F_1$-score}
    \end{subfigure}
    \caption[Further experimental results of hypergraph partitioning on synthetic data]{The average performance of each algorithm, when $n = 2,000$, $r = 5$, and $p = 10^{-11}$.
    \label{fig:synthetic_rank5}
    }
\end{figure}

To compare the algorithms' runtime, we fix the parameter $n = 2,000$, the ratio $q = 2p$, and report the runtime of the \diffalgapproxshortname\ and \algcc\ algorithms on a variety of hypergraphs in Table~\ref{tab:runtime}.
Our new algorithm takes more time than the baseline \algcc\ algorithm, but both algorithms appear to scale linearly in the size of the input hypergraph\footnote{Although $n$ is fixed, the \algcliquecut\ algorithm's runtime is not constant since the time 
to compute an eigenvalue of the sparse adjacency matrix scales with the number and rank of the hyperedges.}; this suggests that our algorithm's runtime is roughly a constant factor multiple of the baseline.

\begin{table} [hb]
  \caption[The runtime of the \diffalgapproxshortname\ and \algcc\ algorithms]{The runtime in seconds of the \diffalgapproxshortname\ and \algcc\ algorithms}
  \label{tab:runtime}
  \centering
  \begin{tabular}{lllll}
    \toprule
    & & & \multicolumn{2}{c}{Avg.\ Runtime} \\
    \cmidrule(r){4-5}
    $r$ & $p$ & Avg.\ $\cardinality{\edgeset_\graphh}$ & \diffalgapproxshortname\ & \algcc \\
    \midrule
    & $10^{-9}$ & $1239$ & $1.15$ & $0.12$ \\
    $4$ & $10^{-8}$ & $12479$ & $10.14$ & $0.86$ \\
    & $10^{-7}$ & $124717$ & $89.92$ & $9.08$ \\
    \midrule
    & $10^{-11}$ & $5177$ & $3.99$ & $0.62$ \\
    $5$ & $10^{-10}$ & $51385$ & $44.10$ & $6.50$ \\
    & $10^{-9}$ & $514375$ & $368.48$ & $69.25$ \\
    \bottomrule
  \end{tabular}
\end{table}

\subsection{Real-world Datasets} \label{sec:real_world_experiments}
Next, we demonstrate the broad utility of our new algorithm on complex real-world datasets with higher-order relationships which are most naturally represented by hypergraphs.
Moreover, the hypergraphs are \firstdef{inhomogeneous}, meaning that they contain vertices of different types, although this information is not available to the algorithm and so an algorithm has to treat every vertex identically.
We see that Algorithm~\ref{algo:main} is able to find clusters which correspond to the vertices of different types.
Table~\ref{tab:real_world} shows the $F_1$-score of the clustering produced by the algorithm on each dataset, and demonstrates that it consistently outperforms the \algcliquecut\ algorithm. 

\begin{table}[t]
  \caption{The performance of the \diffalgapproxshortname\ and \algcc\ algorithms on real-world datasets} 
  \label{tab:real_world}
  \centering
  \begin{tabular}{llll}
    \toprule
    & & \multicolumn{2}{c}{$F_1$-Score} \\
    \cmidrule(r){3-4}
    Dataset & Cluster & \diffalgapproxshortname\ & \algcc \\
    \midrule
    \multirow{2}{1.6cm}{Penn Treebank} & Verbs & $\mathbf{0.73}$ & $0.69$ \\
      & Non-Verbs & $\mathbf{0.59}$ & $0.56$ \\
      \midrule
    \multirow{2}{*}{DBLP} & Conferences & $\mathbf{1.00}$ & $0.25$ \\
      & Authors & $\mathbf{1.00}$ & $0.98$ \\
    \bottomrule
  \end{tabular}
\end{table}

\paragraph[Penn Treebank.]{Penn Treebank.}
The Penn Treebank dataset is an English-language corpus with examples of written American English from several sources, including fiction and journalism~\cite{marcusBuildingLargeAnnotated1993}.
The dataset contains 49,208 sentences and over 1 million words, which are labelled with their part of speech.
We construct a hypergraph in the following way: the vertex set consists of all the  verbs, adverbs, and adjectives which occur at least 10 times in the corpus, and for every 4-gram (a sequence of 4 words) we add a hyperedge containing the co-occurring words.
This results in a hypergraph with 4,686 vertices and 176,286 edges. 
The clustering returned by the new algorithm correctly distinguishes between verbs and non-verbs with an accuracy of $67$\%.
This experiment demonstrates that our \emph{unsupervised} general purpose algorithm can recover non-trivial structure in a dataset which would ordinarily be clustered using domain knowledge and a complex pre-trained model~\cite{akbikContextualStringEmbeddings2018, goldwaterFullyBayesianApproach2007}.

\paragraph{DBLP.}
We construct a hypergraph from a subset of the DBLP network consisting of 14,376 papers published in artificial intelligence and machine learning conferences~\cite{DBLPComputerScience, wangHeterogeneousGraphAttention2019}.
For each paper, we include a hyperedge linking the authors of the paper with the conference in which it is published, giving a hypergraph with 14,495 vertices and 14,376 edges.
The clusters returned by our algorithm successfully separate the authors from the conferences with an accuracy of $100$\%.

\section{Future Work} 
In this chapter, we introduce a new Laplacian-type operator which generalises the signless Laplacian operator from graphs to hypergraphs.
We show that the resulting algorithm is effective for recovering the structure of complex datasets, and this further justifies the recent effort to develop a  spectral theory of hypergraphs.
The applicability of spectral hypergraph theory for clustering remains an important area for future research.
For example, as demonstrated in this chapter, the Laplacian operator of Chan \etal~\cite{chanSpectralPropertiesHypergraph2018} can have more than two eigenvectors.
It would be interesting to investigate whether these other eigenvectors can be applied for learning clusters in hypergraphs.
\begin{openquestion}
Can the eigenvectors of the non-linear hypergraph Laplacian be used to develop a general algorithm for learning clusters in hypergraphs?
\end{openquestion}
Such an algorithm could avoid the downsides of using graph reductions for hypergraph clustering and make full use of the higher-order information encoded in the hypergraph edges.


\titleformat{\chapter}[display]{\fontsize{18pt}{21pt}\bfseries \sffamily \centering}{\chaptertitlename\ \thechapter}{10pt}{}{}

\bibliographystyle{alpha}

\singlespace

\bibliography{references}

\newcommand{\etalchar}[1]{$^{#1}$}
\begin{thebibliography}{AOGPT16}

\bibitem[Abb17]{abbeCommunityDetectionStochastic2017}
Emmanuel Abbe.
\newblock Community detection and stochastic block models: Recent developments.
\newblock {\em The Journal of Machine Learning Research}, 18:177:1--86, 2017.

\bibitem[ABC{\etalchar{+}}08]{andersenTrustbasedRecommendationSystems2008}
Reid Andersen, Christian Borgs, Jennifer Chayes, Uriel Feige, Abraham Flaxman,
  Adam Kalai, Vahab Mirrokni, and Moshe Tennenholtz.
\newblock Trust-based recommendation systems: An axiomatic approach.
\newblock In {\em 17th International Conference on {{World Wide Web}}
  ({{WWW}}'08)}, pages 199--208, 2008.

\bibitem[ABV18]{akbikContextualStringEmbeddings2018}
Alan Akbik, Duncan Blythe, and Roland Vollgraf.
\newblock Contextual string embeddings for sequence labeling.
\newblock In {\em 27th {{International Conference}} on {{Computational
  Linguistics}} ({{COLING}}'18)}, pages 1638--1649, 2018.

\bibitem[AC09]{andersenFindingDenseSubgraphs2009}
Reid Andersen and Kumar Chellapilla.
\newblock Finding dense subgraphs with size bounds.
\newblock In {\em 6th {{International Workshop}} on {{Algorithms}} and
  {{Models}} for the {{Web-Graph}} ({{WAW}}'09)}, pages 25--37, 2009.

\bibitem[ACL06]{andersenLocalGraphPartitioning2006}
Reid Andersen, Fan Chung, and Kevin Lang.
\newblock Local graph partitioning using {{Pagerank}} vectors.
\newblock In {\em 47th {{Annual IEEE Symposium}} on {{Foundations}} of
  {{Computer Science}} ({{FOCS}}'06)}, pages 475--486, 2006.

\bibitem[ACL07]{andersenLocalPartitioningDirected2007}
Reid Andersen, Fan Chung, and Kevin Lang.
\newblock Local partitioning for directed graphs using {{Pagerank}}.
\newblock In {\em 5th {{International Workshop}} on {{Algorithms}} and
  {{Models}} for the {{Web-Graph}} ({{WAW}}'07)}, pages 166--178, 2007.

\bibitem[AHN19]{arrigoNonbacktrackingPagerank2019}
Francesca Arrigo, Desmond~J. Higham, and Vanni Noferini.
\newblock Non-backtracking {{Pagerank}}.
\newblock {\em Journal of Scientific Computing}, 80(3):1419--1437, 2019.

\bibitem[ALM13]{zhuLocalAlgorithmFinding2013}
Zeyuan {Allen-Zhu}, Silvio Lattanzi, and Vahab Mirrokni.
\newblock A local algorithm for finding well-connected clusters.
\newblock In {\em 30th {{International Conference}} on {{Machine Learning}}
  ({{ICML}}'13)}, pages 396--404, 2013.

\bibitem[Alo86]{alonEigenvaluesExpanders1986}
Noga Alon.
\newblock Eigenvalues and expanders.
\newblock {\em Combinatorica}, 6(2):83--96, 1986.

\bibitem[AM07]{abbassiRecommenderSystemBased2007}
Zeinab Abbassi and Vahab~S. Mirrokni.
\newblock A recommender system based on local random walks and spectral
  methods.
\newblock In {\em 9th {{WebKDD}} and 1st {{SNA-KDD}} Workshop on {{Web}} Mining
  and Social Network Analysis}, pages 102--108, 2007.

\bibitem[AMFM11]{arbelaezContourDetectionHierarchical2011}
Pablo Arbelaez, Michael Maire, Charless Fowlkes, and Jitendra Malik.
\newblock Contour detection and hierarchical image segmentation.
\newblock {\em IEEE Transactions on Pattern Analysis and Machine Intelligence},
  33(5):898--916, 2011.

\bibitem[And10]{andersenLocalAlgorithmFinding2010}
Reid Andersen.
\newblock A local algorithm for finding dense subgraphs.
\newblock {\em ACM Transactions on Algorithms}, 6(4):1--12, 2010.

\bibitem[AOGPT16]{andersenAlmostOptimalLocal2016}
Reid Andersen, Shayan Oveis~Gharan, Yuval Peres, and Luca Trevisan.
\newblock Almost optimal local graph clustering using evolving sets.
\newblock {\em Journal of the ACM}, 63(2):1--31, 2016.

\bibitem[AP09]{andersenFindingSparseCuts2009}
Reid Andersen and Yuval Peres.
\newblock Finding sparse cuts locally using evolving sets.
\newblock In {\em 41st {{Annual ACM Symposium}} on {{Theory}} of {{Computing}}
  ({{STOC}}'09)}, pages 235--244, 2009.

\bibitem[BJ13]{bauerBipartiteNeighborhoodGraphs2013}
Frank Bauer and J{\"u}rgen Jost.
\newblock Bipartite and neighborhood graphs and the spectrum of the normalized
  graph {{Laplace}} operator.
\newblock {\em Communications in Analysis and Geometry}, 21(4):787--845, 2013.

\bibitem[CB]{UnitedStatesCensus2000}
U.~S. Census~Bureau.
\newblock United {{States}} 2000 {{Census}}.
\newblock
  \url{https://web.archive.org/web/20150905081016/https://www.census.gov/population/www/cen2000/commuting/}.

\bibitem[CGKR21]{cohen-addadOnlineKmeansClustering2021}
Vincent {Cohen-Addad}, Benjamin Guedj, Varun Kanade, and Guy Rom.
\newblock Online k-means clustering.
\newblock In {\em 24th {{International Conference}} on {{Artificial
  Intelligence}} and {{Statistics}} ({{AISTATS}}'21)}, pages 1126--1134, 2021.

\bibitem[Che70]{cheegerLowerBoundSmallest1970}
Jeff Cheeger.
\newblock A lower bound for the smallest eigenvalue of the {{Laplacian}}.
\newblock {\em Problems in Analysis}, 625(195-199):110, 1970.

\bibitem[Chu93]{chungLaplacianHypergraph1993}
Fan Chung.
\newblock The {{Laplacian}} of a hypergraph.
\newblock {\em AMS Series in Discrete Mathematics and Theoretical Computer
  Science}, pages 21--36, 1993.

\bibitem[Chu96]{chungLaplaciansGraphsCheeger1996}
Fan Chung.
\newblock Laplacians of graphs and {{Cheeger}}'s inequalities.
\newblock {\em Combinatorics, Paul Erd\H{o}s is Eighty}, 2(157-172):13--2,
  1996.

\bibitem[Chu97]{chungSpectralGraphTheory1997}
Fan R~K Chung.
\newblock {\em Spectral Graph Theory}.
\newblock {American Mathematical Society}, 1997.

\bibitem[Chu07]{chungHeatKernelPagerank2007}
Fan Chung.
\newblock The heat kernel as the {{Pagerank}} of a graph.
\newblock {\em Proceedings of the National Academy of Sciences},
  104(50):19735--19740, 2007.

\bibitem[Chu09]{chungLocalGraphPartitioning2009}
Fan Chung.
\newblock A local graph partitioning algorithm using heat kernel {{Pagerank}}.
\newblock {\em Internet Mathematics}, 6(3):315--330, 2009.

\bibitem[CKP{\etalchar{+}}17]{cohenAlmostlineartimeAlgorithmsMarkov2017}
Michael~B Cohen, Jonathan Kelner, John Peebles, Richard Peng, Anup~B Rao, Aaron
  Sidford, and Adrian Vladu.
\newblock Almost-linear-time algorithms for {{Markov}} chains and new spectral
  primitives for directed graphs.
\newblock In {\em 49th {{Annual ACM Symposium}} on {{Theory}} of {{Computing}}
  ({{STOC}}'17)}, pages 410--419, 2017.

\bibitem[CL20]{chanGeneralizingHypergraphLaplacian2020}
T-H~Hubert Chan and Zhibin Liang.
\newblock Generalizing the hypergraph {{Laplacian}} via a diffusion process
  with mediators.
\newblock {\em Theoretical Computer Science}, 806:416--428, 2020.

\bibitem[CLSZ20]{cucuringuHermitianMatricesClustering2020}
Mihai Cucuringu, Huan Li, He~Sun, and Luca Zanetti.
\newblock Hermitian matrices for clustering directed graphs: {{Insights}} and
  applications.
\newblock In {\em 23rd {{International Conference}} on {{Artificial
  Intelligence}} and {{Statistics}} ({{AISTATS}}'20)}, pages 983--992, 2020.

\bibitem[CLTZ18]{chanSpectralPropertiesHypergraph2018}
T-H.~Hubert Chan, Anand Louis, Zhihao~Gavin Tang, and Chenzi Zhang.
\newblock Spectral properties of hypergraph {{Laplacian}} and approximation
  algorithms.
\newblock {\em Journal of the ACM}, 65(3):1--48, 2018.

\bibitem[CLW18]{chienCommunityDetectionHypergraphs2018}
I.~Chien, Chung-Yi Lin, and I.-Hsiang Wang.
\newblock Community detection in hypergraphs: {{Optimal}} statistical limit and
  efficient algorithms.
\newblock In {\em 21st {{International Conference}} on {{Artificial
  Intelligence}} and {{Statistics}} ({{AISTATS}}'18)}, pages 871--879, 2018.

\bibitem[CPS15]{czumajTestingClusterStructure2015}
Artur Czumaj, Pan Peng, and Christian Sohler.
\newblock Testing cluster structure of graphs.
\newblock In {\em 47th {{Annual ACM Symposium}} on {{Theory}} of {{Computing}}
  ({{STOC}}'15)}, pages 723--732, 2015.

\bibitem[CR19]{chitraRandomWalksHypergraphs2019}
Uthsav Chitra and Benjamin~J Raphael.
\newblock Random walks on hypergraphs with edge-dependent vertex weights.
\newblock In {\em 36th {{International Conference}} on {{Machine Learning}}
  ({{ICML}}'19)}, pages 2002--2011, 2019.

\bibitem[CSWZ16]{chenCommunicationoptimalDistributedClustering2016}
Jiecao Chen, He~Sun, David Woodruff, and Qin Zhang.
\newblock Communication-optimal distributed clustering.
\newblock In {\em 30th {{Advances}} in {{Neural Information Processing
  Systems}} ({{NeurIPS}}'16)}, pages 3727--3735, 2016.

\bibitem[CTWZ19]{chanDiffusionOperatorSpectral2019}
T-H.~Hubert Chan, Zhihao~Gavin Tang, Xiaowei Wu, and Chenzi Zhang.
\newblock Diffusion operator and spectral analysis for directed hypergraph
  {{Laplacian}}.
\newblock {\em Theoretical Computer Science}, 784:46--64, 2019.

\bibitem[{DBL}]{DBLPComputerScience}
The {DBLP team}.
\newblock {{DBLP}} computer science bibliography.
\newblock \url{https://dblp.uni-trier.de/}.

\bibitem[DH72]{donathAlgorithmsPartitioningGraphs1972}
William~E. Donath and Alan~J. Hoffman.
\newblock Algorithms for partitioning of graphs and computer logic based on
  eigenvectors of connection matrices.
\newblock {\em IBM Technical Disclosure Bulletin}, 15(3):938--944, 1972.

\bibitem[Dod84]{dodziukDifferenceEquationsIsoperimetric1984}
Jozef Dodziuk.
\newblock Difference equations, isoperimetric inequality and transience of
  certain random walks.
\newblock {\em Transactions of the American Mathematical Society},
  284(2):787--794, 1984.

\bibitem[DPRS19]{deySpectralConcentrationGreedy2019}
Tamal~K. Dey, Pan Peng, Alfred Rossi, and Anastasios Sidiropoulos.
\newblock Spectral concentration and greedy $k$-clustering.
\newblock {\em Computational Geometry}, 76:19–32, 2019.

\bibitem[Fie73]{fiedlerAlgebraicConnectivityGraphs1973}
Miroslav Fiedler.
\newblock Algebraic connectivity of graphs.
\newblock {\em Czechoslovak Mathematical Journal}, 23(2):298--305, 1973.

\bibitem[FWY20]{fountoulakisPNormFlowDiffusion2020}
Kimon Fountoulakis, Di~Wang, and Shenghao Yang.
\newblock $p$-{{Norm}} flow diffusion for local graph clustering.
\newblock In {\em 37th International Conference on Machine Learning (ICML'20)},
  page 3222–3232, 2020.

\bibitem[GA17]{gatesImpactRandomModels2017}
Alexander~J. Gates and Yong-Yeol Ahn.
\newblock The impact of random models on clustering similarity.
\newblock {\em Journal of Machine Learning Research}, 18(1):3049--3076, 2017.

\bibitem[GG07]{goldwaterFullyBayesianApproach2007}
Sharon Goldwater and Tom Griffiths.
\newblock A fully {{Bayesian}} approach to unsupervised part-of-speech tagging.
\newblock In {\em 45th {{Annual Meeting}} of the {{Association}} of
  {{Computational Linguistics}} ({{ACL}}'07)}, 2007.

\bibitem[GJ79]{gareyComputersIntractability1979}
Michael~R. Garey and David~S. Johnson.
\newblock {\em Computers and {{Intractability}}: {{A Guide}} to the {{Theory}}
  of {{NP-Completeness}}}.
\newblock {W. H. Freeman}, 1979.

\bibitem[GPP{\etalchar{+}}22]{goglevaKnowledgeGraphbasedRecommendation2022}
Anna Gogleva, Dimitris Polychronopoulos, Matthias Pfeifer, Vladimir Poroshin,
  Micha{\"e}l Ughetto, Matthew~J. Martin, Hannah Thorpe, Aurelie Bornot,
  Paul~D. Smith, and Ben Sidders.
\newblock Knowledge graph-based recommendation framework identifies drivers of
  resistance in {{EGFR}} mutant non-small cell lung cancer.
\newblock {\em Nature Communications}, 13(1):1--14, 2022.

\bibitem[HdK21]{highamEpidemicsHypergraphsSpectral}
Desmond~J. Higham and Henry-Louis de~Kergorlay.
\newblock Epidemics on hypergraphs: {{Spectral}} thresholds for extinction.
\newblock {\em Proceedings of the Royal Society A}, 477(2252), 2021.

\bibitem[HRC17]{huDeepGenerativeModels2017}
Changwei Hu, Piyush Rai, and Lawrence Carin.
\newblock Deep generative models for relational data with side information.
\newblock In {\em 34th {{International Conference}} on {{Machine Learning}}
  ({{ICML}}'17)}, pages 1578--1586, 2017.

\bibitem[HSJR13]{heinTotalVariationHypergraphslearning2013}
Matthias Hein, Simon Setzer, Leonardo Jost, and Syama~Sundar Rangapuram.
\newblock The total variation on hypergraphs - learning on hypergraphs
  revisited.
\newblock In {\em 27th {{Advances}} in {{Neural Information Processing
  Systems}} ({{NeurIPS}}'13)}, pages 2427--2435, 2013.

\bibitem[Hul94]{hullDatabaseHandwrittenText1994}
Jonathan~J. Hull.
\newblock A database for handwritten text recognition research.
\newblock {\em IEEE Transactions on Pattern Analysis and Machine Intelligence},
  16(5):550--554, 1994.

\bibitem[JM19]{jostHypergraphLaplaceOperators2019}
J{\"u}rgen Jost and Raffaella Mulas.
\newblock Hypergraph {{Laplace}} operators for chemical reaction networks.
\newblock {\em Advances in Mathematics}, 351:870--896, 2019.

\bibitem[KG14]{klosterHeatKernelBased2014}
Kyle Kloster and David~F. Gleich.
\newblock Heat kernel based community detection.
\newblock In {\em 20th {{ACM SIGKDD International Conference}} on {{Knowledge
  Discovery}} and {{Data Mining}} ({{KDD}}'14)}, pages 1386--1395, 2014.

\bibitem[KLL{\etalchar{+}}13]{kwokImprovedCheegerInequality2013a}
Tsz~Chiu Kwok, Lap~Chi Lau, Yin~Tat Lee, Shayan Oveis~Gharan, and Luca
  Trevisan.
\newblock Improved {{Cheeger}}'s inequality: Analysis of spectral partitioning
  algorithms through higher order spectral gap.
\newblock In {\em 45th {{Annual ACM Symposium}} on {{Theory}} of {{Computing}}
  ({{STOC}}'13)}, pages 11--20, 2013.

\bibitem[KM16]{kolevNoteSpectralClustering2016}
Pavel Kolev and Kurt Mehlhorn.
\newblock A note on spectral clustering.
\newblock In {\em 24th {{Annual European Symposium}} on {{Algorithms}}
  ({{ESA}}'16)}, pages 1--14, 2016.

\bibitem[KMN{\etalchar{+}}04]{kanungoLocalSearchApproximation2004}
Tapas Kanungo, David~M. Mount, Nathan~S. Netanyahu, Christine~D. Piatko, Ruth
  Silverman, and Angela~Y. Wu.
\newblock A local search approximation algorithm for $k$-means clustering.
\newblock {\em Computational Geometry}, 28(2-3):89–112, 2004.

\bibitem[KMS16]{kanadeGlobalLocalInformation2016}
Varun Kanade, Elchanan Mossel, and Tselil Schramm.
\newblock Global and local information in clustering labeled block models.
\newblock {\em IEEE Transactions on Information Theory}, 62(10):5906--5917,
  2016.

\bibitem[KSS04]{kumarSimpleLinearTime2004}
Amit Kumar, Yogish Sabharwal, and Sandeep Sen.
\newblock A simple linear time $(1+\epsilon)$-approximation algorithm for
  $k$-means clustering in any dimensions.
\newblock In {\em 45th Symposium on Foundations of Computer Science (FOCS'04)},
  page 454–462, 2004.

\bibitem[KUK17]{kloumannBlockModelsPersonalized2017}
Isabel~M. Kloumann, Johan Ugander, and Jon~M. Kleinberg.
\newblock Block models and personalized {{Pagerank}}.
\newblock {\em Proceedings of the National Academy of Sciences}, 114(1):33--38,
  2017.

\bibitem[LBBH98]{lecunGradientbasedLearningApplied1998}
Yann LeCun, L{\'e}on Bottou, Yoshua Bengio, and Patrick Haffner.
\newblock Gradient-based learning applied to document recognition.
\newblock {\em Proceedings of the IEEE}, 86(11):2278--2324, 1998.

\bibitem[LG20]{liuStronglyLocalPnormcut2020}
Meng Liu and David~F. Gleich.
\newblock Strongly local $p$-norm-cut algorithms for semi-supervised learning
  and local graph clustering.
\newblock In {\em 33rd Advances in Neural Information Processing Systems
  (NeurIPS'20)}, pages 5023--5035, 2020.

\bibitem[Liu15]{liuMultiwayDualCheeger2015}
Shiping Liu.
\newblock Multi-way dual {{Cheeger}} constants and spectral bounds of graphs.
\newblock {\em Advances in Mathematics}, 268:306--338, 2015.

\bibitem[LM17]{liInhomogeneousHypergraphClustering2017}
Pan Li and Olgica Milenkovic.
\newblock Inhomogeneous hypergraph clustering with applications.
\newblock In {\em 31st {{Advances}} in {{Neural Information Processing
  Systems}} ({{NeurIPS}}'17)}, pages 2308--2318, 2017.

\bibitem[LM18]{liSubmodularHypergraphsPLaplacians2018}
Pan Li and Olgica Milenkovic.
\newblock Submodular hypergraphs: $p$-{{Laplacians}}, {{Cheeger}} inequalities
  and spectral clustering.
\newblock In {\em 35th International Conference on Machine Learning (ICML'18)},
  pages 4690--4719, 2018.

\bibitem[LOGT14]{leeMultiwaySpectralPartitioning2014}
James~R Lee, Shayan Oveis~Gharan, and Luca Trevisan.
\newblock Multiway spectral partitioning and higher-order {{Cheeger}}
  inequalities.
\newblock {\em Journal of the ACM}, 61(6):1--30, 2014.

\bibitem[LP13]{liDetectingCharacterizingSmall2013}
Angsheng Li and Pan Peng.
\newblock Detecting and characterizing small dense bipartite-like subgraphs by
  the bipartiteness ratio measure.
\newblock In {\em 24th {{International Symposium}} on {{Algorithms}} and
  {{Computation}} ({{ISAAC}}'13)}, pages 655--665, 2013.

\bibitem[LS90]{lovaszMixingRateMarkov1990}
L{\'a}szl{\'o} Lov{\'a}sz and Mikl{\'o}s Simonovits.
\newblock The mixing rate of {{Markov}} chains, an isoperimetric inequality,
  and computing the volume.
\newblock In {\em 31st {{Annual IEEE Symposium}} on {{Foundations}} of
  {{Computer Science}} ({{FOCS}}'90)}, pages 346--354, 1990.

\bibitem[LS20]{laenenHigherOrderSpectralClustering2020}
Steinar Laenen and He~Sun.
\newblock Higher-order spectral clustering of directed graphs.
\newblock In {\em 34th {{Advances}} in {{Neural Information Processing
  Systems}} ({{NeurIPS}}'20)}, pages 941--951, 2020.

\bibitem[LV19]{louisPlantedModelsKway2019}
Anand Louis and Rakesh Venkat.
\newblock Planted models for $k$-way edge and vertex expansion.
\newblock In {\em 39th Annual Conference on Foundations of Software Technology
  and Theoretical Computer Science (FSTTCS'19)}, page 23:1–23:15, 2019.

\bibitem[Mac90]{macgregorStaticsDynamicsDeeply1990}
Peter~G. Macgregor.
\newblock {\em The Statics and Dynamics of Deeply Submerged Directly Tethered
  Subsea Units}.
\newblock PhD thesis, University of Strathclyde, 1990.

\bibitem[ME11]{menonLinkPredictionMatrix2011}
Aditya~Krishna Menon and Charles Elkan.
\newblock Link prediction via matrix factorization.
\newblock In {\em Machine {{Learning}} and {{Knowledge Discovery}} in
  {{Databases}}}, pages 437--452, 2011.

\bibitem[Miz21]{mizutaniImprovedAnalysisSpectral2021}
Tomohiko Mizutani.
\newblock Improved analysis of spectral algorithm for clustering.
\newblock {\em Optimization Letters}, 15(4):1303--1325, 2021.

\bibitem[MJK{\etalchar{+}}19]{maozDyadicMilitarizedInterstate2019}
Zeev Maoz, Paul~L. Johnson, Jasper Kaplan, Fiona Ogunkoya, and Aaron~P. Shreve.
\newblock The dyadic militarized interstate disputes ({{MIDs}}) dataset version
  3.0: {{Logic}}, characteristics, and comparisons to alternative datasets.
\newblock {\em Journal of Conflict Resolution}, 63(3):811--835, 2019.

\bibitem[MMR02]{mansfieldWhyDemocraciesCooperate2002}
Edward~D. Mansfield, Helen~V. Milner, and B.~Peter Rosendorff.
\newblock Why democracies cooperate more: {{Electoral}} control and
  international trade agreements.
\newblock {\em International Organization}, 56(3):477--513, 2002.

\bibitem[MMS93]{marcusBuildingLargeAnnotated1993}
Mitchell~P. Marcus, Mary~Ann Marcinkiewicz, and Beatrice Santorini.
\newblock Building a large annotated corpus of {{English}}: The {{Penn}}
  treebank.
\newblock {\em Computational Linguistics}, 19(2):313--330, 1993.

\bibitem[MMT08]{martinMakeTradeNot2008}
Philippe Martin, Thierry Mayer, and Mathias Thoenig.
\newblock Make trade not war?
\newblock {\em The Review of Economic Studies}, 75(3):865--900, 2008.

\bibitem[MS21a]{macgregorFindingBipartiteComponents2021}
Peter Macgregor and He~Sun.
\newblock Finding bipartite components in hypergraphs.
\newblock In {\em 34th {{Advances}} in {{Neural Information Processing
  Systems}} ({{NeurIPS}}'21)}, pages 7912--7923, 2021.

\bibitem[MS21b]{macgregorLocalAlgorithmsFinding2021}
Peter Macgregor and He~Sun.
\newblock Local algorithms for finding densely connected clusters.
\newblock In {\em 38th {{International Conference}} on {{Machine Learning}}
  ({{ICML}}'21)}, pages 7268--7278, 2021.

\bibitem[MS22]{macgregorTighterAnalysisSpectral2022}
Peter Macgregor and He~Sun.
\newblock A tighter analysis of spectral clustering, and beyond.
\newblock In {\em 39th {{International Conference}} on {{Machine Learning}}
  ({{ICML}}'22)}, pages 14717--14742, 2022.

\bibitem[MYZ{\etalchar{+}}11]{mooreActiveLearningNode2011}
Cristopher Moore, Xiaoran Yan, Yaojia Zhu, Jean-Baptiste Rouquier, and Terran
  Lane.
\newblock Active learning for node classification in assortative and
  disassortative networks.
\newblock In {\em 17th {{ACM SIGKDD International Conference}} on {{Knowledge
  Discovery}} and {{Data Mining}} ({{KDD}}'11)}, pages 841--849, 2011.

\bibitem[NJW01]{ngSpectralClusteringAnalysis2001}
Andrew~Y Ng, Michael~I Jordan, and Yair Weiss.
\newblock On spectral clustering: {{Analysis}} and an algorithm.
\newblock In {\em 15th {{Advances}} in {{Neural Information Processing
  Systems}} ({{NeurIPS}}'01)}, pages 849--856, 2001.

\bibitem[OA14]{orecchiaFlowbasedAlgorithmsLocal2014}
Lorenzo Orecchia and Zeyuan {Allen-Zhu}.
\newblock Flow-based algorithms for local graph clustering.
\newblock In {\em 25th {{ACM-SIAM Symposium}} on {{Discrete Algorithms}}
  ({{SODA}}'14)}, pages 1267--1286, 2014.

\bibitem[OGT14]{gharanPartitioningExpanders2014}
Shayan Oveis~Gharan and Luca Trevisan.
\newblock Partitioning into expanders.
\newblock In {\em 25th {{ACM-SIAM Symposium}} on {{Discrete Algorithms}}
  ({{SODA}}'14)}, pages 1256--1266, 2014.

\bibitem[PBMW99]{pagePageRankCitationRanking1999}
Lawrence Page, Sergey Brin, Rajeev Motwani, and Terry Winograd.
\newblock The {{Pagerank}} citation ranking: {{Bringing}} order to the web.
\newblock Technical report, {Stanford InfoLab}, 1999.

\bibitem[Pen20]{pengRobustClusteringOracle2020}
Pan Peng.
\newblock Robust clustering oracle and local reconstructor of cluster structure
  of graphs.
\newblock In {\em 31st {{Annual ACM-SIAM Symposium}} on {{Discrete Algorithms}}
  ({{SODA}}'20)}, pages 2953--2972, 2020.

\bibitem[PSZ17]{pengPartitioningWellClusteredGraphs2017}
Richard Peng, He~Sun, and Luca Zanetti.
\newblock Partitioning well-clustered graphs: {{Spectral}} clustering works!
\newblock {\em SIAM Journal on Computing}, 46(2):710--743, 2017.

\bibitem[PWC{\etalchar{+}}19]{peiGeomGCNGeometricGraph2019}
Hongbin Pei, Bingzhe Wei, Kevin Chen-Chuan Chang, Yu~Lei, and Bo~Yang.
\newblock Geom-{{GCN}}: {{Geometric}} graph convolutional networks.
\newblock In {\em 7th {{International Conference}} on {{Learning
  Representations}} ({{ICLR}}'19)}, 2019.

\bibitem[PY20]{pengAverageSensitivitySpectral2020}
Pan Peng and Yuichi Yoshida.
\newblock Average sensitivity of spectral clustering.
\newblock In {\em 26th {{ACM SIGKDD Conference}} on {{Knowledge Discovery}} and
  {{Data Mining}} ({{KDD}}'20)}, pages 1132--1140, 2020.

\bibitem[Ran71]{randObjectiveCriteriaEvaluation1971}
William~M. Rand.
\newblock Objective criteria for the evaluation of clustering methods.
\newblock {\em Journal of the American Statistical Association},
  66(336):846--850, 1971.

\bibitem[SM00]{shiNormalizedCutsImage2000}
Jianbo Shi and Jitendra Malik.
\newblock Normalized cuts and image segmentation.
\newblock {\em IEEE Transactions on Pattern Analysis and Machine Intelligence},
  22(8):888--905, 2000.

\bibitem[Sot15]{sotoImprovedAnalysisMaxCut2015}
Jos{\'e}~A Soto.
\newblock Improved analysis of a max-cut algorithm based on spectral
  partitioning.
\newblock {\em SIAM Journal on Discrete Mathematics}, 29(1):259--268, 2015.

\bibitem[ST13]{spielmanLocalClusteringAlgorithm2013}
Daniel~A. Spielman and Shang-Hua Teng.
\newblock A local clustering algorithm for massive graphs and its application
  to nearly linear time graph partitioning.
\newblock {\em SIAM Journal on Computing (SICOMP)}, 42(1):1--26, 2013.

\bibitem[SZ19]{sunDistributedGraphClustering2019}
He~Sun and Luca Zanetti.
\newblock Distributed graph clustering and sparsification.
\newblock {\em ACM Transactions on Parallel Computing}, 6(3):17:1--17:23, 2019.

\bibitem[TB09]{traagCommunityDetectionNetworks2009}
V.~A. Traag and Jeroen Bruggeman.
\newblock Community detection in networks with positive and negative links.
\newblock {\em Physical Review E}, 80(3):036115, 2009.

\bibitem[TMIY20]{takaiHypergraphClusteringBased2020}
Yuuki Takai, Atsushi Miyauchi, Masahiro Ikeda, and Yuichi Yoshida.
\newblock Hypergraph clustering based on {{Pagerank}}.
\newblock In {\em 26th {{ACM SIGKDD International Conference}} on {{Knowledge
  Discovery}} and {{Data Mining}} ({{KDD}}'20)}, pages 1970--1978, 2020.

\bibitem[Tre12]{trevisanMaxCutSmallest2012}
Luca Trevisan.
\newblock Max cut and the smallest eigenvalue.
\newblock {\em SIAM Journal on Computing}, 41(6):1769--1786, 2012.

\bibitem[TWC10]{tungEnablingScalableSpectral2010}
Frederick Tung, Alexander Wong, and David~A. Clausi.
\newblock Enabling scalable spectral clustering for image segmentation.
\newblock {\em Pattern Recognition}, 43(12):4069--4076, 2010.

\bibitem[vL07]{vonluxburgTutorialSpectralClustering2007}
Ulrike von Luxburg.
\newblock A tutorial on spectral clustering.
\newblock {\em Statistics and Computing}, 17(4):395--416, 2007.

\bibitem[vR04]{vanrijsbergenGeometryInformationRetrieval2004}
Cornelis~Joost van Rijsbergen.
\newblock {\em The Geometry of Information Retrieval}.
\newblock {Cambridge University Press}, 2004.

\bibitem[WFH{\etalchar{+}}17]{wangCapacityReleasingDiffusion2017}
Di~Wang, Kimon Fountoulakis, Monika Henzinger, Michael~W. Mahoney, and Satish
  Rao.
\newblock Capacity releasing diffusion for speed and locality.
\newblock In {\em 34th {{International Conference}} on {{Machine Learning}}
  ({{ICML}}'17)}, pages 3598--3607, 2017.

\bibitem[WJS{\etalchar{+}}19]{wangHeterogeneousGraphAttention2019}
Xiao Wang, Houye Ji, Chuan Shi, Bai Wang, Yanfang Ye, Peng Cui, and Philip~S.
  Yu.
\newblock Heterogeneous graph attention network.
\newblock In {\em 28th {{World Wide Web Conference}} ({{WWW}}'19)}, pages
  2022--2032, 2019.

\bibitem[YBLG17]{yinLocalHigherorderGraph2017}
Hao Yin, Austin~R Benson, Jure Leskovec, and David~F Gleich.
\newblock Local higher-order graph clustering.
\newblock In {\em 23rd {{ACM SIGKDD International Conference}} on {{Knowledge
  Discovery}} and {{Data Mining}} ({{KDD}}'17)}, pages 555--564, 2017.

\bibitem[Yos19]{yoshidaCheegerInequalitiesSubmodular2019}
Yuichi Yoshida.
\newblock Cheeger inequalities for submodular transformations.
\newblock In {\em 30th {{Annual ACM-SIAM Symposium}} on {{Discrete Algorithms}}
  ({{SODA}}'19)}, pages 2582--2601, 2019.

\bibitem[ZHS06]{zhouLearningHypergraphsClustering2006}
Dengyong Zhou, Jiayuan Huang, and Bernhard Sch{\"o}lkopf.
\newblock Learning with hypergraphs: {{Clustering}}, classification, and
  embedding.
\newblock In {\em 20th {{Advances}} in {{Neural Information Processing
  Systems}} ({{NeurIPS}}'06)}, pages 1601--1608, 2006.

\bibitem[ZYZ{\etalchar{+}}20]{zhuHomophilyGraphNeural2020}
Jiong Zhu, Yujun Yan, Lingxiao Zhao, Mark Heimann, Leman Akoglu, and Danai
  Koutra.
\newblock Beyond homophily in graph neural networks: {{Current}} limitations
  and effective designs.
\newblock In {\em 33rd {{Advances}} in {{Neural Information Processing
  Systems}} ({{NeurIPS}}'20)}, pages 7793--7804, 2020.

\bibitem[ZZL{\etalchar{+}}21]{zhangImageSegmentationBased2021}
Chongyang Zhang, Guofeng Zhu, Bobo Lian, Minxin Chen, Hong Chen, and Chenjian
  Wu.
\newblock Image segmentation based on multiscale fast spectral clustering.
\newblock {\em Multimedia Tools and Applications}, 80(16):24969--24994, 2021.

\end{thebibliography}

\appendix

\end{document}